\documentclass[11pt,a4paper]{amsart}
\usepackage[dvips,colorlinks=true,linkcolor=blue,urlcolor=red]{hyperref}
\usepackage[english]{babel}
\usepackage{graphicx,subfigure,float,wrapfig} 
\restylefloat{figure}
\usepackage{stmaryrd}
\usepackage{amssymb,amsfonts}
\numberwithin{equation}{section}

\newcommand{\car}{\mathbf{1}}
\newcommand{\R}{\mathbb{R}}

\newcommand{\N}{\mathbb{N}}

\newcommand{\Z}{\mathbb{Z}}

\newcommand{\C}{\mathbb{C}}

\newcommand{\vers}{\operatornamewithlimits{\to}}
\newcommand{\equ}{\operatornamewithlimits{\sim}}
\newcommand{\D}{\displaystyle}

\newcommand{\Sc}{{\mathcal S}}
\newcommand{\Coi}{{\mathcal C}_0^{\infty}}
\newcommand{\esp}{\mathbb{E}}
\newcommand{\pro}{\mathbb{P}}
\newcommand{\tr}{\text{tr}}

\def\im{\operatorname{Im}}

\theoremstyle{plain}
\newtheorem{Th}{Theorem}[section]
\newtheorem{Le}{Lemma}[section]
\newtheorem{Pro}{Proposition}[section]
\newtheorem{Cor}{Corollary}[section]
\theoremstyle{definition}
\newtheorem{Rem}{Remark}[section]
\newtheorem{Def}{Definition}[section]
\title{Resonances for large one-dimensional ``ergodic'' systems}
\author{Fr{\'e}d{\'e}ric Klopp} 
\address[Fr{\'e}d{\'e}ric Klopp]{ \vskip.1cm Sorbonne Universit{\'e}s,
  UPMC Univ. Paris 06, UMR 7586, IMJ-PRG, F-75005, Paris, France
  \vskip.1cm Univ. Paris Diderot, Sorbonne Paris Cit{\'e}, UMR 7586,
  IMJ-PRG, F-75205 Paris, France
  \vskip.1cm CNRS, UMR 7586, IMJ-PRG, F-75005, Paris, France}
\email{\href{mailto:frederic.klopp@imj-prg.fr}{frederic.klopp@imj-prg.fr}}
\thanks{It is a pleasure to thank N. Filonov for interesting
  discussions at the early stages of this work and, T.T. Phong,
  C. Shirley and M. Vogel for pointing out misprints in previous
  versions of the article. \\ This work was partially supported by the
  grant ANR-08-BLAN-0261-01.}
\begin{document}
\dedicatory{Dedicated to Johannes Sj{\"o}strand on the occasion of his
  seventieth birthday.}
\begin{abstract}
  The present paper is devoted to the study of resonances for
  one-dimensional quantum systems with a potential that is the
  restriction to some large box of an ergodic potential. For discrete
  models both on a half-line and on the whole line, we study the
  distributions of the resonances in the limit when the size of the
  box where the potential does not vanish goes to infinity. For
  periodic and random potentials, we analyze how the spectral theory
  of the limit operator influences the distribution of the resonances.
  \vskip.5cm
  \par\noindent \textsc{R{\'e}sum{\'e}.}
  Dans cet article, nous {\'e}tudions les r{\'e}sonances d'un syst{\`e}me
  unidimensionnel plong{\'e} dans un potentiel qui est la restriction {\`a} un
  grand intervalle d'un potentiel ergodique. Pour des mod{\`e}les discrets
  sur la droite et la demie droite, nous {\'e}tudions la distribution des
  r{\'e}sonances dans la limite de la taille de bo{\^\i}te infinie. Pour des
  potentiels p{\'e}riodiques et al{\'e}atoires, nous analysons l'influence de
  la th{\'e}orie spectrale de l'op{\'e}ra\-teur limite sur la distribution des
  r{\'e}sonances.
\end{abstract}
\setcounter{section}{-1}
\maketitle
\section{Introduction}
\label{sec:introduction}
Consider $V:\ \Z\to\R$ a bounded potential and, on $\ell^2(\Z)$, the
Schr{\"o}dinger operator $H=-\Delta+V$ defined by
\begin{equation*}
  (H u)(n)=u(n+1)+u(n-1)+V(n)u(n), \quad\forall n\in\Z,
\end{equation*}
for $u\in\ell^2(\Z)$.\\
The potentials $V$ we will deal with are of two types:
\begin{itemize}
\item $V$ periodic;
\item $V=V_\omega$, the random Anderson model, i.e., the entries of
  the diagonal matrix $V$ are independent identically distributed non
  constant random variable.
\end{itemize}
The spectral theory of such models has been studied extensively (see,
e.g.,~\cite{MR2509110}) and it is well known that
\begin{itemize}
\item when $V$ is periodic, the spectrum of $H$ is purely absolutely
  continuous;
\item when $V=V_\omega$ is random, the spectrum of $H$ is almost
  surely pure point, i.e., the operator only has eigenvalues; moreover,
  the eigenfunctions decay exponentially at infinity.
\end{itemize}
Pick $L\in\N^*$. The main object of our study is the operator
\begin{equation}
  \label{eq:2}
  H_L=-\Delta+V\car_{\llbracket-L+1,L\rrbracket}  
\end{equation}
when $L$ is large. Here, $\llbracket-L+1,L\rrbracket$ is the integer
interval $\{-L+1,\cdots,L\}$ and $\car_{\llbracket
  a,b\rrbracket}(n)=1$ if $a\leq n\leq b$ and $0$ if not.  \\
For $L$ large, the operator $H_L$ is a simple Hamiltonian modeling a
large sample of periodic or random material in the void. It is well
known in this case (see, e.g.,~\cite{MR1957536}) that not only does
the spectrum of $H_L$ be of importance but also its (quantum)
resonances that we will now define.\\
As $V\car_{\llbracket-L+1,L\rrbracket}$ has finite rank, the essential
spectrum of $H_L$ is the same as that of the discrete Laplace
operator, that is, $[-2,2]$, and it is purely absolutely
continuous. Outside this absolutely continuous spectrum, $H_L$ has
only discrete eigenvalues associated to exponentially
decaying eigenfunctions.\\
We are interested in the resonances of the operator $H_L$ in the limit
when $L\to+\infty$. They are defined to be the poles of the
meromorphic continuation of the resolvent of $H_L$ through $(-2,2)$,
the continuous spectrum of $H_L$ (see Theorem~\ref{thr:1} and,
e.g.,~\cite{MR1957536}). The resonances widths, that is, their
imaginary part, play an important role in the large time behavior of
$\D e^{-itH_L}$, especially the resonances of smallest width that give
the leading order contribution (see~\cite{MR1957536}).
\begin{figure}[h]
  \centering
  \includegraphics[height=1.7cm]{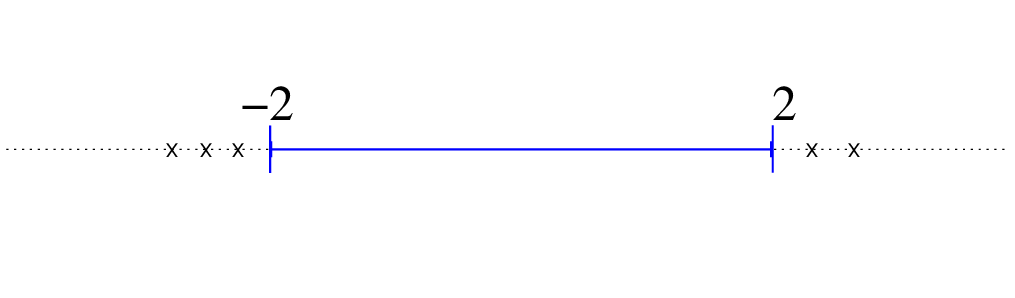}
  \hskip.1cm
  \includegraphics[height=1.6cm]{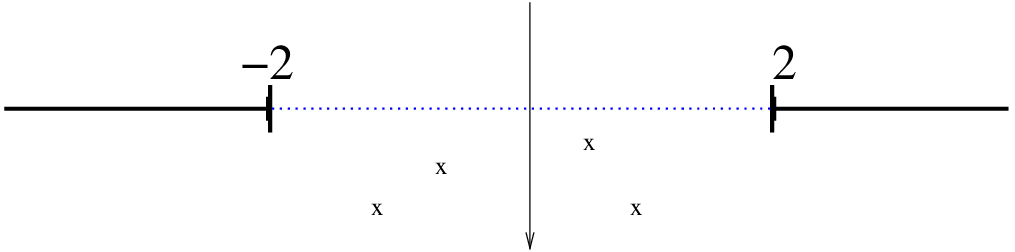}
  \caption{The meromorphic continuation}
  \label{fig:1}
\end{figure}
\\\noindent Quantum resonances are basic objects in quantum
theory. They have been the focus of vast number of studies both
mathematical and physical (see, e.g.,~\cite{MR1957536} and references
therein). Our purpose here is to study the resonances of $H_L$ in the
asymptotic regime $L\to+\infty$. As $L\to+\infty$, $H_L$ converges to
$H$ in the strong resolvent sense. Thus, it is natural to expect that
the differences in the spectral nature between the cases $V$ periodic
and $V$ random should reflect into differences in the behavior of the
resonances in both cases. We shall see below that this is the case. To
illustrate this as simply as possible, we begin with stating three
theorems, one for periodic potentials, two for random potentials, that
underline these different behaviors. These results can be considered
as paradigmatic for our main results presented in
section~\ref{sec:main-results}.\\
The scattering theory or the closely related questions of resonances
for the operator~\eqref{eq:2} or for closely related one-dimensional
models has already been discussed in various works both in the
mathematical and physical literature (see,
e.g.,~\cite{MR1265236,MR1025232,citeulike:6565491,MR2252707,1999PhRvL..82.4220T,MR1618595,PhysRevB.77.054203,MR1687454,MR2197681,PhysRevB.61.R2444}).
We will make more comments on the literature as we will develop our
results in section~\ref{sec:main-results}.
\subsection{When $V$ is periodic}
\label{sec:when-v-periodic}
Assume that $V$ is $p$-periodic ($p\in\N^*$) and does not vanish
identically. Consider $H=-\Delta+V$ and let $\Sigma_\Z$ be its
spectrum, $\overset{\circ}{\Sigma}_\Z$ be its interior and $E\mapsto
N(E)$ be its integrated density of states, i.e., the number of states
of the system per unit of volume below energy $E$ (see
section~\ref{sec:periodic-case} and, e.g.,~\cite{MR1711536} for precise
definitions and details).
\begin{Th}
  \label{thr:22}
  There exist
  \begin{itemize}
  \item $\mathcal{D}$, a discrete (possibly empty) set of energies in
    $(-2,2)\cap\overset{\circ}{\Sigma}_\Z$,
  \item a function $h$ that is real analytic in a complex neighborhood
    of $(-2,2)$ and that does vanish on $(-2,2)\setminus\mathcal{D}$
  \end{itemize}
  such that, for $I\subset (-2,2)\setminus\mathcal{D}$, a compact interval
  such that either $I\cap \Sigma_\Z=\emptyset$ or
  $I\subset\overset{\circ}{\Sigma}_\Z$, there exists $c_0>0$ such that
  for $L$ sufficiently large s.t. $2L\in p\N$, one has
  \begin{itemize}
  \item if $I\cap \Sigma_\Z=\emptyset$, then $H_L$ has no resonance in
    $I+i[-c_0,0]$
  \item if $I\subset\overset{\circ}{\Sigma}_\Z$, one has
    \begin{itemize}
    \item there are plenty of resonances in $I+i[-c_0,0]$ ; more
      precisely,
      \begin{equation}
        \label{eq:7}
        \frac{\#\{z\in I+i[-c_0,0],\ z\text{ resonance of
          }H_L\}}{2L}=\int_I dN(E)+o(1)
      \end{equation}
      where $o(1)\to0$ as $L\to+\infty$;
    \item let $(z_j)_j$ the resonances of $H_L$ in $I+i[-c_0,0]$
      ordered by increasing real part; then,
      \begin{equation}
        \label{eq:6}
        L\cdot\text{Re}\,(z_{j+1}-z_j)\asymp 1\quad\text{ and }\quad
        L\cdot\text{Im}\,z_j=h(\text{Re}\,z_j)+o(1),
      \end{equation}
      the estimates in~\eqref{eq:6} being uniform for all the
      resonances in $I+i[-c_0,0]$ when $L\to+\infty$.
    \end{itemize}
  \end{itemize}
\end{Th}
\noindent After rescaling their width by $L$, resonances are nicely
inter-spaced points lying on an analytic curve (see
Fig.~\ref{fig:2-3}). We give a more precise description of the
resonances in Theorem~\ref{thr:2} and Propositions~\ref{pro:2}
and~\ref{pro:3}. In particular, we describe the set of energies
$\mathcal{D}$ and the resonances near these energies: they lie further
away from the real axis, the maximal distance being of order
$L^{-1}\log L$ (see Fig.~\ref{fig:10}). Theorem~\ref{thr:22} only
describes the resonances closest to the real axis. In
section~\ref{sec:periodic-case}, we also give results on the
resonances located deeper into the lower half of the complex plane.
\begin{figure}[h]
  \centering
  \includegraphics[width=.42\textwidth]{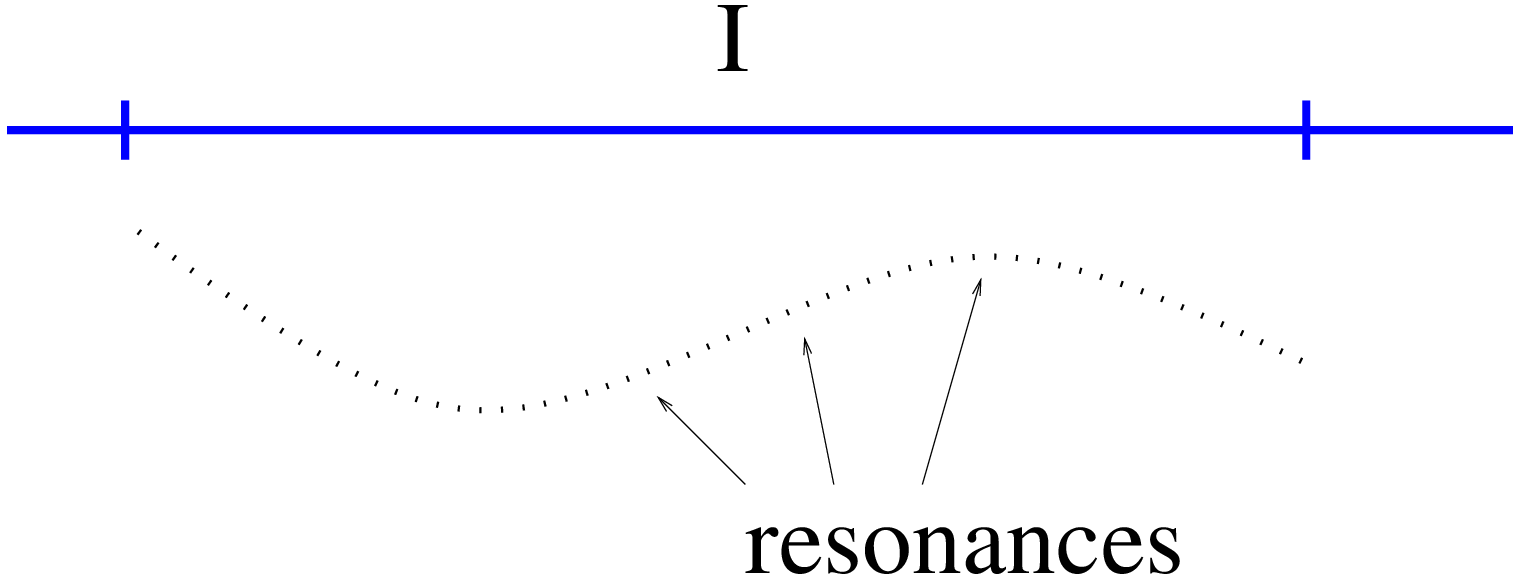}
  \hskip.5cm
  \includegraphics[width=.42\textwidth]{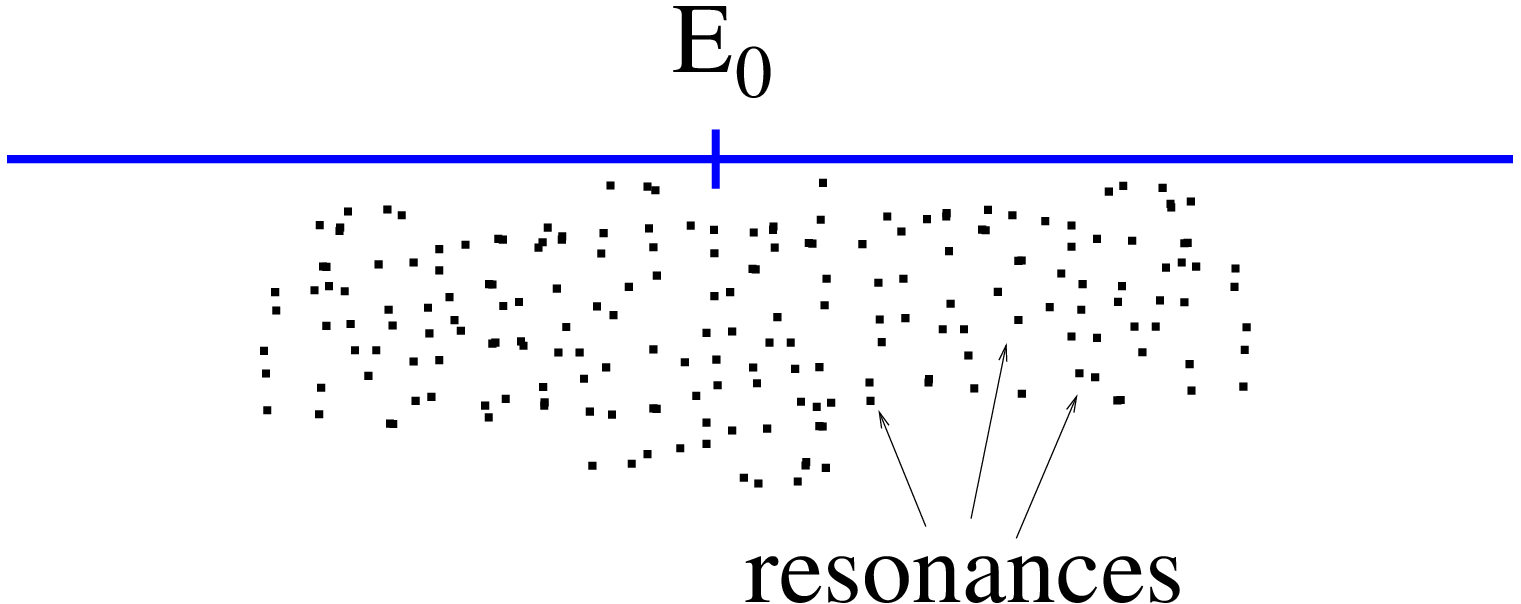}
  \caption{The rescaled resonances for the periodic (left part) and
    the random (right part) potential}
  \label{fig:2-3}
\end{figure}
\subsection{When $V$ is random}
\label{sec:when-v-random}
Assume now that $V=V_\omega$ is the Anderson potential, i.e., its
entries are i.i.d. distributed uniformly on $[0,1]$ to fix
ideas. Consider $H=-\Delta+V_\omega$. Let $\Sigma$ be its almost sure
spectrum (see, e.g.,~\cite{MR94h:47068}), $E\mapsto n(E)$, its density
of states (i.e. the derivative of the integrated density of states,
see section~\ref{sec:periodic-case} and, e.g.,~\cite{MR94h:47068}) and
$E\mapsto \rho(E)$, its Lyapunov exponent (see
section~\ref{sec:random-case} and, e.g.,~\cite{MR94h:47068}). The
Lyapunov exponent is known to be continuous and positive (see,
e.g.,~\cite{MR88f:60013}); the density of states satisfies $n(E)>0$
for
a.e. $E\in\Sigma$ (see, e.g.,~\cite{MR88f:60013}).\\
Define
$H_{\omega,L}:=-\Delta+V_\omega\car_{\llbracket-L+1,L\rrbracket}$. We
prove
\begin{Th}
  \label{thr:29}
  Pick $I\subset(-2,2)$, a compact interval. Then,
  \begin{itemize}
  \item if $I\cap\Sigma=\emptyset$ then, there exists $c_I>0$ such
    that, $\omega$-a.s., for $L$ sufficiently large,
    \begin{equation*}
      \left\{z\text{ resonance of }H_{\omega,L}
        \text{ in }I+i\left(-c_I,0\right]
      \right\}=\emptyset;
    \end{equation*}
  \item if $I\subset\overset{\circ}{\Sigma}$ then, for any $c>0$,
    $\omega$-a.s., one has
    \begin{equation*}
      \lim_{L\to+\infty}
      \frac1L\#\left\{z\text{ resonance of }H_{\omega,L}
        \text{ in }I+i\left(-\infty,-e^{-cL}\right]
      \right\}=\int_I\min\left(\frac{c}{\rho(E)},1\right)n(E)dE.
    \end{equation*}
  \end{itemize}
 
\end{Th}
\noindent As the first statement of Theorem~\ref{thr:29} is clear, let
us discuss the second. Define $\D c_+:=\max_{E\in I}\rho(E)$. For $c\geq
c_+$, $\omega$-a.s., for $L$ large, the number of resonances in the
strip $\{\text{Re}\,z\in I,\ \text{Im}\,z\leq -e^{-cL}\}$ is approximately
$\D L\int_In(E)dE$; thus, in $\{\text{Re}\,z\in I,\
-e^{c_+L}\leq\text{Im}\,z<0\}$, one finds at most $o(L)$
resonances. We shall see that, for $\delta>0$, $\omega$-a.s., for $L$
large, the strip $\{\text{Re}\,z\in I,\
-e^{(c_++\delta)L}\leq\text{Im}\,z<0\}$ actually
contains no resonance (see Theorem~\ref{thr:4}).\\
Define $\D c_-:=\min_{E\in I}\rho(E)$. For $c\leq c_-$, $\omega$-a.s., for
$L$ large, the strip $\{\text{Re}\,z\in I,\ \text{Im}\,z\leq -e^{-cL}\}$
contains approximately $\D c\,L\int_I\frac{n(E)}{\rho(E)}dE$
resonances. We shall see that, for $\kappa\in[0,1)$, the number of
resonances in the strip $\{\text{Re}\,z\in I,\
\text{Im}\,z\leq -e^{-L^{\kappa}}\}$ is
$O(L^\kappa)$, thus, $o(L)$ (cf. Theorem~\ref{thr:27}).\\
One can also describe the resonances locally. Therefore, fix $E_0\in
(-2,2)\cap\overset{\circ}{\Sigma}$ such that $n(E_0)>0$.  Let
$(z_l^L(\omega))_l$ be the resonances of $H_{\omega,L}$. We first
rescale them. Define
\begin{equation}
  \label{eq:148}
  x_l^L(\omega) =2n(E_0)\,L(\text{Re}\,z_l^L(\omega)-E_0)
  \quad\text{and} \quad
  y_l^L(\omega)=-\frac1{L\rho(E_0)}\log|\text{Im}\,z_l^L(\omega)|.
\end{equation}
Consider now the two-dimensional point process
\begin{equation*}
  \xi_L(E_0,\omega)=\sum_{z_l^L \text{ resonances of }H_{\omega,L}}
  \delta_{(x_l^L(\omega),y_l^L(\omega))}.  
\end{equation*}
We prove
\begin{Th}
  \label{thr:23}
  The point process $\xi_L$ converges weakly to a Poisson process of
  intensity $1$ in $\R\times[0,1]$.
\end{Th}
\noindent In the random case, the structure of the (properly rescaled)
resonances is quite different from that in the periodic case (see
Fig.~\ref{fig:2-3}). The real parts of the resonances are scaled in
such a way that that their average spacing becomes of order one. By
Theorem~\ref{thr:29}, the imaginary parts are typically exponentially
small (in $L$); when the resonances are rescaled as in~\eqref{eq:148},
their imaginary parts are rewritten on a logarithmic scale so as to
become of order $1$ too. Once rescaled in this way, the local picture
of the resonances of $H_{\omega,L}$ is that of a two-dimensional cloud
of Poisson points (see the right hand side of
fig.~\ref{fig:2-3}).\\
Theorem~\ref{thr:23} is the analogue for resonances of the well known
result on the distribution of eigenvalues and localization centers for
the Anderson model in the localized phase
(see, e.g.,~\cite{MR97d:82046,MR2299191,Ge-Kl:10}).\\
As in the case of the periodic potential, Theorem~\ref{thr:23} only
describes the resonances closest to the real axis. In
section~\ref{sec:random-case}, we also give results on resonances
located deeper into the lower half of the complex plane. Up to
distances of order $L^{-\infty}$ to the real axis, the cloud of
resonances (once properly rescaled) will have the same Poissonian
behavior as described above (see Theorem~\ref{thr:3}).
\vskip.1cm\noindent Besides proving Theorems~\ref{thr:22}
and~\ref{thr:23}, the goal of the paper is to describe the statistical
properties of the resonances and relate them (the distribution of the
resonances, the distribution of the widths) to the spectral
characteristics of $H=-\Delta+V$, possibly to the distribution of its
eigenvalues (see, e.g.,~\cite{MR2885251}).
\vskip.2cm\noindent As they can be analyzed in a very similar way, we
will discuss three models:
\begin{itemize}
\item the model $H_L$ defined above,
\item its analogue on the half-line $\N$, i.e., on $H_L$, we impose an
  additional Dirichlet boundary condition at $0$,
\item the ``half-infinite'' model on $\ell^2(\Z)$, that is,
  \begin{equation}
    \label{eq:128}
    H^\infty=-\Delta+W\text{ where }
    \begin{cases}
      W(n)=0  \text{ for }n\geq 0\\
      W(n)=V(n) \text{ for }n\leq -1
    \end{cases}
  \end{equation}
  where $V$ is chosen as above, periodic or random.
\end{itemize}
Though in the present paper we restrict ourselves to discrete models,
it is clear that continuous one-dimensional models can be dealt with
essentially using the methods developed in the present paper.
\section{The main results}
\label{sec:main-results}
We now turn to our main results, a number of which were announced
in~\cite{Kl:11d}. Pick $V:\Z\to\R$ a bounded potential and, for
$L\in\N$, consider the following operators:
\begin{itemize}
\item $H^\Z_L=-\Delta+V\car_{\llbracket 0,L\rrbracket}$ on
  $\ell^2(\Z)$;
\item $H^{\N}_L=-\Delta+V\car_{\llbracket0,L\rrbracket}$ on
  $\ell^2(\N)$ with Dirichlet boundary conditions at $0$;
\item $H^\infty$ defined in~\eqref{eq:128}.
\end{itemize}
\begin{Rem}
  \label{rem:3}
  Here, with ``Dirichlet boundary condition at $0$'', we mean that
  $H^{\N}_L$ is the operator $H^\Z_L$ restricted to the subspace
  $\ell^2(\N)$, i.e., if $\Pi: \ell^2(\Z)\to \ell^2(\N)$ is the
  orthogonal projector on $\ell^2(\N)$, one has $H^{\N}_L=\Pi H^{\Z}_L
  \Pi$. In the literature, this is sometime called ``Dirichlet
  boundary condition at $-1$'' (see, e.g.,~\cite{MR1711536}).\\
  For the sake of simplicity, in the half line case, we only consider
  Dirichlet boundary conditions at $0$. But the proofs show that these
  are not crucial; any self-adjoint boundary condition at $0$ would do
  and, mutandi mutandis, the results would be the same.\\
  Note also that, by a shift of the potential $V$, replacing $L$ by
  $L+L'$, we see that studying $H^\Z_L$ is equivalent to studying
  $H_{L,L'}=-\Delta+V\car_{\llbracket -L',L\rrbracket}$ on
  $\ell^2(\Z)$. Thus, to derive the results of
  section~\ref{sec:introduction} from those in the present section, it
  suffices to consider the models above, in particular, $H_L^\Z$.
\end{Rem}
For the models $H^{\N}_L$ and $H^\Z_L$, we start with a discussion of
the existence of a meromorphic continuation of the resolvent, then,
study the resonances when $V$ is periodic and finally turn to the case
when $V$ is random.\\
As $H^\infty$ is not a relatively compact perturbation of the
Laplacian, the existence of a meromorphic continuation of its
resolvent depends on the nature of $V$; so, it will be discussed when
specializing to $V$ periodic or random.
\begin{Rem}[Notations]
  \label{rem:8}
  In the sequel, we write $a\lesssim b$ if for some $C>0$ (independent
  of the parameters coming into $a$ or $b$), one has $a\leq C b$. We write
  $a\asymp b$ if $a\lesssim b$ and $b\lesssim a$.
\end{Rem}
\subsection{The meromorphic continuation of the resolvent}
\label{sec:merom-cont-resolv}
One proves the well known and simple
\begin{Th}
  \label{thr:1}
  The operator valued functions $z\in\C^+\mapsto (z-H^{\N}_L)^{-1}$
  and $z\in\C^+\mapsto (z-H^\Z_L)^{-1}$) admit a meromorphic
  continuation from $\C^+$ to $\D\C\setminus\left((-\infty,-2]
    \cup[2,+\infty)\right)$ through $(-2,2)$ (see Fig.~\ref{fig:1})
  with values in the operators from $l^2_\text{comp}$ to
  $l^2_\text{loc}$.\\
  Moreover, the number of poles of each of these meromorphic
  continuations in the lower half-plane is at most equal to $L$.
\end{Th}
\noindent The resonances are defined to be the poles of this
meromorphic continuation (see Fig.~\ref{fig:1}).
\subsection{The periodic case}
\label{sec:periodic-case}
We assume that, for some $p>0$, one has
\begin{equation}
  \label{eq:3}
  V_{n+p}=V_n\quad\text{for all}\quad n\geq0  .
\end{equation}
Let $\Sigma_\N$ be the spectrum of $H^\N=-\Delta+V$ acting on
$\ell^2(\N)$ with Dirichlet boundary condition at $0$ and $\Sigma_\Z$
be the spectrum of $H^\Z=-\Delta+V$ acting on $\ell^2(\Z)$. One has
the following description for these spectra:
\begin{itemize}
\item $\Sigma_\Z$ is a union of intervals, i.e., $\displaystyle
  \Sigma_\Z:=\sigma(H)=\bigcup_{j=1}^p[E_j^-,E_j^+]$ where
  $E_j^-<E_j^+$ ($1\leq j\leq p$) and $a^+_{j-1}\leq E_j^-$ ($2\leq j\leq p$) (see
 , e.g.,~\cite{MR0650253}); the spectrum of $H^\Z$ is purely absolutely
  continuous and the spectral resolution can be obtained via a
  Bloch-Floquet decomposition (see, e.g.,~\cite{MR0650253});
\item on $\ell^2(\N)$ (see, e.g.,~\cite{MR1301837}), one has
  \begin{itemize}
  \item $\Sigma_\N=\Sigma_\Z\cup\{v_j; 1\leq j\leq n\}$ and $\Sigma_\Z$ is the
    a.c. spectrum of $H$;
  \item the $(v_j)_{0\leq j \leq n}$ are isolated simple eigenvalues
    associated to exponentially decaying eigenfunctions.
  \end{itemize}
\end{itemize}
It may happen that some of the gaps are closed, i.e., that the number
of connected components of $\Sigma_\Z$ be strictly less than
$p$. There still is a natural way to write $\displaystyle
\Sigma_\Z:=\sigma(H)=\bigcup_{j=1}^p[E_j^-,E_j^+]$ (see
section~\ref{sec:spectral-theory-Z}), but in this case, for some
$j$'s, one has $E_{j-1}^+=E_j^-$; the energies $E_{j-1}^+=E_j^-$, we
shall call {\it closed gaps} (see Definition~\ref{def:1}). The
existence of closed gaps is non generic (see~\cite{MR0650253}).\\
The operators $H^\bullet$ (for $\bullet\in\{\N,\Z\}$) admit an
integrated density of states defined by
\begin{equation}
  \label{eq:144}
  N(E)=\lim_{L\to+\infty}\frac{\#\{\text{eigenvalues of
    }(-\Delta+V)_{|\llbracket-L,L\rrbracket\cap\bullet}\text{ in
    }(-\infty,E]\}}{\#(\llbracket-L,L\rrbracket\cap\bullet)}.
\end{equation}
Here, the restriction of $-\Delta+V$ to
$\llbracket-L,L\rrbracket\cap\bullet$ is taken with Dirichlet boundary
conditions; this is to fix ideas as it is known that, in the limit
$L\to+\infty$, other self-adjoint boundary conditions would yield the
same result for the limit~\eqref{eq:144}.\\
The integrated density of states is the same for $H^\N$ and $H^\Z$
(see, e.g.,~\cite{MR94h:47068}). It defines the distribution function of
some probability measure on $\Sigma_\Z$ that is real analytic on
$\overset{\circ}{\Sigma}_\Z$. Let $n$ denote the density of states of
$H^\N$ and $H^\Z$, that is, $\D n(E)=\frac{dN}{dE}(E)$.
\begin{Rem}
  \label{rem:7}
  When $L$ gets large, as $H^{\N}_L$ tends to $H^\N$ in strong
  resolvent sense, interesting phenomena for the resonances of
  $H_L^\N$ should take place near energies in
  $\Sigma_\N$. \\
  Define $\tau_k$ to be the shift by $k$ steps to the left, that is,
  $\tau_kV(\cdot)= V(\cdot+k)$. Then, for $(\ell_L)_L$
  s.t. $l_L\to+\infty$ and $L-\ell_L\to+\infty$ when $L\to+\infty$,
  $\tau^*_{l_L}H^{\Z}_L\tau_{l_L}$ tend to $H^\Z$ in strong resolvent
  sense. Thus, interesting phenomena for the resonances of $H_L^\Z$
  should take place near energies in
  $\Sigma_\Z$.
\end{Rem}
\subsubsection{Resonance free regions}
\label{sec:reson-free-regi}
We start with a description of resonance free regions near the real
axis. Therefore, we introduce some operators on the positive and the
negative half-lattice.\\
Above we have defined $H_\N$; we shall need another auxiliary
operator. On $\ell^2(\Z_-)$ (where $\Z_-=\{n\leq 0\}$), consider the
operator $H_k^-=-\Delta+\tau_kV$ with Dirichlet boundary condition at
$0$ (where $\tau_k$ is defined to be the shift by $k$ steps to the
left, that is, $\tau_kV(\cdot)=
V(\cdot+k)$). Let $\Sigma_k^-=\sigma(H^-_k)$.\\
As is the case for $H^\N$, one knows that
$\sigma_{\text{ess}}(H^-_k)=\Sigma_\Z$ and that
$\sigma_{\text{ess}}(H^-_k)$ is purely absolutely continuous (see,
e.g.,~\cite[Chapter 7]{MR1711536}). $H^-_k$ may also have discrete
eigenvalues in $\R\setminus\Sigma_\Z$.\\
We prove
\begin{Th}
  \label{thr:5}
  Let $I$ be a compact interval in $(-2,2)$.  Then,
  \begin{enumerate}
  \item if $I\subset\R\setminus\Sigma_\N$
    (resp. $I\subset\R\setminus\Sigma_\Z$), then, there exists $c>0$
    such that, for $L$ sufficiently large, $H^{\N}_L$ (resp. $H^\Z_L$)
    has no resonances in the rectangle $\{\text{Re}\,z\in I,\ \text
    {Im}\,z\in[-c,0]\}$;
  \item if $I\subset\Sigma_\Z$, then, there exists $c>0$ such that,
    for $L$ sufficiently large, $H^{\N}_L$ and $H^\Z_L$ have no
    resonances in the rectangle $\{\text{Re}\,z\in I,\ \text
    {Im}\,z\in[-c/L,0]\}$;
  \item fix $0\leq k\leq p-1$ and assume the compact interval $I$ to be such
    that $\{v_j\}=\overset{\circ}{I}\cap\Sigma_\N=I\cap\Sigma_\N$ and
    $I\cap\Sigma_\Z=\emptyset$ ($(v_j)_j$ are defined in the beginning
    of section~\ref{sec:periodic-case}):
    \begin{enumerate}
    \item if $I\cap\Sigma^-_k=\emptyset$ then, there exists $c>0$ such
      that, for $L$ sufficiently large such that $L\equiv k\mod p$,
      $H^{\N}_L$ has a unique resonance in the rectangle
      $\{\text{Re}\,z\in I,\ -c\leq\text{Im}\,z\leq0\}$; moreover, this
      resonance, say $z_j$, is simple and satisfies $\text{Im}
      \,z_j\asymp - e^{-\rho_j L}$ and $|z_j-\lambda_j|\asymp
      e^{-\rho_j L}$ for some $\rho_j>0$ independent of $L$;
    \item if $I\cap\Sigma^-_k\not=\emptyset$ then, there exists $c>0$
      such that, for $L$ sufficiently large such that $L\equiv k\mod
      p$, $H^{\N}_L$ has no resonance in the rectangle
      $\{\text{Re}\,z\in I,\ -c\leq\text{Im}\,z\leq0\}$.
    \end{enumerate}
  \end{enumerate}
\end{Th}
\noindent So, below the spectral interval $(-2,2)$, there exists a
resonance free region of width at least of order $L^{-1}$. For
$H^\N_L$, if $L\equiv k\mod p$, each discrete eigenvalue of $H^\N$
that is not an eigenvalue of $H^-_k$ generates a resonance for
$H^{\N}_L$ exponentially close to the real axis (when $L$ is
large). When the eigenvalue of $H_k^-$ is also an eigenvalue of
$H^\N=H_0^+$, it may also generate a resonance but only much further
away in the complex plane, at least at a distance of order $1$ to the
real axis.\\
In case (3)(a) of Theorem~\ref{thr:5}, one can give an asymptotic
expansion for the resonances (see section~\ref{sec:proof-theorem-5}).\\
\noindent We now turn to the description of the resonances of
$H^\bullet_L$ near $[-2,2]$. Therefore, it will be useful to introduce
a number of auxiliary functions and operators.
\subsubsection{Some auxiliary functions}
\label{sec:auxil-op}
To $H_k^-$ defined above, we associate $N^-_k$, the distribution
function of its spectral measure (that is a probability measure),
i.e., for $\varphi\in\Coi(\R)$, we define
$\D\int_\R\varphi(\lambda)dN^-_k(\lambda):=\varphi(H_k^-)(0,0)$ where
$(\varphi(H_k^-)(x,y))_{(x,y)\in(\Z_-)^2}$ denotes the kernel of
the operator $\varphi(H_k^-)$.\\
On $\overset{\circ}{\Sigma}_\Z$, the spectral measure $dN_k^-$ admits
a density with respect to the Lebesgue measure, say, $n^-_k$, and this
density is real analytic (see Proposition~\ref{pro:1}). \\
For $E\in\overset{\circ}{\Sigma}_\Z$, define
\begin{equation}
  \label{eq:5}
  S^-_k(E):=\text{p.v.}
  \left(\int_\R\frac{dN^-_k(\lambda)}{\lambda-E}\right)
  =\lim_{\varepsilon\downarrow0}\left(\int_{-\infty}^{E-\varepsilon}
    \frac{dN^-_k(\lambda)}{\lambda-E}-
    \int_{E+\varepsilon}^{+\infty}\frac{dN^-_k(\lambda)}{\lambda-E}\right).
\end{equation}
The existence and analyticity of the Cauchy principal value $S^-_k$ on
$\overset{\circ}{\Sigma}_\Z$ is guaranteed by the analyticity of
$n^-_k$ (see, e.g.,~\cite{MR2542214}). Moreover, for
$E\in\overset{\circ}{\Sigma}_\Z$, one has
\begin{equation}
  \label{eq:51}
  S^-_k(E)=\lim_{\varepsilon\to0^+}
  \int_\R\frac{dN^-_k(\lambda)}{\lambda-E-i\varepsilon}-i\pi n_k^-(E).
\end{equation}
In the lower half-plane $\{$Im$\,E<0\}$, define the function
\begin{equation}
  \label{eq:8}
  \Xi^-_k(E):=\int_\R\frac{dN_k^-(\lambda)}{\lambda-E}+e^{-i\arccos(E/2)}
  =\int_\R\frac{dN_k^-(\lambda)}{\lambda-E}+E/2+\sqrt{(E/2)^2-1}    
\end{equation}
where
\begin{itemize}
\item in the first formula, the function $z\mapsto\arccos z$ is the
  analytic continuation to the lower half-plane of the determination
  taking values in $[-\pi,0]$ on the interval $[-1,1]$;
\item in the second formula, the branch of the square root
  $z\mapsto\sqrt{z^2-1}$ has positive imaginary part for $z\in(-1,1)$.
\end{itemize}
The function $\Xi^-_k$ is analytic in $\{$Im$\,E<0\}$ and in a
neighborhood of $(-2,2)\cap\overset{\circ}{\Sigma}_\Z$. Moreover,
$\Xi^-_k$ vanishes identically if and only if $V\equiv0$ (see
Proposition~\ref{pro:4}).\\
From now on we assume that $V\not\equiv0$. In this case, in
$\{$Im$\,E<0\}$ and on $(-2,2)\cap\overset{\circ}{\Sigma}_\Z$, the
analytic function $\Xi_k^-$ has only finitely many zeros, each of
finite multiplicity (see Proposition~\ref{pro:4}). \\
We shall need the analogues of the above defined functions the already
introduced operator $H_0^+:=H^\N=-\Delta+V$ considered on $\ell^2(\N)$
with Dirichlet boundary conditions at $0$. We define the function
$N^+_0$ as the distribution function of the spectral measure of
$H_0^+$, i.e., for $\varphi\in\Coi(\R)$, we define
$\D\int_\R\varphi(\lambda)dN^+_0(\lambda):=\varphi(H_0^+)(0,0)$.  In
the same way as we have defined $n_k^-$, $S_k^-$ and $\Xi_k^-$ from
$H_k^-$, one can define $n_0^+$, $S_0^+$ and $\Xi_0^+$ from $H_0^+$.
They also satisfy Proposition~\ref{pro:1}, relation~\eqref{eq:51} and
Proposition~\ref{pro:4}.\\
For the description of the resonances, it will be convenient to define
the following functions on $\overset{\circ}{\Sigma}_\Z$
\begin{equation}
  \label{eq:232}
  c^{\N}(E):=i+\frac{\Xi_k^-(E)}{\pi\,n_k^-(E)}=\frac1{\pi\,n_k^-(E)}
  \left(S_k^-(E)+e^{-i\arccos(E/2)}\right)
\end{equation}
and
\begin{equation}
  \label{eq:231}
  c^{\Z}(E):=\frac{\D\frac{\left(S^+_0(E)+e^{-i\arccos(E/2)}\right)
      \left(S_k^-(E)+e^{-i\arccos(E/2)}\right)}{n_0^+(E)\,n_k^-(E)}
    -\pi^2}{\D\frac{\pi\left(S^+_0(E)+e^{-i\arccos(E/2)}\right)}{n_0^+(E)}
    +\frac{\pi\left(S_k^-(E)+e^{-i\arccos(E/2)}\right)}{n_k^-(E)}}.
\end{equation}
We shall see that the the zeros of $c^\bullet-i$ play a special role
for the resonances of $H^\bullet_L$: therefore, we define
\begin{equation}
  \label{eq:230}
  \mathcal{D}^\bullet=\left\{z\in\overset{\circ}{\Sigma}_\Z;\
    c^\bullet(z)=i\right\} 
\end{equation}
The set $\mathcal{D}$ introduced in Theorem~\ref{thr:22} is the set
$\mathcal{D}^\Z\cap(-2,2)$.
\begin{Rem}
  \label{rem:6} Before describing the resonances, let us explain why
  the operators $H_0^+$ and $H_k^-$ naturally occur in this
  study. They respectrively are the strong resolvent limits (when
  $L\to+\infty$ s.t. $L\in p\N+k$) of the operator $H^{\Z}_L$
  restricted to $\llbracket 0,L\rrbracket$ with Dirichlet boundary
  conditions at $0$ and $L$ ``seen'' respectively from
  the left and the right hand side. \\
  Indeed, define $H_L$ to be the operator $H^{\N}_L$ restricted to
  $\llbracket 0,L\rrbracket$ with Dirichlet boundary conditions at $L$
  (see Remark~\ref{rem:3}). Note that $H_L$ is also the operator
  $H^{\Z}_L$ restricted to $\llbracket 0,L\rrbracket$ with Dirichlet
  boundary  conditions at $0$ and $L$.\\
  Clearly, the operator $H_0^+$ is the strong resolvent limit of $H_L$
  when
  $L\to+\infty$.\\
  If $\tilde\tau_L$ denotes the translation by $-L$ that unitarily
  maps $\ell^2(\llbracket 0,L\rrbracket)$ into
  $\ell^2(\llbracket-L,0\rrbracket)$, then, $\tilde H_L=\tilde\tau_L
  H_L\tilde\tau^*_L$ converges in the strong resolvent sense to
  $H_k^-$ when $L\to+\infty$ and $L\equiv k\mod(p)$. Indeed,
  $\tau_LV=\tau_kV$ as $V$ is $p$ periodic.
\end{Rem}
\subsubsection{Description of the resonances closest to the real axis}
\label{sec:descr-reson}
Let $(\lambda_l)_{0\leq l\leq L}=(\lambda_l^L)_{0\leq l\leq L}$ be the eigenvalues of
$H_L$ (that is, the eigenvalues of $H^{\N}_L$ or $H^\Z_L$ restricted
to $\llbracket 0,L\rrbracket$ with Dirichlet boundary conditions, see
remark~\ref{rem:3}) listed in increasing order. They are described in
Theorem~\ref{thr:16}; those away from the edges of $\Sigma_\Z$ are
shown to be nicely interspaced points at a distance roughly $L^{-1}$
from one another. \\
We first state our most general result describing the resonances in a
uniform way. We, then, derive two corollaries describing the behavior
of the resonance, first, far from the set of exceptional energies
$\mathcal{D}^\bullet$, second, close to an exceptional energy.\\
Pick a compact interval
$I\subset(-2,2)\cap\overset{\circ}{\Sigma}_\Z$. For
$\bullet\in\{\N,\Z\}$ and $\lambda_l\in I$, for $L$ large, define the
complex number
\begin{equation}
  \label{eq:182}
  \tilde z^\bullet_l=\lambda_l+\frac1{\pi\,n(\lambda_l)\,L}\cot^{-1}\circ\,
  c^\bullet\left[\lambda_l+\frac1{\pi\,n(\lambda_l)\,L}\cot^{-1}\circ\,
    c^\bullet\left(\lambda_l-i\frac{\log L}L\right)\right]
\end{equation}
where the determination of $\cot^{-1}$ is the inverse of the
determination $z\mapsto\cot(z)$ mapping
$\left[0,\pi\right)\times(0,-\infty)$ onto
$\C^+\setminus\{i\}$.\\
Note that, by Proposition~\ref{pro:5}, for $L$ sufficiently large, we
know that, for any $l$ such that $\lambda_l\in I$, one has
\begin{equation*}
  \text{Im}\,c^\bullet\left(\lambda_l-i\frac{\log L}L\right)
  \in(0,+\infty)\setminus\{1\}
\end{equation*}
and
\begin{equation*}
  \text{Im}\,c^\bullet\left[\lambda_l+\frac1{\pi\,n(\lambda_l)\,L}
    \cot^{-1}\circ\,
    c^\bullet\left(\lambda_l-i\frac{\log L}L\right)\right]
  \in(0,+\infty)\setminus\{1\}.
\end{equation*}
Thus, the formula~\eqref{eq:182} defines $\tilde z^\bullet_l$ properly
and in a unique way. Moreover, as the zeros of $E\mapsto
c^\bullet(E)-i$ are of finite order, one checks that
\begin{equation}
  \label{eq:56}
  -\log L\lesssim\,L\cdot\text{Im}\,\tilde z^\bullet_l\lesssim -1\quad
  \text{ and }\quad 1\lesssim L\cdot\text{Re}\,\left(\tilde z^\bullet_{l+1}
    -\tilde z^\bullet_l\right)
\end{equation}
where the constants are uniform for $l$ such that $\lambda_l\in I$.\\
We prove the
\begin{Th}
  \label{thr:2}
  Pick $\bullet\in\{\N,\Z\}$ and $k\in\{0,\cdots,p-1\}$. Let
  $E_0\in(-2,2)\cap\overset{\circ}{\Sigma}_\Z$.\\
  Then, there exists $\eta_0>0$ and $L_0>0$ such that, for $L>L_0$
  satisfying $L=k\mod(p)$, for each $\lambda_l\in
  I:=[E_0-\eta_0,E_0+\eta_0]$, there exists a unique resonance of
  $H^\bullet_L$, say $z^\bullet_l$, in the rectangle
  \begin{equation*}
    \left[\frac{\text{Re}\,(\tilde z^\bullet_l+\tilde
        z^\bullet_{l-1})}2, \frac{\text{Re}\,(\tilde z^\bullet_l+\tilde
        z^\bullet_{l+1})}2\right] +i\left[-\eta_0,0\right];
  \end{equation*}
  this resonance is simple and it satisfies
  $\D\left|z^\bullet_l-\tilde z^\bullet_l\right|\lesssim\frac1{L\log
    L}$.
\end{Th}
\noindent This result calls for a few comments. First, the picture one
gets for the resonances can be described as follows (see also
Figure~\ref{fig:10}). As long as $\lambda_l$ stays away from any zero
of $E\mapsto c^\bullet(E)-i$, the resonances are nicely spaced points
as the following proposition proves.
\begin{Pro}
  \label{pro:2}
  Pick $\bullet\in\{\N,\Z\}$ and $k\in\{0,\cdots,p-1\}$.  Let
  $I\subset(-2,2)\cap\overset{\circ}{\Sigma}_\Z$ be a compact interval
  such that $I\cap\mathcal{D}^\bullet=\emptyset$.\\
  Then, for $L$ sufficiently large, for each $\lambda_l\in I$, the
  resonance $z^\bullet_l$ admits a complete asymptotic expansion in
  powers of $L^{-1}$ and one has
  \begin{equation}
    \label{eq:176}
    z^\bullet_l=\lambda_l+\frac1{\pi\,n(\lambda_l)\,L}\cot^{-1}\circ\,
    c^\bullet(\lambda_l) +O\left(\frac1{L^2}\right)
  \end{equation}
  where the remainder term is uniform in $l$.
\end{Pro}
\begin{figure}[h]
  \centering
  \includegraphics[height=4cm]{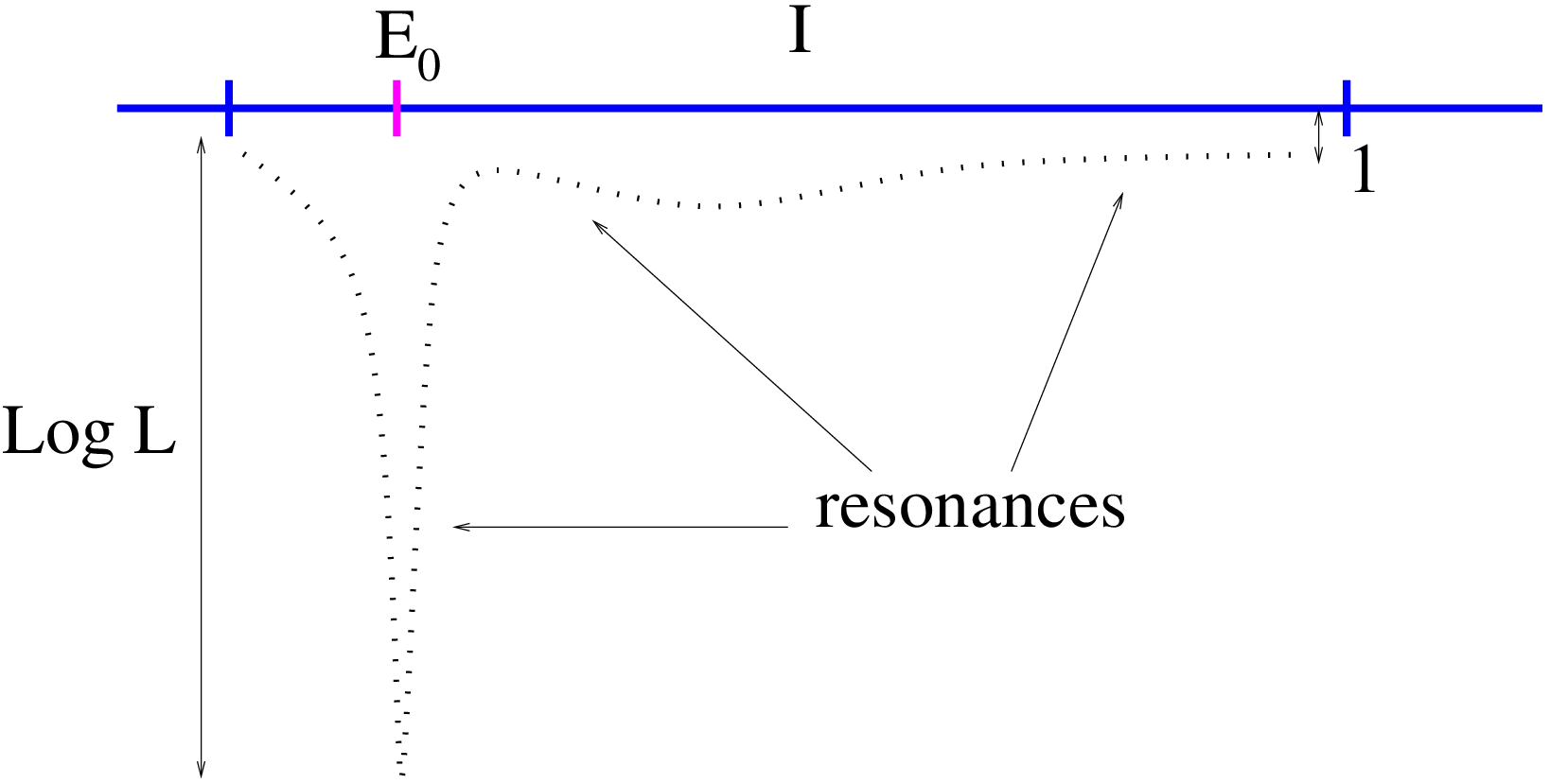}
  \caption{The resonances close to the real axis in the periodic case
    (after rescaling their imaginary parts by $L$)}
  \label{fig:10}
\end{figure}
\noindent The proof of Proposition~\ref{pro:2} actually yields a
complete asymptotic expansion in powers of $L^{-1}$ for the resonances
in this zone (see section~\ref{sec:proofs-propositions}).\\
Proposition~\ref{pro:2} implies Theorem~\ref{thr:22}: we chose
$\bullet=\Z$, $k=0$ and the set $\mathcal{D}$ of exceptional points in
Theorem~\ref{thr:22} is exactly $\mathcal{D}^\Z\cap(-2,2)$; to
obtain~\eqref{eq:6}, it suffices to use the asymptotic form of the
Dirichlet eigenvalues given by Theorem~\ref{thr:16}.\\
Near the zeros of $E\mapsto c^\bullet(E)-i$, the resonances take a
``plunge'' into the lower half of the complex plane (see
Figure~\ref{fig:10}) and their imaginary part becomes of order
$L^{-1}\log L$. Indeed, Theorem~\ref{thr:2} and~\eqref{eq:182} imply
\begin{Pro}
  \label{pro:3}
  Pick $\bullet\in\{\N,\Z\}$ and $k\in\{0,\cdots,p-1\}$.  Let
  $E_0\in\mathcal{D}^\bullet$ be a zero of $E\mapsto c^\bullet(E)-i$
  of order $q$ in
  $(-2,2)\cap\overset{\circ}{\Sigma}_\Z$.\\
  Then, for $\alpha>0$, for $L$ sufficiently large, if $l$ is such
  that $|\lambda_l-E_0|\leq L^{-\alpha}$, the resonance $z^\bullet_l$
  satisfies
  \begin{equation}
    \label{eq:181}
    \text{Im}\,z^\bullet_l=\frac{q}{2\pi\,n(\lambda_l)}\cdot\frac{\log\left(
        |\lambda_l-E_0|^2+\left(\frac{q\,\log L}{2\pi\,n(\lambda_l)\,L}
        \right)^2\right)}{2\,L}\cdot (1+o(1))
  \end{equation}
  where the remainder term is uniform in $l$ such that
  $|\lambda_l-E_0|\leq L^{-\alpha}$.
\end{Pro}
\noindent When $\bullet=\Z$, the asymptotic~\eqref{eq:181} shows that
there can be a ``resonance'' phenomenon for resonances: when the two
functions $\Xi_k^-$ and $\Xi_0^+$ share a zero at the same real
energy, the maximal width of the resonances increases; indeed, the
factor in front of $L^{-1}\log L$ is proportional to the multiplicity
of the zero of $\Xi_k^-\,\Xi_0^+$.
\subsubsection{Description of the low lying resonances}
\label{sec:descr-reson-ll}
The resonances found in Theorem~\ref{thr:2} are not necessarily the
only ones: deeper into the lower complex plane, one may find more
resonances. They are related to the zeros of $\Xi_k^-$ when
$\bullet=\N$ and $\Xi_k^-\,\Xi_0^+$ when $\bullet=\Z$ (see
Proposition~\ref{pro:5}).\\
We now study what happens below the line $\{$Im$\,z=-\eta_0\}$ (see
Theorem~\ref{thr:2}) for the resonances of $H^{\N}_L$ and $H^\Z_L$.\\
The functions $\Xi_k^-$ and $\Xi_0^+$ are analytic in the lower half
plane and, by Proposition~\ref{pro:4}, they don't vanish in an
neighborhood of $-i\infty$. Hence, the functions $\Xi_k^-$ and
$\Xi_0^+$ have only finitely many zeros in the lower half plane.\\
We prove
\begin{Th}
  \label{thr:3}
  Pick $\bullet\in\{\N,\Z\}$ and $k\in\{0,\cdots,p-1\}$. Let
  $(E^\bullet_j)_{1\leq j\leq J}$ be the zeros of $E\mapsto c^\bullet(E)-i$ in
  $I+i(-\infty,0)$. Pick
  $E_0\in(-2,2)\cap\overset{\circ}{\Sigma}_\Z$.\\
  There exists $\eta_0>0$ such that, for $I=E_0+[-\eta_0,\eta_0]$, for
  $L$ sufficiently large s.t. $L\equiv k\mod(p)$, one has,
  \begin{itemize}
  \item if $E_0\not\in\{\text{Re}\,E^\bullet_j; \ 1\leq j\leq J\}$, then, in
    the rectangle $I+i(-\infty,0]$, the only resonances of $H^{\N}_L$
    and $H^{\Z}_L$ are those given by Theorem~\ref{thr:2};
  \item if $E_0\in\{\text{Re}\,E^\bullet_j; \ 1\leq j\leq J\}$, then,
    \begin{itemize}
    \item in the rectangle $I+i[-\eta_0,0]$, the only resonances of
      $H^{\N}_L$ and $H^{\Z}_L$ are those given by
      Theorem~\ref{thr:2};
    \item in the strip $I+i[-\infty,-\eta_0]$, the resonances of
      $H^{\bullet}_L$ are contained in $\D\bigcup_{j=1}^J
      D\left(E^\bullet_j,e^{-\eta_0 L}\right)$
    \item in $\D D\left(E^\bullet_j,e^{-\eta_0 L}\right)$, the number of
      resonances (counted with multipli\-city) is equal to the order
      of $E^\bullet_j$ as a zero of $E\mapsto c^\bullet(E)-i$.
    \end{itemize}
  \end{itemize}
\end{Th}
\noindent We see that the total number of resonances below a compact
subset of $(-2,2)\cap\overset{\circ}{\Sigma}_\Z$ that do not tend to
the real axis when $L\to+\infty$ is finite. These resonances are
related to the resonances of $H^\infty$ to which we turn now.
\subsubsection{The half-line periodic perturbation}
\label{sec:half-line-pert}
Fix $p\in\N^*$. On $\ell^2(\Z)$, we now consider the operator
$H^\infty=-\Delta+V$ where $V(n)=0$ for $n\geq0$ and $V(n+p)=V(n)$ for
$n\leq-1$. We prove
\begin{Th}
  \label{thr:25}
  The resolvent of $H^\infty$ can be analytically continued from the
  upper half-plane through $(-2,2)\cap\overset{\circ}{\Sigma}_Z$ to
  the lower half plane. The resulting operator does not have any poles
  in the lower half-plane or on
  $(-2,2)\cap\overset{\circ}{\Sigma}_Z$.\\
  The resolvent of $H^\infty$ can be analytically continued from the
  upper half-plane through $(-2,2)\setminus\Sigma_\Z$ (resp.
  $\overset{\circ}{\Sigma}_Z\setminus[-2,2]$) to the lower half plane;
  the poles of the continuation through $(-2,2)\setminus\Sigma_\Z$
  (resp.  $\overset{\circ}{\Sigma}_Z\setminus[-2,2]$) are exactly the
  zeros of the function $\D E\mapsto1-e^{i\theta(E)}
  \int_\R\frac{dN^-_{p-1} (\lambda)}{\lambda-E}$ when continued from
  the upper half-plane through $(-2,2)\setminus\Sigma_\Z$ (resp.
  $\overset{\circ}{\Sigma}_Z\setminus[-2,2]$) to the lower half-plane.
\end{Th}
\begin{Rem}
  In Theorem~\ref{thr:25} and below, every time we consider the
  analytic continuation of a resolvent through some open subset of the
  real line, we implicitly assume the open subset to be non empty.
\end{Rem}
\noindent In figure~\ref{fig:11}, to illustrate Theorem~\ref{thr:25},
assuming that $\Sigma_\Z$ (in blue) has a single gap that is contained
in $(-2,2)$, we drew the various analytic continuations of the
resolvent of $H^\infty$ and the presence or absence of resonances for
the different continuations .
\begin{figure}[h]
  \centering
  \includegraphics[height=5.5cm]{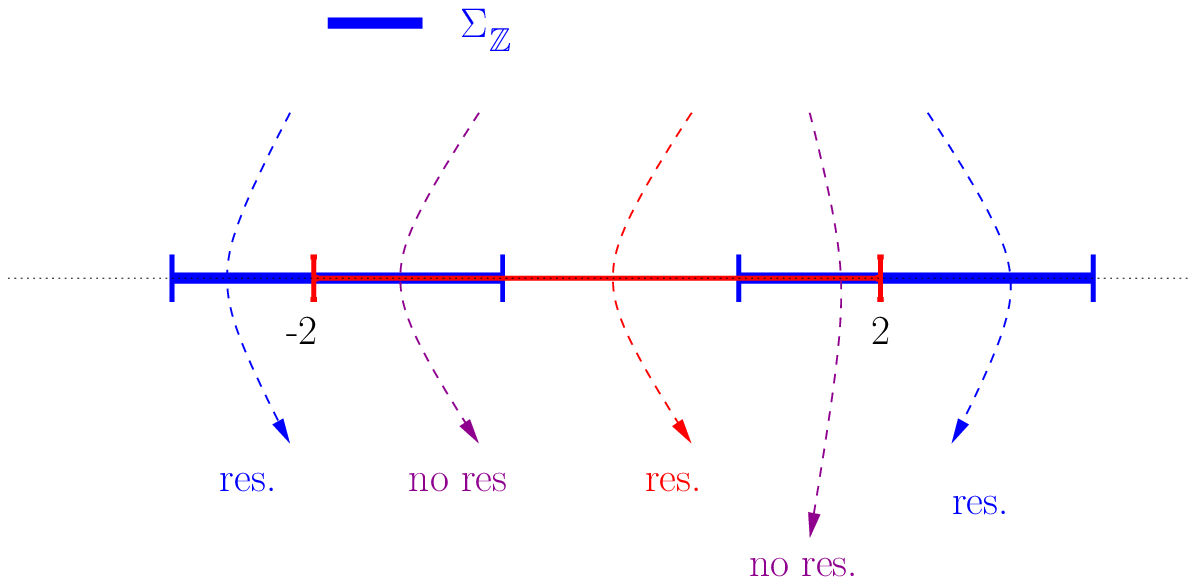}
  \caption{The analytic continuation of the resolvent and resonances
    for $H^\infty$}
  \label{fig:11}
\end{figure}
Using the same arguments as in the proof of Proposition~\ref{pro:4},
one easily sees that the continuations of the function $\D
E\mapsto1-e^{i\theta(E)}\int_\R\frac{dN^-_{p-1} (\lambda)}{\lambda-E}$
to the lower half plane through $(-2,2)\setminus\Sigma_\Z$ and
$\overset{\circ}{\Sigma}_Z\setminus[-2,2]$ have at most finitely many
zeros and that these zeros are away from the real axis.\\
This also implies that the spectrum on $H^\infty$ in
$[-2,2]\cup\Sigma_\Z$ is purely absolutely continuous except possibly
at the points of $\partial\Sigma_\Z\cup\{-2,2\}$ where
$\partial\Sigma_\Z$ is the set of edges of $\Sigma_\Z$.
\subsection{The random case}
\label{sec:random-case}
We now turn to the random case. Let $V=V_\omega$ where
$(V_\omega(n))_{n\in\Z}$ are bounded independent and identically
distributed random variables. Assume that the common law of the random
variables admits a bounded compactly
supported density, say, $g$.\\
Set $H^\N_\omega=-\Delta+V_\omega$ on $\ell^2(\N)$ (with Dirichlet
boundary condition at $0$ to fix ideas). Let $\sigma(H^\N_\omega)$ be
the spectrum of $H^\N_\omega$. Consider also
$H^\Z_\omega=-\Delta+V_\omega$ acting on $\ell^2(\Z)$. Then, one knows
(see, e.g.,~\cite{MR2509110}) that, $\omega$ almost surely,
\begin{equation}
  \label{eq:186}
  \sigma(H^\Z_\omega)=\Sigma:=[-2,2]+\text{supp}\,g.
\end{equation}
One has the following description for the spectra
$\sigma(H^\N_\omega)$ and $\sigma(H^\Z_\omega)$:
\begin{itemize}
\item $\omega$-almost surely, $\sigma(H^\Z_\omega)=\Sigma$; the
  spectrum is purely punctual; it consists of simple eigenvalues
  associated to exponentially decaying eigenfunctions (Anderson
  localization, see, e.g.,~\cite{MR94h:47068,MR2509110}); one can prove
  that, under the assumptions made above, the whole spectrum is
  dynamically localized (see, e.g.,~\cite{MR883643} and references
  therein);
\item for $H^\N_\omega$ (see, e.g.,~\cite{MR94h:47068,MR1102675}), one
  has, $\omega$-almost surely, $\sigma(H^\N_\omega)=\Sigma\cup
  K_\omega$ where
  \begin{itemize}
  \item $\Sigma$ is the essential spectrum of $H^\N_\omega$; it
    consists of simple eigenvalues associated to exponentially
    decaying eigenfunctions;
  \item the set $K_\omega$ is the discrete spectrum of $H^\N_\omega$;
    it may be empty and depends on $\omega$.
  \end{itemize}
\end{itemize}
\subsubsection{The integrated density of states and the Lyapunov
  exponent}
\label{sec:lyapunov-exponent}
It is well known (see, e.g.,~\cite{MR94h:47068}) that the integrated
density of states of $H$, say, $N(E)$ is defined as the following
limit
\begin{equation}
  \label{eq:4}
  N(E)=\lim_{L\to+\infty}\frac{\#\{\text{eigenvalues of
    }H^\Z_{\omega|\llbracket-L,L\rrbracket}\text{ in }(-\infty,E]\}}{2L+1}.
\end{equation}
The above limit does not depend on the boundary conditions used to
define the restriction $H^\Z_{\omega|\llbracket-L,L\rrbracket}$. It
defines the distribution function of a probability measure supported
on $\Sigma$. Under our assumptions on the random potential, $N$ is
known to be Lipschitz continuous (\cite{MR94h:47068,MR2509110}). Let
$\displaystyle n(E)=\frac{dN}{dE}(E)$ be its derivative; it exists for
almost all energies. If one assumes more regularity on $g$ the density
of the random variables $(\omega_n)_n$, then, the density of states
$n$ can be shown to exist everywhere and to be regular (see,
e.g.,~\cite{MR883643}).\\
One also defines the Lyapunov exponent, say $\rho(E)$ as follows
\begin{equation*}
  \rho(E):=\lim_{L\to+\infty}\frac{\log\|T_L(E,\omega)\|}{L+1}  
\end{equation*}
where
\begin{equation}
  \label{eq:9}
  T_L(E;\omega):= \begin{pmatrix}
    E-V_\omega(L) & -1   \\ 1 &0    \end{pmatrix}\times\cdots\times
  \begin{pmatrix}
    E-V_\omega(0) & -1 \\ 1 &0 \end{pmatrix}
\end{equation}
For any $E$, $\omega$-almost surely, the Lyapunov exponent is known to
exist and to be independent of $\omega$ (see,
e.g.,~\cite{MR883643,MR94h:47068,MR1102675}). It is positive at all
energies. Moreover, by the Thouless formula~\cite{MR883643}, it is
positive and continuous for all $E$ and it is the harmonic conjugate
of $n(E)$.\vskip.1cm\noindent
For $\bullet\in\{\N,\Z\}$, we now define $H^\bullet_{\omega,L}$ to be
the operator $-\Delta^\bullet+V_\omega\car_{\llbracket
  0,L\rrbracket}$. The goal of the next sections is to describe the
resonances of these operators in the limit $L\to+\infty$.\par
As in the case of a periodic potential $V$, the resonances are defined
as the poles of the analytic continuation of
$z\mapsto(H^\bullet_{\omega,L}-z)^{-1}$ from $\C^+$ through $(-2,2)$
(see Theorem~\ref{thr:1}).
\subsubsection{Resonance free regions}
\label{sec:reson-free-regi1}
We again start with a description of the resonance free region near a
compact interval in $(-2,2)$. As in the periodic case, the size of the
$H^\bullet_{\omega,L}$-resonance free region below a given energy will
depend on whether this energy belongs to $\sigma(H^\bullet_\omega)$ or
not. We prove
\begin{Th}
  \label{thr:4}
  Fix $\bullet\in\{\N,\Z\}$. Let $I$ be a compact interval in
  $(-2,2)$. Then, $\omega$-a.s., one has
  \begin{enumerate}
  \item for $\bullet\in\{\N,\Z\}$, if
    $I\subset\R\setminus\sigma(H^\bullet_\omega)$, then, there exists
    $C>0$ such that, for $L$ sufficiently large, there are no
    resonances of $H^\bullet_{\omega,L}$ in the rectangle
    $\{\text{Re}\,z\in I,\ 0\geq\text {Im}\,z\geq-1/C\}$;
  \item if $I\subset\overset{\circ}{\Sigma}$, then, for
    $\varepsilon\in(0,1)$, there exists $L_0>0$ such that, for $L\geq
    L_0$, there are no resonances of $H^\bullet_{\omega,L}$ in the
    rectangle $\{\text{Re}\,z\in I,\ 0\geq\text {Im}\,z\geq -e^{-2
      \eta_\bullet\rho L(1+\varepsilon)})\}$ where
    \begin{itemize}
    \item $\rho$ is the maximum of the Lyapunov exponent $\rho(E)$ on
      $I$
    \item $\D\eta_\bullet=\begin{cases}1\text{ if }\bullet=\N,\\
        1/2\text{ if }\bullet=\Z.\end{cases}$
    \end{itemize}
  \item pick $v_j=v_j(\omega)\in K_\omega$ (see the description of the
    spectrum of $H_\omega^\N$ just above
    section~\ref{sec:lyapunov-exponent}) and assume that
    $\{v_j\}=\overset{\circ}{I}\cap\sigma(H^\N_\omega)
    =I\cap\sigma(H^\N_\omega)$ and $I\cap\Sigma=\emptyset$, then,
    there exists $c>0$ such that, for $L$ sufficiently large,
    $H^{\N}_{\omega,L}$ has a unique resonance in $\{\text{Re}\,z\in
    I,\ -c\leq\text{Im}\,z\leq0\}$; moreover, this resonance, say
    $z_j$, is simple and satisfies $\text{Im}
    \,z_j\asymp-e^{-\rho_j(\omega)L}$ and $|z_j-\lambda_j|\asymp
    e^{-\rho_j(\omega)L}$ for some $\rho_j(\omega)>0$ independent of
    $L$.
  \end{enumerate}
\end{Th}
\noindent When comparing point (2) of this result with point (2) of
Theorem~\ref{thr:5}, it is striking that the width of the resonance
free region below $\Sigma$ is much smaller in the random case (it is
exponentially small in $L$) than in the periodic case (it is
polynomially small in $L$). This a consequence of the localized nature
of the spectrum, i.e., of the exponential decay of the eigenfunctions of
$H_\omega^\bullet$.
\subsubsection{Description of the resonances closest to the real axis}
\label{sec:descr-reson1}
We will now see that below the resonance free strip exhibited in
Theorem~\ref{thr:4} one does find resonances, actually, many of
them. We prove
\begin{Th}
  \label{thr:6}
  Fix $\bullet\in\{\N,\Z\}$. Let $I$ be a compact interval in
  $(-2,2)\cap\overset {\circ}{\Sigma}$.  Then,
  \begin{enumerate}
  \item for any $\kappa\in(0,1)$, $\omega$-a.s., one has
    \begin{equation*}
      \frac{\#\left\{z\text{ resonance of }H^\bullet_{\omega,L}\text{
            s.t. Re}\,z\in I,\ 0>\text{Im}\,z\geq
          -e^{-L^\kappa}\right\}}
      L\to \int_In(E)dE;
    \end{equation*}
  \item for $E\in I$ such that $n(E)>0$ and $\lambda\in(0,1)$, define the
    rectangle
    \begin{equation*}
      R^\bullet(E,\lambda,L,\varepsilon,\delta):=\left\{z\in\C;\ 
        \begin{aligned}
          n(E)&|\text{Re}\,z-E|\leq\varepsilon/2\\
          -e^{\eta_\bullet\rho(E)\delta L} &\leq
          e^{2\eta_\bullet\rho(E)\lambda\,L} \text{Im}\,z\leq
          -e^{-\eta_\bullet\rho(E)\delta L}
        \end{aligned}
      \right\}
    \end{equation*}
    where $\eta^\bullet$ is defined in Theorem~\ref{thr:4}; then,
    $\omega$-a.s., one has
    \begin{equation}
      \label{eq:113}
      \lim_{\delta\to0^+}\lim_{\varepsilon\to0^+}
      \lim_{L\to+\infty}\frac{\#\left\{z
          \text{ resonances of }H^\bullet_{\omega,L}
          \text{ in }R^\bullet(E,\lambda,L,\varepsilon,\delta)
        \right\}}{L\,\varepsilon\,\delta}=1.
    \end{equation}
  \item for $E\in I$ such that $n(E)>0$, define
    \begin{equation*}
      R_\pm^\bullet(E,1,L,\varepsilon,\delta)=\left\{z\in\C;\ 
        \begin{aligned}
          n(E)|\text{Re}\,z-E|\leq\varepsilon/2\\
          -e^{-2\eta_\bullet\rho(E)(1\pm\delta)L}\leq\text{Im}\,z<0
        \end{aligned}
      \right\};
    \end{equation*}
    then, $\omega$-a.s., one has
    \begin{equation}
      \label{eq:216}
      \lim_{\delta\to0^+}\lim_{\varepsilon\to0^+}
      \lim_{L\to+\infty}\frac{\#\left\{
          \text{resonances  in }R_\pm^\bullet(E,1,L,\varepsilon,\delta)
        \right\}}{L\,\varepsilon\,\delta}=
      \begin{cases}
        1\text{ if }\pm=-,\\ 0\text{ if }\pm=+.\\
      \end{cases}
    \end{equation}
  \item for $c>0$, $\omega$-a.s., one has
    \begin{equation}
      \label{eq:192}
      \lim_{L\to+\infty} \frac1L\#\left\{z\text{ resonances of
        }H^\bullet_{\omega,L} \text{ in
        }I+i\left(-\infty,-e^{-cL}\right]
      \right\}=\int_I\min\left(\frac{c}{\rho(E)},1\right)n(E)dE.
    \end{equation}
  \end{enumerate}
\end{Th}
\noindent The striking fact is that the resonances are much closer to
the real axis than in the periodic case; the lifetime of these
resonances is much larger. The resonant states are quite stable with
lifetimes that are exponentially large in the width of the random
perturbation. Point (4) is an integrated version of point (2). Let us
also note here that when $\bullet=\Z$, point (4) of
Theorem~\ref{thr:6} is the statement of Theorem~\ref{thr:29}.\\
Note that the rectangles $R^\bullet(E,\lambda,L,\varepsilon,\delta)$
are very stretched along the real axis; their side-length in imaginary
part is exponentially small in $L$ whereas their side-length in real
part is of order $1$.\\
To understand point (2) of Theorem~\ref{thr:6}, rescale the resonances
of $H^\bullet_{\omega,L}$, say, $(z_{l,L}^\bullet(\omega))_l$ as
follows
\begin{equation}
  \label{eq:10}
  \begin{aligned}
    x^\bullet_l&=x_{l,L}^\bullet(E,\omega)
    =n(E)\,L\cdot(\text{Re}\,z_{l,L}^\bullet(\omega)-E)
    \quad\text{and}\\
    y^\bullet_l&=y_{l,L}^\bullet(E,\omega)
    =-\frac1{2\eta_\bullet\,\rho(E)\,L}\log|\text{Im}\,z_{l,L}^\bullet(\omega)|.
  \end{aligned}
\end{equation}
For $\lambda\in(0,1)$, this rescaling maps the rectangle
$R^\bullet(E,\lambda,L,\varepsilon,\delta)$ into
$\{|x|\leq L\varepsilon/2,\ |y-\lambda|\leq \delta/2\}$; and the rectangles
$R_\pm^\bullet(E,1,L,\varepsilon,\delta)$ are respectively mapped into
$\{|x|\leq L\varepsilon/2,\ 1\mp\delta\leq y\}$. The denominator of the
quotient in~\eqref{eq:113} is just the area of the rescaled
$R^\bullet(E,\lambda,L,\varepsilon,\delta)$ for $\lambda\in(0,1)$ or
the rescaled $R_+^\bullet(E,1,L,\varepsilon,\delta)\setminus
R_-^\bullet(E,1,L,\varepsilon,0)$. So, point (2) states that in the
limit $\varepsilon$ and $\delta$ small and $L$ large, the rescaled
resonances become uniformly distributed in the rescaled rectangles.
We see that the structure of the set of resonances is very different
from the one observed in the periodic case (see
Fig.~\ref{fig:2-3}). We will now zoom in on the resonance even more so
as to make this structure clearer. Therefore, we consider the
two-dimensional point process $\xi^\bullet_L(E,\omega)$ defined by
\begin{equation}
  \label{eq:11}
  \xi^\bullet_L(E,\omega)=\sum_{z^\bullet_{l,L}\text{ resonance of }H^\bullet_{\omega,L}}
  \delta_{(x^\bullet_l,y^\bullet_l)}
\end{equation}
where $x^\bullet_l,$ and $y^\bullet_l$ are defined by~\eqref{eq:10}.\\
We prove
\begin{Th}
  \label{thr:7}
  Fix $E\in(-2,2)\cap\overset{\circ}{\Sigma}$ such that
  $n(E)>0$. Then, the point process $\xi^\bullet_L(E,\omega)$
  converges weakly to a Poisson process in $\R\times(0,1]$ with
  intensity $1$. That is, for any $p\geq0$, if $(I_n)_{1\leq n\leq p}$
  resp. $(C_n)_{1\leq n\leq p}$, are disjoint intervals of the real line $\R$
  resp. of $[0,1]$, then
  \begin{equation*}
    \lim_{L\to+\infty}
    \pro\left(\left\{\omega;\
        \begin{aligned}
          &\#\left\{j;
            \begin{aligned}
              x_{l,L}^\bullet(E,\omega)&\in I_1 \\
              y_{l,L}^\bullet(E,\omega)&\in C_1 \end{aligned}
          \right\}=k_1\\&\hskip1cm\vdots\hskip2cm\vdots\\
          &\#\left\{j;
            \begin{aligned}
              x_{l,L}^\bullet(E,\omega)&\in I_p \\
              y_{l,L}^\bullet(E,\omega)&\in C_p
            \end{aligned}
          \right\}=k_p
        \end{aligned}
      \right\}\right)=
    \prod_{n=1}^pe^{-\mu_n}
    \frac{(\mu_n)^{k_n}}{k_n!},
  \end{equation*}
  where $\mu_n:=|I_n||C_n|$ for $1\leq n\leq p$.
\end{Th}
\noindent This is the analogue of the celebrated result on the Poisson
structure of the eigenvalues and localization centers of a random
system (see, e.g.,~\cite{MR84e:34081,MR97d:82046,Ge-Kl:10}).\\
When considering the model for $\bullet=\Z$, Theorem~\ref{thr:7} is
Theorem~\ref{thr:23}.
\vskip.1cm\noindent In~\cite{MR2775121}, we proved decorrelation
estimates that can be used in the present setting to prove
\begin{Th}
  \label{thr:8}
  Fix $E\in(-2,2)\cap\overset{\circ}{\Sigma}$ and
  $E'\in(-2,2)\cap\overset{\circ}{\Sigma}$ such that $E\not=E'$,
  $n(E)>0$ and $n(E')>0$. Then, the limits of the processes
  $\xi^\bullet_L(E,\omega)$ and $\xi^\bullet_L(E',\omega)$ are
  stochastically independent.
\end{Th}
\noindent Due to the rescaling, the above results give only a picture
of the resonances in a zone of the type
\begin{equation}
  \label{eq:236}
  E+L^{-1}\left[-\varepsilon^{-1},\varepsilon^{-1}\right]
  -i\,\left[e^{-2\eta_\bullet(1+\varepsilon)\rho(E)L}
    ,e^{-2\varepsilon\eta_\bullet\rho(E)L}\right] 
\end{equation}
for $\varepsilon>0$ arbitrarily small.\\
When $L$ gets large, this rectangle is of a very small width and
located very close to the real axis. Theorems~\ref{thr:6},~\ref{thr:7}
and~\ref{thr:8} describe the resonances lying closest to the real
axis. As a comparison between points (1) and (2) in
Theorem~\ref{thr:6} shows, these resonances are the most numerous.\\
One can get a number of other statistics (e.g. the distribution of the
spacings between the resonances) using the techniques developed for
the study of the spectral statistics of a random system in the
localized phase (see~\cite{MR2885251,Ge-Kl:10,Kl:10a}) combined with
the analysis developed in section~\ref{sec:random-case-2}.
\subsubsection{The description of the low lying resonances}
\label{sec:descr-low-lying}
It is natural to question what happens deeper in the complex plane. To
answer this question, fix an increasing sequence of scales $(\ell_L)_L$
such that 
\begin{equation}
  \label{eq:14}
  \frac{\ell_L}{\log L}\vers_{L\to+\infty}+\infty
  \quad\text{ and }\quad \frac{\ell_L}L\vers_{L\to+\infty}0.
\end{equation}
We first show that there are only few resonances below the line
$\{$Im$\,z=e^{-\ell_L}\}$, namely
\begin{Th}
  \label{thr:27}
  Pick $(\ell_L)_L$ a sequence of scales satisfying~\eqref{eq:14} and
  $I$ as above.\\
  $\omega$ almost surely, for $L$ large, one has
  \begin{equation}
    \label{eq:193}
    \left\{z \text{ resonances of }H^\bullet_{\omega,L} \text{ in
      }\left\{\text{Re}\,z\in I,\ \text{Im}\,z\leq
        -e^{-\ell_L}\right\}\right\}=O(\ell_L).
  \end{equation}
\end{Th}
\noindent As we shall show now, after proper rescaling, the structure
of theses resonances is the same as that of the resonances closer to
the real axis.\\
Fix $E\in I$ so that $n(E)>0$. Recall that $(z^\bullet_{l,L}(\omega))_l$
be the resonances of $H_{\omega,L}$. We now rescale the resonances
using the sequence $(\ell_L)_L$; this rescaling will select resonances
that are further away from the real axis. Define
\begin{equation}
  \label{eq:12}
  \begin{aligned}
    x^\bullet_l&=x^\bullet_{l,\ell_L}(\omega)=
    n(E)\ell_L(\text{Re}\,z_{l,L}^\bullet(\omega)-E)
    \quad\text{and}\\
    y^\bullet_j&=y^\bullet_{l,\ell_L}(\omega)=
    \frac1{2\eta_\bullet\ell_L\rho(E)}
    \log|\text{Im}\,z_{l,L}^\bullet(\omega)|.
  \end{aligned}
\end{equation}
Consider now the two-dimensional point process
\begin{equation}
  \label{eq:13}
  \xi^\bullet_{L,\ell}(E,\omega)=
  \sum_{z^\bullet_{l,L}\text{ resonance of
    }H^\bullet_{\omega,L}}\delta_{(x^\bullet_{l,\ell_L},y^\bullet_{l,\ell_L})}.  
\end{equation}
We prove the following analogue of the results of
Theorems~\ref{thr:6},~\ref{thr:7} and~\ref{thr:8} for resonances lying
further away from the real axis.
\begin{Th}
  \label{thr:9}
  Fix $E\in(-2,2)\cap\overset{\circ}{\Sigma}$ and
  $E'\in(-2,2)\cap\overset{\circ}{\Sigma}$ such that $E\not=E'$,
  $n(E)>0$ and $n(E')>0$. Fix a sequence of scales $(\ell_L)_L$
  satisfying~\eqref{eq:14}. Then, one has
  \begin{enumerate}
  \item for $\lambda\in(0,1]$, $\omega$-almost surely
    \begin{equation*}
      \lim_{\delta\to0^+}\lim_{\varepsilon\to0^+}
      \lim_{L\to+\infty}\frac{\#\left\{z \text{ resonances of
          }H^\bullet_{\omega,L} \text{ in
          }R^\bullet(E,\lambda,\ell_L,\varepsilon,\delta)
        \right\}}{\ell_L\,\varepsilon\,\delta}=1
    \end{equation*}
    where $R^\bullet(E,\lambda,L,\varepsilon,\delta)$ is defined in
    Theorem~\ref{thr:6};
  \item the point processes $\xi^\bullet_{L,\ell}(E,\omega)$ and
    $\xi^\bullet_{L,\ell}(E',\omega)$ converge weakly to Poisson
    processes in $\R\times(0,+\infty)$ of intensity $1$;
  \item the limits of the processes $\xi^\bullet_{L,\ell}(E,\omega)$
    and $\xi^\bullet_{L,\ell}(E',\omega)$ are stochastically
    independent.
  \end{enumerate}
\end{Th}
\noindent Point (1) shows that, in~\eqref{eq:193}, one actually has
\begin{equation*}
  \left\{z \text{ resonances of }H^\bullet_{\omega,L} \text{ in
    }\left\{\text{Re}\,z\in I,\ \text{Im}\,z\leq
      -e^{-\ell_L}\right\}\right\}\asymp\ell_L.
\end{equation*}
Notice also that the effect of the scaling~\eqref{eq:12} is to select
resonances that live in the rectangle
\begin{equation*}
  E+\ell_L^{-1}\left[-\varepsilon^{-1},\varepsilon^{-1}\right]
  -i\,\left[e^{-2\eta_\bullet(1+\varepsilon)\rho(E)\ell_L}
    ,e^{-2\varepsilon\eta_\bullet\rho(E)\ell_L}\right] 
\end{equation*}
This rectangle is now much further away from the real axis than the
one considered in section~\ref{sec:descr-reson1}.\\
Modulo rescaling, the picture one gets for resonances in such
rectangles is the same one got above in the
rectangles~\eqref{eq:236}. This description is valid almost all the
way from distances to the real axis that are exponentially small in
$L$ up to distances that are of order $e^{-(\log L)^\alpha}$,
$\alpha>1$ (see~\eqref{eq:14}).
\subsubsection{Deep resonances}
\label{sec:deep-resonances}
One can also study the resonances that are even further away from the
real axis in a way similar to what was done in the periodic case in
section~\ref{sec:descr-reson-ll}. Define the following random
potentials on $\N$ and $\Z$
\begin{equation}
  \label{eq:200}
  \begin{aligned}
    \tilde V^\N_{\omega,L}(n)&= \begin{cases}
      \omega_{L-n}\text{ for }0\leq n\leq L\\
      0\text{ for }L+1\leq n
    \end{cases}\quad\text{ and}\\
    \tilde V^\Z_{\omega,\tilde\omega,L}(n)&= \begin{cases}
      0\text{ for }n\leq -1\\
      \tilde\omega_{n}\text{ for }0\leq n\leq [L/2]\\
      \omega_{L-n}\text{ for }[L/2]+1\leq n\leq L\\
      0\text{ for }L+1\leq n
    \end{cases}
  \end{aligned}
\end{equation}
where $\omega=(\omega_n)_{n\in \N}$ and
$\tilde\omega=(\tilde\omega_n)_{n\in \N}$ are i.i.d. and satisfy the
assumptions of the beginning of section~\ref{sec:random-case}.\\
Consider the operators
\begin{itemize}
\item $\tilde H^\N_{\omega,L}=-\Delta+\tilde V^\N_{\omega,L}$ on
  $\ell^2(\N)$ with Dirichlet boundary condition at $0$,
\item $\tilde H^\Z_{\omega,\tilde\omega,L}=-\Delta+\tilde
  V^\Z_{\omega,\tilde\omega,L}$ on $\ell^2(\Z)$.
\end{itemize}
Clearly, the random operator $\tilde H^\N_{\omega,L}$ (resp. $\tilde
H^\Z_{\omega,L}$) has the same distribution as $H^\N_{\omega,L}$
(resp.  $H^\Z_{\omega,L}$). Thus, for the low lying resonances, we are
now going to describe those of $\tilde H^\N_{\omega,L}$ (resp. $\tilde
H^\Z_{\omega,L}$) instead of those of $H^\N_{\omega,L}$ (resp.
$H^\Z_{\omega,L}$).
\begin{Rem}
  \label{rem:4}
  The reason for this change of operators is the same as the one why,
  in the case of the periodic potential, we had to distinguish various
  auxiliary operators depending on the congruence of $L$ modulo $p$,
  the period : this gives a meaning to the limiting operators when
  $L\to+\infty$.
\end{Rem}
\noindent Define the probability measure $dN_\omega(\lambda)$ using
its Borel transform by, for Im$ z\not=0$,
\begin{equation}
  \label{eq:114}
  \int_\R\frac{dN_\omega(\lambda)}{\lambda-z}:=\langle\delta_0,
  (H^\N_\omega-E)^{-1}\delta_0\rangle.
\end{equation}
Consider the function
\begin{equation}
  \label{eq:185}
  \Xi_\omega(E)=\int_\R\frac{dN_\omega(\lambda)}{\lambda-E}+e^{-i\arccos(E/2)}
  =\int_\R\frac{dN_\omega(\lambda)}{\lambda-E}+E/2+\sqrt{(E/2)^2-1}
\end{equation}
where the determinations of $z\mapsto\arccos z$ and
$z\mapsto\sqrt{z^2-1}$ are those described after~\eqref{eq:8}.\\
This random function $\Xi_\omega$ is the analogue of $\Xi_k^-$ in the
periodic case. One proves the analogue of Proposition~\ref{pro:4}
\begin{Pro}
  \label{pro:6}
  If $\omega_0\not=0$, one has $\D\Xi_\omega(E)\equ_{\substack{|E|\to\infty\\
      \text{Im}\,E<0}}-\omega_0\,E^{-2}$. Thus, $\omega$ almost
  surely, $\Xi_\omega$ does not vanish identically in $\{$Im$\,E<0\}$.\\
  Pick $I\subset\overset{\circ}{\Sigma}\cap(-2,2)$ compact. Then,
  $\omega$ almost surely, the number of zeros of $\Xi_\omega$ (counted
  with multiplicity) in $I+i\left(-\infty,\varepsilon\right]$ is
  asymptotic to $\D\int_I\frac{n(E)}{\rho(E)}dE\,|\log\varepsilon|$ as
  $ \varepsilon\to0^+$; moreover, $\omega$ almost surely, there exists
  $\varepsilon_\omega>0$ such that all the zeroes of $\Xi_\omega$ in
  $I+i[-\varepsilon_\omega,0)$ are simple.
\end{Pro}
\noindent It seems reasonable to believe that, except for the zero at
$-i\infty$, $\omega$ almost surely, all the zeros of $\Xi_\omega$ are
simple; we do not prove it\\
For the ``deep'' resonances, we then prove
\begin{Th}
  \label{thr:21}
  Fix $I\subset\overset{\circ}{\Sigma}\cap(-2,2)$ a compact
  interval. There exists $c>0$ such that, with probability 1, there
  exists $c_\omega>0$ such that, for $L$ sufficiently large, one has
  \begin{enumerate}
  \item for each resonance of $\tilde H^\N_{\omega,L}$ (resp. $\tilde
    H^\Z_{\omega,\tilde\omega,L}$) in $I+i\left(-\infty,-e^{-c
        L}\right]$, say $E$, there exists a unique zero of
    $\Xi_\omega$ (resp. $\Xi_\omega\,\Xi_{\tilde\omega}$), say $\tilde
    E$, such that $|E-\tilde E|\leq e^{-c_\omega L}$;
  \item reciprocally, to each zero (counted with multiplicity) of
    $\Xi_\omega$ (resp. $\Xi_\omega\,\Xi_{\tilde\omega}$) in the
    rectangle $I+i\left(-\infty,-e^{-c L}\right]$, say $\tilde E$, one
    can associate a unique resonance of $\tilde H^\N_{\omega,L}$
    (resp. $\tilde H^\Z_{\omega,\tilde\omega,L}$), say $E$, such that
    $|E-\tilde E|\leq e^{-c_\omega L}$.
  \end{enumerate}
\end{Th}
\noindent One can combine this result with the description of the
asymptotic distribution of the resonances given by Theorem~\ref{thr:9}
to obtain the asymptotic distributions of the zeros of the function
$\Xi_\omega$ near a point $E-i\varepsilon$ when
$\varepsilon\to0^+$. Indeed, let $(z_l(\omega))_l$ be the zeros of
$\Xi_\omega$ in $\{$Im$\,E<0\}$. Rescale the zeros:
\begin{equation}
  \label{eq:140}
  x_{l,\varepsilon}(\omega)
  =n(E)|\log\varepsilon|\cdot(\text{Re}\,z_l(\omega)-E)
  \quad\text{and}\quad y_{l,\varepsilon}(\omega)
  =-\frac1{2\rho(E)|\log\varepsilon|}\log|\text{Im}\,z_l(\omega)|
\end{equation}
and consider the two-dimensional point process
$\xi_\varepsilon(E,\omega)$ defined by
\begin{equation}
  \label{eq:209}
  \xi_\varepsilon(E,\omega)=\sum_{z_l(\omega)\text{ zeros of
    }\Xi_\omega} \delta_{(x_{l,\varepsilon},y_{l,\varepsilon})}.  
\end{equation}
Then, one has
\begin{Cor}
  \label{cor:3}
  Fix $E\in I$ such that $n(E)>0$. Then, the point process
  $\xi_\varepsilon(E,\omega)$ converges weakly to a Poisson process in
  $\R\times\R$ with intensity $1$.
\end{Cor}
\noindent The function $\Xi_\omega$ has been studied
in~\cite{MR2252707,PhysRevB.77.054203} where the average density of
its zeros was computed. Here, we obtain a more precise result.
\subsubsection{The half-line random perturbation}
\label{sec:half-line-rand-pert}
On $\ell^2(\Z)$, we now consider the operator
$H_\omega^\infty=-\Delta+V_\omega$ where $V_\omega(n)=0$ for $n\geq0$
and $V_\omega(n)=\omega_n$ for $n\leq-1$ and $(\omega_n)_{n\geq0}$ are
i.i.d. and have the same distribution as above. Recall that $\Sigma$
is the almost sure spectrum of $H_\omega^\Z$ (on $\ell^2(\Z)$). We
prove
\begin{Th}
  \label{thr:26}
  First, $\omega$ almost surely, the resolvent of $H_\omega^\infty$
  does not admit an analytic continuation from the upper half-plane
  through $(-2,2)\cap\overset{\circ}{\Sigma}$ to any subset of the
  lower half plane. Nevertheless, $\omega$-almost surely, the spectrum
  of $H^\infty_\omega$ in $(-2,2)\cap\overset{\circ}{\Sigma}$ is
  purely absolutely continuous.\\
  Second, $\omega$ almost surely, the resolvent of $H_\omega^\infty$
  does admit a meromorphic continuation from the upper half-plane
  through $(-2,2)\setminus\Sigma$ to the lower half plane; the poles
  of this continuation are exactly the zeros of the function $\D
  E\mapsto1-e^{i\theta(E)}\int_\R\frac{dN_\omega
    (\lambda)}{\lambda-E}$ when continued from the upper half-plane
  through $(-2,2)\setminus\Sigma$ to the lower half-plane.\\
  Third, $\omega$ almost surely, the spectrum of $H_\omega^\infty$ in
  $\overset{\circ}{\Sigma}\setminus[-2,2]$ is pure point associated to
  exponentially decaying eigenfunctions; hence, the resolvent of
  $H_\omega^\infty$ cannot be be continued through
  $\overset{\circ}{\Sigma}\setminus[-2,2]$.
\end{Th}
\noindent In figure~\ref{fig:12}, to illustrate Theorem~\ref{thr:26},
assuming that $\Sigma_\Z$ (in blue) has a single gap that is contained
in $(-2,2)$, we drew the analytic continuation of the resolvent of
$H_\omega^\infty$ and the associated resonances; we also indicate the
real intervals of spectrum through which the the resolvent of
$H_\omega^\infty$ does not admit an analytic continuation and the
spectral type of $H_\omega^\infty$ in the intervals.
\begin{figure}[ht]
  \centering
  \includegraphics[height=5.5cm]{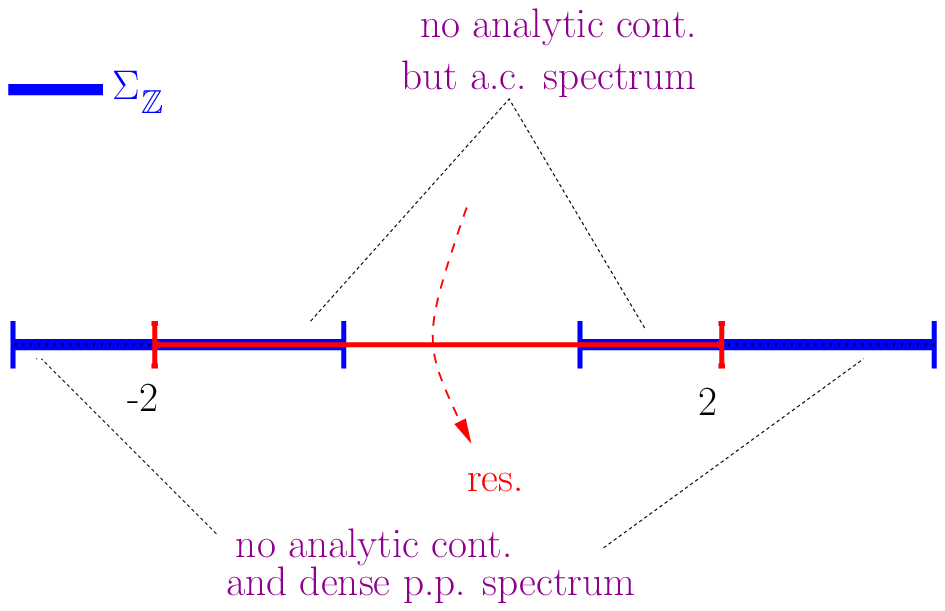}
  \caption{The analytic continuation of the resolvent and resonances
    for $H_\omega^\infty$}
  \label{fig:12}
\end{figure}
Let us also note here that if $0\in\,$supp$g$ (where $g$ is the
density of the random variables defining the random potential), then,
by~\eqref{eq:186}, one has $[-2,2]\subset\Sigma$. In this case, there
is no possibility to continue the resolvent of
$H_\omega^\infty$ to the lower half plane passing through $[-2,2]$. \\
Comparing Theorem~\ref{thr:26} to Theorem~\ref{thr:25}, we see that,
as the operator $H^\infty$, when continued through
$(-2,2)\cap\overset{\circ}{\Sigma}$, the operator $H^\infty_\omega$
does not have any resonances but for very different reasons.\\
When one does the continuation through $(-2,2)\setminus\Sigma$, one
sees that the number of resonances is finite; ``near'' the real axis,
the continuation of the function $\D
E\mapsto1-e^{i\theta(E)}\int_\R\frac{dN_\omega(\lambda)}{\lambda-E}$
has non trivial imaginary part and near $\infty$ it does not vanish.\\
Theorem~\ref{thr:26} also shows that the equation studied
in~\cite{MR2252707,PhysRevB.77.054203}, i.e., the equation
$\Xi_\omega(E)=0$, does not describe the resonances of
$H_\omega^\infty$ as is claimed in these papers: these resonances do
not exist as there is no analytic continuation of the resolvent of
$H_\omega^\infty$ through $(-2,2)\cap\Sigma$! As is shown in
Theorem~\ref{thr:21}, the solutions to the equation $\Xi_\omega(E)=0$
give an approximation to the resonances of $H^\N_{\omega,L}$ (see
Theorem~\ref{thr:21}).
\subsection{Outline of and reading guide to the paper}
\label{sec:outl-read-guide}
In the present section, we shall explain the main ideas leading to the
proofs of the results presented above.\\
In section~\ref{sec:proof-theorem}, we prove Theorem~\ref{thr:1}; this
proof is classical. As a consequence of the proof, one sees that, in
the case of the half-lattice $\N$ (resp. lattice $\Z$), the resonances
are the eigenvalues of a rank one (resp. two) perturbation of
$(-\Delta+V)_{|\llbracket 0,L\rrbracket}$ with Dirichlet b.c. The
perturbation depends in an explicit way on the resonance. This yields
a closed equation for the resonances in terms of the eigenvalues and
normalized eigenfunctions of the Dirichlet restriction
$(-\Delta+V)_{|\llbracket 0,L\rrbracket}$. To obtain a description of
the resonances we then are in need of a ``precise'' description of the
eigenvalues and normalized eigenfunctions. Actually the only
information needed on the normalized eigenfunctions is their weight at
the point $L$ (and the point $0$ in the full lattice case), $0$ and
$L$ being the endpoints of $\llbracket 0,L\rrbracket$. \\
In section~\ref{sec:char-reson}, we solve the two equations obtained
previously under the condition that the weight of the normalized
eigenfunctions at $L$ (and $0$) be much smaller than the spacing
between the Dirichlet eigenvalues. This condition entails that the
resonance equation we want to solve essentially factorizes and become
very easy to solve (see Theorems~\ref{thr:13},~\ref{thr:14}
and~\ref{thr:12}), i.e., it suffices to solve it near any given
Dirichlet eigenvalue.
\par For periodic potentials, the condition that the eigenvalue
spacing is much larger than the weight of the normalized
eigenfunctions at $L$ (and $0$) is not satisfied: both quantities are
of the same order of magnitude (see Theorem~\ref{thr:16}) for the
Dirichlet eigenvalues in the bulk of the spectrum, i.e., the vast
majority of them. This is a consequence of the extended nature of the
eigenfunctions in this case.  Therefore, we find another way to solve
the resonance equation. This way goes through a more precise
description of the Dirichlet eigenvalues and normalized eigenfunctions
which is the purpose of Theorems~\ref{thr:16}. We use this description
to reduce the resonance equation to an effective equation (see
Theorem~\ref{thr:17}) up to errors of order $O(L^{-\infty})$. It is
important to obtain errors of at most that size. Indeed, the effective
equation may have solutions to any order (the order is finite and only
depends on $V$ but it is unknown); thus, to obtain solutions to the
true equation from solutions to the effective equation with a good
precision, one needs the two equations to differ by at most
$O(L^{-\infty})$. We then solve the effective equation and, in
section~\ref{sec:proof-theorem-3}, prove the results of
section~\ref{sec:periodic-case}.
\par On the other hand, for random potentials, it is well known that
the eigenfunctions of the Dirichlet restriction
$(-\Delta+V)_{|\llbracket 0,L\rrbracket}$ are exponentially localized
and, for most of them localized, far from the edge of $\llbracket
0,L\rrbracket$. Thus, their weight at $L$ (and $0$ in the full lattice
case) is typically exponentially small in $L$; the eigenvalue spacing
however is typically of order $L^{-1}$. We can then use the results of
section~\ref{sec:char-reson} to solve the resonance equation. The real
part of a given resonance is directly related to a Dirichlet
eigenvalue and its imaginary part to the weight of the corresponding
eigenfunction at $L$ (and $0$ in the full lattice case). The main
difficulty is to find the asymptotic behavior of this weight. Indeed,
while it is known that, in the random case, eigenfunctions decay
exponentially away from a localization center and while it is known
that, for the full random Hamiltonian (i.e. the Hamiltonian on the
line or half-line with a random potential), at infinity, this decay
rate is given by the Lyapunov exponent, to the best of our knowledge,
before the present work, it was not known at which length scale this
Lyapunov behavior sets in (with a good probability). Answering this
question is the purpose of Theorems~\ref{thr:10} and~\ref{thr:20}
proved in section~\ref{sec:estim-growth-eigenf}: we show that, for the
$1$-dimensional Anderson model, for $\delta>0$ arbitrary, on a box of
size $L$ sufficiently large, all the eigenfunctions exhibit an
exponential decay (we obtain both an upper and a lower bound on the
eigenfunctions) at a rate equal to the Lyapunov exponent at the
corresponding energy (up to an error of size $\delta$) as soon as one
is at a distance $\delta L$ from the corresponding localization center. \\
These bounds give estimates on the weight of most eigenfunctions at
the point $L$ (and $0$ in the full lattice case): it is directly
related to the distance of the corresponding localization center to
the points $L$ (and $0$). One can then transform the known results on
the statistics of the (rescaled) eigenvalues and (rescaled)
localization centers into statistics of the (rescaled) resonances.
This is done in section~\ref{sec:proofs-theorems} and proves most of
the results in section~\ref{sec:random-case}.\\
Finally, section~\ref{sec:half-line-rand} is devoted to the study of
the full line Hamiltonian obtained from the free Hamiltonian on one
half-line and a random Hamiltonian on the other half-line; it contains
in particular the proof of Theorem~\ref{thr:26}.
\section{The analytic continuation of the resolvent}
\label{sec:proof-theorem}
Resonances for Jacobi matrices were considered in various works (see,
e.g.,~\cite{MR2164835,MR2869822} and references tehrein). For the sake
of completeness, we provide an independent proof of
Theorem~\ref{thr:1}. It follows standard ideas that were first applied
in the continuum setting, i.e., for partial differential operators
instead of finite difference operators (see, e.g.,~\cite{MR1115789} and
references therein).\\
The proof relies on the fact that the resolvent of free Laplace
operator can be continued holomorphically from $\C^+$ to
$\D\C\setminus\left((-\infty,-2]\cup[2,+\infty)\right)$ as an operator
valued function from $l^2_\text{comp}$ to $l^2_\text{loc}$. This is an
immediate consequence of the fact that, by discrete Fourier
transformation, $-\Delta$ is the Fourier multiplier by the function
$\theta\mapsto2\cos\theta$.\\
Indeed, for $-\Delta$ on $\ell^2(\Z)$ and Im$\,E>0$, one has, for
$(n,m)\in\Z$ (assume $n-m\geq0$)
\begin{equation}
  \label{eq:166}
  \begin{split}
    \langle\delta_n,(-\Delta-E)^{-1}\delta_m\rangle&=\frac1{2\pi}
    \int_0^{2\pi}\frac{e^{-i(n-m)\theta}}{2\cos\theta-E}d\theta
    =\frac1{2i\pi}\int_{|z|=1}\frac{z^{n-m}}{z^2-Ez+1}dz\\
    &=\frac1{2\sqrt{(E/2)^2-1}}\left(E/2-\sqrt{(E/2)^2-1}\right)^{n-m}
    =\frac{e^{i(n-m)\theta(E)}}{\sin\theta(E)}
  \end{split}
\end{equation}
where $E=2\cos\,\theta(E)$ and the determination $\theta=\theta(E)$ is
chosen so that Im$\,\theta>0$ and Re$\,\theta\in(-\pi,0)$ for
Im$\,E>0$. The determination satisfies
$\D\theta\left(\overline{E}\right)=\overline{\theta(E)}$.\\
The map $E\mapsto\theta(E)$ can continued analytically from $\C^+$ to
the cut plane $\D\C\setminus\left((-\infty,-2]\cup[2,+\infty)\right)$
as shown in Figure~\ref{fig:5}.
\begin{figure}[h]
  \centering
  \includegraphics[height=2.5cm]{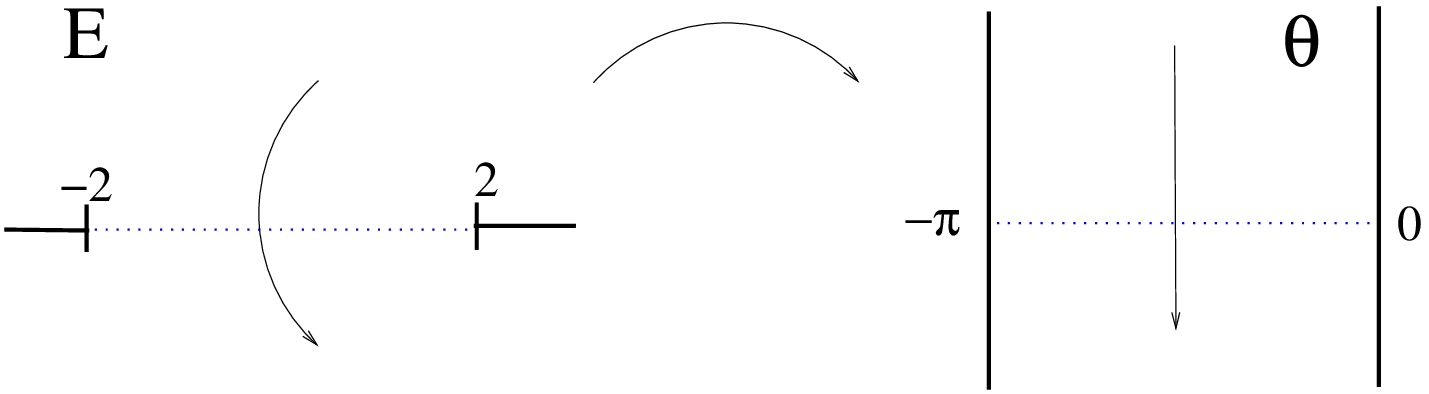}
  \caption{The mapping $E\mapsto\theta(E)$}
  \label{fig:5}
\end{figure}
\\ The continuation is one-to-one and onto from
$\D\C\setminus\left((-\infty,-2]\cup[2,+\infty)\right)$ to
$(-\pi,0)+i\R$. It defines a determination of
$E\mapsto\arccos(E/2)=\theta(E)$. \\
Clearly, using~\eqref{eq:166}, this continuation yields an analytic
continuation of $R_0^{\Z}:=(-\Delta-E)^{-1}$ from $\{$Im$\,E>0\}$ to
$\D\C\setminus\left((-\infty,-2]\cup[2,+\infty)\right)$ as an operator
from $l^2_\text{comp}$ to $l^2_\text{loc}$. \vskip.1cm\noindent
Let us now turn to the half-line operator, i.e., $-\Delta$ on $\N$ with
Dirichlet condition at $0$. Pick $E$ such that Im$\,E>0$ and set
$E=2\cos\theta$ where the determination $\theta=\theta(E)$ is chosen
as above. If for $v\in\C^\N$ bounded and $n\geq-1$, one sets
$v_{-1}=0$ and
\begin{equation}
  \label{eq:15}
  [R_0^\N(E)(v)]_n=\frac1{2i\sin\theta(E)}\sum_{j=-1}^n
  v_j\cdot\sin((n-j)\theta(E))
  -e^{i\theta(E)}\frac{\sin((n+1)\theta(E))}{2i\sin\,\theta(E)}
  \sum_{j\geq0}e^{ij\theta(E)}v_j.
\end{equation}
Then, for Im$\,E>0$, a direct computations shows that
\begin{enumerate}
\item for $v\in\ell^2(\N)$, the vector $R_0^\N(E)(v)$ is in the domain
  of the Dirichlet Laplacian on $\ell^2(\N)$
 , i.e., $[R_0^\N(E)(v)]_{-1}=0$;
\item for $n\geq0$, one checks that
  \begin{equation}
    \label{eq:16}
    [R_0^\N(E)(v)]_{n+1}+[R_0^\N(E)(v)]_{n-1}-E[R_0^\N(E)(v)]_n=v_n.
  \end{equation}
  %
%
\item $R_0^\N(E)$ defines a bounded map from $\ell^2(\N)$ to itself;
\end{enumerate}
Thus, $R_0^\N(E)$ is the resolvent of the Dirichlet Laplacian on $\N$
at energy $E$ for Im$\,E>0$.
Using the continuation of $E\mapsto\theta(E)$, formula~\eqref{eq:15}
yields an analytic continuation of the resolvent $R_0^\N(E)$ as an
operator from $l^2_\text{comp}$ to $l^2_\text{loc}$.
\begin{Rem}
  \label{rem:5}
  Note that the resolvent $R_0^\N(E)$ at an energy $E$ s.t. Im$\,E<0$
  is given by formula~\eqref{eq:15} where $\theta(E)$ is replaced by
  $-\theta(E)$. For~\eqref{eq:15}, one has to assume that
  $(v_j)_{j\in\N}$ decays fast enough at $\infty$.
\end{Rem}
\noindent To deal with the perturbation $V$, we proceed in the same
way on $\Z$ and on $\N$.  Set $V^L=V\car_{\llbracket 0,L\rrbracket}$
(seen as a function on $\N$ or $\Z$ depending on the case).  Letting
$R_0(E)$ be either $R_0^\Z(E)$ or $R_0^\N(E)$, we compute
\begin{equation*}
  -\Delta+V^L-E=(-\Delta-E)(1+R_0(E)V^L)=(1+V^LR_0(E))(-\Delta_L-E).
\end{equation*}
Thus, it suffices to check that the operator $R_0(E)V^L$
(resp. $V^LR_0(E)$) can be analytically continued as an operator from
$l^2_\text{loc}$ to $l^2_\text{loc}$ (resp. $l^2_\text{comp}$ to
$l^2_\text{comp}$). This follows directly from~\eqref{eq:15} and the
fact $V^L$ has finite rank.\\
To complete the proof of Theorem~\ref{thr:1}, we just note that, as
\begin{itemize}
\item $E\mapsto R_0(E)V^L$ (resp. $E\mapsto V^LR_0(E)$) is a finite
  rank operator valued function analytic on the connected set
  $\D\C\setminus\left((-\infty,-2]\cup[2,+\infty)\right)$,
\item $-1$ is not an eigenvalue of $R_0(E)V^L$ (resp. $V^LR_0(E)$) for
  Im$\,E>0$,
\end{itemize}
by the Fredholm principle, the set of energies $E$ for which $-1$ is
an eigenvalue of $R_0(E)V^L$ (resp. $V^LR_0(E)$) is discrete. Hence,
the set of resonances is discrete.\\
This completes the proof of the first part of Theorem~\ref{thr:1}. To
prove the second part, we will first write a characteristic equation
for resonances. The bound on the number of resonances will then be
obtained through a bound on the number of solutions to this equation.
\subsection{A characteristic equation for resonances}
\label{sec:char-equat-reson}
In the literature, we did not find a characteristic equation for the
resonances in a form suitable for our needs. The characteristic
equation we derive will take different forms depending on whether we
deal with the half-line or the full line operator. But in both cases,
the coefficients of the characteristic equation will be constructed
from the spectral data (i.e. the eigenvalues and eigenfunctions) of
the operator $H_L$ (see Remark~\ref{rem:6}).
\subsection{In the half-line case}
\label{sec:half-line-case}
We first consider $H^{\N}_L$ on $\ell^2(\N)$ and prove
\begin{Th}
  \label{thr:11}
  Consider the operator $H_L$ defined as $H^{\N}_L$ restricted to
  $\llbracket 0,L\rrbracket$ with Dirichlet boundary conditions at $L$
  and define
  \begin{itemize}
  \item $(\lambda_j)_{0\leq j\leq L}=(\lambda_j(L))_{0\leq j\leq L}$ are the Dirichlet
    eigenvalues of $H^{\N}_L$ ordered so that
    $\lambda_j<\lambda_{j+1}$;
  \item $a^\N_j=a^\N_j(L)=|\varphi_j(L)|^2$ where
    $\varphi_j=(\varphi_j(n))_{0\leq n\leq L}$ is a normalized eigenvector
    associated to $\lambda_j$.
  \end{itemize}
  Then, an energy $E$ is a resonance of $H^\N_L$ if and only if
  \begin{equation}
    \label{eq:1}
    S_L(E):=\sum_{j=0}^L\frac{a^\N_j}{\lambda_j-E}=-e^{-i\theta(E)},\quad
    E=2\cos\theta(E),
  \end{equation}
  the determination of $\theta(E)$ being chosen so that
  Im$\,\theta(E)>0$ and Re$\,\theta(E)\in(-\pi,0)$ when Im$\,E>0$.
\end{Th}
\noindent Let us note that
\begin{equation}
  \label{eq:20}
  \forall0\leq j\leq L,\ a^\N_j(L)>0 \quad\text{and}\quad
  \sum_{j=0}^La^\N_j(L)=\sum_{j=0}^L|\varphi_j(L)|^2=1.
\end{equation}
\begin{proof}[Proof of Theorem~\ref{thr:11}]
  By the proof of the first statement of Theorem~\ref{thr:1} (see the
  beginning of section~\ref{sec:proof-theorem}), we know that an
  energy $E$ is a resonance if and only if $-1$ if an eigenvalue of
  $R_0(E)V^L$ where $R_0(E)$ is defined by~\eqref{eq:15}. Pick $E$ an
  resonance and let $u=(u_n)_{n\geq0}$ be a resonant state that is an
  eigenvector of $R_0(E)V^L$ associated to the eigenvalue $-1$. As
  $V^L_n=0$ for $n\geq L+1$, equation~\eqref{eq:15} yields that, for
  $n\geq L+1$, $u_n=\beta e^{in\theta(E)}$ for some fixed
  $\beta\in\C^*$. As $u=-R_0(E)V^Lu$, for $n\geq L+1$, it satisfies
  $u_{n+1}+u_{n-1}=E\,u_n$.  Thus, $u_{L+1}=e^{i\theta(E)}u_L$ and
  by~\eqref{eq:16}, $u$ is a solution to the eigenvalues problem
  \begin{equation*}
    \begin{cases}
      u_{n+1}+u_{n-1}+V_nu_n=E\,u_n,\ \forall n\in\llbracket
      0,L\rrbracket \\ u_{-1}=0,\quad u_{L+1}=e^{i\theta(E)}u_L
    \end{cases}
  \end{equation*}
  This can be equivalently be rewritten as
  \begin{equation}
    \label{eq:19}
    \begin{pmatrix}
      V_0 & 1 & 0 & \cdots &0 \\
      1 & V_1 & 1 & 0 & &\\
      \vdots & \ddots &\ddots &\ddots &\\
      &0 &1 & V_{L-1} &1\\
      0 & \cdots &0 & 1 &V_L+e^{i\theta(E)} \\
    \end{pmatrix}
    \begin{pmatrix}
      u_0\\ \\ \vdots\\ \\ u_L
    \end{pmatrix}= E\begin{pmatrix} u_0\\ \\ \vdots \\ \\ u_L
    \end{pmatrix}
  \end{equation}
  The matrix in~\eqref{eq:19} is the Dirichlet restriction of
  $H^{\N}_L$ to $\llbracket 0,L\rrbracket$ perturbed by the rank one
  operator $e^{i\theta(E)}\delta_L\otimes\delta_L$. Thus, by rank one
  perturbation theory (see, e.g.,~\cite{MR97c:47008}), an energy $E$ is
  a resonance if and only if if satisfies~\eqref{eq:1}.\\
  This completes the proof of Theorem~\ref{thr:11}.
\end{proof}
\noindent Let us now complete the proof of Theorem~\ref{thr:1} for the
operator on the half-line. Let us first note that, for Im$\,E>0$, the
imaginary part of the left hand side of~\eqref{eq:1} is positive
by~\eqref{eq:18}. On the other hand, the imaginary part of the right
hand side of~\eqref{eq:1} is equal to
$-e^{\text{Im}\,\theta(E)}\sin($Re$\,\theta(E))$ and, thus, is
negative (recall that Re$\,\theta(E)\in(-\pi,0)$ (see
fig.~\ref{fig:1}). Thus, as already underlined, equation~\eqref{eq:1}
has no solution in the upper half-plane or on the interval $(-2,2)$.\\
Clearly, equation~\eqref{eq:1} is equivalent to the following
polynomial equation of degree $2L+2$ in the variable
$z=e^{-i\theta(E)}$
\begin{equation}
  \label{eq:18}
  \prod_{k=0}^L\left(z^2-2\lambda_k z+1\right)-
  \sum_{j=0}^La^\N_j\prod_{\substack{0\leq k\leq L\\k\not=j}}
  \left(z^2-2\lambda_k z+1\right)=0.
\end{equation}
We are looking for the solutions to~\eqref{eq:18} in the upper
half-plane. As the polynomial in the right hand side of~\eqref{eq:18}
has real coefficients, its zeros are symmetric with respect to the
real axis. Moreover, one notices that, by~\eqref{eq:20}, $0$ is a
solution to~\eqref{eq:18}. Hence, the number of solutions
to~\eqref{eq:18} in the upper half-plane is bounded by $L$. This
completes the proof of Theorem~\ref{thr:1}.
\subsection{On the whole line}
\label{sec:whole-line}
Now, consider $H^\Z_L$ on $\ell^2(\Z)$. We prove
\begin{Th}
  \label{thr:24}
  Using the notations of Theorem~\ref{thr:11}, an energy $E$ is a
  resonance of $H^\Z_L$ if and only if
  \begin{equation}
    \label{eq:135}
    \text{det}
    \left(\sum_{j=0}^L\frac1{\lambda_j-E}
      \begin{pmatrix} |\varphi_j(L)|^2& \overline{\varphi_j(0)}
        \varphi_j(L) \\ \varphi_j(0) \overline{\varphi_j(L)} &
        |\varphi_j(0)|^2
      \end{pmatrix}+e^{-i\theta(E)}\right)=0
  \end{equation}
  where $det(\cdot)$ denotes the determinant of a square matrix,
  $E=2\cos\theta(E)$ and the determination of $\theta(E)$ is chosen as
  in Theorem~\ref{thr:11}.
\end{Th}
\noindent So, an energy $E$ is a resonance of $H^\Z_L$ if and only if
$-e^{-i\theta(E)}$ belongs to the spectrum of the $2\times2$ matrix
\begin{equation}
  \label{eq:142}
  \Gamma_L(E):=\sum_{j=0}^L\frac1{\lambda_j-E}
  \begin{pmatrix} |\varphi_j(L)|^2& \overline{\varphi_j(0)}
    \varphi_j(L) \\ \varphi_j(0) \overline{\varphi_j(L)} &
    |\varphi_j(0)|^2
  \end{pmatrix}.
\end{equation}
\begin{proof}[Proof of Theorem~\ref{thr:24}]
  The proof is the same as that of Theorem~\ref{thr:11} except that
  now $E$ is a resonance if there exists $u$ a non trivial solution to
  the eigenvalues problem
  \begin{equation*}
    \begin{cases}
      u_{n+1}+u_{n-1}+V_nu_n=E\,u_n,\ \forall n\in\llbracket
      0,L\rrbracket \\ u_{-1}=e^{i\theta(E)}u_{0}\text{ and }
      u_{L+1}=e^{i\theta(E)}u_L
    \end{cases}
  \end{equation*}
  This can be equivalently be rewritten as
  \begin{equation*}
    \begin{pmatrix}
      V_0+e^{i\theta(E)} & 1 & 0 & \cdots &0 \\
      1 & V_1 & 1 & 0 & &\\
      \vdots & \ddots &\ddots &\ddots &\\
      &0 &1 & V_{L-1} &1\\
      0 & \cdots &0 & 1 &V_L+e^{i\theta(E)} \\
    \end{pmatrix}
    \begin{pmatrix}
      u_{0}\\ \\ \vdots\\ \\ u_L
    \end{pmatrix}= E\begin{pmatrix} u_{0}\\ \\ \vdots \\ \\ u_L
    \end{pmatrix}
  \end{equation*}
  Thus, using rank one perturbations twice, we find that an energy $E$
  is a resonance if and only if
  \begin{equation*}
    \left( 1+e^{i\theta(E)}
      \sum_{j=0}^L\frac{|\varphi_j(0)|^2}{\lambda_j-E} \right)
    \left(1+e^{i\theta(E)}
      \sum_{j=0}^L\frac{|\varphi_j(L)|^2}{\lambda_j-E}
    \right)\\=e^{2i\theta(E)}\sum_{0\leq j,j'\leq L} \frac{\varphi_j(L)
      \varphi_{j'}(0) \overline{\varphi_{j'}(L)\varphi_j(0)}}
    {(\lambda_j-E)(\lambda_{j'}-E)},
  \end{equation*}
  that is, if and only is~\eqref{eq:135} holds. This completes the
  proof of Theorem~\ref{thr:24}.
\end{proof}
\noindent Let us now complete the proof of Theorem~\ref{thr:1} for the
operator on the full-line. Let us first show that~\eqref{eq:135} has
no solution in the upper half-plane. Therefore, if $-e^{-i\theta(E)}$
belongs to the spectrum of the matrix defined by~\eqref{eq:135} and if
$u\in\C^2$ is a normalized eigenvector associated to
$-e^{-i\theta(E)}$, one has
\begin{equation*}
  \sum_{j=0}^L\frac1{\lambda_j-E}
  \left|\left\langle\begin{pmatrix} \varphi_j(L)\\
        \varphi_j(0) \end{pmatrix}
      ,u\right\rangle\right|^2=-e^{-i\theta(E)}.
\end{equation*}
This is impossible in the upper half-plane and on $(-2,2)$ as the two
sides of the equation have imaginary parts of opposite signs.\\
Note that
\begin{equation*}
  \sum_{j=0}^L\begin{pmatrix} \varphi_j(L)\\
    \varphi_j(0) \end{pmatrix}\overline{\begin{pmatrix}\varphi_j(L) &
      \varphi_j(0)\end{pmatrix}}=\begin{pmatrix} 1&0\\0&1 \end{pmatrix}.
\end{equation*}
Note also that $-e^{-i\theta(E)}$ is an eigenvalue of~\eqref{eq:135}
if and only if it satisfies
\begin{equation}
  \label{eq:122}
  1+e^{i\theta(E)}
  \sum_{j=0}^L\frac{|\varphi_j(L)|^2+|\varphi_j(0)|^2}{\lambda_j-E}
  =-\frac12e^{2i\theta(E)}\sum_{0\leq j,j'\leq L}\frac{\left|
      \begin{matrix}\varphi_j(0)&\varphi_{j'}(0)\\
        \varphi_j(L)&\varphi_{j'}(L) \end{matrix}\right|^2}
  {(\lambda_j-E)(\lambda_{j'}-E)}.
\end{equation}
As the eigenvalues of $H_L$ are simple, one computes
\begin{equation}
  \label{eq:137}
  \sum_{0\leq j,j'\leq L}\frac{\left|
      \begin{matrix}\varphi_j(0)&\varphi_{j'}(0)\\
        \varphi_j(L)&\varphi_{j'}(L) \end{matrix}\right|^2}
  {(\lambda_j-E)(\lambda_{j'}-E)}= 2\sum_{0\leq j\leq
    L}\frac{1}{\lambda_j-E}
  \sum_{j'\not=j}\frac1{\lambda_{j'}-\lambda_j}\left|
    \begin{matrix}\varphi_j(0)&\varphi_{j'}(0)\\
      \varphi_j(L)&\varphi_{j'}(L) \end{matrix}\right|^2.
\end{equation}
Thus, equation~\eqref{eq:122} is equivalent to the following
polynomial equation of degree $2(L+1)$ in the variable
$z=e^{-i\theta(E)}$
\begin{equation}
  \label{eq:138}
  z\prod_{k=0}^L\left(z^2-\lambda_k z+1\right)-
  \sum_{j=0}^L(2a^\Z_j z+b^\Z_j)\prod_{\substack{0\leq k\leq L\\k\not=j}}
  \left(z^2-\lambda_k z+1\right)=0.
\end{equation}
where we have defined
\begin{equation}
  \label{eq:145}
  \begin{split}
    a^\Z_j&:=\frac12\left(|\varphi_j(L)|^2+|\varphi_j(0)|^2\right)
    =\frac12 \left\| \begin{pmatrix}
        \varphi_j(L)\\\varphi_j(0)\end{pmatrix}\right\|^2 =\frac12
    \left\| \begin{pmatrix} |\varphi_j(L)|^2& \overline{\varphi_j(0)}
        \varphi_j(L) \\ \varphi_j(0) \overline{\varphi_j(L)} &
        |\varphi_j(0)|^2 \end{pmatrix} \right\|.
  \end{split}
\end{equation}
and
\begin{equation*}
  b^\Z_j:=\sum_{j'\not=j}\frac1{\lambda_{j'}-\lambda_j}\left|
    \begin{matrix}\varphi_j(0)&\varphi_{j'}(0)\\
      \varphi_j(L)&\varphi_{j'}(L) \end{matrix}\right|^2.
\end{equation*}
The sequence $(a^\Z_j)_j$ also satisfies~\eqref{eq:20}.  Taking $|E|$
to $+\infty$ in~\eqref{eq:137}, one notes that
\begin{equation}
  \label{eq:98}
  \sum_{j=0}^Lb^\Z_j=0\quad\text{and}\quad\sum_{j=0}^L\lambda_j b^\Z_j=
  -\frac12\sum_{0\leq j,j'\leq L}\left|
    \begin{matrix}\varphi_j(0)&\varphi_{j'}(0)\\
      \varphi_j(L)&\varphi_{j'}(L) \end{matrix}\right|^2=-1.
\end{equation}
We are looking for the solutions to~\eqref{eq:138} in the upper
half-plane. As the polynomial in the right hand side of~\eqref{eq:138}
has real coefficients, its zeros are symmetric with respect to the
real axis. Moreover, one notices that, by~\eqref{eq:98}, $0$ is a root
of order two of the polynomial in~\eqref{eq:138}. Hence, as the
polynomial has degree $2L+3$, the number of solutions
to~\eqref{eq:138} in the upper half-plane is bounded by $L$. This
completes the proof of Theorem~\ref{thr:1}.

\section{General estimates on resonances}
\label{sec:char-reson}
By Theorems~\ref{thr:11} and~\ref{thr:24}, we want to solve
equations~\eqref{eq:1} and~\eqref{eq:135} in the lower half-plane. We
first derive some general estimates for zones in the lower half-plane
free of solutions to equations~\eqref{eq:1} and~\eqref{eq:135} (i.e.
resonant free zones for the operators $H^{\N}_L$ and $H^\Z_L$) and
later a result on the existence of solutions to equations~\eqref{eq:1}
and~\eqref{eq:135} (i.e. resonances for the operators $H^{\N}_L$ and
$H^\Z_L$).
\subsection{General estimates for resonant free regions}
\label{sec:gener-estim-reson}
We keep the notations of Theorems~\ref{thr:11} and~\ref{thr:24}. To
simplify the notations in the theorems of this section, we will write
$a_j$ for either $a^\N_j$ when solving~\eqref{eq:1} or $a^\Z_j$ when
solving~\eqref{eq:135}. We will specify the superscript only when
there is risk of confusion.\\
We first prove
\begin{Th}
  \label{thr:13}
  Fix $\delta>0$. Then, there exists $C>0$ (independent of $V$ and
  $L$) such that, for any $L$ and $j\in\{0,\cdots,L\}$ such that
  $-4+\delta\leq \lambda_{j-1}+\lambda_j<\lambda_{j+1}+\lambda_j
  \leq4-\delta$, equations~\eqref{eq:1} and~\eqref{eq:135} have no
  solution in the set
  \begin{equation}
    \label{eq:22}
    U_j:=\left\{E\in\C;\
      \begin{aligned}
        \text{Re}\,E&\in\left[\frac{\lambda_j+
            \lambda_{j-1}}2,\frac{\lambda_j+\lambda_{j+1}}2\right]\\
        \quad 0\geq C\cdot\theta'_\delta\,&\text{Im}\,E>-a_j\,d^2_j\,
        |\sin\,\text{Re}\;\theta(E)|
      \end{aligned}
    \right\}
  \end{equation}
  where the map $E\mapsto\theta(E)$ is defined in
  section~\ref{sec:proof-theorem} and we have set
  \begin{equation}
    \label{eq:21}
    d_j:=\min\left(\lambda_{j+1}-\lambda_j,\lambda_j-\lambda_{j-1},1\right)\quad
    \text{and}\quad\theta'_\delta:=\max_{|E|\leq2-\frac\delta2}|\theta'(E)|.
  \end{equation}
\end{Th}
\begin{wrapfigure}{r}{.45\textwidth}
  \centering
  \includegraphics[width=.30\textwidth]{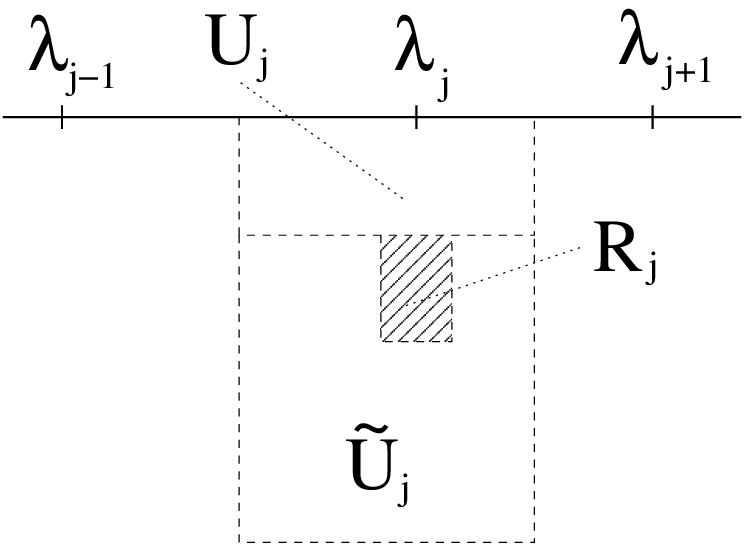}
  \caption{The resonance free zones $U_j$ and $\tilde U_j$.}
  \label{fig:6}
\end{wrapfigure}
\noindent In Theorem~\ref{thr:13} there are no conditions on the
numbers $(a_j)_j$ or $(d_j)_j$ except their being positive. In our
application to resonances, this holds. Theorem~\ref{thr:13} becomes
optimal when $a_j\ll d^2_j$. In our application to resonances, for
periodic operators, one has $a_j\asymp L^{-1}$ and $d_j\asymp L^{-1}$
(see Theorem~\ref{thr:19}) and for random operators, one has
$a_j\asymp e^{-cL}$ and $d_j\gtrsim L^{-4}$ (see Theorem~\ref{thr:10}
and~\eqref{eq:105}). Thus, in the random case, Theorem~\ref{thr:13}
will provide an optimal strip free of resonances whereas in the
periodic case we will use a much more precise computation (see
Theorem~\ref{thr:17}) to obtain sharp
results.\\
When $a_j\ll d_j^2$, one proves the existence of another resonant free
region near a energy $\lambda_j$, namely,
\begin{Th}
  \label{thr:14}
  Fix $\delta>0$. Pick $j\in\{0,\cdots,L\}$ such that
  $-4+\delta<\lambda_{j-1}+\lambda_j<\lambda_{j+1}+\lambda_j<4-\delta$.
  There exists $C>0$ (depending only on $\delta$) such that, for any
  $L$, if $a_j\leq d^2_j/C^2$, equations~\eqref{eq:1} and~\eqref{eq:135}
  have no solution in the set
  \begin{equation}
    \label{eq:30}
    \begin{split}
      \tilde U_j&:=\left\{E\in\C;\
        \begin{split}
          \text{Re}\,E&\in\left[\frac{\lambda_j+
              \lambda_{j-1}}2,\lambda_j-Ca_j\right]\cup
          \left[\lambda_j+Ca_j,\frac{\lambda_j+\lambda_{j+1}}2\right]\\
          \quad &-Ca_j\leq \text{Im}\,E\leq -a_j\,d^2_j/C
        \end{split}
      \right\}\\&\hskip3cm\bigcup\left\{E\in\C;\
        \begin{split}
          \text{Re}\,E&\in\left[\frac{\lambda_j+
              \lambda_{j-1}}2,\frac{\lambda_j+\lambda_{j+1}}2\right]\\
          \quad &-d^2_j/C\leq\text{Im}\,E\leq -Ca_j
        \end{split}
      \right\}
    \end{split}
  \end{equation}
\end{Th}
\noindent Theorem~\ref{thr:14} becomes optimal when $a_j$ is small and
$d_j$ is of order one. This will be sufficient to deal with the
isolated eigenvalues for both the periodic and the random
potential. It will also be sufficient to give a sharp description of
the resonant free region for random potentials. For the periodic
potential, we will rely a much more precise computations (see
Theorem~\ref{thr:17}). \\
Note that Theorem~\ref{thr:14} guarantees that, if $d_j$ is not too
small, outside $R_j$, resonances are quite far below the real axis.
\begin{proof}[Proof of Theorem~\ref{thr:13}]
  The basic idea of the proof is that, for $E$ close to $\lambda_j$,
  $S_L(E)$ and the matrix $\Gamma_L(E)$ are either large or have a
  very small imaginary part while, as
  $-4<\lambda_{j-1}+\lambda_j<\lambda_{j+1}+\lambda_j<4$,
  $e^{-i\theta(E)}$ has a large imaginary part. Thus,~\eqref{eq:1}
  and~\eqref{eq:135} have no solution in this region.\\
  We start with equation~\eqref{eq:1}. Pick $E\in U_j$ for some $C$ large
  to be chosen later on. Assume first that $|E-\lambda_j|\leq
  a_jd_j(2+C_0 d_j)^{-1}$ for $C_0:=2e^{1/C}$. Recall that
  $0<a_j,d_j\leq 1$. Note that, for $C$ sufficiently large, for $E\in
  U_j$, one has
  \begin{equation}
    \label{eq:24}
    \begin{split}
      \left|\text{Im}\,e^{-i\theta(E)}\right|&=e^{\text{Im}\,\theta(E)}
      |\sin\text{Re}\,\theta(E)|=
      e^{\text{Im}[\theta(E)-\theta(\text{Re}\,E)]}
      |\sin\text{Re}\,\theta(E)|\\&\geq e^{\theta'_\delta\text{Im} E}
      |\sin\text{Re}\,\theta(E)|\geq e^{-1/C}
      |\sin\text{Re}\,\theta(E)|
    \end{split}
  \end{equation}
  and
  \begin{equation}
    \label{eq:117}
    \left|e^{-i\theta(E)}\right|\leq 1\leq e^{1/C}.
  \end{equation}
  One estimates
  \begin{equation}
    \label{eq:146}
    |S_L(E)|\geq\frac{a_j}{|\lambda_j-E|}-\sum_{k\not=j}
    \frac{a_k}{|\lambda_k-E|}\geq \frac2{d_j}+C_0-\sum_{k\not=j}
    \frac{2a_k}{\D\min_{k\not=j}|\lambda_k-\lambda_j|}\geq
    C_0=2e^{1/C}.
  \end{equation}
  Thus, comparing~\eqref{eq:146} and~\eqref{eq:117}, we see that
  equation~\eqref{eq:1} has
  no solution in the set $U_j\cap\{|E-\lambda_j|\leq a_jd_j(2+C d_j)^{-1}\}$.\\
  Assume now that $|E-\lambda_j|> a_jd_j(2+C_0d_j)^{-1}$. Then, for
  $E\in U_j$, one has
  \begin{equation}
    \label{eq:147}
    |\text{Im}\,E|\leq \frac1{\theta'_\delta C} 
    a_jd^2_j|\sin(\text{Re}\,\theta(E))|.
  \end{equation}
  Thus, for $E\in U_j\cap\{|E-\lambda_j| > a_jd_j(2+C_0d_j)^{-1}\}$, one
  computes
  \begin{equation}
    \label{eq:194}
    \begin{split}
      |\text{Im}\,S_L(E)|&\leq|\text{Im}\,E|
      \left(\frac{a_j}{|\lambda_j-\text{Re}\,E|^2+
          |\text{Im}\,E|^2}+\frac{4}{d^2_j+
          |\text{Im}\,E|^2}\right)\\&\leq \frac1{\theta'_\delta C}
      a_jd^2_j|\sin(\text{Re}\,\theta(E))|
      \left(\frac{(2+C_0d_j)^2a_j}{a^2_jd^2_j}+\frac4{d_j^2}\right)\\
      &\leq \frac4{\theta'_\delta C} (1+e^{1/C})^2
      |\sin(\text{Re}\,\theta(E))|
      \leq\frac12e^{-1/C}|\sin(\text{Re}\,\theta(E))|
    \end{split}
  \end{equation}
  provided $C$ satisfies $\D 8e^{1/C}(1+e^{1/C})^2<\theta'_\delta C$.\\
  Hence, the comparison of~\eqref{eq:24} with~\eqref{eq:194} shows
  that~\eqref{eq:1} has no solution in $U_j\cap\{|E-\lambda_j| >
  a_jd_j(2+C_0d_j)^{-1}\}$ if we choose $C$ large enough (independent
  of $(a_j)_j$ and $(\lambda_j)_j$). Thus, we have proved that for
  some $C>0$ large enough (independent of $(a_j)_j$ and
  $(\lambda_j)_j$), ~\eqref{eq:1} has no solution in $U_j$.
  \vskip.1cm\noindent
  Let us now turn to the case of equation~\eqref{eq:135}. The basic
  ideas are the same as for equation~\eqref{eq:1}. Consider the matrix
  $\Gamma_L(E)$ defined by~\eqref{eq:142}. The summands
  in~\eqref{eq:142} are hermitian, of rank $1$ and their norm is given
  by~\eqref{eq:145}.\\
  Assume that $E\in U_j$ is a solution to~\eqref{eq:135}. Define the
  vectors
  \begin{equation*}
    v_j:=a_j^{-1/2}\begin{pmatrix}
      \varphi_j(L)\\\varphi_j(0) \end{pmatrix}\quad\text{for}\quad
    j\in\{0,\cdots,L\}.
  \end{equation*}
  Here $a_j=a_j^{\Z}$.\\
  Note that by definition of $a_j$, one has $\|v_j\|^2=2$. Pick $u$ in
  $C^2$, a normalized eigenvector of $\Gamma_L(E)$ associated to the
  eigenvalue $-e^{-i\theta(E)}$. Thus, $u$ satisfies
  \begin{equation}
    \label{eq:111}
    \sum_{j=0}^L\frac{a_j\,\langle v_j,u\rangle \, v_j}{\lambda_j-E}
    =-e^{-i\theta(E)}u
  \end{equation}
  Note that, by assumption, one has
  \begin{equation}
    \label{eq:169}
    \sup_{E\in U_j}\left\|\sum_{k\not=j}\frac{a_k\,\langle
        v_k,u\rangle\,v_k}{\lambda_k-E}\right\|\lesssim\frac1{d_j}
    \quad\text{and}\quad
    \left|\text{Im}\left(\sum_{k\not=j}\frac{a_k\,|\langle
          v_k,u\rangle|^2}{\lambda_k-E}\right)\right|\lesssim
    \frac{|\text{Im}\, E|}{d^2_j}
  \end{equation}
  where the constants are independent of $C$, the one defining
  $U_j$.\\
  Taking the (real) scalar product of equation~\eqref{eq:111} with
  $\overline{u}$, and then the imaginary part, we obtain
  \begin{equation*}
    -\frac{a_j\,|\langle
      v_j,u\rangle|^2\text{Im}\,E}{|\lambda_j-E|^2}+
    \text{Im}\left(e^{-i\theta(E)}\right)
    =O\left(\frac{|\text{Im}\, E|}{d^2_j}\right)
  \end{equation*}
  Thus, for $E\in U_j$, as $a_j\leq1$, for $C$ in~\eqref{eq:22}
  sufficiently large (depending only on $\delta$),
  \begin{equation*}
    \frac{a_j\,|\langle
      v_j,u\rangle|^2|\text{Im}\,E|}{|\lambda_j-E|^2}\geq
    \frac12|\sin(\text{Re}\,\theta(E))|.
  \end{equation*}
  Hence, for a solution to~\eqref{eq:135} in $U_j$ and $u$ as above,
  one has
  \begin{equation*}
    |\lambda_j-E|\leq |\langle v_j,u\rangle|
    \sqrt{\frac{2a_j|\text{Im}\,E|}{|\sin(\text{Re}\,\theta(E))|}}
    \leq2\sqrt{\frac{a_j|\text{Im}\,E|}{|\sin(\text{Re}\,\theta(E))|}}.
  \end{equation*}
  Hence, by the definition of $U_j$, for $C$ large, we get
  \begin{equation}
    \label{eq:212}
    \left|\frac{a_j}{\lambda_j-E}\right|\geq\frac{C\theta'_\delta}{d_j}\gg \frac1{d_j}.
  \end{equation}
  By~\eqref{eq:169}, the operator $\Gamma_L(E)$ can be written as
  \begin{equation}
    \label{eq:215}
    \Gamma_L(E)=\frac{a_j}{\lambda_j-E} v_j\otimes v_j+
    R_j(E)+i I_j(E)
  \end{equation}
  where $R_j(E)$ and $I_j(E)$ are self-adjoint ($I_j$ is non negative)
  and satisfy
  \begin{equation}
    \label{eq:214}
    \|R_j(E)\|\lesssim \frac1{d_j}\quad\text{ and }\quad
    \|I_j(E)\|\lesssim \frac{|\text{Im}\,E|}{d^2_j}.
  \end{equation}
  An explicit computation shows that the eigenvalues of the two-by-two
  matrix $\D\frac{a_j}{\lambda_j-E}v_j\otimes v_j+R_j(E)$ satisfy
  \begin{equation*}
    \text{either}\quad \lambda=\frac{a_j}{\lambda_j-E} 
    \left(1+O\left(\frac{d_j}{C\theta'_\delta}\right)\right)\quad
    \text{or}\quad|\text{Im}\,\lambda|\lesssim\frac{|\im E|}{a_j}
  \end{equation*}
  where the implicit constants are independent of the one defining
  $U_j$.\\
  Thus, by~\eqref{eq:215}, using~\eqref{eq:212} and the second
  estimate in~\eqref{eq:214}, we see that the eigenvalues of the
  matrix $\D\Gamma_L(E)$ satisfy
  \begin{equation*}
    \text{either}\quad \lambda=\frac{a_j}{\lambda_j-E} 
    \left(1+O\left(\frac{d_j}{C\theta'_\delta}\right)\right)\quad
    \text{or}\quad|\text{Im}\,\lambda|\leq \frac2{C\theta'_\delta}.     
  \end{equation*}
  Clearly, for $C$ large, no such value can be equal to
  $-e^{-i\theta(E)}$ being to large by~\eqref{eq:212} in the first
  case or having too small imaginary part in the second. The proof of
  Theorem~\ref{thr:13} is complete.
\end{proof}
\begin{proof}[Proof of Theorem~\ref{thr:14}]
  Again, we start with the solutions to~\eqref{eq:1}. For $z\in\tilde
  U_j$, we compute
  \begin{equation}
    \label{eq:136}
    \begin{split}
      \text{Im}\,S_L(E)&=\sum_{k=0}^L\frac{a_k\,\text{Im}\,
        E}{(\lambda_k-\text{Re}\,E)^2+\text{Im}^2
        E}\\&=\frac{a_j\,\text{Im}\,
        E}{(\lambda_j-\text{Re}\,E)^2+\text{Im}^2 E}+
      \sum_{\substack{0\leq k\leq L\\ k\not=j}}\frac{-a_k\,\text{Im}\,
        E}{(\lambda_k-\text{Re}\,E)^2+\text{Im}^2 E}.
    \end{split}
  \end{equation}
  When $-d^2_j/C\leq\text{Im}\,E\leq -Ca_j$, the second equality above
  and~\eqref{eq:20} yield, for $C$ sufficiently large,
  \begin{equation}
    \label{eq:196}
    0\leq \text{Im}\,S_L(E)\lesssim
    \frac{a_j}{|\text{Im}\,E|}+\frac{|\text{Im}\,E|}{d^2_j+\text{Im}^2E}\leq
    \frac2C.
  \end{equation}
  On the other hand, for some $K>0$, one has
  \begin{equation*}
    \left|\text{Im}\,e^{-i\theta(E)}\right|\geq
    |\text{Im}\,e^{-i\theta(\text{Re}\,E)}|-K d^2_j/C.
  \end{equation*}
  Now, as, under the assumptions of Theorem~\ref{thr:14}, one has
  \begin{equation}
    \label{eq:197}
    \min_{E\in\left[\frac{\lambda_j+
          \lambda_{j-1}}2,\frac{\lambda_j+\lambda_{j+1}}2\right]}
    \left|\text{Im}\,e^{-i\theta(E)}\right|\geq\frac14\min
    \left(\sqrt{16-(\lambda_j+\lambda_{j-1})^2},\sqrt{16-(\lambda_j+
        \lambda_{j+1})^2}\right),
  \end{equation}
  we obtain that~\eqref{eq:1} has no solution in $\tilde
  U_j\cap\{-d_j/C\leq\text{Im}\,E\leq -Ca_j\}$. \\
  Pick now $E\in \tilde U_j$ such that $-Ca_j\leq\text{Im}\,E\leq
  -a_jd_j^2/C$. Then,~\eqref{eq:117} and~\eqref{eq:20} yield, for $C$
  sufficiently large,
  \begin{equation*}
    \text{Im}\,S_L(E)\lesssim
    \frac{a_j\text{Im}\,E}{C^2a^2_j+\text{Im}^2E}+
    \frac{Ca_j}{d^2_j}\leq
    \frac1C+\frac1{2C}.
  \end{equation*}
  The imaginary part of $e^{-i\theta(E)}$ is estimated as above. Thus,
  for $C$ sufficiently large, ~\eqref{eq:1} has no solution in $\tilde
  U_j\cap\{-Ca_j\leq\text{Im}\,E\leq
  -a_jd_j^2/C\}$.\vskip.1cm\noindent
  The case of equation~\eqref{eq:135} is studied in exactly the same
  way except that, as in the proof of Theorem~\ref{thr:13}, one has to
  replace the study of $S_L(E)$ by that of
  $\langle\Gamma_L(E)u,u\rangle$ for $u$ a normalized eigenvector of
  $\Gamma_L(E)$ associated to $-e^{-i\theta(E)}$ and, thus, the
  coefficient $a_k$ in~\eqref{eq:136} gets multiplied by a factor
  $|\langle v_k,u\rangle|^2$ that is bounded by $2$.  \\
  This completes the proof of Theorem~\ref{thr:14}.
\end{proof}
\subsection{The resonances near an ``isolated'' eigenvalue}
\label{sec:reson-near-discr}
We will now solve equation~\eqref{eq:1} near a given $\lambda_j$ under
the additional assumptions that $a_j\ll d^2_j$. By
Theorems~\ref{thr:13} and~\ref{thr:14}, we will do so in the rectangle
$R_j$ (see Fig.~\ref{fig:6}). Actually, we prove that, in $R_j$, there
is exactly one resonance and give an asymptotic for this resonance in
terms of $a_j$, $d_j$ and $\lambda_j$. This result is going to be
applied to the case of random $V$ and to that of isolated eigenvalues
(for any $V$).\\
Using the notations of section~\ref{sec:char-reson}, for
$j\in\{0,\cdots,L\}$, we define
\begin{equation}
  \label{eq:25}
  S_{L,j}(E):=\sum_{k\not=j}\frac{a^\N_k}{\lambda_k-E}
  \quad\text{and}\quad
  \Gamma_{L,j}(E):=\sum_{k\not=j}\frac1{\lambda_k-E}
  \begin{pmatrix} |\varphi_k(L)|^2& \overline{\varphi_k(0)}
    \varphi_k(L) \\ \varphi_k(0) \overline{\varphi_k(L)} &
    |\varphi_k(0)|^2
  \end{pmatrix}.
\end{equation}
We prove
\begin{Th}
  \label{thr:12}
  Pick $j\in\{0,\cdots,L\}$ such that
  $-4<\lambda_{j-1}+\lambda_j<\lambda_{j+1}+\lambda_j<4$.  There
  exists $C>1$ (depending only on $(\lambda_{j-1}+\lambda_j)+4$ and
  $4-(\lambda_{j+1}+\lambda_j)$) such that, for any $L$, if $a_j\leq
  d^2_j/C$, equation~\eqref{eq:1} and~\eqref{eq:135} has exactly one
  solution in the set
  \begin{equation}
    \label{eq:29}
    R_j:=\left\{E\in\C;\
      \begin{aligned}
        \text{Re}\,E&\in\lambda_j+Ca_j\left[-1,1\right]\\
        \quad -Ca_j\leq &\text{Im}\,E\leq -a_j\,d^2_j/C
      \end{aligned}
    \right\}.
  \end{equation}
  Moreover, the solution to~\eqref{eq:1}, say $z^\N_j$, satisfies
  \begin{equation}
    \label{eq:31}
    z^\N_j=\lambda_j+\frac{a^\N_j}{S_{L,j}(\lambda_j)+e^{-i\theta(\lambda_j)}}
    +O\left(\left(a^\N_jd^{-1}_j\right)^2\right).
  \end{equation}
  and the solution to~\eqref{eq:135}, say $z^\Z_j$, satisfies
  \begin{equation}
    \label{eq:149}
    z^\Z_j=\lambda_j+\left\langle \begin{pmatrix}\overline{\varphi_j(L)}
        \\ \varphi_j(0)\end{pmatrix} ,\left(\Gamma_{L,j}(\lambda_j) 
        +e^{-i\theta(\lambda_j)}\right)^{-1}
      \begin{pmatrix}\overline{\varphi_j(L)}
        \\ \varphi_j(0)\end{pmatrix}\right\rangle+
    O\left(\left(a^\Z_jd^{-1}_j\right)^2\right).
  \end{equation}
\end{Th}
\noindent Note that, if $a^\N_jd^{-2}_j$ is small,
formula~\eqref{eq:31} gives the asymptotic of the width of the
solution $z^\N_j$, namely,
\begin{equation}
  \label{eq:150}
  \text{Im}\,z^\N_j=\frac{a^\N_j\cdot\sin\theta(\lambda_j)}
  {[S_{L,j}(\lambda_j)+\cos\theta(\lambda_j)]^2+
    \sin^2\theta(\lambda_j)}  (1+o(1)).
\end{equation}
Recall that $\sin\theta(\lambda_j)<0$ (see Theorem~\ref{thr:11}). For
$H_L^\Z$, using the bounds~\eqref{eq:156} and~\eqref{eq:153}, we see
that the asymptotic of the imaginary part of the solution $z^\Z_j$
satisfies
\begin{equation}
  \label{eq:155}
  -\frac1C a^\Z_j\leq\text{Im}\,z^\Z_j\leq -Ca^\Z_j d_j^2.
\end{equation}
This and~\eqref{eq:150} will be useful when $a_j\ll d^2_j$ as will be
the case for random potentials. The case when $a_j$ and $d_j$ are of
the same order of magnitude requires more information. This is the
case that we meet in the next section when dealing with periodic
potentials.\\
The proof of Theorem~\ref{thr:12} also yields the behavior of the
functions $E\mapsto S_L(E)+e^{-i\theta(E)}$ and $\D E\mapsto
\text{det}\left( \Gamma_L(E)+e^{-i\theta(E)}\right)$ near their zeros
in $R_j$ and, in particular shows the following
\begin{Pro}
  \label{pro:7}
  Fix $\delta>0$. Under the assumptions of Theorem~\ref{thr:12}, there
  exists $c>0$ such that, for
  $-4+\delta<\lambda_{j-1}+\lambda_j<\lambda_{j+1}+\lambda_j<4-\delta$,
  one has
  \begin{equation*}
    \begin{aligned}
      \inf_{0<r<c a^\N_jd^{-1}_j}\min_{|E-z^\N_j|=r}
      \frac{\left|S_L(E)+e^{-i\theta(E)}\right|}r&\geq c\quad\text{and}\\
      \inf_{0<r<c a^\Z_jd^{-1}_j}
      \min_{|E-z^\Z_j|=r}\frac{\left|\text{det}
          \left(\Gamma_L(E)+e^{-i\theta(E)}\right)\right|}r&\geq c.
    \end{aligned}
  \end{equation*}
\end{Pro}
\noindent Proposition~\ref{pro:7} is a consequence of the analogues
of~\eqref{eq:201} and~\eqref{eq:202} on the rectangles
\begin{equation*}
  \tilde R_j=\tilde z_j+ca^\bullet_jd^{-1}_j[-1,1]\times[-1,1] 
\end{equation*}
for $\bullet\in\{\N,\Z\}$ and $c$ sufficiently small.
\begin{proof}[Proof of Theorem~\ref{thr:12}]
  Let us start with equation~\eqref{eq:1}. To prove the statement on
  equation~\eqref{eq:1}, in $R_j$, we compare the function $E\mapsto
  S_L(E)+e^{-i\theta(E)}$ to the function
  \begin{equation*}
    E\mapsto\tilde S_{L,j}(E)=\frac{a^\N_j}{\lambda_j-E}
    +S_{L,j}(\lambda_j)+e^{-i\theta(\lambda_j)}.
  \end{equation*}
  Clearly, in $\C$, the equation $\tilde S_{L,j}(E)=0$ admits a unique
  solution given by
  \begin{equation*}
    \tilde z_j=\lambda_j+\frac{a^\N_j}{S_{L,j}(\lambda_j)+e^{-i\theta(\lambda_j)}}.
  \end{equation*}
  For $E\in\partial R_j$, the boundary of $R_j$, one has
  \begin{equation}
    \label{eq:151}
    \begin{split}
      \left|\tilde S_{L,j}(E)\right|\geq
      \frac1{2C}\quad&\text{and}\quad
      \left|\frac{a^\N_j}{\lambda_j-E}\right|\geq \frac1{2C},\\
      \left| e^{-i\theta(E)}-e^{-i\theta(\lambda_j)}\right|\leq C a^\N_j
      &\quad\text{and}\quad \left|
        S_{L,j}(E)-S_{L,j}(\lambda_j)\right|\leq Ca^\N_jd^{-2}_j.
    \end{split}
  \end{equation}
  Hence, as $d_j\leq1$, one gets
  \begin{equation*}
    \max_{E\in\partial R_j}\frac{\left|\tilde
        S_{L,j}(E)-S_L(E)-e^{-i\theta(E)}\right|}{|\tilde S_{L,j}(E)|}\leq
    4Ca^\N_jd_j^{-2} 
  \end{equation*}
  Thus, by Rouch{\'e}'s theorem, equation~\eqref{eq:1} has a unique
  solution in $R_j$. \\
  To obtain the asymptotics of the solution, it suffices to use
  Rouch{\'e}'s theorem again with the functions $\tilde S_{L,j}$ and
  $S_L(E)+e^{-i\theta(E)}$ on the smaller rectangle $\tilde R_j=\tilde
  z_j+K(a^\N_jd^{-1}_j)^2[-1,1]\times[-1,1]$. One then estimates
  \begin{equation}
    \label{eq:201}
    \max_{E\in\partial\tilde R_j}\frac{\left|\tilde S_{L,j}(E)-S_L(E)
        -e^{-i\theta(E)}\right|}{|\tilde S_{L,j}(E)|}\leq4CK^{-1}. 
  \end{equation}
  Thus, for $K$ sufficiently large, this completes the proof of the
  statements on the solutions to equation~\eqref{eq:1} contained in
  Theorem~\ref{thr:12}.\vskip.1cm\noindent
  Let us turn to equation~\eqref{eq:135}. On $R_j$, we now compare
  $\Gamma_L(E)+e^{-i\theta(E)}$ to the matrix valued function
  \begin{equation*}
    E\mapsto\tilde \Gamma_{L,j}(E):=
    \frac1{\lambda_j-E}
    \begin{pmatrix} |\varphi_j(L)|^2& \overline{\varphi_j(0)}
      \varphi_j(L) \\ \varphi_j(0) \overline{\varphi_j(L)} &
      |\varphi_j(0)|^2
    \end{pmatrix}+\Gamma_{L,j}(\lambda_j)+e^{-i\theta(\lambda_j)}.
  \end{equation*}
  The matrix $\D\begin{pmatrix} |\varphi_j(L)|^2&
    \overline{\varphi_j(0)} \varphi_j(L) \\ \varphi_j(0)
    \overline{\varphi_j(L)} & |\varphi_j(0)|^2
  \end{pmatrix}$ is rank $1$ and can be diagonalized as
  \begin{equation*}
    \begin{pmatrix} |\varphi_j(L)|^2&
      \overline{\varphi_j(0)} \varphi_j(L) \\ \varphi_j(0)
      \overline{\varphi_j(L)} & |\varphi_j(0)|^2
    \end{pmatrix}=P_j\begin{pmatrix} a^\Z_j&0 \\0&0\end{pmatrix}P^*_j
  \end{equation*}
  where $a^\Z_j$ is given by~\eqref{eq:145} and
  \begin{equation*}
    P_j=\frac1{\sqrt{a^\Z_j}} \begin{pmatrix} \varphi_j(L)&
      -\overline{\varphi_j(0)} \\ \varphi_j(0) &
      \overline{\varphi_j(L)}
    \end{pmatrix}.
  \end{equation*}
  Thus, $\tilde \Gamma_{L,j}(E)$ is unitarily equivalent to
  \begin{equation}
    \label{eq:154}
    M:=\frac1{\lambda_j-E}\begin{pmatrix} a^\Z_j&0 \\0&0\end{pmatrix}+
    P_j^*\Gamma_{L,j}(\lambda_j)P_j+e^{-i\theta(\lambda_j)}.
  \end{equation}
  As $P_j^*\Gamma_{L,j}(\lambda_j)P_j$ is real and the imaginary part
  of $e^{-i\theta(\lambda_j)}$ does not vanish, the matrix $
  M_0:=P_j^*\Gamma_{L,j}(\lambda_j)P_j+e^{-i\theta(\lambda_j)}$ is
  invertible. By rank $1$ perturbation theory (see
 , e.g.,~\cite{MR2154153}), we know that $M$ is invertible if and only
  if $a^\Z_j\left[M_0^{-1}\right]_{11}+\lambda_j\not=E$ (where
  $[M]_{11}$ is the upper right coefficient of the $2\times 2$ matrix
  $M$). In this case, one has
  \begin{equation}
    \label{eq:152}
    M^{-1}=M_0^{-1}
    -\frac{a^\Z_j}{a^\Z_j\left[M_0^{-1}\right]_{11}+\lambda_j-E}
    M_0^{-1}\begin{pmatrix} 1& 0\\0 &0 \end{pmatrix}
    M_0^{-1}.  
  \end{equation}
  Hence, $0$ is an eigenvalue of $M$ if and only if
  \begin{equation}
    \label{eq:159}
    \begin{split}
      E&=\lambda_j+a^\Z_j\left[\left(P_j^*\Gamma_{L,j}(\lambda_j)P_j
          +e^{-i\theta(\lambda_j)}\right)^{-1}\right]_{11}\\&=
      \lambda_j+\left\langle \begin{pmatrix}\overline{\varphi_j(L)} \\
          \varphi_j(0)\end{pmatrix} ,\left(\Gamma_{L,j}(\lambda_j)
          +e^{-i\theta(\lambda_j)}\right)^{-1}\begin{pmatrix}\overline{\varphi_j(L)}
          \\ \varphi_j(0)\end{pmatrix}\right\rangle.
    \end{split}
  \end{equation}
  Note that, as $\Gamma_{L,j}(\lambda_j)$ is real symmetric and
  $\|\Gamma_{L,j}(\lambda_j)\|\leq Cd_j^{-1}$, one has
  \begin{equation}
    \label{eq:156}
    \left|\left\langle \begin{pmatrix}\overline{\varphi_j(L)} \\
          \varphi_j(0)\end{pmatrix} ,\left(\Gamma_{L,j}(\lambda_j)
          +e^{-i\theta(\lambda_j)}\right)^{-1}
        \begin{pmatrix}\overline{\varphi_j(L)}
          \\ \varphi_j(0)\end{pmatrix}\right\rangle\right|\leq
    \frac{a^\Z_j}{\left|\sin\theta(\lambda_j)\right|}.
  \end{equation}
  and
  \begin{equation}
    \label{eq:153}
    \text{Im}\left(\left\langle \begin{pmatrix}\overline{\varphi_j(L)} \\
          \varphi_j(0)\end{pmatrix} ,\left(\Gamma_{L,j}(\lambda_j)
          +e^{-i\theta(\lambda_j)}\right)^{-1}
        \begin{pmatrix}\overline{\varphi_j(L)}
          \\ \varphi_j(0)\end{pmatrix}\right\rangle\right)\leq
    \frac{a^\Z_jd_j^2\sin\theta(\lambda_j)} {1+d_j^2}.
  \end{equation}
  Using~\eqref{eq:154},~\eqref{eq:152},~\eqref{eq:156}
  and~\eqref{eq:153},we see that, for $E\in\partial R_j$, the boundary
  of $R_j$, $\tilde\Gamma_{L,j}(E)$ is invertible and that one has
  \begin{gather*}
    \left\|\left[\tilde\Gamma_{L,j}(E)\right]^{-1}\right\|\leq
    2C\quad\text{and}\quad \left\|
      \Gamma_{L,j}(E)-\Gamma_{L,j}(\lambda_j)\right\|\leq
    Ca^\Z_jd^{-2}_j.
  \end{gather*}
  Hence, as $d_j\leq1$, taking~\eqref{eq:151} into account, one gets
  \begin{equation*}
    \max_{E\in\partial
      R_j}\left\|1-\left[\tilde\Gamma_{L,j}(E)\right]^{-1}
      \left(\Gamma_L(E)+e^{-i\theta(E)}\right)\right\|\leq
    4C^2a^\Z_jd_j^{-2} 
  \end{equation*}
  In the same way, one proves
  \begin{equation}
    \label{eq:202}
    \max_{E\in\partial\tilde R_j}\left\|1-\left[\tilde\Gamma_{L,j}(E)
      \right]^{-1}\left(\Gamma_L(E)+e^{-i\theta(E)}\right)\right\|\lesssim
    K^{-1} 
  \end{equation}
  where we recall that $\tilde R_j=\tilde
  z_j+K(a^\N_jd^{-1}_j)^2[-1,1]\times[-1,1]$.\\%
  Thus, we can apply Rouch{\'e}'s Theorem to compare the following two
  functions on $\partial R_j$ and $\partial \tilde R_j$ (for $K$
  sufficiently large),
  \begin{equation*}
    \text{det}\left(\tilde\Gamma_{L,j}(E)\right)\quad\text{and}\quad
    \text{det}\left(\Gamma_{L}(E)+e^{-i\theta(E)}\right)
  \end{equation*}
  as
  \begin{multline*}
    \frac{\left|\text{det}\left(\tilde\Gamma_{L,j}(E)\right)
        -\text{det}\left(\Gamma_{L}(E)+e^{-i\theta(E)}\right)\right|}
    {\left|\text{det}\left(\tilde\Gamma_{L,j}(E)\right)\right|}\\=
    \left|1-\text{det}\left(1-\left[
          1-\left[\tilde\Gamma_{L,j}(E)\right]^{-1}
          \left(\Gamma_L(E)+e^{-i\theta(E)}\right)\right]\right)\right|.
  \end{multline*}
  We then conclude as in the case of equation~\eqref{eq:1}. This
  completes the proof of Theorem~\ref{thr:12}.
\end{proof}
\noindent Combining Theorems~\ref{thr:12},~\ref{thr:13}
and~\ref{thr:14}, we get a pretty clear picture of the resonances near
the Dirichlet eigenvalues in $(-2,2)$ as long as the associated $a_j$
and $d_j$ behave correctly. As said, this and the knowledge of the
spectral statistics for random operators will enable us to prove the
results described in section~\ref{sec:random-case}. For the periodic
case, Theorems~\ref{thr:13},~\ref{thr:14} and~\ref{thr:12} will prove
not too be sufficient. As we shall see, in this case, $a_j$ and $d_j$
are of the same order of magnitude. Thus, neighboring Dirichlet
eigenvalues have a sizable effect on the location of
resonances. Therefore, in the next section, we compute the Dirichlet
spectral data for the truncated periodic potential.
\section{The Dirichlet spectral data for periodic potentials}
\label{sec:periodic-case-1}
As we did not find any suitable reference for this material, we first
derive a suitable description of the spectral data (i.e. the $(a_j)_j$
and $(\lambda_j)_j$) for the Dirichlet restriction of a periodic
operator to the interval $\llbracket 0,L\rrbracket$ when $L$ becomes
large.\\
Consider a potential $V:\N\to\R$ such that, for some $p\geq1$, one has
$V_k=V_{k+p}$ for all $k\geq0$. We assume $p$ to be minimal, i.e., to be
the period of $V$. In our first result, we describe the spectrum of
$H^\Z=-\Delta+V$ on $\ell^2(\Z)$ and $H^\N=-\Delta+V$ on $\ell^2(\N)$
(with Dirichlet boundary conditions at $0$). In the second result we
turn to $H_L$, the Dirichlet restriction $H^\N$ to $\llbracket
0,L\rrbracket$ and described its spectral data, i.e., its eigenvalues
and eigenfunctions. \\
We recall
\begin{Th}
  \label{thr:15}
  The spectrum of $H^\Z$, say $\Sigma_\Z$, is a union of at most $p$
  disjoint intervals that all consist in purely absolutely continuous
  spectrum.\\
  The spectrum of $H^\N$ is the union of $\Sigma_\Z$ and at most
  finitely many simple eigenvalues outside $\Sigma_\Z$, say,
  $(v_j)_{0\leq j \leq n}$. $\Sigma_\Z$ consists of purely absolutely
  continuous spectrum of $H^\N$ and the eigenfunctions associated to
  $(v_j)_{0\leq j \leq n}$, say $(\psi_j)_{0\leq j \leq n}$, are exponentially
  decaying at infinity.
\end{Th}
\noindent Except for the exponential decay of the eigenfunctions, the
proof of the statement for the periodic operator on $\Z$ and $\N$ is
classical and can e.g. be found in a more general setting
in~\cite[chapters 2, 3 and 7]{MR1711536} (see
also~\cite{MR0650253,MR58:12429c}). The exponential decay is an
immediate consequence of Floquet theory for the periodic Hamiltonian
on $\Z$ and the fact that the eigenvalues lie in gaps of $\Sigma_\Z$.\\
For $H^\Z$ one can define its Bloch quasi-momentum (see the beginning
of section~\ref{sec:proof-theorem-2} for details) that we denote by
$\theta_p$; it is continuous and strictly increasing on $\Sigma_\Z$
and real analytic on $\overset{\circ}{\Sigma}_\Z$. Decompose
$\Sigma_\Z$ into its connected components, i.e.,
$\D\Sigma_\Z=\bigcup_{r=1}^q B_r$ where $q\leq p$. Let $c_q$ be the
number of closed gaps contained in $q$. Then, $\theta_p$ is continuous
and strictly increasing on $B_r$ and real analytic on
$\overset{\circ}{B}_r$, the interior of the $r$-th band.  Moreover, on
this set, its derivative can be expressed in terms of the density of
states defined in~\eqref{eq:144} as
\begin{equation}
  \label{eq:183}
  n(\lambda)=\frac1\pi\theta'_p(\lambda).
\end{equation}
We first describe the eigenvalues of $H_L$.
\begin{Th}
  \label{thr:16}
  One has
  \begin{enumerate}
  \item For any $k\in\{0,\cdots,p-1\}$, there exists $h_k:\;
    \Sigma_\Z\to\R$, a continuous function that is real analytic in a
    neighborhood of $\overset{\circ}{\Sigma}_\Z$ such that, for $L$
    sufficiently large s.t.  $L\equiv k\mod p$,
    \begin{enumerate}
    \item for $1\leq r\leq q$, the function $h_k$ maps $B_r$ into
      $(-(c_r+1)\pi,(c_r+1)\pi)$;
    \item define the function
      \begin{equation}
        \label{eq:229}
        \theta_{p,L}:=\theta_p-\frac{h_k}{L-k};    
      \end{equation}
      it is continuous and strictly monotonous on each $B_r$ ($1\leq
      r\leq q$);
    \item for $1\leq r\leq q$, the eigenvalues of $H_L$ in $B_r$, the $r$-th
      band of $\Sigma_\Z$, say $(\lambda^r_j)_j$, are the solutions
      (in $\Sigma_\Z$) to the quantization conditions
      \begin{equation}
        \label{eq:23}
        \theta_{p,L}(\lambda^r_j)=\frac{j\pi}{L-k},\quad j\in\Z.
      \end{equation}
    \end{enumerate}
  \item There exists $c>0$ such that, if $\lambda$ is an eigenvalue of
    $H_L$ outside $\Sigma_\Z$, then, for $L=Np+k$ sufficiently large,
    there exists
    $\lambda_\infty\in\Sigma_0^+\cup\Sigma^-_k\setminus\Sigma_\Z$
    s.t., one has $|\lambda-\lambda_\infty|\leq e^{-c L}$.
  \end{enumerate}
\end{Th}
\noindent Recall that $\Sigma_0^+$ and $\Sigma^-_k$ are respectively the
spectra of $H_0^+$ and $H^-_k$ defined in section~\ref{sec:auxil-op}. \\
In Theorem~\ref{thr:16}, when solving equation~\eqref{eq:23}, one has
to do it for each band $B_r$, and, for each band and each $j$ such
that $\D\frac{j\pi}{L-k}\in \theta_{p,L}(B_r)$, equation~\eqref{eq:23}
admits a unique solution. But, it may happen that one has two
solutions to~\eqref{eq:23} for a given $j$ belonging to neighboring
bands. In the sequel to simplify the notations, we will not
distinguish between the different bands, i.e., we will write
eigenvalues $(\lambda_j)_j$ not referring to the band they belong to.\\
Let us now describe the associated eigenfunctions.
\begin{Th}
  \label{thr:30}
  Recall that $(\lambda_j)_j$ are the eigenvalues of $H_L$ in
  $\Sigma_\Z$ (enumerated as in Theorem~\ref{thr:16}).
  \begin{enumerate}
  \item There exist $p+2$ positive functions, say, $f^+_0$,
    $(f^-_k)_{0\leq k\leq p-1}$ and $\tilde f$, that are real analytic in a
    neighborhood of $\overset{\circ}{\Sigma}_\Z$ such that, there
    exists $\sigma_r\in\{+1,-1\}$ such that, for $L=Np+k$ sufficiently
    large, for $\lambda_j$ in $\overset{\circ}{B_r}$, the interior of
    $r$-th band of $\Sigma_\Z$, one has
    \begin{equation}
      \label{eq:26}
      \begin{split}
        &|\varphi_l(L)|^2=\frac{f^-_k(\lambda_j)}{L-k}\left(1+
          \frac{\tilde f(\lambda_j)}{L-k}\right)^{-1},\quad
        |\varphi_l(0)|^2=\frac{f^+_0(\lambda_j)}
        {f^-_k(\lambda_j)}|\varphi_l(L)|^2,\\
        &\varphi_l(L)\overline{\varphi_l(0)}=\sigma_r\, e^{i\pi
          l} |\varphi_l(L)||\varphi_l(0)|=\sigma_r\,
        e^{i(L-k)\theta_p(\lambda_j)-
          h_k(\lambda_j)}|\varphi_l(L)||\varphi_l(0)|.
      \end{split}
    \end{equation}
  \item Let $\lambda$ be an eigenvalue of $H_L$ outside $\Sigma_\Z$
    (see point (3) in Theorem~\ref{thr:16}). If $\varphi$ is a
    normalized eigenfunction associated to $\lambda$ and $H_L$, one
    has one of the following alternatives for $L$ large
    \begin{enumerate}
    \item if $\lambda_\infty\in\Sigma_0^+\setminus\Sigma^-_k$, one has
      \begin{equation}
        \label{eq:27}
        |\varphi(L)|\asymp e^{-cL}\quad\text{and}\quad|\varphi(0)|\asymp 1;
      \end{equation}
    \item if $\lambda_\infty\in\Sigma^-_k\setminus\Sigma^+_0$, one has
      \begin{equation}
        \label{eq:134}
        |\varphi(L)|\asymp 1\quad\text{and}\quad|\varphi(0)|\asymp e^{-cL};
      \end{equation}
    \item if $\lambda_\infty\in\Sigma^-_k\cap\Sigma^+_0$, one has
      \begin{equation}
        \label{eq:158}
        |\varphi(L)|\asymp 1\quad\text{and}\quad|\varphi(0)|\asymp 1.
      \end{equation}
    \end{enumerate}
  \end{enumerate}
\end{Th}
\noindent For later use, let us define $\theta_{p,L}$, $f_{0,L}$ and
$f_{k,L}$ by
\begin{equation}
  \label{eq:171}
  f_{k,L}(\lambda)=f^-_k(\lambda)\left(1+
    \frac{\tilde
      f(\lambda)}{L-k}\right)^{-1}\quad\text{and}\quad
  f_{0,L}(\lambda)=f^+_0(\lambda)\left(1+ \frac{\tilde
      f(\lambda)}{L-k}\right)^{-1}    
\end{equation}
where $\theta_p$, $h_k$, $f_{0}$, $f_{k}$ and $\tilde f$ are
defined in Theorem~\ref{thr:16}.\\
As a consequence of Theorem~\ref{thr:16}, we obtain
\begin{Cor}
  \label{cor:1}
  For $\lambda\in\overset{\circ}{\Sigma}_\Z$, for $L\equiv k\mod(p)$
  sufficiently large, one has
  \begin{gather}
    \label{eq:64}
    \frac{dN_k^-}{d\lambda}(\lambda)=n_k^-(\lambda)=f^-_k(\lambda)n(\lambda)
    =\frac1\pi f^-_k(\lambda)\theta_p'(\lambda)=\frac1\pi
    f_{k,L}(\lambda)\theta_{p,L}'(\lambda),
    \\
    \label{eq:53}
    \frac{dN_0^+}{d\lambda}(\lambda)=n_0^+(\lambda)=f^+_0(\lambda)n(\lambda)
    =\frac1\pi f^+_0(\lambda)\theta_p'(\lambda)=\frac1\pi
    f_{0,L}(\lambda)\theta_{p,L}'(\lambda).
  \end{gather}
  Here, $\theta_P$, $f^+_0$ and $f^-_k$ are defined the functions
  defined in Theorem~\ref{thr:16}.
\end{Cor}
\begin{proof}[Proof of Corollary~\ref{cor:1}]
  To prove the first equalities in~\eqref{eq:64} and~\eqref{eq:53}, it
  suffices to prove that, for any
  $\chi\in\Coi(\overset{\circ}{\Sigma}_\Z)$,
  \begin{gather}
    \label{eq:67}
    \begin{aligned}
      \langle\delta_0,\chi(H_k^-)\delta_0\rangle&=
      \int_\R\chi(\lambda)dN_k^-(\lambda)=\frac1\pi
      \int_\R\chi(\theta_p^{-1}(k))f^-_k(\theta_p^{-1}(k))d
      k\\&=\frac1\pi
      \int_\R\chi(\lambda)f^-_k(\lambda)\theta_p'(\lambda)d\lambda,
    \end{aligned}
    \\
    \label{eq:69}
    \begin{aligned}
      \langle\delta_0,\chi(H_0^+)\delta_0\rangle&=
      \int_\R\chi(\lambda)dN_0^+(\lambda)=\frac1\pi
      \int_\R\chi(\theta_p^{-1}(k))f^+_0(\theta_p^{-1}(k))d
      k\\&=\frac1\pi
      \int_\R\chi(\lambda)f^+_0(\lambda)\theta_p'(\lambda)d\lambda,
    \end{aligned}
  \end{gather}
  the full statement then following by standard density argument. The
  operator $H_L$ converges to $H_0^+$ in norm resolvent sense. Thus,
  we know that $\D \langle\delta_0,\chi(H_0^+)\delta_0\rangle=
  \lim_{L\to+\infty} \langle\delta_0,\chi(H_L)\delta_0\rangle$. Now,
  by Theorem~\ref{thr:16}, as $\chi$ is supported in
  $\overset{\circ}{\Sigma}_\Z$, using the Poisson formula, one
  computes
  \begin{equation*}
    \begin{split}
      \langle\delta_0,&\chi(H_L)\delta_0\rangle=
      \sum_{j}\chi(\lambda_j)||\varphi_j(0)|^2
      =\frac1{L-k}\sum_{l}\chi\left(\theta_{p,L}^{-1}
        \left(\frac{l\pi}{L-k}\right)\right)
      f_{0,L}\left(\theta_{p,L}^{-1}\left(\frac{l\pi}{L-k}\right)\right)\\
      &=\frac1{L-k}\sum_{j\in\Z}\int_{\R} e^{-i2\pi j
        \lambda}\chi\left(\theta_{p,L}^{-1}
        \left(\frac{\pi\,\lambda}{L-k}\right)\right)
      f_{0,L}\left(\theta_{p,L}^{-1}
        \left(\frac{\pi\,\lambda}{L-k}\right)\right)d\lambda\\&
      =\frac1\pi\sum_{j\in\Z}\int_{\R} e^{-i2(L-k)j
        \theta_{p,L}(\lambda)}\chi\left( \lambda\right)
      f_{0,L}\left(\lambda\right)\theta'_{p,L}(\lambda)d\lambda.
    \end{split}
  \end{equation*}
  Thus, using the non stationary phase, i.e., integrating by parts, one
  gets, for any $N\geq2$,
  \begin{equation}
    \label{eq:173}
    \begin{split}
      \left|\langle\delta_0,\chi(H_L)\delta_0\rangle-\frac1\pi\int_{\R}
        \chi\left( \lambda\right)
        f_{0,L}\left(\lambda\right)\theta'_{p,L}(\lambda)d\lambda\right|&\leq
      \sum_{j\geq1}C_{N,K}\|\chi\|_{\mathcal{C}^N}(|j|(L-k))^{-N}
      \\&\leq C_{N,K}\|\chi\|_{\mathcal{C}^N}((L-k))^{-N}.
    \end{split}
  \end{equation}
  Here, we have used the analyticity of the functions $\theta_{p,L}$
  and $f_{0,L}$.\\
  To deal with $H_k^-$, we recall the operator $\tilde H_L$ (that is
  unitarily equivalent to $H_L$) defined in Remark~\ref{rem:6}. One
  has $\langle\delta_L,H_L\delta_L\rangle= \langle\delta_0,\chi(\tilde
  H_L)\delta_0\rangle$, thus, as $H_k^-$ is the strong resolvent sense
  limit of $\tilde H_L$, one gets
  $\D\langle\delta_0,\chi(H_k^-)\delta_0\rangle=
  \lim_{L\to+\infty} \langle\delta_L,\chi(H_L)\delta_L\rangle$.\\
  Then,~\eqref{eq:67} and~\eqref{eq:69} and, thus, the first
  equalities in~\eqref{eq:64} and~\eqref{eq:53}, follow as
  $\theta'_{p,L}$, $f_{0,L}$ and $f_{k,L}$ converge (locally uniformly
  on $\overset{\circ}{\Sigma}_\Z$) respectively to $\theta'_p$,
  $f^+_0$
  and $f^-_k$ (see~\eqref{eq:171} and Theorem~\ref{thr:16}).\\
  Let us now prove the second equalities in~\eqref{eq:64}
  and~\eqref{eq:53}. Therefore, we use an {\it almost analytic
    extension} (see~\cite{Mat:71}) of $\chi$, say, $\tilde\chi$, that
  is, a function $\tilde\chi:\ \C\to\C$ satisfying
  ( \begin{enumerate}
  \item for $z\in\R$, $\tilde\chi(z)=\chi(z)$,
  \item supp$(\tilde\chi)\subset\{z\in\C;\
    \vert\text{Im}(z)\vert<1\}$,
  \item $\tilde\chi\in{\mathcal S}(\{z\in\C;\
    \vert\text{Im}(z)\vert<1\})$,
  \item The family of functions $\D
    x\mapsto\frac{\partial\tilde\chi}{\partial\overline z}(x+iy)\cdot \vert
    y\vert^{-n}$ (for $0<\vert y\vert<1$) is bounded in $\Sc(\R)$ for
    any $n\in{\mathbb N}$.
  \end{enumerate}
  Moreover, $\tilde\chi$ can be chosen so that one has the following
  estimates: for $n\geq 0$, $\alpha\geq0$, $\beta\geq0$, there exists
  $C_{n,\alpha,\beta}>0$ such that
  \begin{equation}
    \label{estunif}
    \sup_{0<\vert y\vert\leq 1}\sup_{x\in\R}\left\vert
      x^\alpha\frac{\partial^\beta}{\partial x^\beta}\left(\vert
        y\vert^{-n}\cdot \frac{\partial\tilde\chi}{\partial\overline
          z}(x+iy)\right)\right\vert\\\leq
    C_{n,\alpha,\beta}\sup_{\substack{\beta'\leq n+\beta+2
        \alpha'\leq\alpha}}\sup_{x\in\R}\left\vert
      x^{\alpha'}\frac{\partial^{\beta'}\chi}{\partial
        x^{\beta'}}(x)\right\vert.
  \end{equation}
  By the definition of $\chi$, the right hand side of~\eqref{estunif}
  is bounded uniformly in $E$ complex.\\
  Let $\chi\in\Coi(\R)$ and $\tilde\chi$ be an almost analytic
  extension of $\chi(x)$. Then, by \cite{He-Sj:90} and
  \cite{MR96m:82033}, we know that, for any $n$ and $\omega\in\Omega$,
  the following formula hold,
  \begin{equation}
    \label{hesj0}
    \chi(H_\bullet)= \frac i{2\pi}\int_\C
    \frac{\partial\tilde\chi}{\partial{\overline z}}(z)\cdot 
    (z-H_\bullet)^{-1}dz\wedge d{\overline z}
  \end{equation}
  where $H_\bullet=H_L, \ \tilde H_L,\ H_0^+$ or $H_k^-$.\\
  Using the geometric resolvent equation (see, e.g.,~\cite[Theorem
  5.20]{MR2509110}) and the Combes-Thomas estimate (see
 , e.g.,~\cite[Theorem 11.2]{MR2509110}), we know that for some $C>0$,
  for Im$z\not=0$,
  \begin{multline}
    \label{eq:168}
    \left|\left\langle\delta_0,\left[(\tilde
          H_L-z)^{-1}-(H_k^--z)^{-1}
        \right]\delta_0\right\rangle\right|\\+
    \left|\left\langle\delta_0,\left[(H_L-z)^{-1}-(H_0^+-z)^{-1}
        \right]\delta_0\right\rangle\right|\leq
    \frac{C}{|\text{Im}z|}e^{-L|\text{Im}z|/C}.
  \end{multline}
  Plugging~\eqref{eq:168} into~\eqref{hesj0} and
  using~\eqref{estunif}, we get
  \begin{equation*}
    \left|\sum_{j=0}^L\chi(\lambda_j)|\varphi_j(0)|^2-
      \int_{\R}\chi(\lambda)dN^+_0(\lambda)
    \right|\leq \tilde C_N\int_{|y|\leq 1}|y|^{N-1}
    e^{-L|y|/C}dy\leq C_N L^{-N}
  \end{equation*}
  Thus, by~\eqref{eq:69} and~\eqref{eq:173}, we obtain that, for
  $\chi\in\Coi(\overset{\circ}{\Sigma}_\Z)$ and any $N\geq0$, there
  exists $C_N>0$ such that
  \begin{equation}
    \label{eq:71}
    \begin{split}
      &\left|\int_{\R} \chi\left( \lambda\right)\left[
          f_{0,L}\left(\lambda\right)\theta'_{p,L}(\lambda)-
          f^+_0\left(\lambda\right)\theta'_p(\lambda)
        \right]d\lambda\right|\\&= \left|\int_{\R} \chi\left(
          \lambda\right)
        f_{0,L}\left(\lambda\right)\theta'_{p,L}(\lambda)d\lambda-
        \int_{\R}\chi(\lambda)dN^+_0(\lambda) \right|\leq C_N L^{-N}.
    \end{split}
  \end{equation}
  Now, by~\eqref{eq:23} and~\eqref{eq:171}, the function
  $f_{0,L}\theta'_{p,L}-f^+_0\theta'_p$ admits an expansion in inverse
  powers of $L$ that is converging uniformly on compact subsets of
  $\overset{\circ}{\Sigma}_\Z$, namely,
  \begin{equation*}
    f_{0,L}\theta'_{p,L}-f^+_0\theta'_p=\sum_{k\geq1} L^{-k}\alpha_k.
  \end{equation*}
  Thus,~\eqref{eq:71} immediately yields that, for any $k\geq1$, one
  has $\alpha_k\equiv0$ on $\overset{\circ}{\Sigma}_\Z$. Hence,
  $f_{0,L}\theta'_{p,L}\equiv f^+_0\theta'_p$ on
  $\overset{\circ}{\Sigma}_\Z$. This completes the proof of
  Corollary~\ref{cor:1}.
\end{proof}
\subsection{The proofs of Theorems~\ref{thr:16} and~\ref{thr:30}}
\label{sec:proof-theorem-2}
We will first describe some objects from the spectral theory of
$H^\Z$, use them to describe the spectral theory of $H^\N$, prove
Theorem~\ref{thr:16} and finally prove Theorem~\ref{thr:30}.
\subsubsection{The spectral theory of $H^\Z$}
\label{sec:spectral-theory-Z}
This material is classical (see, e.g.,~\cite{MR0650253,MR1711536}); we
only recall it for the readers convenience. For $0\leq j\leq p-1$, define
$\tilde T_j=\tilde T_j(E)$ to be a monodromy matrix for the periodic
finite difference operator $H^\Z$, that is ,
\begin{equation}
  \label{eq:223}
  \tilde T_j(E)= T_{j+p-1,j}(E)=T_{j+p-1}(E)\cdots
  T_j(E)=:\begin{pmatrix}a^j_p(E)&b^j_p(E)\\
    a^j_{p-1}(E)&b^j_{p-1}(E)\end{pmatrix} 
\end{equation}
where
\begin{equation}
  \label{eq:242}
  T_j(E)=\begin{pmatrix}E-V_j&-1\\1&0\end{pmatrix}.  
\end{equation}
The coefficients of $\tilde T_j(E)$ are monic polynomials in the
energy $E$: $a^j_p(E)$ has degree $p$ and $b^j_p(E)$ has degree
$p-1$. Clearly, det$\,\tilde T_j(E)=1$. As $j\mapsto V_j$ is
$p$-periodic, so is $j\mapsto \tilde T_j(E)$. Moreover, for $j'<j$,
one has
\begin{equation}
  \label{eq:39}
  \tilde T_j(E)\,T_{j,j}(E)=T_{j+p-1,j'+p-1}(E)\,\tilde T_{j'}(E)
  =T_{j,j'}(E)\,\tilde T_{j'}(E).
\end{equation}
Thus, the discriminant $\underline{\Delta}(E):=\,$tr$\,\tilde
T_j(E)=a^j_p(E)+ b^j_{p-1}(E)$ is a polynomial of degree $p$ that is
independent of $j$; so are $\rho(E)$ and $\rho^{-1}(E)$, the
eigenvalues of $\tilde T_j(E)$. One defines the Bloch quasi-momentum
$E\mapsto\theta_p(E)$ by
\begin{equation}
  \label{eq:187}
  \underline{\Delta}(E)=\rho(E)+\rho^{-1}(E)=2\cos(p\,\theta_p(E)).
\end{equation}
Let us recall some basic properties of the discriminant $\Delta$ and
the coefficients of $\tilde T_j$, the proofs of which can be found
in~\cite{MR0650253}:
\begin{enumerate}
\item if $\Delta'(E)=0$ then $|\underline{\Delta}(E)|\geq2$;
\item the zeros of $\Delta'$ are simple;
\item $E$ is a zero of $\Delta'$ s.t.  $|\underline{\Delta}(E)|=2$ if
  and only if $\tilde T_j(E)\in\{+\,$Id,$-\,$Id$\}$ (for any $j$);
\item the polynomials $b_p^j$ and $a_{p-1}^j$ only vanish in the set
  $\{|\underline{\Delta}(E)|\geq2\}$ ; they keep a fixed sign in each
  of the connected components of the set
  $\{|\underline{\Delta}(E)|<2\}$.
\end{enumerate}
Note that $\underline{\Delta}(E)$ is real when $E$ is real. Thus, for
$E$ real, $|\underline{\Delta}(E)|\leq 2$ implies that
$\rho^{-1}(E)=\overline{\rho(E)}$ and $|\underline{\Delta}(E)|>2$ that
$\rho(E)$ is real. When $|\underline{\Delta}(E)|\leq 2$, we will fix
$\rho(E):=e^{ip\theta_p(E)}$ and when $|\underline{\Delta}(E)|> 2$, we
will fix $\rho(E)$ so that $|\rho(E)|<1$.\\
$E$ belongs to the spectrum of $H^\Z$ (i.e. $-\Delta+V$ on
$\ell^2(\Z)$) if and only if $|\underline{\Delta}(E)|\leq 2$ (see,
e.g.,~\cite{MR1711536}).\\
Properties (1)-(3) above imply that, for $E_0$ a zero of $\Delta'$
such that $\underline{\Delta}(E_0)=\pm2$, $\theta_p$ is real analytic
near $E_0$ and $\theta'_p(E_0)\not=0$.
\begin{Def}
  \label{def:1}
  $E_0$ is said to be a closed gap if and only if
  $|\underline{\Delta}(E_0)|=2$ and $\Delta'(E_0)=0$ or equivalently
  if and only if $\tilde T_0(E_0)$ is diagonal.
\end{Def}
\noindent Consider $\partial\Sigma_\Z$. It is the set of energies
solutions to $|\underline{\Delta}(E)|=2$ where $\tilde T_0(E)$ is not
diagonal; it is also the set of roots of $|\underline{\Delta}(E)|=2$
that are not closed gaps. From the upper half of the complex plane,
one can continue $E\mapsto\theta_p(E)$ analytically to the universal
cover of $\C\setminus\partial\Sigma_\Z$. Each of the points in
$\partial\Sigma_\Z$ is a branch point of $\theta_p$ of square root
type.  Moreover, for $E\not\in\partial\Sigma_\Z$, there exists two
linearly independent solutions to the eigenvalue equation
$(-\Delta+V-E)u=0$, say $\varphi_\pm(E)$, satisfying, for $n\in\Z$
\begin{equation}
  \label{eq:28}
  \varphi_\pm(n+p,E)=e^{\pm i p\theta_p(E)}\varphi_\pm(n,E).
\end{equation}
\subsubsection{The spectral theory of $H^\N$}
\label{sec:spectral-theory}
Let us now turn to the spectrum of the operator on the half-lattice.
\paragraph{\it The operator $H_0^+$}
\label{sec:operator}
For the operator $H_0^+=H^\N$ (that is $-\Delta+V$ on $\ell^2(\N)$
with Dirichlet boundary conditions at $0$), $E$ is in the spectrum if
and only if
\begin{itemize}
\item either $|\underline{\Delta}(E)|\leq 2$
\item or $|\underline{\Delta}(E)|>2$ and $\D [\tilde
  T_0(E)]^n\begin{pmatrix}1\\0\end{pmatrix}$ stays bounded as
  $n\to+\infty$.
\end{itemize}
The second condition is equivalent to asking that $\D [\tilde
T_j(E)]^nT_{j-1}(E)\cdots T_0(E)\begin{pmatrix}1\\0\end{pmatrix}$
stay bounded as $n\to+\infty$.\\
When $|\underline{\Delta}(E)|\not=2$ and $a^0_{p-1}(E)\not=0$, one can
diagonalize $\tilde T_0(E)$ in the following way
\begin{multline}
  \label{eq:33}
  \begin{pmatrix} a^0_{p-1}(E)&\rho(E)-a^0_p(E)\\
    -a^0_{p-1}(E)&a^0_p(E)-\rho^{-1}(E) \end{pmatrix}\times \tilde T_0(E)\\=
  \begin{pmatrix}\rho(E)& 0\\0&\rho^{-1}(E) \end{pmatrix}\times
  \begin{pmatrix} a^0_{p-1}(E)&\rho(E)-a^0_p(E)\\
    -a^0_{p-1}(E)&a^0_p(E)-\rho^{-1}(E) \end{pmatrix}.
\end{multline}
Thus, using
\begin{equation}
  \label{eq:43}
  \left|\begin{matrix}
      \rho(E)-a_p^0(E)&-b_p^0(E) \\-a_{p-1}^0(E) &\rho(E)-b_{p-1}^0(E)
    \end{matrix}\right|=
  \left|\begin{matrix}
      \rho(E)-a_p^0(E)&-b_p^0(E) \\-a_{p-1}^0(E) &a_p^0(E)-\rho^{-1}(E)
    \end{matrix}\right|=0
\end{equation}
for $n\in\Z$, one computes
\begin{equation}
  \label{eq:42}
  \left(\tilde T_0(E)\right)^n=
  \begin{pmatrix}\tilde t_{0,n}^{11}(E)&\tilde t_{0,n}^{12}(E)\\\tilde
    t_{0,n}^{21}(E)& \tilde t_{0,n}^{22}(E)\end{pmatrix}
\end{equation}
where
\begin{equation}
  \label{eq:210}
  \begin{aligned} & \tilde
    t_{0,n}^{11}(E):=\rho^n(E)
    \frac{a^0_p(E)-\rho^{-1}(E)}{\rho(E)-\rho^{-1}(E)}
    +\rho^{-n}(E)\frac{\rho(E)-a^0_p(E)}{\rho(E)-\rho^{-1}(E)},
    \\
    &\tilde t_{0,n}^{12}(E):=
    \left(\rho^{-n}(E)-\rho^n(E)\right)\frac{b^0_p(E)}{\rho(E)-\rho^{-1}(E)},\\
    &\tilde t_{0,n}^{21}(E):=\left(\rho^n(E)-\rho^{-n}(E)\right)
    \frac{a^0_{p-1}(E)}{\rho(E)-\rho^{-1}(E)},\\&
    \tilde
    t_{0,n}^{22}(E):=\rho^{-n}(E)\frac{a^0_p(E)-\rho^{-1}(E)}
    {\rho(E)-\rho^{-1}(E)}
    +\rho^n(E)\frac{\rho(E)-a^0_p(E)}{\rho(E)-\rho^{-1}(E)}.
  \end{aligned}
\end{equation}
Clearly, the formulas~\eqref{eq:33},~\eqref{eq:42} and~\eqref{eq:210}
stay valid even if $a^0_{p-1}(E)=0$. They also stay valid if
$|\underline{\Delta}(E)|=2$ and $\Delta'(E)=0$. Indeed, by points
(1)-(3) in section~\ref{sec:spectral-theory-Z}, the functions
$\rho-\rho^{-1}$, $a^0_p-\rho^{-1}$, $-\rho-a^0_p$, $b^0_p$ and
$a^0_{p-1}$ are analytic
near and have simple zeros at such points. \\
We have thus proved that
\begin{Le}
  \label{le:14}
  For $E\not\in\partial\Sigma_\Z$, $\left(\tilde T_0(E)\right)^n$ has
  the form~\eqref{eq:42}~-~\eqref{eq:210}
\end{Le}
\noindent Simple computations then show that $E$ is in the spectrum of
$H_0^+$, that is, $-\Delta+V$ on $\ell^2(\N)$ with Dirichlet boundary
conditions at $0$ if and only if one of the following conditions is
satisfied:
\begin{enumerate}
\item $|\underline{\Delta}(E)|\leq 2$: moreover, the set $\{E\in\R;\
  |\underline{\Delta}(E)|\leq 2\}$ is contained in the absolutely
  continuous spectrum of $H_0^+$;
\item $|\underline{\Delta}(E)|>2$ and
  \begin{equation}
    \label{eq:119}
    a^0_{p-1}(E)=0\quad\text{ and }\quad |a^0_p(E)|<1.
  \end{equation}
\end{enumerate}
Thus, on $\Sigma_\Z$, the spectrum of $H_0^+$ is purely absolutely
continuous; it does not contain any embedded eigenvalues. \\
Note that, in case (2), $\D [\tilde
T_0(E)]^n\begin{pmatrix}1\\0\end{pmatrix}$ actually decays
exponentially fast. In this case, $E$ is an eigenvalue associated to
the (non normalized) eigenfunction $(u_l)_{l\in\N}$ where, for
$n\geq0$ and $j\in\{0,\cdots,p-1\}$,
\begin{equation}
  \label{eq:127}
  \begin{split}
    u_{np+j}(E)&= \left\langle T_{j-1}(E)\cdots T_0(E)
      \begin{pmatrix}1\\0 \end{pmatrix},
      \begin{pmatrix} 1\\0 \end{pmatrix}\right\rangle\cdot
    \left[a^0_p(E)\right]^n\\&=a_j(E)\left[a^0_p(E)\right]^n
  \end{split}
\end{equation}
writing
\begin{equation}
  \label{eq:118}
  T_{j-1}(E)\cdots T_0(E)=: \begin{pmatrix} a_j(E)&b_j(E)\\a_{j-1}(E)
    &b_{j-1}(E) \end{pmatrix}.
\end{equation}
It is well know that, for any $j$, the zeros of $a_j$ and $b_j$ are
simple (see, e.g.,~\cite[section 4]{MR1711536}), and the roots of
$a_{j+1}$ (resp. $b_{j+1}$) interlace those of $a_j$ (resp. $b_j$).
Let $E'$ be an eigenvalue of $H_0^+$. Differentiating~\eqref{eq:43} at
the energy $E'$, we compute
\begin{equation}
  \label{eq:115}
  b_p^0(E')\frac{d a_{p-1}^0}{dE}(E')
  +(\rho(E')-\rho^{-1}(E'))\frac{d (\rho-a_p^0)}{dE}(E')=0.
\end{equation}
\paragraph{\it The eigenvalues of the operator $H_k^-$}
\label{sec:operator2}
Let us now turn to $H^-_k$. Recalling~\eqref{eq:118} and using the
representation~\eqref{eq:42}, we obtain that the eigenvalues of
$H^-_k$ outside $\Sigma_\Z$ satisfy
\begin{equation}
  \label{eq:34}
  \begin{pmatrix}
    \rho(E)-a_p^0(E)&-a_{p-1}^0(E) \\-b_p^0(E) &a_p^0(E)-\rho^{-1}(E)
  \end{pmatrix}
  \begin{pmatrix}
    a_{k+1}(E)\\ b_{k+1}(E)
  \end{pmatrix}=0.
\end{equation}
As for $H_0^+$, the eigenfunction associated to $E$ and $H_k^-$ decays
exponentially fast. Indeed, the eigenvalues of $H_k^-$ in the region
$|\underline{\Delta}(E)|>2$ can be analyzed as we analyzed those of
$H_0^+$, i.e., they are the energies such that $\D [\tilde
T_k(E)]^{-n}\begin{pmatrix}0\\1\end{pmatrix}$ stays bounded; this
yields the quantization conditions $ b^k_p(E)=0\text{ and
}|b^k_{p-1}(E)|<1$. In this case, $E$ is an eigenvalue associated to
the (non normalized) eigenfunction $(u_l)_{-l\in\N}$ where, for
$n\geq0$ and $k\in\{0,\cdots,p-1\}$,
\begin{equation}
  \label{eq:126}
  u_{-np-k}(E)=b_k(E)\left[b^k_{p-1}(E)\right]^{-n}.
\end{equation}
\paragraph{\it Common eigenvalues to $H_0^+$ and $H_k^-$}
\label{sec:operator3}
Assume now that $E'$ is simultaneously an eigenvalue of $H_k^-$ and
$H_0^+$. In this case, one has $a_{p-1}^0(E')=0$, $|a_p^0(E')|<1$ and
$b_p^0(E')b_{k+1}(E')=a_{k+1}(E')(\rho^{-1}(E')-\rho(E'))$. So~\eqref{eq:34}
(see also~\eqref{eq:115}) becomes
\begin{equation}
  \label{eq:116}
  \begin{pmatrix}
    \frac{d (\rho-a_p^0)}{dE}(E')&-\frac{d a_{p-1}^0}{dE}(E')
    \\-b_p^0(E) &a_p^0(E')-\rho^{-1}(E')
  \end{pmatrix} \begin{pmatrix} a_{k+1}(E')\\
    b_{k+1}(E') \end{pmatrix}=0.
\end{equation}
Hence, the analytic function $E\mapsto
a_{k+1}(E)(a_p^0(E)-\rho(E))-b_{k+1}(E)a_{p-1}^0(E)$ has a root of
order at least 2 at $E'$. It also implies that
$a_{k+1}(E')\not=0$. Indeed, if $a_{k+1}(E')=0$,~\eqref{eq:116}
implies
$b_{k+1}(E')=0$ as $\frac{d a_{p-1}^0}{dE}(E')\not=0$.\\
Conversely, if $E'\in\sigma(H_0^+)$ such that
$|\underline{\Delta}(E')|>2$ and $E\mapsto
a_{k+1}(E)(a_p^0(E)-\rho(E))-b_{k+1}(E)a_{p-1}^0(E)$ has a root of
order at least 2 at $E'$, then~\eqref{eq:116} holds and $E'$
is an eigenvalue of $H_k^-$.\\
We have thus proved
\begin{Le}
  \label{le:18}
  $E_0\in\sigma(H_0^+)\cap\sigma(H_k^-)\setminus\Z$ if and only if
  $|\underline{\Delta}(E_0)|>2$ and $E_0$ is a double root of
  $E\mapsto a_{k+1}(E)(a_p^0(E)-\rho(E))-b_{k+1}(E)a_{p-1}^0(E)$.
\end{Le}
\subsubsection{The Dirichlet eigenvalues for a periodic potential :
  the proof of Theorem~\ref{thr:16}}
\label{sec:dirichl-eigenv}
Let us now turn to the study of the eigenvalues and eigenvectors of
$H_L$, i.e., to the proof of Theorem~\ref{thr:16}.  We first prove the
statements for the eigenvalues and then, in the next section,
turn to the eigenvectors.\\
Recall that $L\equiv k\mod p$; we write $L=Np+k$. By definition, $E$
is an eigenvalue of $-\Delta+V$ on $\llbracket 0,L\rrbracket$ with
Dirichlet boundary conditions if and only if
\begin{equation}
  \label{eq:32}
  \begin{split}
    0&=\text{det}\left(T_{L+1}(E)T_L(E)T_{L-1}(E)\cdots
      T_0(E)\begin{pmatrix}1 \\0\end{pmatrix},
      \begin{pmatrix}0\\1\end{pmatrix}\right)\\&= \text{det}\left(
      T_{k}(E)\cdots T_0(E)\cdot [\tilde T_0(E)]^N
      \begin{pmatrix}1\\0\end{pmatrix},
      \begin{pmatrix}0\\1\end{pmatrix}\right)
  \end{split}
\end{equation}
where $\tilde T_k(E)$ is  the monodromy matrix defined above.\\
We use the notations of sections~\ref{sec:spectral-theory}
and~\ref{sec:spectral-theory-Z}. Let us first show point (1) of
Theorem~\ref{thr:16}, namely,
\begin{Le}
  \label{le:11} For $L$ large, one has
  \begin{equation*}
    \partial\Sigma_\Z\cap\sigma(H_L)=\{E_0;\
    a_{k+1}(E_0)=a_{p-1}^0(E_0)=0\text{ and }b^0_p(E_0)\not=0\}.
  \end{equation*}
\end{Le}
\begin{proof}
  For $E_0\in \partial\Sigma_\Z$, we know that
  $|\underline{\Delta}(E_0)|=2$ and $\tilde T_0(E_0)$ is not
  diagonal. Assume $\underline{\Delta}(E_0)=2$ (the case
  $\underline{\Delta}(E_0)=-2$ is dealt with in the same way); hence,
  $\tilde T_0(E_0)$ has a Jordan normal form, i.e., there exists $P$,
  a $2\times2$ invertible matrix and $\alpha\in\R^*$ such that
  \begin{equation}
    \label{eq:218}
    \tilde T_0(E_0)=P^{-1} \begin{pmatrix} 1 & 0 \\ \alpha & 1
    \end{pmatrix} P\quad\text{where}\quad P=
    \begin{pmatrix} p_{11}&p_{12}\\ p_{21}&p_{22} \end{pmatrix}.
  \end{equation}
  Thus, by~\eqref{eq:32}, $E_0\in\sigma(H_L)$ is and only if
  \begin{equation}
    \label{eq:219}
    \begin{split}
      0&=\left|\begin{pmatrix} a_{k+1}(E_0) & b_{k+1}(E_0)\\a_{k}(E_0)
          & b_{k}(E_0)\end{pmatrix}\, \left(\tilde T_0(E_0)
        \right)^N\,\begin{pmatrix}1\\0\end{pmatrix},
        \begin{pmatrix}0\\1\end{pmatrix} \right| \\&=
      \left| \begin{pmatrix} a_{k+1}(E_0) & b_{k+1}(E_0)\\a_{k}(E_0)
          & b_{k}(E_0)\end{pmatrix}\, P^{-1}\begin{pmatrix} 1 & 0\\
          N \alpha &1\end{pmatrix}P\,\begin{pmatrix}1\\0\end{pmatrix},
        \begin{pmatrix}0\\1\end{pmatrix} \right|,
    \end{split}
  \end{equation}
  that is,
  \begin{equation*}
    \begin{split}
      0&=\left|\begin{pmatrix} 1 & 0\\ N \alpha & 1\end{pmatrix}
        P\, \begin{pmatrix}1\\0\end{pmatrix},P\, \begin{pmatrix}
          -b_{k+1}(E_0) \\a_{k+1}(E_0)\end{pmatrix}\right| \\&=
      (\text{det}\,P)\, a_{k+1}(E_0)-N\,\alpha\, p_{11}\,
      (-p_{11}b_{k+1}(E_0)+ p_{12} a_{k+1}(E_0)).
    \end{split}
  \end{equation*}
  For $N$ large, this expression vanishes if and only if
  $(\text{det}\,P)\, a_{k+1}(E_0)=0$ and $\alpha\, p_{11}\,
  (-p_{11}b_{k+1}(E_0)+ p_{12} a_{k+1}(E_0))=0$. As $P$ is invertible,
  as $|b_{k+1}(E_0)|+|a_{k+1}(E_0)|\not=0$ and as $\alpha\not=0$, one
  has $a_{k+1}(E_0)=0$ and $p_{11}=0$.\\
  In this case, using $b_{k+1}(E_0)\not=0$, we can then rewrite the
  eigenvalue equation~\eqref{eq:219} as
  \begin{equation}
    \label{eq:220}
    0=\left| (\tilde T_0(E_0))^N\begin{pmatrix}1\\0\end{pmatrix},
      \begin{pmatrix} 1\\0\end{pmatrix}\right|=\tilde t_{0,N}^{21}(E_0)
  \end{equation}
  For $E\in\overset{\circ}{\Sigma}_\Z$ close to $E_0$,
  by~\eqref{eq:210}, we have
  \begin{equation*}
    \begin{split}
      t_{0,N}^{21}(E)= \frac{\left(\rho^N(E)-\rho^{-N}(E)\right)\,
        a^0_{p-1}(E)}{\rho(E)-\rho^{-1}(E)}=\rho^{N-1}
      \left(\sum_{j=0}^{N-1}\rho^{-2j}(E) \right) a^0_{p-1}(E).
    \end{split}
  \end{equation*}
  As $\rho$ is continuous at $E_0$ and $\rho^2(E_0)=1$, taking $E$ to
  $E_0$, we get
  \begin{equation*}
    a^0_{p-1}(E_0)=0.
  \end{equation*}
  As $\tilde T_0(E_0)$ is not diagonal, this implies
  $b^0_p(E_0)\not=0$. This completes the proof of Lemma~\ref{le:11}.
\end{proof}
\noindent Now, pick $E\not\in\partial\Sigma_\Z$. Then, by
Lemma~\ref{le:14}, the quantization condition~\eqref{eq:32} becomes
\begin{equation}
  \label{eq:35}
  \left|\begin{matrix}\D \rho^N(E)
      \frac{a^0_p(E)-\rho^{-1}(E)}{\rho(E)-\rho^{-1}(E)}
      +\rho^{-N}(E)\frac{\rho(E)-a^0_p(E)}{\rho(E)-\rho^{-1}(E)}
      & -b_{k+1}(E) \\\D \left(\rho^N(E)-\rho^{-N}(E)\right)
      \frac{a^0_{p-1}(E)}{\rho(E)-\rho^{-1}(E)}
      & a_{k+1}(E)\end{matrix} \right|=0.
\end{equation}
\paragraph{\it The eigenvalues outside of $\Sigma_\Z$}
\label{sec:eiog-outs-}
Let us first study the eigenvalues outside $\Sigma_\Z$, i.e., in the
region $|\underline{\Delta}(E)|>2$. If, for $j\in\N$, we define
\begin{equation}
  \label{eq:74}
  \begin{split}
    \alpha_j(E):=a_j(E)\frac{a^0_p(E)-\rho^{-1}(E)}{\rho(E)-\rho^{-1}(E)}
    +b_j(E)\frac{a^0_{p-1}(E)}{\rho(E)-\rho^{-1}(E)}\\
    \text{and}\quad
    \beta_j(E):=a_j(E)\frac{\rho(E)-a^0_p(E)}{\rho(E)-\rho^{-1}(E)}-b_j(E)
    \frac{a^0_{p-1}(E)}{\rho(E)-\rho^{-1}(E)},
  \end{split}
\end{equation}
equation~\eqref{eq:35} can be rewritten as
$\beta_{k+1}(E)+\rho^{2N}(E) \alpha_{k+1}(E)=0$; using
\begin{equation}
  \label{eq:37}
  \alpha_{k+1}(E)+\beta_{k+1}(E)=a_{k+1}(E),
\end{equation}
\eqref{eq:35} becomes
\begin{equation}
  \label{eq:36}
  \beta_{k+1}(E)=-\frac{\rho^{2N}(E)}{1-\rho^{2N}(E)}\,a_{k+1}(E).
\end{equation}
We first show
\begin{Le}
  \label{le:19}
  There exists $\eta>0$ such that, for $L$ sufficiently large,
  $\sigma(H_L)\cap[(\Sigma_\Z+[-\eta,\eta])\setminus\Sigma_\Z]=\emptyset$.
\end{Le}
\begin{proof}
  Using~\eqref{eq:74}, we rewrite~\eqref{eq:36} as
  \begin{equation}
    \label{eq:234}
    a_{k+1}(E)(\rho(E)-a_p^0(E))-b_{k+1}(E)a_{p-1}^0(E)
    =\rho^{2N+1}(E)\frac{1-\rho^2(E)}{1-\rho^{2N}(E)}\,a_{k+1}(E).
  \end{equation}
  Pick $E_0\in\partial\Sigma_Z$. Then, by our choice for $\rho$, for
  $\eta>0$ small, we know that, for
  $E\in([E_0-\eta,E_0+\eta])\setminus\Sigma_\Z$,
  $\rho^2(E)=e^{-c_0\sqrt{|E-E_0|}(1+O(\sqrt{|E-E_0|}))}$. Hence, for
  $E\in([E_0-\eta,E_0+\eta])\setminus\Sigma_\Z$, one has
  \begin{equation}
    \label{eq:235}
    \left|\rho^{2N+1}(E)\frac{1-\rho^2(E)}{1-\rho^{2N}(E)}\right|
    \lesssim \min\left(\sqrt{|E-E_0|},\frac1N\right).
  \end{equation}
  Thus, if $a_{k+1}(E_0)(\rho(E_0)-a_p^0(E_0))-b_{k+1}(E_0)
  a_{p-1}^0(E_0)\not=0$, equation~\eqref{eq:234} has no solution in
  $[E_0-\eta,E_0+\eta]\setminus\Sigma_\Z$ for $\eta$ small and $L$
  sufficiently large.\\
  Let us now assume that
  $a_{k+1}(E_0)(\rho(E_0)-a_p^0(E_0))-b_{k+1}(E_0)a_{p-1}^0(E_0)=0$. Hence,
  \begin{itemize}
  \item if $a_{k+1}(E_0)\not=0$: one computes
    \begin{gather*}
      a_{k+1}(E)(\rho(E)-a_p^0(E))-b_{k+1}(E)a_{p-1}^0(E)
      =a_{k+1}(E_0)(\rho(E)-\rho(E_0))(1+o(1))\\\intertext{and}
      \rho^{2N+1}(E)\frac{1-\rho^2(E)}{1-\rho^{2N}(E)}\,a_{k+1}(E)=
      -(\rho(E)-\rho(E_0))\,a_{k+1}(E_0)
      \frac{\rho^{2(N+1)}(E)}{1-\rho^{2N}(E)}(1+o(1)).
    \end{gather*}
    Hence, for $\eta>0$ small and
    $E\in[E_0-\eta,E_0+\eta]\setminus\Sigma_\Z$, the two sides of
    equation~\eqref{eq:234} have opposite signs: there is no solution
    to equation~\eqref{eq:234} in this interval;
  \item if $a_{k+1}(E_0)=0$: then $b_{k+1}(E_0)\not=0$, $a^0_{p-1}(E_0)=0$,
    $\rho(E_0)=a^0_p(E_0)$ and $(a^0_{p-1})'(E_0)\not=0$; one computes
    \begin{gather*}
      a_{k+1}(E)(\rho(E)-a_p^0(E))-b_{k+1}(E)a_{p-1}^0(E)
      =-b_{k+1}(E_0)(a^0_{p-1})'(E_0)(E-E_0)(1+o(1)) \\\intertext{and,
        by~\eqref{eq:235}, for $\eta>0$ small and
      $E\in[E_0-\eta,E_0+\eta]\setminus\Sigma_\Z$,}
      \left|\rho^{2N+1}(E)\frac{1-\rho^2(E)}{1-\rho^{2N}(E)}\,a_{k+1}(E)\right|
      \lesssim |E-E_0|\min\left(\sqrt{|E-E_0|},\frac1N\right)
    \end{gather*}
    Hence, for $\eta>0$ small and
    $E\in[E_0-\eta,E_0+\eta]\setminus\Sigma_\Z$, there is no solution
    to equation~\eqref{eq:234} in this interval.
  \end{itemize}
  This completes the proof of Lemma~\ref{le:19}.
\end{proof}
\noindent In Lemma~\ref{le:11}, we saw that, if
$E_0\in\partial\Sigma_\Z$ satisfies $a_{k+1}(E_0)=0$ and
$a_{k+1}(E_0)(\rho(E_0)-a_p^0(E_0))-b_{k+1}(E_0)a_{p-1}^0(E_0)=0$,
then $E_0$ is an eigenvalue of $H_L$ for $L$ large.\\
By Lemma~\ref{le:19}, if now suffices to consider energies such that
$|\underline{\Delta}(E)|>2+\eta$ for some $\eta>0$. In this case, we
note that the left hand side in~\eqref{eq:36} is the left hand side of
the first equation in~\eqref{eq:34} (up to the factor $\rho-\rho^{-1}$
that does not vanish outside $\Sigma_\Z$). On the other hand, the
right hand side in~\eqref{eq:36} is uniformly exponentially small for
large $N$ on $\{E\in\R;\ |\underline{\Delta}(E)|>2+\eta\}$. Thus, for
$L$ large, the solutions to~\eqref{eq:36} are exponentially close to
$E'$ that is either an eigenvalue of $H_0^+$ or one of $H_k^-$. One
distinguishes between the following cases:
\begin{enumerate}
\item if $E'$ is an eigenvalue of $H_0^+$ but not of $H_k^-$, then
  $E'$ is a simple root of the function $E\mapsto \beta_{k+1}(E)$ (see
  section~\ref{sec:operator3}); one has to distinguish two cases
  depending on whether $a_{k+1}(E')$ vanishes or not. Assume first
  $a_{k+1}(E')=0$; then, by~\eqref{eq:127}, we know that the
  eigenvector of $H_0^+$ actually satisfies the Dirichlet boundary
  conditions at $L$; thus, $E'$ is a solution to~\eqref{eq:36}, i.e.,
  an eigenvalue of $H_L$, and~\eqref{eq:127} gives a (non
  normalized) eigenvector.\\
  Assume now that $a_{k+1}(E')\not=0$; then, by Rouch{\'e}'s Theorem, the
  unique solution to~\eqref{eq:36} close to $E'$ satisfies
  \begin{equation}
    \label{eq:121}
    E-E'=-\frac{\rho^{2N}(E')}{\beta'_{k+1}(E')}a_{k+1}(E')
    (1+o(\rho^{2N}(E')));
  \end{equation}
\item if $E'$ is an eigenvalue of $H_k^-$ but not of $H_0^+$, mutandi
  mutandis, the analysis is the same as in point (1);
\item if $E'$ is an eigenvalue of both $H_0^+$ and $H_k^-$, then, we
  are in a resonant tunneling situation. The analysis done in the
  appendix, section~\ref{sec:appendix}, shows that near $E'$, $H_L$
  has two eigenvalues, say $E_\pm$ satisfying, for some constant
  $\alpha>0$,
  \begin{equation}
    \label{eq:120}
    E_\pm-E'=\pm \alpha\,\rho^N(E'))
    \left(1+O\left(N\rho(E')^N\right)\right).
  \end{equation}
\end{enumerate}
This completes the proof of the statements of Theorem~\ref{thr:16} for
the eigenvalues outside $\Sigma_\Z$. \vskip.1cm\noindent
\paragraph{\it The eigenvalues inside $\Sigma_\Z$}
\label{sec:eiog-ins-}
We now study the eigenvalues in the region
$\overset{\circ}{\Sigma}_\Z$. One can express $\rho(E)$ in terms of
the Bloch quasi-momentum $\theta_p(E)$ and use
$\rho^{-1}(E)=\overline{\rho(E)}$. Notice that, on
$\overset{\circ}{\Sigma}_\Z$, one has
\begin{itemize}
\item Im$\,\rho(E)$ does not vanish
\item the function $E\mapsto\rho(E)$ is real analytic,
\item the functions $E\mapsto a^0_p(E)$, $E\mapsto a^0_{p-1}(E)$,
  $E\mapsto a_{k+1}(E)$ and $E\mapsto b_{k+1}(E)$ are real valued
  polynomials.
\end{itemize}
We prove
\begin{Le}
  \label{le:15}
  The function $\alpha_{k+1}$ is analytic and does not vanish on
  $\overset{\circ}{\Sigma}_\Z$.
\end{Le}
\begin{proof}
  Assume that the function $\alpha_{k+1}$ vanishes at a point $E_0$ in
  $\overset{\circ}{\Sigma}_\Z$:
  \begin{itemize}
  \item if $\rho(E_0)\not=\rho^{-1}(E_0)$: then, one has
    $a_{k+1}(E_0)\,(a^0_p(E_0)-\rho^{-1}(E_0))+b_{k+1}(E_0)\,
    a^0_{p-1}(E_0)=0$: as $\rho(E_0)\not= \rho^{-1}(E_0)$ and
    $E_0\in\overset{\circ}{\Sigma}_\Z$, one has
    $\rho^{-1}(E_0)=\overline{\rho(E_0)}\not\in\R$; thus, for
    $a_{k+1}(E_0)\,(a^0_p(E_0)-\rho^{-1}(E_0))-b_{k+1}(E_0)\,
    a^0_{p-1}(E_0)$ to vanish, one needs $a_{k+1}(E_0)=0$ and
    $a^0_{p-1}(E_0)=0$ (as $b_{k+1}$ and $a_{k+1}$ don't vanish
    together); this implies that $\rho(E_0)=\pm1$ and contradicts
    $\rho(E_0)\not=\rho^{-1}(E_0)$;
  \item if $\rho(E_0)=\rho^{-1}(E_0)$: such a point $E_0$ is a simple
    root of the three functions $a^0_{p-1}$, $\rho-\rho^{-1}$ and
    $a^0_p-\rho$ that are analytic near $E_0$ (see points (1)-(4) in
    section~\ref{sec:spectral-theory-Z}). Moreover, one checks that
    the derivatives of these functions at that point are respectively
    real, purely imaginary and neither real, nor purely imaginary: for
    $E$ close to $E_0$, one has
    \begin{equation}
      \label{eq:225}
      \begin{split}
        a^0_{p-1}(E)&=A(E-E_0)(1+O(E-E_0)),\\
        \rho(E)-\rho^{-1}(E)&=2iC(E-E_0)(1+O(E-E_0)),
        \\a^0_p(E)-\rho^{-1}(E)&=(B+iC)(E-E_0)(1+O(E-E_0))
        \quad\text{where}\quad (A,B,C)\in (\R^*)^3.
      \end{split}
    \end{equation}
    Now, as $a_{k+1}$ and $b_{k+1}$ are real valued and can't vanish
    at the same point, we see that $\alpha_{k+1}(E_0)\not=0$.
  \end{itemize}
  This complete the proof of Lemma~\ref{le:15}
\end{proof}
\noindent Now, as $L=Np+k$, the characteristic equation~\eqref{eq:35}
(valid for $E\in\overset{\circ}{\Sigma}_\Z$) becomes
\begin{equation}
  \label{eq:38}
  \begin{split}
    \rho^{2N}(E)&=e^{2 iNp \theta_p(E)}=
    -\frac{\overline{\alpha_{k+1}(E)}}{\alpha_{k+1}(E)}=
    -\frac{\beta_{k+1}(E)} {\overline{\beta_{k+1}(E)}}
    \\&=\frac{a_{k+1}(E)(\rho(E)-a^0_p(E))-b_{k+1}(E)
      a^0_{p-1}(E)}{\overline{a_{k+1}(E)(\rho(E) -a^0_p(E))-b_{k+1}(E)
        a^0_{p-1}(E)}}=:e^{2i h_k(E)}.
  \end{split}
\end{equation}
By Lemma~\ref{le:15}, the function $E\mapsto h_k(E)$ defined
in~\eqref{eq:38} is real analytic on
$\overset{\circ}{\Sigma}_\Z$. Clearly, as inside $\Sigma_\Z$, $\rho$
is real only at bands edges or closed gaps, $h_k$ takes values in
$\pi\Z$ only at bands edges or closed gaps.  This implies point (a) of
Theorem~\ref{thr:16}. We prove
\begin{Le}
  \label{le:16}
  The function $h_k$ can be extended continuously from
  $\overset{\circ}{\Sigma}_\Z$ to $\Sigma_\Z$; for
  $E_0\in\partial\Sigma_\Z$, one has
  \begin{equation*}
    h_k(E_0)\in\begin{cases}
      \frac{\pi}2+\pi\Z&\text{ if }a_{k+1}(E_0)\not=0\text{ and }
      a_{k+1}(E_0)(\rho(E_0)-a^0_p(E_0))-b_{k+1}(E_0)a^0_{p-1}(E_0)=0,\\
      \pi\Z&\text{ if not}.
    \end{cases}
  \end{equation*}
  The function $\theta_{p,L}$ is strictly increasing on the bands of
  $\Sigma_\Z$.
\end{Le}
\begin{proof}
  Pick $E_0\in\partial\Sigma_\Z$.  It suffices to study the behavior
  of $E\in\Sigma_\Z\mapsto
  s(E):=a_{k+1}(E)(\rho(E)-a^0_p(E))-b_{k+1}(E) a^0_{p-1}(E)$ near
  $E_0$ inside $\Sigma_\Z$. Write $E=E_0\pm t^2$ for $t$ real
  positive; here, the sign $\pm$ depends on whether $E_0$ is a left or
  right edge of $\Sigma_\Z$ and is chosen so that $E=E_0\pm
  t^2\in\overset{\circ}{\Sigma}_\Z$ for $t$ small.\\
  First, $t\mapsto\rho(E_0\pm t^2)$ is analytic near $0$; thus, so is
  $t\mapsto s(E_0\pm t^2)$. Solving the characteristic equation
  $\rho^2(E)-\underline{\Delta}(E)\rho(E)+1=0$, one finds
  \begin{equation*}
    \rho(E_0\pm t^2)=\rho(E_0)+i a t+ b t^2+O(t^3),\quad a\in\R^*,\
    b\in\R.
  \end{equation*}
  Thus,
  \begin{equation*}
    s(E_0\pm t^2)=s(E_0)+ i a_{k+1}(E_0)\cdot a\cdot t+c\cdot
    t^2+O(t^3)
  \end{equation*}
  where
  \begin{equation*}
    c:=a'_{k+1}(E_0)(\rho(E_0)-a^0_p(E_0))+
    a_{k+1}(E_0)(b-(a^0_p)'(E_0))- (b'_{k+1}(E_0)
    a^0_{p-1}(E_0)+b_{k+1}(E_0)(a^0_{p-1})'(E_0)).
  \end{equation*}
  Hence,
  \begin{itemize}
  \item if $s(E_0)\not=0$, then $s(E_0\pm t^2)=s(E_0)+O(t)$ ; hence,
    $h_k(E_0\pm t^2)=\pi n+O(t)$ for some $n\in\Z$
  \item if $s(E_0)=0$ and $a_{k+1}(E_0)\not=0$, one has $s(E_0\pm
    t^2)=ia_{k+1}(E_0)\cdot a\cdot t+O(t^2)$; thus, $h_k(E_0\pm
    t^2)=\frac{\pi}2+\pi n+O(t)$ for some $n\in\Z$;
  \item if $s(E_0)=a_{k+1}(E_0)=0$, one has $b_{k+1}(E_0)\not=0$,
    $a^0_{p-1}(E_0)=0$, $\rho(E_0)=a^0_p(E_0)$ and
    $(a^0_{p-1})'(E_0)\not=0$; thus $s(E_0\pm
    t^2)=-b_{k+1}(E_0)(a^0_{p-1})'(E_0) t^2+0(t^2)$; hence,
    $h_k(E_0\pm t^2)=\pi n+O(t)$ for some $n\in\Z$.
  \end{itemize}
  This completes the proof of the statement of Lemma~\ref{le:16} on
  the function $h_k$.
  \vskip.1cm\noindent Let us now control the monotony of
  $\theta_{p,L}$ (see Theorem~\ref{thr:16}) on the bands of
  $\Sigma_\Z$. It is well known that keeping the above notations,
  $\theta_p(E_0\pm t^2)-\theta_p(E_0)=\pm\alpha t(1+t g_0(t))$ with
  $\alpha>$. The computations done in the previous paragraph show that
  $h_k(E_0\pm t^2)=h_k(E_0)+at^k(1+t g_1(t))$, $k\geq1$. Hence,
  \begin{itemize}
  \item if $k>1$, we have $\D\theta_{p,L}(E_0\pm
    t^2)-\theta_{p,L}(E_0)=\pm\alpha t(1+t g_2(t))$,
  \item if $k=1$, we have $\D\theta_{p,L}(E_0\pm
    t^2)-\theta_{p,L}(E_0)=\left(\pm\alpha+\frac{a}{L-k}\right)t(1+t
    g_2(t))$.
  \end{itemize}
  Hence, $\theta_{p,L}$ is strictly increasing inside the band near
  $E_0$ for $L$ sufficiently large. Outside a neighborhood of the
  edges of a band, by analyticity of $h_k$, as the bands are compact,
  we have $|\theta'_{p,L}-\theta'_{p}|\lesssim L^{-1}$.  As $\theta_p$
  is strictly increasing on each band, $\theta_{p,L}$ is also strictly
  increasing outside a neighborhood of the edges of a band. This
  completes the proof of Lemma~\ref{le:16}.
\end{proof}
\noindent One proves
\begin{Le}
  \label{le:17}
  Let $E_0$ be a closed gap for $H^\Z$ (see
  Definition~\ref{def:1}). Then, for any $L=Np+k$ the following
  assertions are equivalent:
  \begin{equation}
    \label{eq:224}
    E_0\in\sigma(H_L)\quad\Longleftrightarrow\quad  h_k(E_0)\in\pi\Z
    \quad\Longleftrightarrow\quad a_{k+1}(E_0)=0
    \quad\Longleftrightarrow\quad \alpha_{k+1}(E_0)\in i\R^*.
  \end{equation}
\end{Le}
\begin{proof}
  The proof of the first equivalence follows immediately from
  Definition~\ref{def:1} and the quantization condition~\eqref{eq:38};
  the second follows from~\eqref{eq:74} and the expansions
  in~\eqref{eq:225}; the third follows
  Lemma~\ref{le:16},~\eqref{eq:74} and~\eqref{eq:38}.
\end{proof}
\noindent Let us note that, in particular, closed gaps where $a_{k+1}$
vanishes are eigenvalues of $H_L$ for all $L=Np+k$.
\begin{Rem}
  \label{rem:12}
  The characteristic equation~\eqref{eq:38} and the computations done
  at the end of the proof of Lemma~\ref{le:15} show that, for $L=Np+k$
  large, an energy $E_0$ such that $\rho(E_0)=\rho^{-1}(E_0)$ is an
  eigenvalues of $H_L$ if and only if $a_{k+1}(E_0)=0$. This is an
  extension of Lemma~\ref{le:11}.
\end{Rem}
\noindent In view of the definition and monotony of $\theta_{p,L}$,
the quantization condition~\eqref{eq:38} is clearly equivalent
to~\eqref{eq:23}. This completes the proof Theorem~\ref{thr:15} on the
eigenvalues of $H_L$. Let us now turn to the computation of the
associated eigenfunctions.
\subsubsection{The Dirichlet eigenfunctions for a truncated periodic
  potential: the proof of Theorem~\ref{thr:30}}
\label{sec:dirichl-eigenf}
Recall that we assume $L=Np+k$. First, if $(u^j_l)_{l=0}^L$ is an
eigenfunction associated to the eigenvalue $\lambda_j$, the eigenvalue
equation reads
\begin{equation*}
  \begin{pmatrix}u^j_{l+1}\\ u^j_l\end{pmatrix}=T_l(\lambda_j)  
  \begin{pmatrix}u^j_l\\u^j_{l-1}  \end{pmatrix}\text{ for }0\leq
  l\leq L\text{ where } u^j_{L+1}=u^j_{-1}=0.
\end{equation*}
To normalize the solution, we assume that $u^j_0=1$. The coefficients
we want to compute are
\begin{equation}
  \label{eq:44}
  |\varphi_j(L)|^2=|u^j_L|^2\left(\sum_{l=0}^L\left|u^j_l\right|^2\right)^{-1}
  \quad\text{and}\quad  
  |\varphi_j(0)|^2=\left(\sum_{l=0}^L\left|u^j_l\right|^2\right)^{-1}.
\end{equation}
Fix $l=np+m$. Thus, using the notations of
section~\ref{sec:dirichl-eigenv} and the
expressions~\eqref{eq:42},~\eqref{eq:210} and~\eqref{eq:33}, one
computes
\begin{equation}
  \label{eq:41}
  \begin{pmatrix}u^j_l\\u^j_{l-1}\end{pmatrix}=
  T_{m-1,0}(\lambda_j)\left(\tilde
    T_0(\lambda_j)\right)^n \begin{pmatrix}1\\0 \end{pmatrix}=
  \begin{pmatrix} \alpha_m(\lambda_j)
    \rho^n(\lambda_j)+\beta_m(\lambda_j)\rho^{-n}(\lambda_j)\\ 
    \alpha_{m-1}(\lambda_j)\rho^n(\lambda_j)+
    \beta_{m-1}(\lambda_j)\rho^{-n}(\lambda_j) \end{pmatrix}     
\end{equation}
where $\alpha_m$ and $\beta_m$ are defined in~\eqref{eq:74}.
\paragraph{\it The eigenvectors associated to eigenvalues inside
  $\Sigma_\Z$}
\label{sec:it-eigenv-assoc}
As $\rho^{-1}(\lambda_j)=\overline{\rho(\lambda_j)}$,
$\beta_m(\lambda_j)=\overline{\alpha_m(\lambda_j)}$ and as the
functions $(\alpha_m)_{0\leq m\leq p-1}$ do not vanish on
$\overset{\circ}{\Sigma}_\Z$, we compute
\begin{equation}
  \label{eq:48}
  \left|u^j_{np+m}\right|^2=2|\alpha_m(\lambda_j)|^2
  \left(1+\text{Re}\left[\frac{\alpha_m(\lambda_j)}
      {\overline{\alpha_m(\lambda_j)}}\rho^{2n}(\lambda_j)\right]
  \right).
\end{equation}
As $L=Np+k$, using the quantization condition~\eqref{eq:38}, we obtain
that
\begin{equation}
  \label{eq:221}
  \begin{split}
    \sum_{l=0}^L\left|u^j_l\right|^2&=
    2\sum_{m=0}^k\left|\alpha_m(\lambda_j)\right|^2 \left(1+\text{Re}
      \left[\frac{\alpha_m(\lambda_j)}{\overline{\alpha_m(\lambda_j)}}
        \rho^{2N}(\lambda_j)\right]\right) \\&\hskip2.5cm+
    2\sum_{m=0}^{p-1}\left|\alpha_m(\lambda_j)
    \right|^2\sum_{n=0}^{N-1}\left(
      1+\text{Re}\left[\frac{\alpha_m(\lambda_j)}
        {\overline{\alpha_m(\lambda_j)}}
        \rho^{2n}(\lambda_j)\right]\right)\\
    &=N\,p\,f(\lambda_j)\left(1+\frac1{Np}\tilde f(\lambda_j) \right)
  \end{split}
\end{equation}
where we have defined
\begin{equation}
  \label{eq:46}
  f(E):=\frac2p\sum_{m=0}^{p-1}\left|\alpha_m(E)\right|^2.
\end{equation}
and, using the quantization condition~\eqref{eq:38}, computed
\begin{equation}
  \label{eq:49}
  \begin{split}
    \tilde f(E)&:=\frac2{f(E)}
    \text{Re}\left[\left(\sum_{m=0}^{p-1}\alpha^2_m(E)\right)
      \frac1{1-\rho^2(E)} \left(1+\frac{\overline{\alpha_{k+1}(E)}}
        {\alpha_{k+1}(E)}\right)\right] \\&\hskip4cm
    +\frac{2}{f(E)}\sum_{m=0}^k \left|\alpha_m(E)\right|^2\left(1-
      \text{Re}\left[\frac{\alpha_m(E)\,\overline{\alpha_{k+1}(E)}}
        {\overline{\alpha_m(E)}\,\alpha_{k+1}(E)}\right]\right)
  \end{split}
\end{equation}
The function $E\mapsto f(E)$ is real analytic and does
not vanish on $\overset{\circ}{\Sigma}_\Z$.\\
We prove
\begin{Pro}
  \label{pro:9}
  For $E_0$, a closed gap, one has $\D
  \sum_{m=0}^{p-1}\alpha^2_m(E_0)=0$.
\end{Pro}
\begin{proof}
  By the definition of $(a_j,b_j)$, see~\eqref{eq:118}, and that of
  $\alpha_j(E)$, see~\eqref{eq:74}, the sequence
  $(\alpha_j(E))_{j\in\Z}$ satisfies the equation
  $\alpha_{j+1}+\alpha_{j-1}+(V_j-E)\alpha_{j}=0$. As $\tilde
  T_0(E)=T_{p-1}(E)\cdots T_0(E)$, by~\eqref{eq:33}, for $j\in\Z$, one
  has $\alpha_{j+p}(E)=\rho(E) \alpha_{j}(E)$. Hence, the column
  vector $A(E)=(\alpha_1(E),\cdots,\alpha_p(E))^t$ satisfies
  \begin{equation*}
    (H_\rho-E)A(E)=0\quad\text{where}\quad H_\rho=\begin{pmatrix}
      V_1 & 1 & 0 & \cdots & 0 &\rho(E)\\
      1& V_2 & 1 & 0 & \cdots & 0 \\
      0 & 1 & V_3 & 1 & \cdots & 0 \\
      \vdots &  &  & \ddots &  & \vdots \\
      0 & \cdots & 0& 1& V_{p-1} & 1 \\
      \rho^{-1}(E) & 0 & \cdots & 0& 1& V_p
    \end{pmatrix}.
  \end{equation*}
  Thus, we have
  \begin{equation}
    \label{eq:226}
    \langle (H_\rho-E)A(E),A(E)\rangle_\R=0   
  \end{equation}
  where $\langle\cdot,\cdot\rangle_\R$ denotes the real scalar product
  over $\C^p$, i.e., $\D \left\langle\begin{pmatrix}z_1\\\vdots\\z_p
    \end{pmatrix},\begin{pmatrix}z'_1\\\vdots\\z'_p\end{pmatrix}
  \right\rangle_\R=\sum_{j=1}^p z_jz'_j$.\\
  The functions $E\mapsto A(E)$ and $E\mapsto \rho(E)$ being analytic
  over $\overset{\circ}{\Sigma_\Z}$, one can
  differentiate~\eqref{eq:226} with respect to $E$ to obtain
  \begin{multline}
    \label{eq:227}
    0=-\langle A(E),A(E)\rangle_\R\\+(\rho(E)-\rho^{-1}(E))
    \left(\rho^{-1}(E)\rho'(E)\alpha_1(E)\alpha_p(E)-
      \alpha_p(E)\alpha'_1(E)+\alpha_1(E)\alpha'_p(E)\right).
  \end{multline}
  Here, we have used the fact that, if $H^t_\rho$ is the transposed of
  the matrix $H_\rho$, then
  \begin{equation*}
    H^t_\rho-H_\rho=(\rho(E)-\rho^{-1}(E))
    \begin{pmatrix}
      0 &  \cdots & 0&-1\\
      0 &  \cdots & 0& 0 \\
      \vdots & & & \vdots \\
      0 &  0& \cdots & 0 \\
      1 & 0 & \cdots & 0
    \end{pmatrix}.
  \end{equation*}
  At $E_0$, a closed gap, one has $\rho(E_0)=\rho^{-1}(E_0)$. Hence,
  \eqref{eq:227} implies
  \begin{equation*}
    0=\langle
    A(E_0),A(E_0)\rangle_\R=\sum_{m=0}^{p-1}\alpha^2_m(E_0).    
  \end{equation*}
  This completes the proof of Proposition~\ref{pro:9}.
\end{proof}
\noindent In view of~\eqref{eq:49}, the function $\tilde f$ is real
analytic on $\overset{\circ}{\Sigma}_\Z$; indeed, the only poles of
the function $E\mapsto [\rho(E)-\rho^{-1}(E)]^{-1}$ in
$\overset{\circ}{\Sigma}_\Z$ are the closed gaps; they are simple
poles of this function and, by Proposition~\ref{pro:9}, the real
analytic function $\D E\mapsto \sum_{m=0}^{p-1}\alpha^2_m(E)$
vanishes at these poles.\\
Now that we have computed the normalization constant, let us compute
the coefficient $u_L^j$ defined in~\eqref{eq:44}. As $L=Np+k$, the
characteristic equation for $\lambda_j$, that is,~\eqref{eq:38} reads
\begin{equation}
  \label{eq:45}
  \alpha_{k+1}(\lambda_j)\rho^N(\lambda_j)=
  -\beta_{k+1}(\lambda_j)\rho^{-N}(\lambda_j)
  =-\overline{\alpha_{k+1}(\lambda_j)\rho^{N}(\lambda_j)}.
\end{equation}
Hence, one computes
\begin{equation}
  \label{eq:47}
  \begin{split}
    u^j_L&=\alpha_k(\lambda_j)\rho^N(\lambda_j)+
    \overline{\alpha_k(\lambda_j)\rho^N(\lambda_j)}=
    \rho^N(\lambda_j)\frac{\alpha_k(\lambda_j)
      \overline{\alpha_{k+1}(\lambda_j)}
      -\overline{\alpha_k(\lambda_j)}\alpha_{k+1}(\lambda_j)}
    {\overline{\alpha_{k+1}(\lambda_j)}}\\&=
    \frac{-\rho^N(\lambda_j)\,a_{p-1}^0(\lambda_j)}
    {(\rho(\lambda_j)-\rho^{-1}(\lambda_j))\,
      \overline{\alpha_{k+1}(\lambda_j)}}
    =\frac{-e^{i[Np\theta_p(\lambda_j)-h_k(\lambda_j)]}\,
      a_{p-1}^0(\lambda_j)} {\left|a_{k+1}(\lambda_j)(a^0_p(\lambda_j)
        -\rho^{-1}(\lambda_j))+b_{k+1}(\lambda_j)
        a^0_{p-1}(\lambda_j)\right|}\\
    &=\frac{-\,e^{i\pi j}\, a_{p-1}^0(\lambda_j)}
    {\left|a_{k+1}(\lambda_j)(a^0_p(\lambda_j)
        -\rho^{-1}(\lambda_j))+b_{k+1}(\lambda_j)
        a^0_{p-1}(\lambda_j)\right|}
  \end{split}
\end{equation}
where we have used the quantization condition satisfied by
$\lambda_j$, the last equality in~\eqref{eq:38}, and that
\begin{equation*}
  \left|\begin{matrix}
      \alpha_{k+1}(\lambda_j)& \alpha_k(\lambda_j)\\
      \overline{\alpha_{k+1}(\lambda_j)}&\overline{\alpha_k(\lambda_j)}
    \end{matrix}\right|=
  \left|\begin{matrix} \frac{a^0_{p-1}(\lambda_j)}
      {\rho(\lambda_j)-\rho^{-1}(\lambda_j)}&
      \frac{a^0_p(\lambda_j)-\rho^{-1}(\lambda_j)}
      {\rho(\lambda_j)-\rho^{-1}(\lambda_j)}\\
      -\frac{a^0_{p-1}(\lambda_j)}{\rho(\lambda_j)-\rho^{-1}(\lambda_j)}&
      \frac{\rho(\lambda_j)-a^0_p(\lambda_j)}
      {\rho(\lambda_j)-\rho^{-1}(\lambda_j)} \end{matrix} \right|
  \left| \begin{matrix} b_{k+1}(\lambda_j)&b_k(\lambda_j) \\a_{k+1}(\lambda_j)
      &a_k(\lambda_j)\end{matrix}\right|
\end{equation*}
and
\begin{equation*}
  \left|\begin{matrix} 1&
      \frac{a^0_p(\lambda_j)-\rho^{-1}(\lambda_j)}
      {\rho(\lambda_j)-\rho^{-1}(\lambda_j)}\\ -1&
      \frac{\rho(\lambda_j)-a^0_p(\lambda_j)}
      {\rho(\lambda_j)-\rho^{-1}(\lambda_j)} \end{matrix} \right|
  = \left|\begin{matrix} b_k(\lambda_j)&b_{k+1}(\lambda_j)
      \\a_k(\lambda_j)&a_{k+1}(\lambda_j)
    \end{matrix}\right|=1
\end{equation*}
\begin{Le}
  \label{le:21}
  Define the function $\tilde f^-_k(E)$ by
  \begin{equation*}
    \tilde f^-_k(E):=\frac{|a_{p-1}^0(E)|^2}{|a_{k+1}(E)(a_p^0(E)-
      \rho^{-1}(E))+b_{k+1}(E)a_{p-1}^0(E)|^2};
  \end{equation*}  
  Then, the function $\tilde f^-_k$ does not vanish on
  $\overset{\circ}{\Sigma}_\Z$.
\end{Le}
\begin{proof}
  By the definition of $\alpha_{k+1}$, one has $\D \tilde
  f^-_k(E)=\frac{|a_{p-1}^0(E)|^2}
  {|\rho(E)-\rho^{-1}(E))|^2\,|\alpha_{k+1}(E)|^2}$. That this
  expression is well defined and does not vanish on
  $\overset{\circ}{\Sigma}_\Z$ follows from Lemma~\ref{le:15} and the
  computations made in the proof thereof.
\end{proof}
\noindent Plugging~\eqref{eq:47} this and~\eqref{eq:48}
into~\eqref{eq:44}, recalling that $u^j_0=1$, outside the bad closed
gaps, we obtain~\eqref{eq:26} if,
\begin{itemize}
\item in addition to~\eqref{eq:46} and~\eqref{eq:49}, we set $\D
  f^+_0(E):=\frac1{f(E)}$ and $f^-_k(E)=f^+_0(E)\cdot\tilde f^-_k(E)$,
\item we remember that the function $a_{p-1}^0$ only changes sign in
  the gaps of the spectrum $\Sigma_\Z$ (see point (4) in
  section~\ref{sec:spectral-theory-Z}) and set $\sigma_r$ to be the
  sign of $-a_{p-1}^0$ on $B_r$, the $r$-th band.
\end{itemize}
By~\eqref{eq:44} and~\eqref{eq:48}, we obtain~\eqref{eq:26} using
Lemma~\ref{le:21}. This completes the proof of the statements in
Theorem~\ref{thr:30} on the eigenfunctions of $H_L$ associated to
eigenvalues in $\overset{\circ}{\Sigma}_\Z$.
\begin{Rem}
  \label{rem:10}
  To complete our study let us also see what happens the
  eigenfunctions near the edges of the spectrum. Pick
  $E_0\in\partial\Sigma_\Z$. One then knows that, for $E\in\Sigma_\Z$,
  $E$ close to $E_0$, one has
  \begin{equation}
    \label{eq:222}
    \theta_p(E)-\theta_p(E_0)=a\sqrt{|E-E_0|}(1+o(1)) 
  \end{equation}
  (see the proof of Lemma~\ref{le:16}).\\
  Let us rewrite $\tilde f$ (see~\eqref{eq:49}) in the following way
  \begin{equation}
    \label{eq:228}
    \begin{split}
      \tilde
      f(E)&=\frac2{f(E)}\left[\sum_{m=0}^{p-1}\left|\alpha_m(E)\right|^2
        \cos(h_k(E)-2 h_{m-1}(E)-p\theta_p(E))\right]
      \frac{\sin(h_k(E))}{\sin(p\theta_p(E))} \\&\hskip3.5cm
      +\frac{2}{f(E)}\sum_{m=0}^k \left|\alpha_m(E)\right|^2\left(1-
        \cos(2(h_k(E)-h_{m-1}(E)))\right).
    \end{split}
  \end{equation}
  Let us first show
  \begin{Le}
    \label{le:20}
    For any $0\leq m\leq p-1$, $E\mapsto
    \frac{2\left|\alpha_m(E)\right|^2}{p\,f(E)}$ can be extended
    continuously from $\overset{\circ}{\Sigma}_\Z$ to $\Sigma_\Z$.
  \end{Le}
  \begin{proof}
    For $p=1$ there is nothing to be done as
    $\frac{2\left|\alpha_m(E)\right|^2}{p\,f(E)}\equiv1$.\\
    For $p\geq2$, we note that, for $0\leq m\leq m+1\leq p-1$, as $\D
    \left|\begin{matrix} a_{m+1}(E)&b_{m+1}(E)\\a_m(E)
        &b_m(E) \end{matrix} \right|=1$ by~\eqref{eq:118},
    \begin{equation*}
      \begin{split}
        0&=a_{m+1}(E_0)(a^0_p(E_0)-\rho^{-1}(E_0))+b_{m+1}(E_0)
        a^0_{p-1}(E_0)
        \\&=a_m(E_0)(a^0_p(E_0)-\rho^{-1}(E_0))+b_m(E_0)
        a^0_{p-1}(E_0)
      \end{split}
    \end{equation*}
    if and only if $a^0_{p-1}(E_0)=0$ (as this implies
    $a^0_p(E_0)-\rho^{-1}(E_0)=0$).\\
    Let us assume this is the case. As $p\geq2$, we know that
    $\D\sum_{j=0}^{p-1}|a_j(E_0)|^2\not=0$. By~\eqref{eq:225}, for at
    least one $m_0\in\{0,\cdots,p-1\}$, one has $a_{m_0}(E_0)\not=0$
    and $\alpha_{m_0}(E)=bc^{-1}a_{m_0}(E_0)+
    O(\sqrt{|E-E_0|})$. Hence, $E\mapsto
    \frac{2\left|\alpha_m(E)\right|^2} {p\,f(E)}$ can be continued to
    $E_0$ setting $\D \frac{2\left|\alpha_m(E_0)\right|^2}{p\,f(E_0)}=
    \frac{|a_m(E_0)|^2}{|a_0(E_0)|^2+\cdots+|a_{p-1}(E_0)|^2}$. Actually,
    $f(E)$ can be continued at $E_0$ by setting
    \begin{equation}
      \label{eq:239}
      f(E_0)=|a_0(E_0)|^2+\cdots+|a_{p-1}(E_0)|^2.        
    \end{equation}
    Let us now assume that $a^0_{p-1}(E_0)\not=0$. We study the
    behavior of $\alpha_m$ near $E_0$. Recall~\eqref{eq:74}. Then, one
    has
    \begin{enumerate}
    \item either $d_m:=a_m(E_0)(a^0_p(E_0)-\rho^{-1}(E_0))+b_m(E_0)
      a^0_{p-1}(E_0) \not=0$: in this case, by~\eqref{eq:225}, one has
      $\alpha_m(E)=\frac{d_mc^{-1}}{\sqrt{|E-E_0|}}(1+o(1))$;
    \item or $d_m=a_m(E_0)(a^0_p(E_0)-\rho^{-1}(E_0))+b_m(E_0)
      a^0_{p-1}(E_0)=0$: in this case, as for some $A_m\in\R^*$ and
      $k_m\geq1$, one has
      \begin{equation*}
        a_m(E)(a^0_p(E)-\rho^{-1}(E_0))+b_m(E)
        a^0_{p-1}(E)=A_m(E-E_0)^{k_m}(1+o(1)),
      \end{equation*}
      and, by~\eqref{eq:225}, one can continue $\alpha_m$ to $E_0$ by
      setting $\alpha_m(E_0)=a_m(E_0)/2$.
    \end{enumerate}
    As $a^0_{p-1}(E_0)\not=0$, we know that for some
    $m_0\in\{0,\cdots,p-1\}$, we are in case (a). Hence, one has
    \begin{equation}
      \label{eq:238}
      f(E)=\frac{2}{p|E-E_0|}
      \sum_{m=0}^{p-1}|a_m(E_0)(a^0_p(E_0)-\rho^{-1}(E_0))+b_m(E_0)
      a^0_{p-1}(E_0)|^2(1+o(1))        
    \end{equation}
    and $E\mapsto \frac{2\left|\alpha_m(E)\right|^2} {p\,f(E)}$ can be
    continued to $E_0$ setting $\D
    \frac{2\left|\alpha_m(E_0)\right|^2}{p\,f(E_0)}=
    \frac{|d_m|^2}{|d_0|^2+\cdots+|d_{p-1}|^2}$ (using the notation
    introduced in point (a).\\
    This completes the proof of Lemma~\ref{le:20}.
  \end{proof}
  \noindent By Lemma~\ref{le:16}, we know that for $1\leq k\leq p$ and
  $E_0\in\partial\Sigma_\Z$, one has $2h_k(E_0)\in\pi\Z$. Thus, for
  $1\leq k\leq p$, $1\leq m\leq p$ and $E_0\in\partial\Sigma_\Z$, one has
  $\cos(h_k(E_0)-2h_{m-1}(E_0)-p\theta_p(E_0))\sin(h_k(E_0))=0$. Using
  the expansions leading to the proof of Lemma~\ref{le:16}, one gets
  \begin{equation*}
    \cos(h_k(E)-2h_{m-1}(E)-p\theta_p(E))\sin(h_k(E))=c\sqrt{|E-E_0|}(1+o(1)).
  \end{equation*}
  Recalling~\eqref{eq:222} and the fact that $p\theta_p(E_0)\in\pi\Z$,
  Lemma~\ref{le:20} implies that $\tilde f$ can be extended
  continuously up to $E_0$. Hence, the expansion~\eqref{eq:221} again
  yields
  \begin{equation}
    \label{eq:233}
    \sum_{l=0}^L\left|u^j_l\right|^2\asymp Np f(\lambda_j).
  \end{equation}
  Let us now review the computation~\eqref{eq:47} in this case. We
  distinguish two cases:
  \begin{enumerate}
  \item if $a^0_{p-1}(E_0)=0$: then,~\eqref{eq:47} and the fact that
    $a_{k+1}(E_0)\not=0$ (this case was dealt with in point (1)),
    yields that, for $|\lambda_j-E_0|$ sufficiently small,
    \begin{equation*}
      |u^j_L|\asymp\sqrt{|\lambda_j-E_0|}.
    \end{equation*}
    By~\eqref{eq:239} and~\eqref{eq:233}, we obtain
    \begin{equation}
      \label{eq:240}
      |\varphi_j(L)|^2\asymp\frac{|\lambda_j-E_0|}{Np}\quad\text{
        and }\quad|\varphi_j(0)|^2\asymp\frac{1}{Np}.
    \end{equation}
  \item if $a^0_{p-1}(E_0)\not=0$: then
    \begin{enumerate}
    \item if $d_{k+1}\not=0$ (see case (a) in the proof of
      Lemma~\ref{le:20}): by~\eqref{eq:238} and~\eqref{eq:233}, one
      has
      \begin{equation}
        \label{eq:241}
        |\varphi_j(0)|^2\asymp\frac{|\lambda_j-E_0|}{Np}\quad\text{
          and }\quad|\varphi_j(L)|^2\asymp\frac{|\lambda_j-E_0|}{Np}.
      \end{equation}
    \item if $d_{k+1}=0$: by~\eqref{eq:238} and~\eqref{eq:233}, one
      has
      \begin{equation}
        \label{eq:244}
        |\varphi_j(0)|^2\asymp\frac{|\lambda_j-E_0|}{Np}\quad\text{
          and }\quad|\varphi_j(L)|^2\asymp\frac{1}{Np}.
      \end{equation}
    \end{enumerate}
  \end{enumerate}
\end{Rem}
\paragraph{\it The eigenvectors associated to eigenvalues outside
  $\Sigma_\Z$}
\label{sec:it-eigenv-assoc-out}
Let us now turn to the eigenfunctions associated to eigenvalues $H_L$
in the gaps of $\Sigma_\Z$, i.e., in the region
$\{E;\;|\underline{\Delta}(E)|>2\}$.  On $\R\setminus\Sigma_\Z$, the
eigenvalue $E\mapsto\rho(E)$ is real valued (recall that we pick it so
that $|\rho(E)|<1$) and so are all the functions
$(\alpha_m)_{0\leq m\leq p-1}$ and $(\beta_m)_{0\leq m\leq p-1}$
(see~\eqref{eq:74}). For $0\leq m\leq p-1$,~\eqref{eq:41} yields
\begin{equation}
  \label{eq:124}
  \left|u^j_{np+m}\right|^2=\alpha^2_m(E)\rho^{2n}(E)
  +\beta^2_m(E)\rho^{-2n}(E)+2\alpha_m(E) \beta_m(E).
\end{equation}
As when we studied the eigenvalues of $H_L$, let us now distinguish
the cases when $E$ is close to an eigenvalue of $H_0^+$ or to an
eigenvalue of $H_k^-$:
\begin{enumerate}
\item Pick $E'$ an eigenvalue of $H_0^+$ but not an eigenvalue of
  $H_k^-$; then, recall that $a_{p-1}^0(E')=0=a_p^0(E')-\rho(E')$.
  Thus, for $0\leq m\leq p-1$, one has $\beta_m(E')=0$. Assume $E$ be close to
  $E'$. As $E$ satisfies~\eqref{eq:121}, using~\eqref{eq:36},
  ~\eqref{eq:124} becomes
  \begin{equation*}
    \begin{split}
      \left|u^j_{np+m}\right|^2&=\rho^{2n}(E')
      \left|\alpha_m(E')-\frac{\beta'_m(E')}{\beta'_{k+1}(E')}
        a_{k+1}(E')\right.\\&\hskip2.5cm
      \cdot\left[\rho(E')-\rho^{-1}(E')\right]\rho^{2(N-n)}(E')+
      O(\rho^{2N}(E))\Biggl|^2.
    \end{split}
  \end{equation*}
  for $0\leq m\leq p-1$ if $0\leq n\leq N-1$ and
  $0\leq m\leq k$ if $n=N$.\\
  Using~\eqref{eq:37}, one computes
  \begin{equation}
    \label{eq:123}
    \left|u^j_{np+m}\right|^2=\rho^{2n}(E')
    \left|a_m(E')-\frac{\beta'_m(E')}{\beta'_{k+1}(E')}a_{k+1}(E')
      \rho^{2(N-n)}(E')+
      O(\rho^{2N}(E))\right|^2.     
  \end{equation}
  This yields
  \begin{equation*}
    \begin{split}
      \sum_{l=0}^L\left|u^j_l\right|^2
      &=\sum_{m=0}^{p-1}\sum_{n=0}^{N-1} \rho^{2n}(E')a^2_m(E')+
      O(N\rho^{2N}(E))\\&=\frac1{1-\rho^2(E')}
      \sum_{m=0}^{p-1}a^2_m(E')+ O(N\rho^{2N}(E)).
    \end{split}
  \end{equation*}
  Moreover, by~\eqref{eq:44},~\eqref{eq:124} and~\eqref{eq:74}, as
  $a_{p-1}^0(E')=0=a_p^0(E')-\rho(E')$, we obtain
  \begin{equation*}
    \begin{split}
      \left|\varphi_j(L)\right|^2&
      =\rho^{2N}(E')\frac{(1-\rho^2(E'))a^2_{k+1}(E')}{
        \D\left[\beta'_{k+1}(E')\right]^2\sum_{m=0}^{p-1}a^2_m(E')}
      \left|\begin{matrix}\beta'_k(E')&
          a_k(E')\\\beta'_{k+1}(E')&a_{k+1}(E')
        \end{matrix}\right|^2 +O(N\rho^{4N}(E))\\
      &=\gamma\rho^{2N}(E') +O(N\rho^{4N}(E)).
    \end{split}
  \end{equation*}
  where
  \begin{equation*}
    \gamma:=\frac{(1-\rho^2(E'))a^2_{k+1}(E')}{
      \D\left[\beta'_{k+1}(E')\right]^2\sum_{m=0}^{p-1}a^2_m(E')}
    \left(\frac{d a_{p-1}^0}{dE}(E')\right)^2>0.
  \end{equation*}
  Hence, $\left|\varphi_j(L)\right|$ is exponentially small in $L$
  (recall $|\rho(E)|<1$).
\item if $E'$ is an eigenvalue of $H_k^-$ but not of $H_0^+$, then
  inverting the parts of $H_k^-$ and $H_0^+$, we see that
  $|\varphi_j(L)|$ is of order $1$. A precise asymptotic can be
  computed but it won't be needed.
\item if $E'$ is an eigenvalue of $H_0^+$ and of $H_k^-$, the double
  well analysis done in section~\ref{sec:appendix} shows that for
  normalized eigenvectors, say, $\varphi_{1,2}$ associated to the two
  eigenvalues of $H_L$ close to $E'$, the four coefficients
  $|\varphi_{1,2}(0)|$ and $|\varphi_{1,2}(L)|$ are of order
  $1$. Again precise asymptotics can be computed but won't be needed.
\end{enumerate}
This completes the description of the eigenfunctions given by
Theorem~\ref{thr:30} and completes the proof of this result.\qed
\section{Resonances in the periodic case}
\label{sec:reson-peri-case}
We are now in the state to prove the results stated in
section~\ref{sec:periodic-case}. Therefore, we first study the
function $E\mapsto S_L(E)$ and $E\mapsto \Gamma_L(E)$ in the complex
strip $I+i(-\infty,0)$ for $I\subset\overset{\circ}{\Sigma}_\Z$.
\subsection{The matrix $\Gamma_L$ in the periodic case}
\label{sec:funct-s_l-peri}
Using Theorem~\ref{thr:16}, we first prove
\begin{Th}
  \label{thr:17}
  Fix $I\subset\overset{\circ}{\Sigma}_\Z$ a compact interval. There
  exists $\varepsilon_I>0$ and $\sigma_I\in\{+1,-1\}$ such that, for
  any $N\geq0$, there exists $C_N>0$ such that, for $L$ sufficiently
  large s.t $L\equiv k\mod(p)$, one has
  \begin{equation}
    \label{eq:59}
    \sup_{\substack{\text{Re}\,E\in I\\-\varepsilon_I<\text{Im}\,E<0}}
    \left|\Gamma_L(E)-\Gamma^\text{eff}_L(E)\right|\leq C_NL^{-N}.
  \end{equation}
  where
  \begin{equation}
    \label{eq:162}
    \Gamma^\text{eff}_L(E)= -\frac{\theta'_p(E)}{\sin u_L(E)}
    \begin{pmatrix}\D e^{-i u_L(E)} f^-_k(E)&\D\sigma_I\,
      \sqrt{f^-_k(E)f^+_0(E)}\\ \D\sigma_I\,\sqrt{f^-_k(E)f^+_0(E)}&
      \D e^{-i u_L(E)} f^+_0(E)\end{pmatrix}+
    \begin{pmatrix}\D \int_\R\frac{dN^-_k(\lambda)}{\lambda-E}&0\\
      0&\D \int_\R\frac{dN^+_0(\lambda)}{\lambda-E}
    \end{pmatrix}
  \end{equation}
  and $\D u_L(E):=(L-k)\theta_{p,L}(E)$ (see~\eqref{eq:229}),
\end{Th}
\noindent The sign $\sigma_I$ only deepends on the spectral band
containing $I$.\\
Deeper into the lower half-plane, we obtain the following simpler
estimate
\begin{Th}
  \label{thr:19}
  There exists $C>0$ such that, for any $\varepsilon>0$ and for
  $L\geq1$ sufficiently large s.t.  $L=Np+k$, one has
  \begin{equation}
    \label{eq:180}
    \sup_{\substack{\text{Re}\,E\in I\\\text{Im}\,E<-\varepsilon}}
    \left|\Gamma_L(E)-\begin{pmatrix}\D \int_\R\frac{dN^-_k(\lambda)}{\lambda-E}&0\\
        0&\D  \int_\R\frac{dN^+_0(\lambda)}{\lambda-E}
      \end{pmatrix}
    \right|\leq C\varepsilon^{-2}e^{-\varepsilon L/C}.
  \end{equation}
\end{Th}
\noindent In sections~\ref{sec:proof-theorem-3}, the
approximations~\eqref{eq:59} and~\eqref{eq:180} theorems will be used
to prove Theorems~\ref{thr:5},~\ref{thr:2} and~\ref{thr:3}.\\
Let us note that, as $\cot z=i+O\left(e^{-2i\text{Im}\,z}\right)$, for
$\varepsilon\in(0,\varepsilon_I)$, the asymptotics given by
Theorems~\ref{thr:17} and~\ref{thr:19} coincide in the region
$\{$Re$\,E\in I,\ $Im$\,E\in(-\varepsilon_I,-\varepsilon)\}$: indeed one
has,
\begin{equation*}
  \sup_{\substack{\text{Re}\,E\in I\\-\varepsilon_I<\text{Im}\,E<-\varepsilon}}
  \left\|\frac{\theta'_p(E)}{\sin u_L(E)}
    \begin{pmatrix}\D e^{-i u_L(E)} f^-_k(E)&\D\sigma_I\,
      \sqrt{f^-_k(E)f^+_0(E)}\\ \D\sigma_I\,\sqrt{f^-_k(E)f^+_0(E)}&
      \D e^{-iu_L(E)} f^+_0(E) \end{pmatrix} \right\|\leq e^{-\varepsilon L/C}.
\end{equation*}
Let us now turn to the proofs of Theorems~\ref{thr:17}
and~\ref{thr:19}.
\subsubsection{The proof of Theorem~\ref{thr:17}}
\label{sec:proof-theorem-4}
To prove Theorem~\ref{thr:17}, we split the sum $S_L(E)$ into two
parts, one containing the Dirichlet eigenvalues ``close'' to Re$\,E$,
the second one containing those ``far'' from Re$\,E$. By ``far'', we
mean that the distance to Re$\,E$ is lower bounded by a small constant
independent of $L$. The ``close'' eigenvalues are then described by
Theorem~\ref{thr:16}. For the ``far'' eigenvalues, the strong
resolvent convergence of $H_L$ to $H_0^+$, that of $\tilde H_L$ to
$H_k^-$ (see Remark~\ref{rem:6}) and Combes-Thomas estimates enable us
to compute the limit and to show that the prelimit and the limit are
$O(L^{-\infty})$ close to each other. For the ``close'' eigenvalues,
the sum coming up in~\eqref{eq:142}, the definition of $\Gamma_L$, is
a Riemann sum. We use the Poisson summation formula to obtain a
precise approximation.\vskip.1cm\noindent
As $I$ is a compact interval in $\overset{\circ}{\Sigma}_\Z$, we pick
$\varepsilon>0$ such that, for $E\in I$, one has
$[E-6\varepsilon,E+6\varepsilon]\subset\overset{\circ}{\Sigma}_\Z$. Let
$\chi\in\Coi(\R)$ be a non-negative cut-off function such that
$\chi\equiv1$ on $[-4\varepsilon,4\varepsilon]$ and $\chi\equiv0$
outside $[-5\varepsilon,5\varepsilon]$. For $E\in I$, define
$\chi_E(\cdot)=\chi(\cdot-E)$.\vskip.1cm\noindent
We first give the asymptotic for the sum over the Dirichlet
eigenvalues far from Re$\,E$. We prove
\begin{Le}
  \label{le:1}
  For any $N>1$, there exists $C_N>0$ such that, for $L$ sufficiently
  large such that $L\equiv k\mod(p)$, one has
  \begin{equation}
    \label{eq:63}
    \sup_{E\in\C}\left|\sum_{j=1}^L
      \frac{1-\chi_{\text{Re}\,E}(\lambda_j)}{\lambda_j-E}
      \begin{pmatrix} |\varphi_j(L)|^2& \overline{\varphi_j(0)}
        \varphi_j(L) \\ \varphi_j(0) \overline{\varphi_j(L)} &
        |\varphi_j(0)|^2 \end{pmatrix}-\tilde M(E) \right|\leq C_N L^{-N}
  \end{equation}
  where
  \begin{equation}
    \label{eq:211}
    \tilde M(E):=    
    \begin{pmatrix}
      \D
      \int_{\R}(1-\chi_{\text{Re}\,E})(\lambda)\frac{dN^-_k(\lambda)}
      {\lambda-E}& 0\\0 & \D
      \int_{\R}(1-\chi_{\text{Re}\,E})(\lambda)\frac{dN^+_0(\lambda)}
      {\lambda-E}
    \end{pmatrix}.
  \end{equation}
\end{Le}
\begin{proof}[Proof of Lemma~\ref{le:1}]
  Recall (see Theorem~\ref{thr:11}) that $H_L$ is the operator $H_0^+$
  restricted to $\llbracket 0,L\rrbracket$ with Dirichlet boundary
  condition at $L$; as $L\equiv k\mod(p)$, it is unitarily equivalent
  to the operator $H_k^-$ restricted to $\llbracket -L,0\rrbracket$
  with Dirichlet boundary condition at $-L$ (see
  Remark~\ref{rem:6}).\\
  Pick $\tilde\chi\in\Coi$ such that $\tilde\chi\equiv1$ on
  $\sigma(H_0^+)\cup\sigma(H_k^-)$. First, we compute
  \begin{equation*}
    \begin{split}
      &\sum_{j=0}^L(1-\chi_{\text{Re}\,E})(\lambda_j)\frac{|\varphi_j(0)|^2}
      {\lambda_j-E}-
      \int_{\R}(1-\chi_{\text{Re}\,E})(\lambda)\frac{dN^+_0(\lambda)}{\lambda-E}
      \\&=\langle\delta_0,\left[\tilde\chi(1-\chi_{\text{Re}\,E})\right](H_L)(H_L-E)^{-1
      }\delta_0\rangle\\&\hskip5cm
      -\langle\delta_0,\left[\tilde\chi(1-\chi_{\text{Re}\,E})\right](H_0^+)(H_0^+-E)^{-1
      }\delta_0\rangle,
    \end{split}
  \end{equation*}
  \begin{equation*}
    \begin{split}
      &\sum_{j=0}^L(1-\chi_{\text{Re}\,E})(\lambda_j)\frac{
        |\varphi_j(L)|^2}{\lambda_j-E}-
      \int_{\R}(1-\chi_{\text{Re}\,E})(\lambda)\frac{dN^-_k(\lambda)}{\lambda-E}
      \\&=
      \langle\delta_L,\left[\tilde\chi(1-\chi_{\text{Re}\,E})\right](H_L)(H_L-E)^{-1
      }\delta_L\rangle\\&\hskip5cm-\langle\delta_L,\left[\tilde\chi(1-\chi_{\text{Re}\,E})\right]
      (H_k^-)(H_k^--E)^{-1}\delta_L\rangle,
    \end{split}
  \end{equation*}
  and
  \begin{equation*}
    \sum_{j=0}^L(1-\chi_{\text{Re}\,E})(\lambda_j)
    \frac{\varphi_j(L)\overline{\varphi_j(0)}}{\lambda_j-E}=
    \langle\delta_L,\left[\tilde\chi(1-\chi_{\text{Re}\,E})\right](H_L)(H_L-E)^{-1
    }\delta_0\rangle.
  \end{equation*}
  By the definition of $\chi_{\text{Re}\,E}$, the function
  $\lambda\mapsto(\lambda-E)^{-1} \tilde\chi(\lambda)
  (1-\chi_{\text{Re}\,E})(\lambda)$ is $\Coi$ on $\R$; moreover, its
  semi-norms (see~\eqref{estunif}) are bounded uniformly in
  $E\in\C$. Thus, there exists an almost analytic extension of
  $[\tilde\chi(1-\chi_{\text{Re}\,E})](\cdot)(\cdot-E)^{-1}$ such
  that, uniformly in $E$, one has~\eqref{estunif}.\\
  In the same way as we obtained~\eqref{eq:168}, we obtain
  \begin{multline}
    \label{eq:157}
    \left|\left\langle\delta_L,\left[(\tilde
          H_L-z)^{-1}-(H_k^--z)^{-1}
        \right]\delta_L\right\rangle\right|\\+
    \left|\left\langle\delta_0,\left[(H_L-z)^{-1}-(H_0^+-z)^{-1}
        \right]\delta_0\right\rangle\right|\\+
    \left|\left\langle\delta_0,(H_L-z)^{-1}\delta_L\right\rangle\right|\leq
    \frac{C}{|\text{Im}z|^2}e^{-L|\text{Im}z|/C}
  \end{multline}
  Plugging~\eqref{eq:157} into~\eqref{hesj0} and using~\eqref{estunif}
  for $[\tilde\chi(1-\chi_{\text{Re}\,E})](\cdot)(\cdot-E)^{-1}$, we
  get
  \begin{equation*}
    \forall K\in\N,\quad\sup_{\substack{L\geq1\\ L\equiv k\mod(p)}}
    L^{K}\,\left|\sum_{j=0}^L(1-\chi_{\text{Re}\,E})(\lambda_j)\frac{|\varphi_j(0)|^2}
      {\lambda_j-E}-
      \int_{\R}(1-\chi_{\text{Re}\,E})(\lambda)\frac{dN^+_0(\lambda)}{\lambda-E}
    \right|<+\infty
  \end{equation*}
  This entails~\eqref{eq:63} and completes the proof of
  Lemma~\ref{le:1}.
\end{proof}
\noindent Let us now estimate the part of $\Gamma_L(E)$ associated to
the Dirichlet eigenvalues close to Re$\,E$. Therefore, define
\begin{equation}
  \label{eq:61}
  \Gamma_L^\chi(E)=\sum_{j=1}^L\frac{\chi_{\text{Re}\,E}(\lambda_j)}{\lambda_j-E}
  \begin{pmatrix} |\varphi_j(L)|^2& \overline{\varphi_j(0)}
    \varphi_j(L) \\ \varphi_j(0) \overline{\varphi_j(L)} &
    |\varphi_j(0)|^2
  \end{pmatrix}.
\end{equation}
We prove
\begin{Le}
  \label{le:13}
  There exists $\varepsilon>0$ such that, for $N\geq1$, there exists
  $C_N$ such that, for $L$ sufficiently large such that $L\equiv
  k\mod(p)$, one has
  \begin{equation*}
    \label{eq:65}
    \sup_{\substack{\text{Re}\,E\in I\\-\varepsilon<\text{Im}\,E<0}}
    \left|\Gamma^\chi_L(E)-\Gamma^\text{eff}_L(E)+\tilde M(E)
    \right|\leq C_NL^{-N}
  \end{equation*}
  where $\tilde M$ is defined in~\eqref{eq:211}.
\end{Le}
\noindent Clearly Lemmas~\ref{le:1} and~\ref{le:13} immediately yield
Theorem~\ref{thr:17}.
\begin{proof}[Proof of Lemma~\ref{le:13}]
  Recall that the quasi-momentum $\theta_p$ defines a real analytic
  one-to-one monotonic map from the interior of each band of spectrum
  onto the set $(0,\pi)$, $(-\pi,0)$ or $(-\pi,\pi)$ (depending on the
  spectral band containing $I+[-4\varepsilon,4\varepsilon]$ where
  $\varepsilon>0$ has been fixed above) (see
 , e.g.,~\cite{MR1711536}). Moreover, the derivative $\theta'_p$ is
  positive in the interior of a spectral band. Thus, for $L$
  sufficiently large, the real part of the derivative $\theta'_{p,L}$
  (see~\eqref{eq:229}) is positive $I+[-2\varepsilon,2\varepsilon]$
  and $\theta_{p,L}$ is real analytic one-to-one on a complex
  neighborhood of
  $(I+[-3\varepsilon,3\varepsilon])+i[-3\varepsilon,3\varepsilon]$
  (possibly at the expense of reducing $\varepsilon$ somewhat). \\
  By~\eqref{eq:142},~\eqref{eq:171} and Theorem~\ref{thr:16}, one may
  write
  \begin{equation}
    \label{eq:60}
    \Gamma^\chi_L(E)=\frac1{L-k}\sum_{j\in\Z}
    \frac{\chi_{\text{Re}\,E}\left(\theta_{p,L}^{-1}
        \left(\frac{\pi j}{L-k}\right)\right)}{\theta_{p,L}^{-1}
      \left(\frac{\pi j}{L-k}\right)-E}
    M\left(\theta_{p,L}^{-1}
      \left(\frac{\pi j}{L-k}\right)\right)
  \end{equation}
  where
  \begin{equation}
    \label{eq:62}
    M(\lambda):=
    \begin{pmatrix} f_{k,L}(\lambda)&
      \sigma_I\,
      e^{i(L-k)\theta_{p,L}(\lambda)}\sqrt{f_{k,L}(\lambda)f_{0,L}(\lambda)} \\ 
      \sigma_I\,
      e^{i(L-k)\theta_{p,L}(\lambda)}\sqrt{f_{k,L}(\lambda)f_{0,L}(\lambda)}&
      f_{0,L}(\lambda)  \end{pmatrix}.
  \end{equation}
  and the matrix $M$ is analytic in the rectangle
  $(I+[-3\varepsilon,3\varepsilon])+i[-3\varepsilon,3\varepsilon]$. Thus,
  the Poisson formula tells us that
  \begin{equation}
    \label{eq:66}
    \begin{split}
      \Gamma^\chi_L(E)&=\frac1{L-k}\sum_{j\in\Z}\int_\R e^{-2i\pi j x}
      \frac{\chi_{\text{Re}\,E}\left(\theta_{p,L}^{-1} \left(\frac{\pi
              x}{L-k}\right)\right)}{\theta_{p,L}^{-1} \left(\frac{\pi
            x}{L-k}\right)-E}\, M\left(\theta_{p,L}^{-1}
        \left(\frac{\pi x}{L-k}\right)\right)dx\\
      &=\sum_{j\in\Z}\frac1\pi\int_\R e^{-2i j(L-k)
        \theta_{p,L}(\lambda)}
      \frac{\chi_{\text{Re}\,E}(\lambda)}{\lambda-E}\theta_{p,L}'(\lambda)
      \,M\left(\lambda\right)d\lambda\\
      &=\sum_{j\in\Z}\frac1\pi\int_\R M_{j,\chi}(E,\lambda,\lambda)
      d\lambda
    \end{split}
  \end{equation}
  by the definition of $\chi_{\text{Re}\,E}$; here, we have set
  \begin{equation*}
    M_{j,\chi}(E,\lambda,\beta):=e^{-2i j(L-k)
      \theta_{p,L}(\beta+\text{Re}\,E)}
    \frac{\chi(\lambda)}{\beta-i\text{Im}\,E}
    \theta_{p,L}'(\beta+\text{Re}\,E)\,
    M(\beta+\text{Re}\,E).
  \end{equation*}
  Let us now study the individual terms in the last sum
  in~\eqref{eq:66}. Therefore, recall that, on
  $[-4\varepsilon,4\varepsilon]$, $\chi$ is identically $1$ and that
  $\lambda\mapsto \theta_{p,L}(\lambda+\text{Re}\,E)$ and
  $\lambda\mapsto M(\lambda)$ are analytic in
  $(I+[-3\varepsilon,3\varepsilon])+i[-3\varepsilon,3\varepsilon]$;
  moreover, by~\eqref{eq:23}, for some $\delta>0$, one has
  \begin{equation}
    \label{eq:85}
    \liminf_{L\to+\infty}\inf_{\lambda\in[-4\varepsilon,4\varepsilon]}
    \theta'_{p,L}(\lambda+\text{Re}\,E)\geq \liminf_{L\to+\infty}\inf_{E\in I}
    \theta'_{p,L}(E)\geq \delta.
  \end{equation}
  Recall also that Im$\,E<0$. Consider $\tilde\chi:\R\to[0,1]$ smooth
  such that $\tilde\chi=1$ on $[-2\varepsilon,2\varepsilon]$ and
  $\tilde\chi=0$ outside $[-3\varepsilon,3\varepsilon]$.\\
  In the complex plane, consider the paths $\gamma_\pm:\,\R\to\C$
  defined by
  \begin{equation*}
    \gamma_\pm(\lambda)=\lambda\pm2i\varepsilon\tilde\chi(\lambda).    
  \end{equation*}
  As $-\varepsilon\leq$Im$\,E\,<0$, by contour deformation, we have
  \begin{gather*}
    \int_\R M_{j,\chi}(E,\lambda,\lambda)d\lambda= \int_\R
    M_{j,\chi}(E,\lambda,\gamma_+(\lambda))d\lambda,
    \\
    \begin{split}
      \int_\R M_{j,\chi}(E,\lambda,\lambda)d\lambda&= -2i\pi e^{-2i
        j(L-k) \theta_{p,L}(E)}\theta_{p,L}'(E)\,M(E)+ \int_\R
      M_{j,\chi}(E,\lambda,\gamma_-(\lambda))d\lambda.
    \end{split}
  \end{gather*}
  We then estimate
  \begin{itemize}
  \item for $j<0$, using a non-stationary phase argument as the
    integrand is the product of a smooth function with an rapidly
    oscillating function (using $|j|(L-k)$ as the large parameter),
    one then estimates
    \begin{equation*}
      \int_\R M_{j,\chi}(E,\lambda,\gamma_+(\lambda))d\lambda
      =O\left((|j|L)^{-\infty}\right).
    \end{equation*}
    The phase function is complex but its real part is non positive as
    Im$\,\theta_{p,L}(\gamma_+(\cdot)+ \text{Re}\,E)\geq0$ on the
    support of $\chi$ (by~\eqref{eq:85}). Note that the off-diagonal
    terms of $M(\lambda)$ also carry a rapidly oscillating exponential
    (see~\eqref{eq:62}) but it clearly does not suffice to counter the
    main one.
  \item in the same way, for $j>0$, one has
    \begin{equation*}
      \int_\R M_{j,\chi}(E,\lambda,\gamma_-(\lambda))d\lambda=
      O\left((|j|L)^{-\infty}\right).
    \end{equation*}
  \end{itemize}
  Thus, we compute
  \begin{gather}
    \label{eq:174}
    \text{ for }j<0:\ \int_\R M_{j,\chi}(E,\lambda,\lambda)
    d\lambda=O\left((|j|L)^{-\infty}\right),\\
    \label{eq:175}
    \text{ for }j>0:\ \int_\R M_{j,\chi}(E,\lambda,\lambda)d\lambda=
    -2i\pi e^{-2i j(L-k) \theta_{p,L}(E)}\theta_{p,L}'(E)\,M(E)+
    O\left((|j|L)^{-\infty}\right).
  \end{gather}
  Finally, for $j=0$, the contour deformation along $\gamma_+$ yields
  \begin{equation*}
    \begin{split}
      \int_\R\frac{\chi(\lambda)}{\lambda-i\text{Im}\,E}\,
      M(\lambda+\text{Re}\,E)d\lambda&=
      \int_\R\frac{\chi_{\text{Re}\,E}(\lambda)}{\lambda-E}
      \,\theta_{p,L}'(\lambda)
      \begin{pmatrix} f_{k,L}(\lambda)& 0 \\ 0&
        f_{0,L}(\lambda)\end{pmatrix}d\lambda +O\left(L^{-\infty}\right)\\
      &=\int_\R\frac{\chi_{\text{Re}\,E}(\lambda)}{\lambda-E}
      \begin{pmatrix} dN_k^-(\lambda)& 0 \\ 0&
        dN_0^+(\lambda)\end{pmatrix} + O\left(L^{-\infty}\right)
    \end{split}
  \end{equation*}
  by Corollary~\ref{cor:1}.\\
  Plugging this,~\eqref{eq:174} and~\eqref{eq:175} into~\eqref{eq:66}
  and computing the geometric sum immediately yields the following
  asymptotic expansion (where the remainder term is uniform on the
  rectangle $I+i[-\varepsilon,0)$)
  \begin{equation}
    \label{eq:170}
    \begin{split}
      \Gamma^\chi_L(E)&= -2i\sum_{j>0} e^{-2i j(L-k)
        \theta_{p,L}(E)}\theta_{p,L}'(E)\,M(E)\\&\hskip2cm+
      \int_\R\frac{\chi_{\text{Re}\,E}(\lambda)}{\lambda-E}
      \begin{pmatrix} dN_k^-(\lambda)& 0 \\ 0&
        dN_0^+(\lambda)\end{pmatrix} +
      O\left(L^{-\infty}\right)\\
      &= \frac{- e^{-i (L-k) \theta_{p,L}(E)}}{\sin((L-k)
        \theta_{p,L}(E))}\theta_{p,L}'(E)\,M(E)\\&\hskip2cm+
      \int_\R\frac{\chi_{\text{Re}\,E}(\lambda)}{\lambda-E}
      \begin{pmatrix} dN_k^-(\lambda)& 0 \\ 0&
        dN_0^+(\lambda)\end{pmatrix} + O\left(L^{-\infty}\right).
    \end{split}
  \end{equation}
  This completes the proof of Lemma~\ref{eq:65}.
\end{proof}
\subsubsection{The proof of Theorem~\ref{thr:19}}
\label{sec:proof-theorem-7}
To prove~\eqref{eq:59}, for Im$\, E<-\varepsilon$, it suffices to
write
\begin{equation*}
  \begin{split}
    \sum_{j=0}^L\frac{|\varphi_j(0)|^2} {\lambda_j-E}-
    \int_{\R}\frac{dN^+_0(\lambda)}{\lambda-E}
    &=\langle\delta_0,(H_L-E)^{-1}\delta_0\rangle-
    \langle\delta_0,(H_0^+-E)^{-1}\delta_0\rangle
    \\&=\langle\delta_0,(H_L-E)^{-1}\delta_L\rangle
    \langle\delta_{L+1},(H_0^+-E)^{-1 }\delta_0\rangle
  \end{split}
\end{equation*}
and
\begin{gather*}
  \sum_{j=0}^L\frac{ |\varphi_j(L)|^2}{\lambda_j-E}-
  \int_{\R}\frac{dN^-_k(\lambda)}{\lambda-E}=
  \langle\delta_0,(H_L-E)^{-1}\delta_L\rangle
  \langle\delta_{L+1},(H_k^--E)^{-1}\delta_0\rangle,
  \\
  \sum_{j=0}^L\frac{\varphi_j(L)\overline{\varphi_j(0)}}{\lambda_j-E}=
  \langle\delta_L,(H_L-E)^{-1}\delta_0\rangle
\end{gather*}
and to use the Combes-Thomas estimate~\eqref{eq:157}. This completes
the proof of Theorem~\ref{thr:19}.\qed
\subsection{The proofs of Theorems~\ref{thr:5},~\ref{thr:2}
  and~\ref{thr:3}}
\label{sec:proof-theorem-3}
We will now use Theorems~\ref{thr:17} and~\ref{thr:19} to prove
Theorems~\ref{thr:5},~\ref{thr:2} and~\ref{thr:3}.
\subsubsection{The proof of Theorem~\ref{thr:5}}
\label{sec:proof-theorem-5}
The first statement of Theorem~\ref{thr:5} is an immediate consequence
of the characteristic equations for the resonances~\eqref{eq:1}
and~\eqref{eq:135} and the description of the eigenvalues of $H_L$
given in Theorem~\ref{thr:16}.\\
When $\bullet=\N$, i.e., for the operator on the half-line, if
$I\subset(-2,2)$ does not meet $\Sigma_\N$, there exists $C>0$
s.t. for $L$ sufficiently large dist$(I,\sigma(H_L))>1/C$. Thus, on
the set $I-i[0,+\infty)$, one has Im$\,S_L(E)\leq\,$Im$\,E/C$. As on
$I$, one has Im$\,\theta_p(E)>1/C$ (see
section~\ref{sec:proof-theorem}), the characteristic
equation~\eqref{eq:1} admits a solution $E$ such that Re$\,E\in I$ only
if Im$\,E<1/C^2$. This completes the proof of point (1) of
Theorem~\ref{thr:5} for $\bullet=\N$.\\
For $\bullet=\Z$, i.e., to study equation~\eqref{eq:135}, one reasons in
the same way except that one replaces the study of $S_L(E)$ by that of
$\langle\Gamma_L(E)u, u\rangle$ for $u$ an arbitrary vector in $\C^2$
of unit length. This completes the proof of point (1) of
Theorem~\ref{thr:5}\\
Point (3a) is an immediate consequence of Theorems~\ref{thr:12}
and~\ref{thr:14} and the description of the eigenvalues of $H_L$
outside $\Sigma_\Z$. Notice that in the present case $d_j$ in
Theorems~\ref{thr:12} and~\ref{thr:14} is bounded from below by a
constant independent of $L$ and $a^\bullet_j$ is exponentially small
and described by Theorem~\ref{thr:16}. \\
Point (3b) is an immediate consequence of the description of the
eigenvalues of $H_L$ outside $\Sigma_\Z$ in case (3) of
Theorem~\ref{thr:19} and Theorem~\ref{thr:13}. Indeed, in the present
case $d_j$ and and $a^\bullet_j$ are both of order $1$; thus,
Theorem~\ref{thr:13} guarantees, around the common eigenvalue for
$H_k^-$ and $H_0^+$, a rectangle of width of order $1$ free
of resonances. \\
Let us now turn to the proof of point (2). Therefore, we first prove
the following corollary of Theorem~\ref{thr:17}
\begin{Cor}
  \label{cor:2}
  Fix $I\subset\overset{\circ}{\Sigma}_\Z$ compact. There exists
  $\eta_0>0$ such that, for $L$ sufficiently large, one has
  \begin{equation}
    \label{eq:50}
    \min_{\substack{\text{Re}\,E\in I\\\text{Im}\,E\in [-\eta_0/L,0)}}
    \left|S_L(E)+e^{-i\theta(E)}\right|\geq\eta_0\quad\text{ and
    }\quad \min_{\substack{\text{Re}\,E\in I\\\text{Im}\,E\in [-\eta_0/L,0)}}
    \left|\text{det}\left(\Gamma_L(E)+e^{-i\theta(E)}\right)\right|\geq\eta_0.
  \end{equation}
\end{Cor}
\noindent Clearly, Corollary~\ref{cor:2} implies that neither
equation~\eqref{eq:1} nor equation~\eqref{eq:135} can have a solution
in $I+i]-\eta_0/L,0]$. This proves point (2) of
Theorem~\ref{thr:5}.\qed\vskip.2cm\noindent
Before proving Corollary~\ref{cor:2}, we first prove
Propositions~\ref{pro:4} and~\ref{pro:5} as these will be used in the
proof of Corollary~\ref{cor:2}.
\subsubsection{Results on the auxiliary functions defined in section~\ref{sec:auxil-op}}
\label{sec:proof-proposition}
Recall that $N^-_k$ is defined in section~\ref{sec:auxil-op}. We prove
\begin{Pro}
  \label{pro:1}
  For $k\in\{0,\cdots,p-1\}$, $dN^-_k$ is a positive measure that is
  absolutely continuous on $\Sigma_\Z$. Moreover, its density, say,
  $E\mapsto n^-_k(E)$ is real analytic on $\overset{\circ}{\Sigma}_\Z$
  and there exists $f^-_k:\overset{\circ}{\Sigma}_\Z\to\R$ a positive
  real analytic function such that, on $\overset{\circ}{\Sigma}_\Z$,
  one has $\D n^-_k(E) =f^-_k(E)\, n(E)$.
\end{Pro}
\begin{proof}
  Proposition~\ref{pro:1} is an immediate consequence of
  Theorems~\ref{thr:17} and~\ref{thr:19} and Corollary~\ref{cor:1}.
\end{proof}
\noindent For $\Xi^-_k$ defoined in~\eqref{eq:8}, we prove
\begin{Pro}
  \label{pro:4}
  $\Xi^-_k$ vanishes identically if and only if $V\equiv0$, i.e., $V$
  vanishes identically. Moreover, if $V\not\equiv0$ then there exists
  $\xi^-_k\not=0$ and $\alpha^-_k\in\{2,3,\cdots\}$ such that
  $\D\Xi^-_k(E)\equ_{\substack{|E|\to\infty\\
      \text{Im}\,E<0}}\xi_k^-E^{-\alpha^-_k}$.
\end{Pro}
\begin{proof}
  We will do the proofs for the function
  $\Xi_k^-$. Proposition~\ref{pro:4} is an immediate consequence of
  the fact that, in the lower half-plane, the function $\D E\mapsto
  -e^{-i\arccos(E/2)}=-\frac E2-\sqrt{\frac{E^2}4-1}$ (i.e. the
  determination of it defined above) is equal to the Stieltjes (or
  Borel) transform of the spectral measure associated to the Dirichlet
  Laplacian on $\N$ and the vector $\delta_0$; this follows from a
  direct computation (see Remark~\ref{rem:5} and~\eqref{eq:15} for
  $n=0$). Now, if one lets $W$ be the symmetric of $\tau_k V$ with
  respect to $0$, the spectral measure $dN_k^-$ is also the spectral
  measure of the Schr{\"o}dinger operator $H_k=-\Delta+W$ on $\N$
  associated to $\delta_0$. The equality of the Borel transforms
  implies the equality of the measures but $\delta_0$ is cyclic for
  both operators so the operators have equal spectral measures. This
  implies that the two operators are equal and, thus, the symmetric of
  $\tau_k V$ has to vanish identically on $\N$. As $V$ is periodic,
  $V$ must vanish
  identically.\\
  As for the second point, if the function $\Xi_k^-$ were to vanish to
  infinite order at $E=-i\infty$, as each of the terms
  $\D\int_\R\frac{dN^-_k(\lambda)}{\lambda-E}$ and $\D-\frac E2-
  \sqrt{\frac{E^2}4-1}$ admits an infinite asymptotic expansion in
  powers of $E^{-1}$, these two expansions would be equal. The $n$-th
  coefficient of these expansion are respectively the $n$-th moments
  of the spectral measures of $H_k$ and $-\Delta_0^+$ (associated to
  the cyclic vector $\delta_0$). So these moments would coincide and,
  thus, the spectral measures would coincide. One concludes as
  above.\qed
\end{proof}
\noindent For $c^{\bullet}$ defined in~\eqref{eq:232}
and~\eqref{eq:231}, we prove
\begin{Pro}
  \label{pro:5}
  Pick $\bullet\in\{\N,\Z\}$. Let
  $I\subset(-2,2)\cap\overset{\circ}{\Sigma}_\Z$ be a compact
  interval.\\
  There exists a neighborhood of $I$ such that, in this neighborhood,
  the function $E\mapsto c^\bullet(E)$ is analytic and
  has  a positive imaginary part.\\
  The function $c^\N$ (resp.  $c^\Z$) takes the value $i$ only at the
  zeros of $\Xi_k^-$ (resp. $\Xi_k^-\,\Xi_0^+$).
\end{Pro}
\begin{proof}
  On $\{$Im$\,E<0\}$, define the functions
  \begin{gather}
    \label{eq:160}
    g^-_k(E):=i+\frac{\Xi_k^-(E)}{\pi\,n_k^-(E)}=\frac1{\pi\,n_k^-(E)}
    \left(S_k^-(E)+e^{-i\arccos(E/2)}\right),  \\
    \label{eq:161}
    g^+_0(E):=i+\frac{\Xi_0^+(E)}{\pi\,n^+_0(E)}=\frac1{\pi\,n_0^+(E)}
    \left(S^+_0(E)+e^{-i\arccos(E/2)}\right).
  \end{gather}
  First, the analyticity of $g_k^-$ and $g_0^+$ is clear; indeed, all
  the functions involved are analytic and the functions $n_0^+$ and
  $n_k^-$ stay positive on $\overset{\circ}{\Sigma}_\Z$. Moreover,
  these functions can be analytically continued through
  $(-2,2)\cap\overset{\circ}{\Sigma}_Z$.  By~\eqref{eq:51}, for $E$
  real, one has Im$\,g_k^-(E)=$Im$\,g_0^+(E)=$ Im$\,e^{-i\theta(E)}$
  which is positive (see section~\ref{sec:proof-theorem}). Thus, the
  functions $E\mapsto g^-_k(E)$ and $E\mapsto g^+_0(E)$ do not vanish
  on $I$. Moreover, as
  \begin{equation}
    \label{eq:54}
    \frac{g^+_0(E)g^-_k(E)-1}{g^+_0(E)+g^-_k(E)}=-\frac1{g^+_0(E)+g^-_k(E)}+
    \frac1{\D\frac1{g^+_0(E)}+\frac1{g^-_k(E)}};
  \end{equation}
  this function has a positive imaginary part on $I$.\\
  This proves the first two properties of $c^\bullet$ stated in
  Proposition~\ref{pro:5}. By the very definition of $c^\bullet$ and
  $g_k^-$, the last property stated in Proposition~\ref{pro:5} is
  obviously satisfied in the case of the half-line; for the full line
 , i.e., if $\bullet=\Z$, the last property is a consequence of the
  following computation
  \begin{equation}
    \label{eq:73}
    \begin{split}
      c^\Z(E)-i&=\frac{g^+_0(E)g^-_k(E)-1}{g^+_0(E)+g^-_k(E)}-i
      =\frac{(g^+_0(E)-i)(g^-_k(E)-i)}{g^+_0(E)+g^-_k(E)}\\&=
      \frac{\Xi^+_0(E)\Xi^-_k(E)}{2i\pi^2n_0^+(E)n_k^-(E)+\pi
        n_k^-(E)\Xi^+_0(E)+\pi n_0^+(E)\Xi^-_k(E)}.
    \end{split}
  \end{equation}
  This completes the proof of Proposition~\ref{pro:5}.\qed
\end{proof}
\subsubsection{The proof of Corollary~\ref{cor:2}}
\label{sec:proof-corollary}
In view of Theorem~\ref{thr:17}, to obtain~\eqref{eq:50}, it suffices
to prove that there exists $\eta_0>0$ such that, for $L$ sufficiently
large, one has
\begin{equation*}
  \begin{split}
    \min_{\substack{\text{Re}\,E\in I\\\text{Im}\,E\in [-\eta_0/L,0)}}
    \left|\frac{\theta'_{p,L}(E)f^-_k(E)e^{-i u_L(E)}}{\sin u_L(E)}
      -\int_\R\frac{dN^-_k(\lambda)}{\lambda-E}
      -e^{-i\theta(E)}\right|\geq\eta_0
  \end{split}
\end{equation*}
where $\D u_L(E):=(L-k)\theta_{p,L}(E)$.\\
We compute
\begin{equation}
  \label{eq:72}
  \frac{\theta'_{p,L}(E)f^-_k(E)e^{-i u_L(E)}}{\sin u_L(E)}
  -\int_\R\frac{dN^-_k(\lambda)}{\lambda-E} -e^{-i\theta(E)}=
  \theta'_{p,L}(E)f^-_k(E) \left(\cot u_L(E)-g_k^-(E)\right)
\end{equation}
where $g_k^-$ is defined in~\eqref{eq:160}.
Thus,
\begin{equation*}
  \left|\frac{\theta'_{p,L}(E)f^-_k(E)e^{-i u_L(E)}}{\sin u_L(E)}
    -\int_\R\frac{dN^-_k(\lambda)}{\lambda-E}
    -e^{-i\theta(E)}\right|\gtrsim \left|\cot
    u_L(E)-g_k^-(E)\right|
\end{equation*}
as, for $\eta$ sufficiently small and $L\geq1$, one has
\begin{equation*}
  0<\min_{\substack{\text{Re}\,E\in I\\\text{Im}\,E\in [-\eta/L,0)}}
  \left|\theta'_{p,L}(E)f^-_k(E)\right|\leq
  \max_{\substack{\text{Re}\,E\in I\\\text{Im}\,E\in [-\eta/L,0)}}
  \left|\theta'_{p,L}(E)f^-_k(E)\right|<+\infty.
\end{equation*}
Now, notice that, by Corollary~\ref{cor:1}, for $E\in I$, one has
\begin{equation}
  \label{eq:52}
  \text{Im}\left(\int_\R\frac{dN^-_k(\lambda)}{\lambda-E}\right)=
  -\theta'_{p,L}(E)f^-_k(E)=-\frac1\pi n_k^-(E).
\end{equation}
Thus, as $E\mapsto$Im$\,e^{-i\theta(E)}$ is positive on $I$, the
analytic function $\D E\mapsto g_k^-(E)$ has positive imaginary part
larger than, say, $2\tilde\eta$ on $I$; hence, it has imaginary part
larger than, say, $\tilde\eta$ in some neighborhood of
$I+\overline{D(0,\eta_0)}$ (for sufficiently small $\eta_0>0$). Let
$M$ be the maximum modulus of this function on
$I+\overline{D(0,\eta_0)}$. Thus, as
$\D\max_{\substack{\text{Re}\,E\in I\\\text{Im}\,E\in
    [-\eta_0/L,0)}}|\theta'_{p,L}(E)|\lesssim1$, one has
\begin{equation*}
  \max_{\substack{\text{Re}\,E\in I\\\text{Im}\,E\in
      [-\eta_0/L,0)\\|\cot(u_L(E))|<2M}} \left|\text{Im}\,\cot
    u_L(E)\right|\lesssim (M^2+1)\eta_0.  
\end{equation*}
Possibly reducing $\eta_0$, this guarantees that, for $\text{Re}\,E\in
I$ and $\text{Im}\,E\in [-\eta_0/L,0)$, one has
\begin{equation*}
  \begin{aligned}
    &\text{either }\quad \left|\cot u_L(E)-g_k^-(E)\right|\geq
    2M-M\geq M \\&\text{ or }\quad \text{Im}\left(\cot
      u_L(E)-g_k^-(E)\right)\leq
    -\tilde\eta+\tilde\eta/2=-\tilde\eta/2.
  \end{aligned}
\end{equation*}
This completes the proof of the first lower bound in~\eqref{eq:50} in
Corollary~\ref{cor:2}.\\
To prove the second bound in~\eqref{eq:50}, using~\eqref{eq:162}, we
compute
\begin{equation}
  \label{eq:70}
  \begin{split}
    \frac{\text{det}\left(\Gamma_L^{\text{eff}}(E)+e^{-i\theta(E)}\right)}
    {n_k^-(E)n_0^+(E)}&=\left(\cot u_L(E)-g_k^-(E)\right) \left(\cot
      u_L(E)-g_0^+(E)\right)-\frac1{\sin^2 u_L(E)}\\
    &=-\left(g^+_0(E)+g^-_k(E)\right)\left(\cot
      u_L(E)-\frac{g^+_0(E)g^-_k(E)-1}{g^+_0(E)+g^-_k(E)}\right)
  \end{split}
\end{equation}
where $g^-_k$ and $g_0^+$ are defined by~\eqref{eq:160}
and~\eqref{eq:161}.\\
Using Proposition~\ref{pro:5}, one then concludes the non-vanishing of
$E\mapsto\text{det}\left(\Gamma_L^{\text{eff}}(E)+e^{-i\theta(E)}\right)$
in the complex rectangle $\{\text{Re}\,E\in I,\ \text{Im}\,E\in
[-\eta_0/L,0)\}$ (for $\eta_0$ sufficiently small) in the same way as
above. This completes the proof of Corollary~\ref{cor:2}.  \qed
\subsubsection{The proof of Theorem~\ref{thr:2}}
\label{sec:proof-theorem-52}
To solve~\eqref{eq:1} and~\eqref{eq:135}, by Theorem~\ref{thr:17}, we
respectively first solve the equations
\begin{equation}
  \label{eq:68}
  \frac{\theta'_{p,L}(E)f^-_k(E)e^{-i u_L(E)}}{\sin u_L(E)}
  =\int_\R\frac{dN^-_k(\lambda)}{\lambda-E}
  -e^{-i\theta(E)}\quad\text{and}\quad
  \text{det}\left(\Gamma_L^{\text{eff}}(E)+e^{-i\theta(E)}\right)=0
\end{equation}
in a rectangle $I+i[-\eta,-\tilde\eta/L]$. Indeed, in such a
rectangle, by Theorem~\ref{thr:17}, equations~\eqref{eq:1}
and~\eqref{eq:135} are respectively equivalent to
\begin{equation}
  \label{eq:40}
  \begin{aligned}
    \frac{\theta'_{p,L}(E)f^-_k(E)e^{-i u_L(E)}}{\sin u_L(E)}
    &=\int_\R\frac{dN^-_k(\lambda)}{\lambda-E}
    -e^{-i\theta(E)}+O\left(L^{-\infty}\right)\\
    \text{ and }\quad\text{det}\left(\Gamma_L^{\text{eff}}(E)
      +e^{-i\theta(E)}\right)&= O\left(L^{-\infty}\right)
  \end{aligned}
\end{equation}
where the terms $O\left(L^{-\infty}\right)$ are analytic in a
rectangle $\tilde I+i[-2\eta,-0)$ (where $I\subset\tilde I$) and the
bound
$O\left(L^{-\infty}\right)$ holds in the supremum norm.\\
Thanks to~\eqref{eq:72} for $\bullet=\N$ and to~\eqref{eq:70} for
$\bullet=\Z$, to solve the equations~\eqref{eq:68}, it suffices to
solve
\begin{equation}
  \label{eq:139}
  \cot u_L(E)=c^\bullet(E)
\end{equation}
where we recall $\D u_L(E):=(L-k)\theta_{p,L}(E)$ and, $g_0^+$ and
$g_k^-$ being respectively defined in~\eqref{eq:161}
and~\eqref{eq:160}, and, as in section~\ref{sec:descr-reson}, one has
set
\begin{itemize}
\item $c^\N(E):=g_k^-(E)$ in the case of the half-line,
\item $\D c^\Z(E):=\frac{g^+_0(E)g^-_k(E)-1}{g^+_0(E)+g^-_k(E)}$ in
  the case of the line.
\end{itemize}
We want to solve~\eqref{eq:139} is a rectangle $I+i[-\varepsilon,0)$
for some $\varepsilon$ small but fixed. Using Proposition~\ref{pro:5},
we pick $\varepsilon$ so small that, in the rectangle
$I+i[-\varepsilon,0]$, the only zeros of $c^\bullet-i$ are those on
the real line and Im$\, c^\bullet$ is positive in $I+i[-\varepsilon,0)$.\\
To solve~\eqref{eq:139}, we change variables $u=(L-k)\theta_{p,L}(E)$
that is, we write
\begin{equation*}
  E=\theta_{p,L}^{-1}\left(\frac{u}{L-k}\right).
\end{equation*}
As, for $L_0$ sufficiently large, $\D\inf_{\substack{L\geq L_0\\E\in
    I+i[-\varepsilon,0)}}\text{Re}\,\theta'_{p,L}(E)>c>0$, at the cost
of possibly reducing $\varepsilon$, this real analytic change of
variables maps $I+[-\varepsilon,\varepsilon]+i[-\varepsilon,0)$ into,
say, $D_L$ such that $I_L+i[-\eta(L-k),0]\subset D_L$ (for some $\eta>0$)
where $I_L=(L-k)\theta_{p,L}(I+[-\varepsilon/2,\varepsilon/2])$; the
inverse change of variable maps $I_L+i[-\eta(L-k),0]$ into some
domain, say, $\tilde D_L$ such that
$I+[-\varepsilon',\varepsilon']+i[-\varepsilon',0]\subset\tilde D_L$
(for some $0<\varepsilon'<\varepsilon$). Now, to find all the
solutions to~\eqref{eq:139} in $I+i[-\varepsilon',0)$, we first solve
the following equation in $I_L+i[-\eta(L-k),0]$
\begin{equation}
  \label{eq:163}
  \cot u=c^\bullet\circ\theta_{p,L}^{-1}\left(\frac{u}{L-k}\right)
\end{equation}
As $u\mapsto\cot u$ is $\pi$ periodic, we split $I_L+i[-\eta(L-k),0]$
into vertical strips of the type $l\pi+[0,\pi]+i[-\eta(L-k),0]$,
$l_-\leq l\leq l_+$, $(l_-,l_+)\in\Z^2$. Without loss of generality, we may
assume that $I_L=[l_-,l_+]\pi$. To solve~\eqref{eq:163} on the
rectangle $l\pi+[0,\pi]+i[-\eta(L-k),0]$, we shift $u$ by $l\pi$ and
solve the following equation on $[0,\pi]+i[-\eta(L-k),0]$
\begin{equation}
  \label{eq:164}
  \cot u=c^\bullet_{l,L}(u) \quad\text{where}\quad c^\bullet_{l,L}(\cdot):=
  c^\bullet\circ\theta_{p,L}^{-1}\left(\frac{\cdot+l\pi}{L-k}\right).
\end{equation}
In proving Theorem~\ref{thr:5}, we have already shown that for some
$\tilde\eta>0$ (independent of $L$ sufficiently large and $l_-\leq
l\leq l_+$),~\eqref{eq:164} does not have a solution in
$[0,\pi]+i[-\tilde\eta,0]$. The cotangent is an analytic one-to-one
mapping from $[0,\pi)+i(-\infty,0]$ to $\C^+\setminus\{i\}$. Thus, for
$L$ sufficiently large and $\tilde\eta$ sufficiently small, the
cotangent defines a one-to-one mapping from
$[0,\pi)+i[-\eta(L-k),-\tilde\eta]$ onto
$T_L=\overline{D(z_+,r_+)\setminus D(z_-,r_-)}$, analytic in the
interior of $[0,\pi)+i[-\eta(L-k),-\tilde\eta]$ and continuous up to
the boundary where we have defined
\begin{equation*}
  z_+=i\frac{e^{4\eta(L-k)}+1}{e^{4\eta(L-k)}-1},\quad
  z_-=i\frac{e^{4\tilde\eta}-1}{e^{4\tilde\eta}-1},\quad
  r_+=\frac{2e^{2\tilde\eta}}{e^{4\tilde\eta}-1},\quad
  r_-=\frac{2e^{2\eta(L-k)}}{e^{4\eta(L-k)}-1}. 
\end{equation*}
Moreover, the boundaries $\{0\}+i[-\eta(L-k),-\tilde\eta]$ and
$\{\pi\}+i[-\eta(L-k),-\tilde\eta]$ are
mapped onto the interval $[z_-+ir_-,z_++ir_+]$.\\
Let $\tilde Z^\bullet$ denote the finite set of zeros of $E\mapsto
c^\bullet(E)-i$ in $I$. Then, by a Taylor expansion near the zeros of
$c-i$, we know that, for $\eta$ sufficiently small, there exists
$\varepsilon_0>0$ and $\tilde k\geq1$ such that, for $L$ sufficiently
large,
\begin{itemize}
\item for $\varepsilon\in(0,\varepsilon_0)$, there exists $0<\eta_-$
  such that, for $l_-\leq l\leq l_+$, if $\forall \tilde E\in\tilde
  Z^\bullet$, one has
  \begin{equation*}
    \left|\theta_{p,L}^{-1}
      \left(\frac{l\pi}{L-k}\right)-\tilde E\right|\geq\varepsilon
  \end{equation*}
  then $\forall u\in[0,\pi]+ i[-\eta(L-k),0]$, one has $\eta_-\leq
  |\text{Im}\, c^\bullet_{l,L}(u)-1|$;
\item for $u\in[0,\pi]+i[-\eta(L-k),0]$ and $\tilde E$ the point in
  $\tilde Z^\bullet$ closest to $\D
  \theta_{p,L}^{-1}\left(\frac{l\pi}{L-k}\right)$ , one has
  \begin{equation}
    \label{eq:165}
    \varepsilon_0\leq \left(1-\text{Im}\,c^\bullet_{l,L}(u)\right)\cdot\left[
      \left|\theta_{p,L}^{-1}\left(\frac{\text{Re}\,u+l\pi}{L-k}\right)-
        \tilde E\right|+\frac{|\text{Im}\,u|}{L-k}
    \right]^{-\tilde k}\leq\frac1{\varepsilon_0}
  \end{equation}
  where $\tilde k$ is the order of $\tilde E$ as a zero of $E\mapsto
  c^\bullet(E)-i$.
\end{itemize}
As a consequence of the above description of $c^\bullet_{l,L}$, we
obtain
\begin{Le}
  \label{le:12}
  There exists $\tilde\eta$ and $\eta$ small such that, for $L$
  sufficiently large, for all $l_-\leq l\leq l_+$, $\D u\mapsto
  c^\bullet_{l,L}(u)$ maps the rectangle
  $[0,\pi]+i[-\eta(L-k),-\tilde\eta]$ into a compact subset of
  $D(z_+,r_+)\setminus D(z_-,r_-)$ in such a way that
  \begin{equation}
    \label{eq:178}
    \sup_{u\in\partial([0,\pi]+i[-\eta(L-k),-\tilde\eta])}
    \left|\cot u- c^\bullet_{l,L}(u)\right|\gtrsim
    \left(\left|\tilde E-\theta_{p,L}^{-1}\left(
          \frac{l\pi}{L-k}\right)\right|+\frac{\tilde\eta}{L-k}
    \right)^{\tilde k}
  \end{equation}
  where $\tilde E$ is the root of $E\mapsto c^\bullet(E)-i$ closest to
  $\D \theta_{p,L}^{-1}\left(\frac{l\pi}{L-k}\right)$ and $\tilde k$
  is the order of this root.
\end{Le}
\noindent Note that, under the assumptions of
Lemma~\ref{le:12},~\eqref{eq:178} implies that
\begin{equation*}
  \sup_{u\in\partial([0,\pi]+i[-\eta(L-k),-\tilde\eta])}
  \left|\cot u- c^\bullet_{l,L}(u)\right|\gtrsim L^{-\tilde k}
\end{equation*}
Thus, we can define the analytic mapping
$\cot^{-1}\circ\,c^\bullet_{l,L}$ on
$[0,\pi]+i[-\eta(L-k),-\tilde\eta]$; it maps the rectangle
$[0,\pi]+i[-\eta(L-k),-\tilde\eta]$ into a compact subset of
$(0,\pi)+i(-\eta(L-k),-\tilde\eta)$. The equation~\eqref{eq:164} on
$[0,\pi]+i[-\eta(L-k),-\tilde\eta]$ is, thus, equivalent to the
following fixed point equation on the same rectangle
\begin{equation}
  \label{eq:167}
  u=\cot^{-1}\circ\,c^\bullet_{l,L}(u)
\end{equation}
We note that, for $\alpha\in(0,1)$, for $L$ sufficiently large, if for
some $\tilde E\in\tilde Z^\bullet$ of multiplicity $\tilde k$, one has
$\D \left|\theta_{p,L}^{-1}\left(\frac{l\pi}{L-k}\right)- \tilde
  E\right|<L^{-\alpha}$ then, equation~\eqref{eq:164} has no solution
in $[0,\pi]+i[-\eta(L-k),-\tilde\eta]$ outside of the set
\begin{equation*}
  R_{l,L}:=[0,\pi]+i\left[-\eta(L-k),\frac{\alpha\tilde k}4\log\left[
      \left|\theta_{p,L}^{-1}\left(\frac{l\pi}{L-k}\right)-
        \tilde E\right|+\frac1L \right]\right].
\end{equation*}
Indeed, for $u\in([0,\pi]+i[-\eta(L-k),-\tilde\eta])\setminus
R_{l,L}$, by~\eqref{eq:165}, that is, for
\begin{equation*}
  0\leq\text{Re}\,u\leq\pi  \quad\text{and}\quad
  -\frac{\alpha\tilde
    k}4\log L\leq \frac{\alpha\tilde
    k}4\log\left[\left|\theta_{p,L}^{-1}\left(\frac{l\pi}{L-k}\right)-
      \tilde E\right|+\frac1L
  \right]\leq\text{Im}\,u\leq  -\tilde\eta
\end{equation*}
one has $ \left|c^\bullet_{l,L}(u)-i\right|\lesssim L^{-\alpha\tilde
  k}$ and
$|\cot u-i|\gtrsim L^{-\alpha\tilde k/2}$. \\
So, if for some $\tilde E\in\tilde Z^\bullet$, one has $\D
\left|\theta_{p,L}^{-1}\left(\frac{l\pi}{L-k}\right)- \tilde
  E\right|<L^{-\alpha}$, it suffices to solve~\eqref{eq:167} on
$R_{l,L}$. We compute the derivative of $c^\bullet_{l,L}$ in the
interior of $R_{l,L}$
\begin{equation*}
  \begin{split}
    \frac{d}{du}\left(\cot^{-1}\circ\,
      c^\bullet_{l,L}\right)(u)&=-\frac1{L-k}\frac{c'\circ\theta_{p,L}^{-1}
      \left(\frac{u+l\pi}{L-k}\right)}{1+\left(c^\bullet_{l,L}(u)\right)^2}
    \cdot\frac1{\theta'_{p,L}\
      \left(\theta_{p,L}^{-1}\left(\frac{u+l\pi}{L-k}\right)\right)}\\
    &=\frac1{L-k}\frac{c'\circ\theta_{p,L}^{-1}
      \left(\frac{u+l\pi}{L-k}\right)}{c^\bullet_{l,L}(u)-i}
    \cdot\frac1{c^\bullet_{l,L}(u)+i} \cdot\frac1{\theta'_{p,L}\
      \left(\theta_{p,L}^{-1}\left(\frac{u+l\pi}{L-k}\right)\right)}.
  \end{split}
\end{equation*}
Thus, fixing $\alpha\in(0,1)$,
\begin{itemize}
\item if $l$ is such that, for some $\tilde E\in\tilde Z^\bullet$, one
  has $\D \left|\theta_{p,L}^{-1}\left(\frac{l\pi}{L-k}\right)- \tilde
    E\right|<L^{-\alpha}$, for $u\in R_{l,L}$, we estimate
  \begin{equation}
    \label{eq:172}
    \begin{split}
      \left|\frac{d}{du}\left(\cot^{-1}\circ\,
          c^\bullet_{l,L}\right)(u)\right|&\lesssim\frac1{L-k}\left[\left|
          \theta_{p,L}^{-1}\left(\frac{l\pi}{L-k}\right)- \tilde
          E\right|+\frac{|\text{Im}\,u|}{L-k}\right]^{-1}
      \\&\lesssim\frac1{(L-k)\left|
          \theta_{p,L}^{-1}\left(\frac{l\pi}{L-k}\right)- \tilde
          E\right|+\left|\log\left[
            \left|\theta_{p,L}^{-1}\left(\frac{l\pi}{L-k}\right)-
              \tilde E\right|+\frac{\tilde\eta}{L-k} \right]\right|}
      \\&\lesssim\frac1{\log L};
    \end{split}
  \end{equation}
\item if $l$ is such that, for all $\tilde E\in\tilde Z^\bullet$, one
  has $\D \left|\theta_{p,L}^{-1}\left(\frac{l\pi}{L-k}\right)- \tilde
    E\right|\geq L^{-\alpha}$, for $u\in
  [0,\pi]+i[-\eta(L-k),-\tilde\eta]$, we estimate
  \begin{equation}
    \label{eq:177}
    \begin{split}
      \left|\frac{d}{du}\left(\cot^{-1}\circ\,
          c^\bullet_{l,L}\right)(u)\right|&\lesssim\frac1{L-k}\left[\left|
          \theta_{p,L}^{-1}\left(\frac{l\pi}{L-k}\right)- \tilde
          E\right|+\frac{|\text{Im}\,u|}{L-k}\right]^{-1}
      \\&\lesssim\frac1{(L-k)\left|
          \theta_{p,L}^{-1}\left(\frac{l\pi}{L-k}\right)- \tilde
          E\right|}\lesssim\frac1{L^{1-\alpha}}.
    \end{split}
  \end{equation}
\end{itemize}
Hence, for $L$ sufficiently large, $\cot^{-1}\circ\,c^\bullet_{l,L}$
is a contraction on $R_{l,L}$. Equation~\eqref{eq:167} thus admits a
unique solution, say, $\tilde u^\bullet_{l,L}$ in the rectangle
$[0,\pi]+i[-\eta(L-k),-\tilde\eta]$. This solution is a simple root of
$u\mapsto u-\cot^{-1}\circ \,c^\bullet_{l,L}(u)$. Hence, $\tilde
u^\bullet_{l,L}$ is the only solution to equation~\eqref{eq:164} in
$[0,\pi]+i[-\eta(L-k),-\tilde\eta]$.\\
By~\eqref{eq:40}, for $L$ sufficiently large, for $l_-\leq l\leq l_+$, both
the equations
\begin{equation}
  \label{eq:179}
  \begin{aligned}
    S_L\circ\theta^{-1}_{p,L}\left(\frac{u+l\pi}{L-k}\right)
    +e^{-i\theta(\theta^{-1}_{p,L}\left(\frac{u+l\pi}{L-k}\right))}&=0
    \quad\text{and}\\\text{det}\left(\Gamma_L\circ\theta^{-1}_{p,L}
      \left(\frac{u+l\pi}{L-k}\right)+
      e^{-i\theta(\theta^{-1}_{p,L}\left(\frac{u+l\pi}{L-k}\right))}\right)&=0
  \end{aligned}
\end{equation}
can be rewritten as
\begin{equation}
  \label{eq:184}
  u=\cot^{-1}\left(c^\bullet_{l,L}(u)+
    O\left(L^{-\infty} \right)\right)=\cot^{-1}\circ\, c^\bullet_{l,L}(u)+
  O\left( L^{-\infty}\right)
\end{equation}
in $[0,\pi]+i[-\eta(L-k),-\tilde\eta]$.\\
Thus, each of the equations in~\eqref{eq:179} admits a single solution
in $[0,\pi]+i[-\eta(L-k),-\tilde\eta]$ and this root is simple;
moreover, this solution, say, $u_{l,L}$ satisfies
$\left|u^\bullet_{l,L}-\tilde u^\bullet_{l,L}\right|=O\left(
  L^{-\infty}\right)$; indeed, the bounds~\eqref{eq:172}
and~\eqref{eq:177} guarantee that one can apply Rouch{\'e}'s Theorem on
the disk $D(\tilde u^\bullet_{l,L},L^{-k})$ for any $k\geq0$. \\
Thus, we have proved the
\begin{Le}
  \label{le:9}
  Pick $I$ as above. Then, there exists $\eta>0$ such that, for $L$
  sufficiently large s.t. $L=Np+k$, the resonances in $I+i[-\eta,0]$
  are the energies $(z^\bullet_l)_{l_-\leq l\leq l_+}$ defined by
  \begin{equation}
    \label{eq:141}
    z^\bullet_l=\theta^{-1}_{p,L}\left(\frac{u^\bullet_{l,L}+l\pi}{L-k}\right)
  \end{equation}
  belonging to $I+i[-\eta,0]$.
\end{Le}
\noindent Let us complete the proof of Theorem~\ref{thr:6} that is,
prove that, for $\eta$ sufficiently small, for $L$ sufficiently large
such that $L\equiv k\mod(p)$, is the unique resonance in
$\D\left[\frac{\text{Re}\,(\tilde z^\bullet_l+\tilde
    z^\bullet_{l-1})}2, \frac{\text{Re}\,(\tilde z^\bullet_l+\tilde
    z^\bullet_{l+1})}2\right] +i\left[-\eta,0\right]$; recall that
$\tilde z^\bullet_l$ is defined in~\eqref{eq:182}.\\
Therefore, we first note that the Taylor expansion of
$\theta^{-1}_{p,L}$,~\eqref{eq:183} and the quantization
condition~\eqref{eq:23} imply that
\begin{equation*}
  z^\bullet_l=\lambda_l+\frac1{\pi n(\lambda_l)L} u^\bullet_{l,L}
  +O\left(\left(\frac{\log L}L\right)^2\right)
\end{equation*}
as Re$\,u_{l,L}\in[0,\pi)$ and $-\log L\lesssim\,$Im$\,
u_{l,L}\lesssim -1$.\\
Moreover, as $\D c^\bullet_{l,L}(u)=c^\bullet\left[\lambda_l+\frac
  u{\pi\,n(\lambda_l)\,L}+O\left(\frac{u^2}{L^2}\right)\right]$
using~\eqref{eq:182} and~\eqref{eq:141}, we compute
\begin{equation*}
  z^\bullet_l-\tilde z^\bullet_l=\frac1{\pi n(\lambda_l)L}
  \left(u^\bullet_{l,L}- \cot^{-1}\circ\,
    c^\bullet\left[\lambda_l+\frac1{\pi\,n(\lambda_l)\,L}\cot^{-1}\circ\,
      c^\bullet\left(\lambda_l-i\frac{\log L}L\right)\right]
  \right) +O\left(\left(\frac{\log L}L\right)^2\right).
\end{equation*}
Thus, one has
\begin{equation*}
  z^\bullet_l-\tilde z^\bullet_l=\frac1{\pi n(\lambda_l)L}
  \left(u^\bullet_{l,L}- \cot^{-1}\circ\,
    c^\bullet_{l,L}\left[\cot^{-1}\circ\,
      c^\bullet_{l,L}\left(-i\pi\, n(\lambda_l)\log L\right)\right]
  \right)
  +O\left(\left(\frac{\log L}L\right)^2\right).
\end{equation*}
As $u_{l,L}$ solves~\eqref{eq:184}, using~\eqref{eq:172}
and~\eqref{eq:177}, we thus obtain that
\begin{equation*}
  \begin{split}
    \left|z^\bullet_l-\tilde z^\bullet_l\right|&\lesssim \frac1{L \log
      L} \left|u^\bullet_{l,L}-\cot^{-1}\circ\,
      c^\bullet_{l,L}\left(-i\pi\, n(\lambda_l)\log L\right)\right|+
    \left(\frac{\log L}L\right)^2\\&\lesssim
    \frac{\left|u^\bullet_{l,L}\right|+\log L} {L\log^2 L} +
    \left(\frac{\log L}L\right)^2\lesssim \frac1{L\log L}
  \end{split}
\end{equation*}
using again Re$\,u_{l,L}\in[0,\pi)$ and $-\log L\lesssim\,$Im$\,
u_{l,L}\lesssim -1$.\\
Taking into account~\eqref{eq:56}, this complete the proof of
Theorem~\ref{thr:2}.\qed
\subsubsection{The proofs of Propositions~\ref{pro:2} and~\ref{pro:3}}
\label{sec:proofs-propositions}
Proposition~\ref{pro:3} is an immediate consequence of
Theorem~\ref{thr:2}, the definition of $\tilde z^
\bullet_l$~\eqref{eq:182} and the standard asymptotics of $\cot$ near
$-i\infty$, i.e., $\D \cot z=i+2ie^{-2iz}+O\left(e^{-4iz}\right)$.
\vskip.1cm\noindent
To prove Proposition~\ref{pro:2}, it suffices to notice that, under
the assumptions of Proposition~\ref{pro:2}, the bound~\eqref{eq:177}
on the derivative of $\cot^{-1}\circ\, c^\bullet_{l,L}$ on the the
rectangle $R_{l,L}$ becomes
\begin{equation*}
  \left|\frac{d}{du}\left(\cot^{-1}\circ\,
      c^\bullet_{l,L}\right)(u)\right|\lesssim\frac1L.      
\end{equation*}
Thus, as a solution to~\eqref{eq:167}, $u^\bullet_{l,L}$ admits an
asymptotic expansion in inverse powers of $L$. Plugging this
into~\eqref{eq:141} yields the asymptotic expansion for the
resonance. Then,~\eqref{eq:176} follows from the computation of the
first terms. \qed
\subsubsection{The proof of Theorem~\ref{thr:3}}
\label{sec:proof-theorem-53}
Theorem~\ref{thr:3} is an immediate consequence of
Theorem~\ref{thr:19}, the fact that the functions are analytic in the
lower complex half-plane and have only finitely many zeros there and
the argument principle. \qed
\subsection{The half-line periodic perturbation: the proof of
  Theorem~\ref{thr:25}}
\label{sec:half-line-periodic}
Using the same notations as above, we can write
\begin{equation*}
  H^\infty=  \begin{pmatrix}    H_{-1}^-&
    |\delta_{-1}\rangle\langle\delta_0|
    \\  |\delta_0\rangle\langle\delta_{-1}|&-\Delta_0^+
  \end{pmatrix}.
\end{equation*}
where $-\Delta_{0}^+$ is the Dirichlet Laplacian on $\ell^2(\N)$.\\
Define the operators
\begin{gather*}
  \Gamma(E):=H_{-1}^--E-\langle\delta_0|(-\Delta_0^+-E)^{-1}|\delta_0\rangle
  \,|\delta_{-1}\rangle\langle\delta_{-1}|\\\intertext{and}
  \tilde\Gamma(E):=-\Delta_0^+-E-\langle\delta_{-1}|(H_{-1}^--E)^{-1}
  |\delta_{-1}\rangle \,|\delta_{0}\rangle\langle\delta_{0}|.
\end{gather*}
For Im$\,E\not=0$,
$\langle\delta_{-1}|(H_{-1}^--E)^{-1}|\delta_{-1}\rangle$ and
$\langle\delta_0|(-\Delta_0^+-E)^{-1}|\delta_0\rangle$ have a non
vanishing imaginary part of the same sign; hence, the complex number
\begin{equation*}
  (\langle\delta_0|(-\Delta_0^+-E)^{-1}|\delta_0\rangle)^{-1}-
  \langle\delta_{-1}|(H_{-1}^--E)^{-1}|\delta_{-1}\rangle
\end{equation*}
does not vanish. Thus, by rank one perturbation theory, (see,
e.g.,~\cite{MR2154153}), we know that $\Gamma(E)$ and $\tilde\Gamma(E)$
are invertible and their inverses are given by
\begin{equation}
  \label{eq:125}
  \Gamma^{-1}(E):=(H_{-1}^--E)^{-1}+
  \frac{|H_{-1}^--E)^{-1}|\delta_{-1}\rangle\langle\delta_{-1}|(H_{-1}^--E)^{-1}|}{(\langle\delta_0|(-\Delta_0^+-E)^{-1}|\delta_0\rangle)^{-1}-
    \langle\delta_{-1}|(H_{-1}^--E)^{-1}|\delta_{-1}\rangle}.
\end{equation}
and
\begin{equation}
  \tilde\Gamma^{-1}(E):=(-\Delta_0^+-E)^{-1}+
  \frac{|-\Delta_0^+-E)^{-1}|\delta_{0}\rangle\langle\delta_{0}|(-\Delta_0^+-E)^{-1}|}{(\langle\delta_{-1}|(H_{-1}^--E)^{-1}|\delta_{-1}\rangle^{-1}-
    \langle\delta_0|(-\Delta_0^+-E)^{-1}|\delta_0\rangle)}.
\end{equation}
Thus, for Im$\,E\not=0$, using Schur's complement formula, we compute
\begin{equation}
  \label{eq:112}
  (H^\infty-E)^{-1}=\begin{pmatrix} \Gamma(E)^{-1}& \gamma(E) \\
    \gamma^*\left(\overline{E}\right) & \tilde \Gamma(E)^{-1}
  \end{pmatrix}.
\end{equation}
where $\gamma^*\left(\overline{E}\right)$ is the adjoint of
$\gamma\left(\overline{E}\right)$ and
\begin{equation*}
  \gamma(E):=-|\Gamma(E)^{-1}|\delta_{-1}\rangle\langle\delta_0|
  (-\Delta_0^+-E)^{-1}|.
\end{equation*}
Now, when coming from Im$\,E>0$ and passing through
$(-2,2)\cap\overset{\circ}{\Sigma}_\Z$, the complex numbers
$\langle\delta_{-1}|(H_{-1}^--E)^{-1}|\delta_{-1}\rangle$ and
$\langle\delta_0|(-\Delta_0^+-E)^{-1}|\delta_0\rangle$ keep imaginary
parts of the same positive sign; thus, the two operator-valued
functions $E\mapsto\Gamma^{-1}(E)$ and $E\mapsto(H^\infty-E)^{-1}$ can
be analytically continued through
$(-2,2)\cap\overset{\circ}{\Sigma}_\Z$ from the upper to the lower
complex half-plane (as operators respectively from
$\ell^2_{\text{comp}}(\N)$ to $\ell^2_{\text{loc}}(\N)$ and from
$\ell^2_{\text{comp}}(\Z)$ to
$\ell^2_{\text{loc}}(\Z)$). \\
When coming from the upper half-plane and passing through
$(-2,2)\setminus\Sigma_\Z$ and
$\overset{\circ}{\Sigma}_\Z\setminus[-2,2]$,~\eqref{eq:112} also
provides an analytic continuation of
$(H^\infty-E)^{-1}$. Definition~\eqref{eq:125} and
formula~\eqref{eq:112} immediately show that the poles of these
continuations only occur at the zeros of the function
\begin{equation*}
  \D E\mapsto 1-\langle\delta_{-1}|(H_{-1}^--E)^{-1}|\delta_{-1}\rangle  
  \langle\delta_0|(-\Delta_0^+-E)^{-1}|\delta_0\rangle
  =1-e^{i\theta(E)} \int_\R\frac{dN^-_{p-1}
    (\lambda)}{\lambda-E}  
\end{equation*}
when continued from the upper half-plane through the sets
$(-2,2)\setminus\Sigma_\Z$ and
$\overset{\circ}{\Sigma}_Z\setminus[-2,2]$ (these sets are finite
unions of open intervals).\\
This completes the proof of Theorem~\ref{thr:25}.\qed
\section{Resonances in the random case}
\label{sec:random-case-2}
As for the periodic potential, for the random potential, we start with
a description of the function $E\mapsto \Gamma_L(E)$
(see~\eqref{eq:142}), that is, with a description of the spectral data
for the Dirichlet operator $H_{\omega,L}$.
\subsection{The matrix $\Gamma_L$ in the random case}
\label{sec:function-s_l-random}
We recall a number of results on the Dirichlet eigenvalues of
$H_{\omega,L}$ that will be used in our analysis.\\
It is well known that, under our assumptions, in dimension one, the
whole spectrum of $H_\omega$ is in the localization region (see,
e.g.,~\cite{MR582611,MR883643,MR1102675}) that is
\begin{Th}
  \label{thr:18}
  There exists $\rho>0$ and $\alpha\in(0,1)$ such that, one has
  \begin{equation}
    \label{eq:204}
    \sup_{\substack{L\in\N\cup\{+\infty\}\\y\in\llbracket0,L\rrbracket
        \\\text{Im}\,E\not=0}}
    \esp\left\{\sum_{x\in\llbracket0,L\rrbracket}e^{\rho|x-y|}
      |\langle\delta_x,(H_{\omega,L}-E)^{-1} \delta_y \rangle|^\alpha  
    \right\}<\infty
  \end{equation}
  and
  \begin{equation}
    \label{FVdynloc1}
    \sup_{\substack{L\in\N\cup\{+\infty\}\\y\in\llbracket0,L\rrbracket}}
    \esp\left\{\sum_{x\in\llbracket0,L\rrbracket}e^{\rho|x-y|}
      \sup_{\substack{\text{supp}\,f \subset\R \\ |f|\leq 1 }}
      |\langle\delta_x,f(H_{\omega,L}) \delta_y \rangle|  \right\}<\infty. 
  \end{equation}
  where $H_{\omega,+\infty}:=H_\omega^\N$ and
  $\llbracket0,+\infty\rrbracket=\N$. The supremum is taken over the
  functions $f$ that are Borelian and compactly supported.
\end{Th}
\noindent As a consequence, one can define localization centers
e.g. by means of the following results
\begin{Le}[\cite{Ge-Kl:10}]
  \label{le:2}
  Fix $(l_L)_L$ a sequence of scales, i.e., $l_L\to+\infty$ as
  $L\to+\infty$. There exists $\rho>0$ such that, for $L$ sufficiently
  large, with probability larger than $1-e^{-\ell_L}$, if
  \begin{enumerate}
  \item $\varphi_{j,\omega}$ is a normalized eigenvector of
    $H_{\omega,L}$ associated to $E_{j,\omega}$ in $\Sigma$,
  \item $x_j(\omega)\in\llbracket 0,L\rrbracket$ is a maximum of
    $x\mapsto|\varphi_{j,\omega}(x)|$ in $\llbracket 0,L\rrbracket$,
  \end{enumerate}
  then, for $x\in\llbracket 0,L\rrbracket$, one has
  \begin{equation}
    \label{eq:75}
    |\varphi_{j,\omega}(x)|\leq \sqrt{L}e^{2\ell_L}e^{-\rho|x-x_j(\omega)|}.
  \end{equation}
\end{Le}
\noindent Note that Lemma~\ref{le:2} is of interest only if
$\ell_L\lesssim L$; otherwise~\eqref{eq:75} is obvious. This result
can e.g. be applied for the scales $l_L=2\log L$. In this case, the
probability estimate of the bad sets (i.e. when the conclusions of
Lemma~\ref{le:3} does not hold) is summable. The point $x_j(\omega)$
is a localization center for $E_{j,\omega}$ or
$\varphi_{j,\omega}$. It is not defined uniquely, but, one easily
shows that there exists $C>0$ such that for any two localization
centers, say, $x$ and $x'$, one has $|x-x'|\leq C\log L$
(see~\cite{Ge-Kl:10}). To fix ideas, we set the localization center
associated to the eigenvalue $E_{j,\omega}$ to be the left most
maximum of $x\mapsto\|\varphi_{j,\omega}\|_x$.\\
We show
\begin{Le}
  \label{le:3}
  For any $p>0$, there exists $C>0$ and $L_0>0$ (depending on $\alpha$
  and $p$) such that, for $L\geq L_0$, for any sequence
  satisfying~\eqref{eq:14}, with probability at least $1-L^{-p}$,
  there exists at most $C\ell_L$ eigenvalues having a localization
  center in $\llbracket 0,\ell_L\rrbracket\cup\llbracket
  L-\ell_L,L\rrbracket$.
\end{Le}
\noindent We will now use the fact that we are dealing with
one-dimensional systems to improve upon the estimate~\eqref{eq:75}. We
prove
\begin{Th}
  \label{thr:10}
  For any $\delta>0$ and $p\geq0$, there exists $C>0$ and $L_0>0$
  (depending on $p$ and $\delta$) such that, for $L\geq L_0$, with
  probability at least $1-L^{-p}$, if $E_{j,\omega}$ is an eigenvalue
  in $\Sigma$ associated to the eigenfunction $\varphi_{j,\omega}$ and
  the localization center $x_{j,\omega}$ then,
  \begin{itemize}
  \item if $x_{j,\omega}\in\llbracket 0,L-C\log L\rrbracket$, one has
    \begin{equation}
      \label{eq:76}
      -\rho(E_{j,\omega})-\delta
      \leq\frac{ \log|\varphi_{j,\omega}(L)|}{L-x_{j,\omega}}
      \leq -\rho(E_{j,\omega})+\delta.
    \end{equation}
  \item if $x_{j,\omega}\in\llbracket C \log L,L\rrbracket$, one has
    \begin{equation}
      \label{eq:97}
      -\rho(E_{j,\omega})-\delta
      \leq\frac{ \log|\varphi_{j,\omega}(0)|}{x_{j,\omega}}
      \leq -\rho(E_{j,\omega})+\delta.
    \end{equation}
  \end{itemize}
\end{Th}
\noindent To analyze the resonances of $H^\N_{\omega,L}$ (resp.
$H^\Z_{\omega,L}$), we shall use~\eqref{eq:76} (resp. \eqref{eq:76}
and~\eqref{eq:97}).\\
We now use these estimates as the starting point of a short digression
from the main theme of this paper. Let us first state a corollary to
Theorem~\ref{thr:10}, we prove
\begin{Th}
  \label{thr:20}
  For any $\delta>0$ and $p\geq0$, for $L$ sufficiently large
  (depending on $p$ and $\delta$), with probability at least
  $1-L^{-p}$, if $E_{j,\omega}$ is an eigenvalue in $\Sigma$
  associated to the eigenfunction $\varphi_{j,\omega}$ and the
  localization center $x_{j,\omega}$ then, for $|x-x_{j,\omega}|\geq
  \delta L$ and $1\leq x\leq L$, one has
  \begin{equation}
    \label{eq:88}
    -\rho(E_{j,\omega})-\delta\leq\frac{
      \log(|\varphi_{j,\omega}(x)|+|\varphi_{j,\omega}(x-1)|)}
    {|x-x_{j,\omega}|}
    \leq -\rho(E_{j,\omega})+\delta.
  \end{equation}
\end{Th}
\noindent Compare~\eqref{eq:88} to~\eqref{eq:75}. There are two
improvements. First, the unknown rate of decay $\rho$ is replaced by
the Lyapunov exponent $\rho(E_{j,\omega})$ which was expected to be
the correct decay rate. Indeed, for the one-dimensional discrete
Anderson model on the half-axis, it is well known (see,
e.g.,~\cite{MR88f:60013,MR1102675,MR94h:47068}) that, $\omega$-almost
surely, the spectrum is localized and the eigenfunctions decay
exponentially at infinity at a rate given by the Lyapunov exponent. In
Theorem~\ref{thr:20}, we state that, with a good probability, this is
true for finite volume restrictions.\\
Second, in~\eqref{eq:88}, we get both an upper and lower bound on the
eigenfunction. This is more precise than~\eqref{eq:75}.\\
To our knowledge, such a result was not known until the present
paper. The strategy that we use to prove this result can be applied in
a more general one-dimensional setting to obtain analogues
of~\eqref{eq:88} (see~\cite{Kl:12b}).\\
We complement this with the much simpler
\begin{Le}
  \label{le:8}
  For any $C>0$ and $p\geq0$, there exists $K>0$ and $L_0>0$
  (depending on $I$, $p$ and $\delta$) such that, for $L\geq L_0$,
  with probability at least $1-L^{-p}$, if $E_{j,\omega}$ is an
  eigenvalue in $\Sigma$ associated to the eigenfunction
  $\varphi_{j,\omega}$ and the localization center $x_{j,\omega}$
  then,
  \begin{itemize}
  \item if $x_{j,\omega}\in\llbracket L-C\log L, L\rrbracket$, one has
    $\D L^{-K} \leq|\varphi_{j,\omega}(L)|$;
  \item if $x_{j,\omega}\in\llbracket 0,C\log L\rrbracket$, one has
    $\D L^{-K} \leq|\varphi_{j,\omega}(0)|$.
  \end{itemize}
\end{Le}
\noindent The proof of this result is obvious and only uses the fact
that the matrices in the cocycle defining the operator (see
section~\ref{sec:estim-growth-eigenf}) are bounded that is,
equivalently, that the solutions to the Schr{\"o}dinger equation grow at
most exponentially at a rate controlled by the potential.
\vskip.2cm\noindent Let us return to the resonances in the random case
and the description of the function $S_L$. Recall that
in~\eqref{eq:1}, the values $(\lambda_j)_j$ are the eigenvalues
$(E_{j,\omega})_{0\leq j\leq L}$ of $H_{\omega,L}$ and the coefficients
$(a_j^\bullet)_j$ are defined in Theorem~\ref{thr:11} and
by~\eqref{eq:145}. Thus, Theorem~\ref{thr:10} describes the
coefficients $(a^\bullet_j)_j$ coming into $S_L$ and $\Gamma_L$
(see~\eqref{eq:1} and~\eqref{eq:135}). Let us now state a
few consequences of Theorem~\ref{thr:10}.\\
Fix $I$ a compact interval in $\Sigma$ the almost sure spectrum of
$H_\omega$. For $\bullet\in\{\N,\Z\}$, define
\begin{equation}
  \label{eq:188}
  d^\bullet_{j,\omega}=
  \begin{cases}
    L-x_{j,\omega}&\text{for }\bullet=\N,\\
    \min(x_{j,\omega},L-x_{j,\omega})&\text{for }\bullet=\Z.
  \end{cases}
\end{equation}
Taking $p>2$ in Theorem~\ref{thr:10} and using Borel-Cantelli
argument, we obtain that
\begin{equation}
  \label{eq:104}
  \begin{split}
    \omega\text{ almost surely, for }\delta>0 \text{ and }L\text{
      sufficiently large, if }\lambda_j=E_{j,\omega}\in I\\ \text{ and }
    d^\bullet_{j,\omega}\geq C\log L\text{ then }
    -2\rho(\lambda_j)-\delta \leq\frac{\log
      a^\bullet_j}{d^\bullet_{j,\omega}} \leq
    -2\rho(\lambda_j)+\delta.
  \end{split}
\end{equation}
This and the continuity of the Lyapunov exponent (see,
e.g.,~\cite{MR88f:60013,MR1102675,MR94h:47068}) guarantees that
\begin{equation}
  \label{eq:106}
  \omega\text{ almost surely, for any }\delta>0 \text{ and }L\text{
    large, one has }-2\eta_\bullet\sup_{E\in I}\rho(E)(1+\delta)L
  \leq\inf_{\lambda_j\in I}\log a^\bullet_j
\end{equation}
where $\eta_\bullet$ is defined in Theorem~\ref{thr:4}.\\
To use the analysis performed in section~\ref{sec:char-reson}, we also
need a description for the $(\lambda_j)_j$, i.e., the Dirichlet
eigenvalues of $H_{\omega,L}$. Therefore, we will use the results
of~\cite{Ge-Kl:10},~\cite{MR2775121} and~\cite{Kl:10a} (see
also~\cite{MR2885251}).\\
We first recall the Minami estimate satisfied by $H_{\omega,L}$ (see,
e.g.,~\cite{MR2505733} and references therein): there exists $C>0$ such
that, for $I\subset\R$, one has
\begin{equation*}
  \begin{split}
    \pro\left(\tr(\car_I(H_{\omega,L}))\right)\geq2)&\leq
    \esp\left(\tr(\car_I(H_{\omega,L}))[\tr(\car_I(H_{\omega,L}))-1])\right)
    \\&\leq C|I|^2(L+1)^2.
  \end{split}
\end{equation*}
Here, $\car_I(H)$ denotes the spectral projector for the self-adjoint
operator $H$ onto the energy interval $I$.\\
By a simple covering argument, this entails the following estimate
\begin{equation*}
  \pro\left(\exists i\not=j\text{ s.t. }|\lambda_i-\lambda_j|\leq
    L^{-q}\right)\leq C L^{-q+2}.
\end{equation*}
Thus, for $q>3$, a Borel-Cantelli argument yields, that
\begin{equation}
  \label{eq:105}
  \omega\text{ almost surely,   for }L\text{ sufficiently large,
  }\min_{i\not=j}|\lambda_i-\lambda_j|\geq L^{-q}.  
\end{equation}
\subsection{The proofs of the main results in the random case}
\label{sec:proofs-theorems}
We are now going to prove the results stated in
section~\ref{sec:random-case}.
\subsubsection{The proof of Theorem~\ref{thr:4}}
\label{sec:proof-theor-refthr:4}
As for Theorem~\ref{thr:5}, this result follows from
Theorem~\ref{thr:13}. The point (1) is proved exactly as the point (1)
in Theorem~\ref{thr:5}. Point (2) follows immediately from
Theorem~\ref{thr:13} and~\eqref{eq:106}. This completes the proof of
Theorem~\ref{thr:4}.
\subsubsection{The proof of Theorem~\ref{thr:6}}
\label{sec:proof-theor-refthr:6}
Recall that $\kappa\in(0,1)$. To prove (1) we proceed as follows. The
standard result guaranteeing the existence of the density of states
$N$ (see, e.g.,~\cite{MR88f:60013,MR1102675,MR94h:47068}) imply that,
$\omega$ almost surely, one has
\begin{equation}
  \label{eq:107}
  \frac{\#\{\lambda_j\in I\}}{L+1}\to\int_IdN(E).
\end{equation}
This, in particular, shows that, if $I\subset\overset{\circ}{\Sigma}$
is a compact interval, then, $\omega$ almost surely, for $L$
sufficiently large, $I$ is covered by intervals of the form
$[\lambda_j,\lambda_{j+1}]$ and their number is of size $\asymp L$
(actually this holds for $\lambda_j\in I+[-\varepsilon,\varepsilon]$ if
$\varepsilon>0$ is chosen small enough). Moreover, the
estimate~\eqref{eq:105} guarantees that $d_j\geq L^{-q}$ (for any
$q>3$ fixed) for all $\lambda_j\in I$. Thus,
Theorems~\ref{thr:13},~\ref{thr:14} and~\ref{thr:12} and the
estimate~\eqref{eq:104} guarantee that, $\omega$ almost surely, all
the resonances in the strip $I-i[e^{-L^\kappa},0)$ are described by
Theorem~\ref{thr:12}. Indeed, for such a resonance the imaginary part
must be larger than $-e^{-L^\kappa}$; thus, by Theorem~\ref{thr:13},
for every rectangle $[(\lambda_j+\lambda_{j-1})/2,(\lambda_j
+\lambda_{j+1})/2]-i[e^{-L^\kappa},0)$ containing a resonance, one has
$a_j\lesssim e^{-L^\kappa}L^{2q}$ Thus, $a_j\ll d_j^2$ and one can
apply Theorem~\ref{thr:12} to compute the resonance.\\
Let us count the number of those resonances. Therefore, let
$\ell_L=\tau L^\kappa$ where $\tau$ is to be chosen. By~\eqref{eq:104}
and~\eqref{eq:105}, $\omega$ almost surely, one has $a_j \ll d^2_j$
for all $j$ such that $\lambda_j\in I$ as long as the Dirichlet
eigenvalue $\lambda_j$ is associated to a localization center in
$\llbracket 0, L-\ell_L\rrbracket$ (actually it holds for
$\lambda_j\in I+[-\varepsilon,\varepsilon]$ if $\varepsilon>0$ is chosen
small enough); thus, we can apply Theorems~\ref{thr:12}
and~\ref{thr:14} to each of the $(\lambda_j)_j$ that are associated to
a localization center in $\llbracket 0, L-\ell_L\rrbracket$. By
formula~\eqref{eq:31}, each of these eigenvalues gives rise to a
single simple resonance the imaginary part of which is of size $\asymp
a_j$; it lies above the line $\{$Im$z\geq
e^{-\rho\ell_L}=e^{-L^\kappa}\}$ for $\tau\rho=1$.  Actually, the
estimate~\eqref{eq:105} guarantees that $d_j\geq L^{-q}$ (for any
$q>3$ fixed) and Theorem~\ref{thr:14} shows that these resonances are
the only ones above a line Im$z\geq-L^{-q}$. Moreover, by
Lemma~\ref{le:3}, we know there at most $C\ell_L$ eigenvalues
$\lambda_j$ that do not have their localization center in $\llbracket
0, L-\ell_L\rrbracket$. Thus, we obtain, $\omega$ almost surely,
\begin{equation*}
  \lim_{L\to+\infty}
  \frac1L\#\left\{z\text{ resonance of }H_{\omega,L}\text{
      s.t. Re}\,z\in I,\ \text{Im}\,z\geq -e^{-L^\kappa}\right\}=
  \int_IdN(E).
\end{equation*}
Point (2) is proved in the same way. Pick $\lambda\in(0,1)$. In
addition to what was used above, one uses the continuity of the
density of states $E\mapsto n(E)$ and Lyapunov exponent
$E\mapsto\rho(E)$. Assume $E$ is as in point (2). Then, $\omega$
almost surely, the reasoning done above shows that, for any $\eta>0$,
there exists $\varepsilon_0>0$ such that, for
$\varepsilon\in(0,\varepsilon_0)$ and $\delta\in(0,\delta_0)$, for $L$
sufficiently large one has,
\begin{multline*}
  \#\left\{
    \begin{aligned}
      \lambda_l \text{ e.v of }H^\N_{\omega,L}\text{ in }
      E+\frac{\varepsilon}{2\,n(E)}\left[-1+\eta, 1-\eta\right]\text{
        such }\\\text{ that } - e^{\eta_\bullet\rho(E)\delta L}
      \lesssim e^{2\eta_\bullet\rho(E)\lambda\,L} a_l\lesssim
      -e^{-\eta_\bullet\rho(E)\delta L}
    \end{aligned}
  \right\}\\
  \leq \#\left\{z\text{ resonance of }H^\bullet_{\omega,L}\text{ in }
    R^\bullet(E,\lambda,L,\varepsilon,\delta) \right\}\\\leq \#\left\{
    \begin{aligned}
      \lambda_l \text{ e.v of }H^\N_{\omega,L}\text{ in }
      E+\frac{\varepsilon}{2\,n(E)}\left[-1-\eta, 1+\eta\right]\text{
        such }\\\text{ that } -e^{\eta_\bullet\rho(E)\delta L}
      \lesssim e^{2\eta_\bullet\rho(E)\lambda\,L} a_l\lesssim
      -e^{-\eta_\bullet\rho(E)\delta L}
    \end{aligned}
  \right\}
\end{multline*}
Using Theorem~\ref{thr:10} and the continuity of the Lyapunov exponent
in conjunction with the definition of $a_j$ (see~\eqref{eq:1}
and~\eqref{eq:145}), we obtain that, $\omega$ almost surely, for any
$\eta>0$, there exists $\varepsilon_0>0$ such that, for
$\varepsilon\in(0,\varepsilon_0)$ and $\delta\in(0,\delta_0)$, for $L$
sufficiently large one has,
\begin{multline*}
  \#\left\{
    \begin{aligned}
      \text{e.v of }H^\N_{\omega,L}\text{ in }
      E+\frac{\varepsilon}{2\,n(E)}\left[-1+\eta, 1-\eta\right]
      \\\text{ with localization center in } I^\bullet(L,\delta,-\eta)
    \end{aligned}
  \right\}\\
  \leq \#\left\{z\text{ resonance of }H^\bullet_{\omega,L}\text{ in }
    R^\bullet(E,\lambda,L,\varepsilon,\delta) \right\}\\\leq \#\left\{
    \begin{aligned}
      \text{e.v of }H^\N_{\omega,L}\text{ in }
      E+\frac{\varepsilon}{2\,n(E)}\left[-1-\eta,
        1+\eta\right]\\\text{ with localization center in }
      I^\bullet(L,\delta,\eta)
    \end{aligned}
  \right\}
\end{multline*}
where $I^\N(L,\lambda,\delta,\eta)$ is the interval (here $[r]$
denotes the integer part of $r\in\R$)
\begin{gather*}
  I^\N(L,\lambda,\delta,\eta)= [L\lambda]+\llbracket -L\delta(1+\eta),
  L\delta(1+\eta)\rrbracket\\\intertext{and,
    $I^\Z(L,\lambda,\delta,\eta)$ is the union of intervals}
  \begin{aligned}
    I^\Z(L,\lambda,\delta,\eta)&=
    \left(\left[\frac{L\lambda}2\right]+\llbracket
      -L\delta(1+\eta),L\delta(1+\eta))\rrbracket\right)\\
    &\hskip1.5cm\cup \left(
      \left[L\left(1-\frac{\lambda}2\right)\right]+
      \llbracket-L\delta(1+\eta)),L\delta(1+\eta))\rrbracket\right).
  \end{aligned}
\end{gather*}
Now, using the exponential localization of the eigenfunctions, one has
that, $\omega$ almost surely, for any $\eta>0$, there exists
$\varepsilon_0>0$ such that, for $\varepsilon\in(0,\varepsilon_0)$ and
$\delta\in(0,\delta_0)$, for $L$ sufficiently large, one has
\begin{multline}
  \label{eq:191}
  \#\left\{\text{e.v of
    }H^\N_{\omega,L,\lambda,\delta,-2\eta,\bullet}\text{ in }
    E+\frac{\varepsilon}{2\,n(E)}\left[-1+2\eta, 1-2\eta\right]
  \right\}\\
  \leq \#\left\{z\text{ resonance of }H^\bullet_{\omega,L}\text{ in }
    R^\bullet(E,\lambda,L,\varepsilon,\delta) \right\}\\\leq \#\left\{
    \text{e.v of }H^\N_{\omega,L,\lambda,\delta,2\eta,\bullet}\text{
      in } E+\frac{\varepsilon}{2\,n(E)}\left[-1-2\eta, 1+2\eta\right]
  \right\}
\end{multline}
where $\D H^\N_{\omega,L,\lambda,\delta,\eta,\bullet}=
\left(H^\N_{\omega,L}\right)_{|I^\bullet(L,\lambda,\delta,\eta)}$ with
Dirichlet boundary conditions at the edges of the interval
$I^\bullet(L,\lambda,\delta,\eta)$.\\
This immediately yields point (2) for $\lambda\in(0,1)$
using~\eqref{eq:107} for the operators
$H^\N_{\omega,L,\lambda,\delta,\eta,\bullet}$. The
case $\lambda=1$ is dealt with in the same way. \\
As already said, point (3) is an ``integrated'' version of point (2).
Using the same ideas as above, partitioning $I=\cup_{p=0}^P I_p$
s.t. $|I_p|\sim\varepsilon$ centered in $E_p$, one proves
\begin{multline*}
  \sum_{p=0}^P\#\left\{\text{e.v of }H^-_{\omega,p,L,\bullet}\text{ in
    } E_p+\frac{\varepsilon}{2\,n(E_p)}\left[-1+2\eta, 1-2\eta\right]
  \right\}\\
  \leq \#\left\{z\text{ resonance of }H^\bullet_{\omega,L}\text{ in }
    I+\left[-e^{-L^\kappa},-e^{-cL} \right]\right\}\\\leq \sum_{p=0}^P
  \#\left\{ \text{e.v of }H^+_{\omega,p,L,\bullet}\text{ in }
    E_p+\frac{\varepsilon}{2\,n(E_p)}\left[-1-2\eta, 1+2\eta\right]
  \right\}
\end{multline*}
where
\begin{itemize}
\item $H^-_{\omega,p,L,\bullet}$ is the operator $H^\N_\omega$
  restricted to
  \begin{itemize}
  \item $\llbracket 2L^\kappa,
    (\inf(c\rho^{-1}(E_p),1)-\eta)L\rrbracket$ if $\bullet=\N$,
  \item to $\llbracket 2L^\kappa,(\inf(c\rho^{-1}(E_p),1)/2-\eta)L
    \rrbracket\cup \llbracket (1-\inf(c\rho^{-1}(E_p),1)/2+\eta)L,
    L-2L^\kappa \rrbracket$ if $\bullet=\Z$;
  \end{itemize}
\item $H^+_{\omega,p,L,\bullet}$ is the operator $H^\N_\omega$
  restricted to
  \begin{itemize}
  \item $\llbracket L^\kappa/2,
    (\inf(c\rho^{-1}(E_p),1)+\eta)L\rrbracket$ if $\bullet=\N$,
  \item to $\llbracket L^\kappa/2,(\inf(c\rho^{-1}(E_p),1)/2+\eta)L
    \rrbracket\cup \llbracket (1-\inf(c\rho^{-1}(E_p),1)/2-\eta)L,
    L-L^\kappa/2 \rrbracket$ if $\bullet=\Z$;
  \end{itemize}
\end{itemize}
In the computation above, we used the continuity of both, the density
of states $E\mapsto n(E)$ and Lyapunov exponent
$E\mapsto\rho(E)$. Thus, we obtain
\begin{multline*}
  \#\left\{z\text{ resonance of }H^\bullet_{\omega,L}\text{ in }
    I+\left(-\infty,e^{-cL} \right]\right\}\\=
  L\left(\sum_{p=0}^P\inf(c\rho^{-1}(E_p),1)n(E_p)|I_p|+
    o(1)\right)\\+\#\left\{z\text{ resonance of
    }H^\bullet_{\omega,L}\text{ in }
    I+\left(-\infty,e^{-L^\kappa}\right]\right\}.
\end{multline*}
The last term being controlled by Theorem~\ref{thr:27}, one obtains
point (3) as the Riemann sum in the right hand side above converges to
the integral in the right hand side of~\eqref{eq:192} as
$\varepsilon\to0$. This completes the proof of
Theorem~\ref{thr:6}.\qed
\subsubsection{The proof of Theorem~\ref{thr:7}}
\label{sec:proof-theor-refthr:7}
The proof of Theorem~\ref{thr:7} relies on~\cite[Theorem
1.13]{Ge-Kl:10} which describes the local distribution of the
eigenvalues and localization centers $(E_{j,\omega},x_{j,\omega})$:
namely, one has
\begin{equation}
  \label{eq:108}
  \lim_{L\to+\infty}
  \pro\left(\left\{\omega;\
      \begin{aligned}
        &\#\left\{n;
          \begin{aligned}
            E_{j,\omega}&\in E+L^{-1}I_1 \\ x_{j,\omega}&\in L\,
            C_1 \end{aligned}
        \right\}=k_1\\&\hskip1cm\vdots\hskip2cm\vdots\\
        &\#\left\{n;
          \begin{aligned}
            E_{j,\omega}&\in E+L^{-1} I_p \\ x_{j,\omega}&\in L\, C_p
          \end{aligned}
        \right\}=k_p
      \end{aligned}
    \right\}\right)=
  \prod_{n=1}^pe^{-\tilde\mu_n}
  \frac{(\tilde\mu_n)^{k_n}}{k_n!}
\end{equation}
where $\tilde\mu_n:=n(E)|I_n||C_n|$ for $1\leq n\leq p$.\\
Recall that $(z_j^L(\omega))_j$ are the resonances of
$H_{\omega,L}$. By the argument used in the proof of
Theorem~\ref{thr:6}, we know that, $\omega$ almost surely, all the
resonances in $\D K_L:=[E-\varepsilon,E+\varepsilon]+
i\left[-e^{-L^\kappa},0\right]$ are constructed from the
$(\lambda^\bullet_j,a^\bullet_j)$ by formula~\eqref{eq:31}.  Thus, up
to renumbering, the rescaled real and imaginary parts
(see~\eqref{eq:10}) become
\begin{equation*}
  \begin{aligned}
    x_j&=(\text{Re}\,z_{l,L}^\bullet(\omega)-E)L=(\lambda_j-E)L+O(L
    a_j)=(E_{j,\omega}-E)L+O(Le^{-L^\kappa}) \\
    y_j&=-\frac1{2L}\log|\text{Im}\,z_{l,L}^\bullet(\omega)|
    =-\frac{\log a^\bullet_j}{2L}+O(1/L)=\rho(E)
    \frac{d^\bullet_{j,\omega}}{L}+o(1).
  \end{aligned}
\end{equation*}
where $\lambda_j=E_{j,\omega}$ and $x_{j,\omega}$ is the associated
localization center; here we used the continuity of
$E\mapsto\rho(E)$. \\
On the other hand, for the resonances below the line in
$\{$Im$\,z=-e^{-L^\kappa}\}$, one has $y_j\lesssim L^{\kappa-1}$. So
all these resonances are ``pushed upwards'' towards the upper
half-plane. Hence, the statement of Theorem~\ref{thr:7} is an
immediate consequence of~\eqref{eq:108}.\qed
\subsubsection{The proof of Theorem~\ref{thr:8}}
\label{sec:proof-theor-refthr:8}
Using the computations of the previous section, as $E\not=E'$,
Theorem~\ref{thr:8} is a direct consequence of~\cite[Theorem
1.2]{MR2775121} (see also~\cite[Theorem 1.11]{Ge-Kl:10}).
\subsubsection{The proof of Theorem~\ref{thr:27}}
\label{sec:proof-theor-refthr:27}
Consider equations~\eqref{eq:1} and~\eqref{eq:135}. By
Theorem~\ref{thr:10} and Lemma~\ref{le:3}, $\omega$ almost surely, for
$L$ large, the number of $(a_j^\bullet)_j$ larger than $e^{-10\ell_L}$
is bounded by $C\ell_L$. Solving~\eqref{eq:1} and~\eqref{eq:135} in
the strip $\{$Re$\,E\in I,\ $Im$\,E<-e^{-\ell_L}\}$, we can write
$S_L(E)=S^-_L(E)+S^+_L(E)$ where
\begin{equation*}
  S^-_L(E):=\sum_{a^\N_j\leq e^{-10\ell_L}}
  \frac{a^\N_j}{\lambda_j-E}\quad\text{and}\quad 
  S^+_L(E):= \sum_{a^\N_j>e^{-10\ell_L}}\frac{a^\N_j}{\lambda_j-E}
\end{equation*}
and similarly decompose
$\D\Gamma_L(E)=\Gamma^-_L(E)+\Gamma^+_L(E)$. For $L$ large, one then
has
\begin{equation}
  \label{eq:199}
  \sup_{\text{Im}\,E<-e^{-\ell_L}}\|S^-_L(E)\|+\|\Gamma^-_L(E)\|\leq
  e^{-8\ell_L}.
\end{equation}
\begin{wrapfigure}{r}{.35\textwidth}
  \centering
  \includegraphics[width=.30\textwidth]{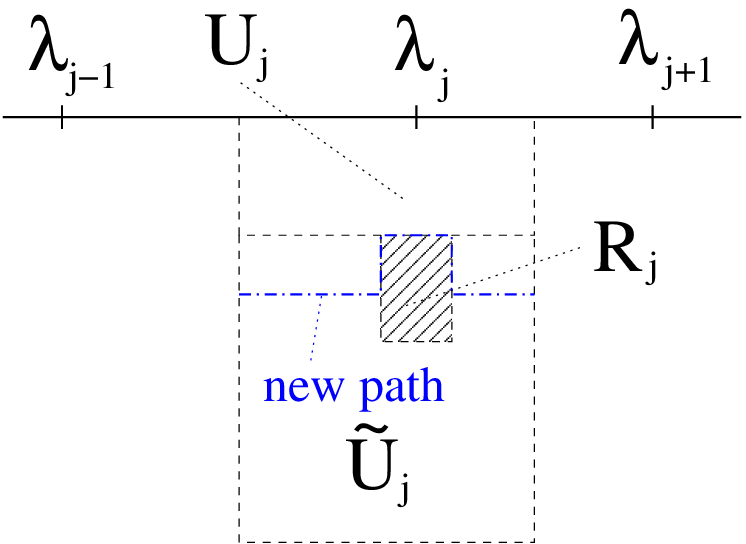}
  \caption{The new path}
  \label{fig:7}
\end{wrapfigure}
\noindent The count of the number of resonances given by the proof of
Theorems~\ref{thr:11} and~\ref{thr:24} then shows that the
equations~\eqref{eq:1} and~\eqref{eq:135} where $S_L$ and $\Gamma_L$
are respectively replaced by $S^+_L$ and $\Gamma^+_L$ have at most
$C\ell_L$ solutions in the lower half plane. The equations where $S_L$
and $\Gamma_L$ are replaced by $S^+_L$ and $\Gamma^+_L$ we will call
the $+$-equations.  The analogue of
Theorems~\ref{thr:13},~\ref{thr:14} and~\ref{thr:12} for the
$+$-equations and Theorem~\ref{thr:10} show that the only solutions to
the $+$-equations in the strip $\{$Re$\,E\in I,\
-e^{-4\ell_L/5}<$Im$\,E<-e^{-3\ell_L/4}\}$ are given by
formulas~\eqref{eq:31} and~\eqref{eq:149} for the eigenvalues of the
Dirichlet problem associated to a localization center in $\llbracket
L-2\ell_L,L-\ell_L/2\rrbracket$ if $\bullet=\N$ and in $\llbracket
\ell_L/2,2\ell_L\rrbracket\cup\llbracket
L-2\ell_L,L-\ell_L/2\rrbracket$ if $\bullet=\Z$. Thus, these zeros are
simple and separated by a distance at least $L^{-4}$ from each other
(recall~\eqref{eq:105}). Moreover, we can cover the interval $I$ by
intervals of the type $[(\lambda_j+\lambda_{j-1})/2,
(\lambda_j+\lambda_{j+1})/2]$, that is, one can write
\begin{equation}
  \label{eq:198}
  I\subset\bigcup_{j=j^-}^{j^+}
  \left[\frac{\lambda_j+\lambda_{j-1}}2,
    \frac{\lambda_j+\lambda_{j+1}}2\right]
\end{equation}
where $\lambda_{j^--1}\not\in I$, $\lambda_{1+j^+}\not\in I$,
$\lambda_{j^-}\in I$ and $\lambda_{j^+}\in I$.
Consider now the line $\{$Im$\,E=-e^{-\ell_L}\}$ and its intersection
with the vertical strip $[(\lambda_j+\lambda_{j-1})/2,
(\lambda_j+\lambda_{j+1})/2]-i\R^+$. Three things may occur:
\begin{enumerate}
\item either $e^{-\ell_L}<a_jdj^2|\sin\theta(\lambda_j)|/C$ (the
  constant $C$ is defined in Theorem~\ref{thr:13}), then, on the
  interval
  $[(\lambda_j+\lambda_{j-1})/2,(\lambda_j+\lambda_{j+1})/2]-ie^{-\ell_L}$,
  one has
  \begin{equation}
    \label{eq:195}
    \left|S^+_L(E)+e^{-i\theta(E)}\right|\gtrsim 1\quad\text{ and }\quad
    \D\left|\text{det}\left(\Gamma^+_L(E)+e^{-i\theta(E)}\right)
    \right|\gtrsim 1 ; 
  \end{equation}
  this follows from the proof of Theorem~\ref{thr:13} (see in
  particular~\eqref{eq:117}, \eqref{eq:146},~\eqref{eq:147}
  and~\eqref{eq:194}) for some fixed $c>0$; recall that, on the
  interval $I+ie^{-\ell_L}$, one has $|\sin\theta(E)|\gtrsim 1$;
\item either $e^{-\ell_L}>C a_j$ (the constant $C$ is defined in
  Theorem~\ref{thr:14}), then, on the interval
  $[(\lambda_j+\lambda_{j-1})/2,(\lambda_j+\lambda_{j+1})/2]-ie^{-\ell_L}$,
  one has again~\eqref{eq:195} for a possibly different constant; this
  follows from the proof of Theorem~\ref{thr:14} (see in
  particular~\eqref{eq:196} and~\eqref{eq:197});
\item if we are neither in case (1) nor in case (2), then the line
  $\{$Im$\,E=-e^{-\ell_L}\}$ may cross $R_j$ (defined in
  Theorem~\ref{thr:12}; see also Fig.~\ref{fig:6}); we change the
  contour $\{$Im$\,E=-e^{-\ell_L}\}$ so as to enter $\tilde U_j$ until
  we reach the boundary of $R_j$ and then follow this boundary getting
  closer to the real axis, turning around $R_j$ and finally reaching
  the line $\{$Im$\,E=-e^{-\ell_L}\}$ again on the other side of $R_j$
  and following it up to the boundary of $\tilde U_j$ (see
  Figure~\ref{fig:7}); on this new line, the bound~\eqref{eq:195}
  again holds; moreover, this new line is closer to the real axis than
  the line $\{$Im$\,E=-e^{-\ell_L}\}$.
\end{enumerate}
\vskip.2cm\noindent Let us call $\mathcal{C}_\ell$ the path obtained
by gluing together the paths constructed in points (1)-(3) for
$j^-\leq j\leq j^+$ and the half-lines $\frac{\lambda_{j^-}+\lambda_{j^--1}}2
-i[e^{-\ell_L},+\infty)$ and $\frac{\lambda_{j^+}+\lambda_{j^++1}}2
-i[e^{-\ell_L},+\infty)$ (see~\eqref{eq:198}).  One can then apply
Rouch{\'e}'s Theorem to compare the $+$ equations to the
equations~\eqref{eq:1} and~\eqref{eq:135}: by~\eqref{eq:199}
and~\eqref{eq:195}, on the line $\mathcal{C}_\ell$, one has
$\D\left|S_L^-\right|<\left|S_L^++e^{-i\theta}\right|$ and
\begin{equation*}
  \left|\text{det}\left(\Gamma_L(E)+e^{-i\theta(E)}\right)
    \text{det}\left(\Gamma^+_L(E)+e^{-i\theta(E)}\right)
  \right|\leq\frac12\left| \text{det}\left(\Gamma_L(E)+e^{-i\theta(E)}
    \right)\right|.    
\end{equation*}
Thus, the number of solutions to equations~\eqref{eq:1}
and~\eqref{eq:135} below the line $\mathcal{C}_\ell$ is bounded by
$C\ell_L$; as $\mathcal{C}_\ell$ lies above
$\{$Im$\,E=-e^{-\ell_L}\}$, in the half-plane
$\{$Im$\,E<-e^{-\ell_L}\}$, the equations~\eqref{eq:1}
and~\eqref{eq:135} have at most $C\ell_L$ solutions.  We have proved
Theorem~\ref{thr:27}.\qed
\subsubsection{The proof of Theorem~\ref{thr:9}}
\label{sec:proof-theor-refthr:9}
The first point in Theorem~\ref{thr:9} is proved in the same way as
point (2) in Theorem~\ref{thr:6} up to the change of scales, $L$ being
replaced by $\ell_L$. Pick scales $(\ell'_L)_L$
satisfying~\eqref{eq:14} such that $\ell'_L\ll\ell_L$. One has
\begin{Le}
  \label{le:4}
  Fix two sequences $(a_L)_L$ and $(b_L)_L$ such that $a_L<b_L$. With
  probability one, for $L$ sufficiently large,
  \begin{multline*}
    \#\left\{\text{e.v. of }H_{\omega, \ell_L-2\ell'_L/\rho}\text{ in
      }\left[a_L+e^{-\ell'_L},
        b_L-e^{-\ell'_L}\right]\right\}\\\leq\#\left\{ \text{e.v. of
      }H_{\omega,L}\text{ in }[a_L,b_L] \text{ with loc. cent. in
      }\llbracket0,\ell_L\rrbracket \right\}\\\leq
    \#\left\{\text{e.v. of }H_{\omega, \ell_L+2\ell'_L/\rho}\text{ in
      }\left[a_L-e^{-\ell'_L}, b_L+e^{-\ell'_L}\right]\right\}
  \end{multline*}
  where $\rho$ is given by Lemma~\ref{le:2}.
\end{Le}
\begin{proof}
  To prove Lemma~\ref{le:4}, we apply Lemma~\ref{le:2} to
  $L=\ell_L+\ell'_L$ (i.e. for the operator $H_\omega$ restricted to
  the interval $\llbracket0,\ell_L+\ell'_L\rrbracket$) and
  $l_L=\ell'_L$. The probability of the bad set is the
  $O\left(L^{-\infty}\right)$, thus, summable in $L$. Using the
  localization estimate~\eqref{eq:75}, one proves that
  \begin{itemize}
  \item each eigenvalue of $H_{\omega, \ell_L-2\ell'_L/\rho}$ is at a
    distance of at most $e^{-\ell'_L}$ of an eigenvalue of
    $H_{\omega,L}$ with loc. cent. in $\llbracket0,\ell_L\rrbracket$;
  \item each eigenvalue of $H_{\omega,L}$ with loc. cent. in
    $\llbracket0,\ell_L\rrbracket$ is at a distance of at most
    $e^{-\ell'_L}$ of an eigenvalue of $H_{\omega,
      \ell_L+2\ell'_L/\rho}$.
  \end{itemize}
  Lemma~\ref{le:4} follows.
\end{proof}
\noindent The first point in Theorem~\ref{thr:9} is then point (2) of
Theorem~\ref{thr:6} for the operator $H_{\omega,
  \ell_L-2\ell'_L/\rho}$ and $H_{\omega, \ell_L+2\ell'_L/\rho}$ and
the fact that $\ell'_L\ll\ell_L$.\vskip.1cm\noindent
The proof of the second statement in Theorem~\ref{thr:9} is very
similar to that of Theorem~\ref{thr:7}.  Fix $I$ a compact interval in
$\overset{\circ}{\Sigma}$. As $\ell_L$ satisfies~\eqref{eq:14}, one
can find $\ell'_L<\ell''_L$ also satisfying~\eqref{eq:14} such that
$\D e^{-\ell''_L}\ll e^{-\ell_L}\ll e^{-\ell'_L}$. For the same
reasons as in the proof of Theorem~\ref{thr:7}, after rescaling, all
the resonances in $\D I-i(-\infty,0)$ outside the strip $\D
I-i\left[e^{-\ell'_L}, e^{-\ell''_L}\right)$ are then pushed to either
$0$ or $i\infty$ as $L\to+\infty$.\\
On the other hand, the resonances in the strip $\D
I-i\left[e^{-\ell'_L},e^{-\ell''_L}\right)$ are described
by~\eqref{eq:31}. The rescaled real and imaginary parts of the
resonances (see~\eqref{eq:12}) now become
$x_j=(E_{j,\omega}-E)\ell_L+o(1)$ and
$y_j=\rho(E)\frac{d_{j,\omega}}{\ell_L}+o(1)$.\\
Now, to compute the limit of $\pro(\#\{j;\ x_j\in I,\ y_j\in J\}=k)$, using
the exponential decay property~\eqref{eq:75}, it suffices to
use~\cite[Theorem 1.14]{Ge-Kl:10}. Let us note here
that~\cite[condition (1.50)]{Ge-Kl:10} on the scales $(\ell_L)_L$ is
slightly stronger than~\eqref{eq:14}. That condition~\eqref{eq:14}
suffices is a consequence of the stronger localization property known
in the present case (compare Theorem~\ref{thr:10} to~\cite[Assumption
(Loc)]{Ge-Kl:10}). This completes the proof of the second point in
Theorem~\ref{thr:9}. The final statement in~\ref{thr:9} is proved in
exactly the same way as Theorem~\ref{thr:8}.\\
The proof of Theorem~\ref{thr:9} is complete.\qed
\subsubsection{The proofs of Proposition~\ref{pro:6} and
  Theorem~\ref{thr:21}}
\label{sec:proof-theor-refthr:21}
Localization for the operator $H^\N_\omega$ can be described by the
following
\begin{Le}
  \label{le:10}
  There exists $\rho>0$ and $q>0$ such that, $\omega$ almost surely,
  there exists $C_\omega>0$ s.t. for $L$ sufficiently large, if
  \begin{enumerate}
  \item $\varphi_{j,\omega}$ is a normalized eigenvector of
    $H_{\omega,L}$ associated to $E_{j,\omega}$ in $\Sigma$,
  \item $x_j(\omega)\in\N$ is a maximum of
    $x\mapsto|\varphi_{j,\omega}(x)|$ in $\N$,
  \end{enumerate}
  then, for $x\in\N$, one has
  \begin{equation}
    \label{eq:190}
    |\varphi_{j,\omega}(x)|\leq C_\omega (1+|x_j(\omega)|^2)^{q/2}
    e^{-\rho|x-x_j(\omega)|}.
  \end{equation}
  Moreover, the mapping $\omega\mapsto C_\omega$ is measurable and
  $\esp(C_\omega)<+\infty$.
\end{Le}
\noindent This result for our model is a consequence of
Theorem~\ref{thr:18} (see, e.g.,~\cite{MR582611,MR883643,MR1102675})
and~\cite[Theorem 6.1]{Ge-Kl:10}.\\
We thus obtain the following representation for the function
$\Xi_\omega$
\begin{equation}
  \label{eq:110}
  \Xi_\omega(E)=\sum_{j}\frac{|\varphi_{j,\omega}(0)|^2}
  {E_{j,\omega}-E}+e^{-i\arccos(E/2)}
\end{equation}
As $H_\omega^\N$ satisfies a Dirichlet boundary condition at $-1$, one
has
\begin{equation}
  \label{eq:189}
  \forall j,\quad |\varphi_{j,\omega}(0)|>0\quad\text{and}\quad
  \sum_j|\varphi_{j,\omega}(0)|^2=1.
\end{equation}
As $E\to -i\infty$, the representation~\eqref{eq:110} yields
\begin{equation*}
  \begin{split}
    \Xi_\omega(E)&=-E^{-2}\sum_{j}|\varphi_{j,\omega}(0)|^2E_{j,\omega}
    +O\left(E^{-3}\right)=-E^{-2}\langle\delta_0,H_\omega^\N\delta_0\rangle
    +O\left(E^{-3}\right)\\&=-\omega_0\,E^{-2}+O\left(E^{-3}\right).
  \end{split}
\end{equation*}
This proves the first point in Proposition~\ref{pro:6}.\\
As a direct consequence of Theorem~\ref{thr:18} and the computation
leading to Theorem~\ref{thr:19} (see
section~\ref{sec:proof-theorem-7}), we obtain that there exists
$\tilde c>0$ s.t. for $L$ sufficiently large, with probability at
least $1-e^{-\tilde c L}$, one has
\begin{equation}
  \label{eq:143}
  \sup_{\text{Im}\,E\leq -e^{-\tilde cL}}
  \left|\int_\R\frac{dN_\omega(\lambda)}{\lambda-E}
    -\langle\delta_0,(H_{\omega,L}-E)^{-1}\delta_0\rangle
  \right|\leq e^{-\tilde c L}.
\end{equation}
Taking
\begin{equation}
  \label{eq:109}
  L=L_\varepsilon\sim c^{-1}|\log\varepsilon| 
\end{equation}
for some sufficiently small $c>0$, this and Rouch{\'e}'s Theorem implies
that, with probability $1-\varepsilon^3$, the number of zeros of
$\Xi_\omega$ (counted with multiplicity) in $I+i(-\infty,\varepsilon]$
is bounded
\begin{itemize}
\item from above by the number of resonances of
  $H_{\omega,L_\varepsilon}$ in
  $I_\varepsilon^++i(-\infty,-\varepsilon-\varepsilon^2]$;
\item from below by the number of resonances of
  $H_{\omega,L_\varepsilon}$ in
  $I_\varepsilon^-+i(-\infty,-\varepsilon+\varepsilon^2]$.
\end{itemize}
where $I_\varepsilon^+=[a-\varepsilon,b+\varepsilon]$ and
$I_\varepsilon^+=[a+\varepsilon,b-\varepsilon]$ if $I=[a,b]$.\\
Here, to apply Rouch{\'e}'s Theorem, we apply the same strategy as in the
proof of Theorem~\ref{thr:27} and construct a path bounding a region
larger (resp. smaller) than $I_\varepsilon^++i(-\infty,
-\varepsilon-\varepsilon^2]$ (resp. $I_\varepsilon^-+i(-\infty,
-\varepsilon+\varepsilon^2]$) on which one can guarantee
$\D\left|S_L(E)+e^{-i\theta(E)}\right|\gtrsim 1$. \\
Now, we choose the constant $c$ (see~\eqref{eq:109}) to be so small
that $\D c<\min_{E\in I}\rho(E)$. Applying point (3) of
Theorem~\ref{thr:6} to $H_{\omega,L_\varepsilon}$ with this constant
$c$, we obtain that the number of resonances of
$H_{\omega,L_\varepsilon}$ in
$I_\varepsilon^++i(-\infty,\varepsilon-\varepsilon^2]$
(resp. $I_\varepsilon^-+i(-\infty,\varepsilon+\varepsilon^2]$) is
upper bounded (resp lower bounded) by
\begin{equation*}
  \begin{split}
    L_\varepsilon\int_I\min\left(\frac{c}{\rho(E)},1\right)n(E)dE\,(1+O(1))
    &=\frac{|\log\varepsilon|}{c}\int_I\frac{c}{\rho(E)}n(E)dE\,(1+O(1))\\
    =&|\log\varepsilon|\int_I\frac{n(E)}{\rho(E)}dE\,(1+O(1)).
  \end{split}
\end{equation*}
Hence, we obtain the second point of Proposition~\ref{pro:6}. The last
point of this proposition is then an immediate consequence of the
arguments developed to obtain the second point if one takes into
account the following facts:
\begin{itemize}
\item the minimal distance between the Dirichlet eigenvalues of
  $H^\N_{\omega,L}$ is bounded from below by $L^{-4}$
  (see~\eqref{eq:105}),
\item the growth of the function $E\mapsto S_L(E)+e^{-i\theta(E)}$
  near the resonances (i.e. its zeros) is well controlled by
  Proposition~\ref{pro:7}.
\end{itemize}
Indeed, this implies that the resonances of $H^\N_{\omega,L}$ are
simple in $I+i[-e^{-\sqrt{L}},0)$ (one can choose larger rectangles)
and that near each resonance one can apply Rouch{\'e}'s Theorem to control
the zero of $\Xi_\omega$. Note that this also yields $\omega$-almost
surely, there exists $c_\omega$ such that
\begin{equation}
  \label{eq:203}
  \min_{\substack{z\text{ zero of }\Xi_\omega\\z\in
      I+i(-\varepsilon_\omega,0)}}
  \inf_{0<r<\varepsilon_\omega(\text{Im}\,z)^{3/2}}\min_{|E-z|=r}
  \frac{|\Xi_\omega(E)|}r\gtrsim 1.
\end{equation}
This completes the proof of
Proposition~\ref{pro:6}.\qed\vskip.2cm\noindent
Theorem~\ref{thr:21} is a consequence of the following
\begin{Th}
  \label{thr:28}
  There exists $\tilde c>0$ such that, $\omega$ almost surely, for
  $L\geq1$ sufficiently large one has
  \begin{equation*}
    \sup_{\substack{\text{Re}\,E\in I\\\text{Im}\,E<-e ^{-\tilde c L}}}
    \left|\Gamma_{L,\omega,\tilde\omega}(E)-
      \begin{pmatrix}\D \int_\R\frac{dN_{\tilde\omega}(\lambda)}{\lambda-E}&0\\
        0&\D \int_\R\frac{dN_\omega(\lambda)}{\lambda-E}
      \end{pmatrix}\right|+\left|S_{L,\omega}(E)-\int_\R
      \frac{dN_\omega(\lambda)}{\lambda-E}\right|\leq e^{-\tilde c L}
  \end{equation*}
  where $\Gamma_{L,\omega,\tilde\omega}(E)$ (resp.  $S_{L,\omega}(E)$)
  is the matrix $\Gamma_L(E)$ (resp. the function $S_L(E)$)
  (see~\eqref{eq:142}) constructed from the Dirichlet data on
  $\llbracket0,L\rrbracket$ of $-\Delta+V^\Z_{\omega,\tilde\omega,L}$
  (resp. $-\Delta+V^\N_{\omega,L}$) (see~\eqref{eq:200}) using
  formula~\eqref{eq:142} (resp.~\eqref{eq:1}).
\end{Th}
\noindent Theorem~\ref{thr:28} is proved exactly as
Theorem~\ref{thr:19} except that one uses the localization
estimates~\eqref{FVdynloc1} instead of the Combes-Thomas estimates.\\
Theorem~\ref{thr:21} is then an immediate consequence of the
estimate~\eqref{eq:143}. Indeed, this implies that if $z$ is a
resonance for e.g. $H^N_{\omega,L}$ in $\D I+i\left(-\infty,e^{\tilde
    cL}\right]$, then $\D|\Xi_\omega(z)|\leq e^{-\tilde cL}$. By the last
point of Proposition~\ref{pro:6}, $\omega$ almost surely, we know that
the multiplicity of the zeros of $\Xi_\omega$ is bounded by
$N_\omega$. Moreover, for the zeros of $\Xi_\omega$ in
$I+i(-\varepsilon_\omega,0)$, we know the bound~\eqref{eq:203}. This
bound and~\eqref{eq:143} imply that
\begin{equation*}
  \max_{\substack{z\text{ zero of }\Xi_\omega\\z\in
      I+i\left(-\varepsilon_\omega,e^{-\tilde c L}\right)}}
  \max_{|E-z|=e^{-\tilde c L}}
  \frac{\left|\Xi_\omega(E)-\left(S_{\omega,L}(E)+e^{-i\theta(E)}\right)\right|}
  {\left|\Xi_\omega(E)\right|}<e^{-\tilde c L}.
\end{equation*}
This yields point (2) in Theorem~\ref{thr:21} by an application of
Rouch{\'e}'s Theorem. Point (1) is obtained in the same way using
Proposition~\ref{pro:7} that gives
\begin{equation*}
  \max_{\substack{z\text{ resonance of }H^\N_{\omega,L}\\z\in
      I+i\left(-\varepsilon_\omega,e^{-\tilde c L}\right)}}
  \max_{|E-z|=e^{-\tilde c L}}
  \frac{\left|\Xi_\omega(E)-\left(S_{\omega,L}(E)+e^{-i\theta(E)}\right)\right|}
  {\left|S_{\omega,L}(E)+e^{-i\theta(E)}\right|}<e^{-\tilde c L}.
\end{equation*}
The case of $H^\Z_{\omega,\tilde\omega,L}$ is dealt with in the same
way.\\
This completes the proof of Theorem~\ref{thr:21}.\qed
\subsection{Estimates on the growth of eigenfunctions}
\label{sec:estim-growth-eigenf}
In the present section we are going to prove Theorems~\ref{thr:10}
and~\ref{thr:20}. At the end of the section, we also prove the simpler
Lemma~\ref{le:3}.\\
The proof of Theorem~\ref{thr:10} relies on locally uniform estimates
on the rate of growth of the cocycle~\eqref{eq:9} attached to the
Schr{\"o}dinger operator that we present now. Define
\begin{equation}
  \label{eq:103}
  T_L(E,\omega)=T(E,\omega_L)\cdots
  T(E,\omega_0)
\end{equation}
where
\begin{equation*}
  T(E,\omega_j)=
  \begin{pmatrix}E-\omega_j&-1\\1&0\end{pmatrix}
\end{equation*}
We start with an upper bound on the large deviations of the growth
rate of the cocycle that is uniform in energy. Fix $\alpha>1$ and
$\delta\in(0,1)$. For one part, the proof of Theorem~\ref{thr:10}
relies on the following
\begin{Le}
  \label{le:6}
  Let $I\subset\R$ be a compact interval. For any $\delta>0$, there
  exists $L_\delta>0$ and $\eta>0$ such that, for $L\geq L_\delta$ and
  any $K>0$, one has
  \begin{equation}
    \label{eq:81}
    \pro\left(
      \begin{aligned}
        &\forall 0\leq k\leq K,\quad \forall E\in I,\quad  \forall \|u\|=1,\\\
        &\frac{\log\|T_L(E;\tau^k(\omega))u\|}{L+1} \leq
        \rho(E)+\delta
      \end{aligned}
    \right)\geq 1-K e^{-\eta(L+1)}
  \end{equation}
  where we recall that $\tau:\Omega\to\Omega$ denotes the left shift
  (i.e. if $\omega=(\omega_n)_{n\geq0}$ then $\D[\tau(\omega)]_n=
  \omega_{n+1}$ for $n\geq0$) and $\tau^n=\tau\circ\cdots\circ\tau$
  $n$ times.
\end{Le}
\noindent At the heart of this result is a large deviation principle
for the growth rate of the cocycle (see~\cite[section I and Theorem
6.1]{MR88f:60013}). As it also serves in the proof of
Theorem~\ref{thr:10}, we recall it now. One has
\begin{Le}
  \label{le:7}
  Let $I\subset\R$ be a compact interval. For any $\delta>0$, there
  exists $L_\delta>0$ and $\eta>0$ such that, for $L\geq L_\delta$,
  one has
  \begin{equation}
    \label{eq:92}
    \sup_{\substack{E\in I\\\|u\|=1}}\pro\left(
      \left|\frac{\log\|T_L(E;\omega)u\|}{L+1}
        -\rho(E)\right|\geq \delta\right)\leq e^{-\eta(L+1)}.
  \end{equation}
\end{Le}
\noindent While this result is not stated as is in~\cite{MR88f:60013},
it can be obtained from~\cite[Lemma 6.2 and Theorem
6.1]{MR88f:60013}. Indeed, by inspecting the proof of~\cite[Lemma 6.2
and Theorem 6.1]{MR88f:60013}, it is clear that the quantities
involved (in particular, the principal eigenvalue of $T(z;E)=T(z)$
in~\cite[Theorem 4.3]{MR88f:60013}) are continuous functions of the
energy $E$. Thus, taking this into account, the proof of~\cite[Theorem
6.1]{MR88f:60013} yields, for our cocycle, a convergence that is
locally uniform in energy, that is,~\eqref{eq:92}.
\vskip.2cm\noindent To prove Theorem~\ref{thr:10}, in addition to
Lemma~\ref{le:6}, we also need to guarantee a uniform lower bound on
the growth rate of the cocycle. We need this bound at least on the
spectrum of $H_{\omega,L}$ with a good probability. Actually, this is
the best one can hope for: a uniform bound in the style
of~\eqref{eq:81} will not hold.\\
We prove
\begin{Le}
  \label{le:5}
  Fix $I$ a compact interval and $\delta>0$. Pick $u\in\C^2$
  s.t. $\|u\|=1$. For $0\leq j\leq L$, if $j\leq L-1$, define
  \begin{gather*}
    \mathcal{K}^+_j(\omega,L,\delta,u):=\left\{E\in I;\
      \left|\frac{\log\left\|T^{-1}_{L-(j+1)}(E,\tau^{j+1}(\omega)) u
          \right\|}{L-j}- \rho(E)\right|>\delta\right\} \intertext{
      and, if $1\leq j$, define}
    \mathcal{K}^-_j(\omega,L,\delta,u):=\left\{E\in I;\
      \left|\frac{\log\left\|T_{j-1}(E,\omega)u\right\|}j
        -\rho(E)\right|>\delta\right\};
  \end{gather*}
  finally, define $\mathcal{K}^+_L(\omega,L,\delta,u)=\emptyset
  =\mathcal{K}^-_0(\omega,L,\delta,u)$.\\
  Recall that $(E_{j,\omega})_{0\leq j\leq L}$ are the eigenvalues of
  $H_{\omega,L}$ and let $x_{j,\omega}$ be the associated localization
  centers.\\
  For $0\leq \ell\leq L$, define the sets
  \begin{equation*}
    \Omega^+_B(L,\ell,\delta,u):=\left\{\omega;
      \begin{aligned}
        \exists j\text{ s.t. }  L-x_{j,\omega}\geq\ell\text{ and } \\
        E_{j,\omega}\in
        \mathcal{K}^+_{x_{j,\omega}}(\omega,L,\delta,u)
      \end{aligned}
    \right\} 
  \end{equation*}
  and
  \begin{equation*}
    \Omega^-_B(L,\ell,\delta,u):=\left\{\omega;
      \begin{aligned}
        \exists j\text{ s.t. }  x_{j,\omega}\geq\ell\text{ and } \\
        E_{j,\omega}\in
        \mathcal{K}^-_{x_{j,\omega}}(\omega,L,\delta,u)
      \end{aligned}
    \right\}.
  \end{equation*}
  Then, the sets $\Omega^\pm_B(L,\ell,\delta,u)$ are measurable and,
  for any $\delta>0$, there exists $\eta>0$ and $\ell_0>0$ such that,
  for $L\geq\ell\geq \ell_0$, one has
  \begin{equation}
    \label{eq:100}
    \max\left(\pro(\Omega^+_B(L,\ell,\delta,u)),
      \pro(\Omega^-_B(L,\ell,\delta,u))\right)\leq
    \frac{(L+1)|I| e^{-\eta(\ell-1)}}{1-e^{-\eta}}.
  \end{equation}
  Here, the constant $\eta$ is the one given by~\eqref{eq:92}.
\end{Le}
\noindent First, let us explain the meaning of Lemma~\ref{le:5}. As,
by Lemma~\ref{le:6}, we already control the growth of the cocycle from
above, we see that in the definitions of the set
$\mathcal{K}^-_j(\omega,L,\delta,u)$ resp.
$\mathcal{K}^+_j(\omega,L,\delta,u)$, it would have sufficed to
require
\begin{equation*}
  \frac{\log\left\|T_{j-1}(E,\omega)u\right\|}j-\rho(E)\leq-\delta    
\end{equation*}
resp.
\begin{equation*}
  \frac{\log\left\|T^{-1}_{L-(j+1)}(E,\tau^{j+1}(\omega))u\right\|}{L-(j+1)}
  -\rho(E)\leq-\delta.  
\end{equation*}
Hence, what Lemma~\ref{le:5} measures is that the probability that the
cocycle at energy $E_{n,\omega}$ leading from a localization center
$x_{n,\omega}$ to either $0$ or $L$ decays at a rate smaller than the
rate predicted by the Lyapunov exponent.\\
The sets $\Omega^\pm_B(L,\ell,\delta,u)$ are the sets of bad
configurations, i.e., the events when the rate of decay of the solution
is far from the Lyapunov exponent. Indeed, for $\omega$ outside
$\Omega^\pm_B(L, \ell,\delta)$, i.e., if the reverse of the inequalities
defining $\mathcal{K}^\pm_j(\omega,L,\delta,u)$ hold, when
$j=x_{n,\omega}$ and $E=E_{n,\omega}$, then, we know that the
eigenfunction $\varphi_{n,\omega}$ has to decay from the center of
localization $x_{n,\omega}$ (which is a local maximum of its modulus)
towards the edges of the intervals at a rate larger than
$\gamma(E_{n,\omega})-\delta$. The eigenfunction being normalized, at
the localization center, it is of size at least $L^{-1/2}$. This will
entail the estimates~\eqref{eq:76} and~\eqref{eq:97} with a good
probability.
\vskip.1cm\noindent There is a major difference in the uniformity in
energy obtained in Lemmas~\ref{le:5} and~\ref{le:6}. In
Lemma~\ref{le:5}, we do not get a lower bound on the decay rate that
is uniform all over $I$: it is merely uniform over the spectrum inside
$I$ (which is sufficient for our purpose as we shall see). The reason
for this difference in the uniformity between Lemma~\ref{le:6}
and~\ref{le:5} is the same that makes the Lyapunov exponent
$E\mapsto\rho(E)$ in general only upper semi-continuous and not lower
semi-continuous (in the present situation, it actually is continuous).
\vskip.2cm\noindent We postpone the proofs of Lemmas~\ref{le:6}
and~\ref{le:5} to the end of this section and turn to the proofs of
Theorems~\ref{thr:10} and~\ref{thr:20}.
\subsubsection{The proof of Theorem~\ref{thr:10}}
\label{sec:proof-lemma-refle:4}
By Lemma~\ref{le:6}, as $T_L(E,\omega)\in SL(2,\R)$, with probability at
least $1-KL e^{-\eta(L+1)}$, for $L\geq L_\delta$ and any $K>0$, one
also has
\begin{equation*}
  \forall 0\leq k\leq K,\quad \forall E\in I,\quad  \forall
  \|u\|=1,\quad \frac{\log\|T^{-1}_L(E;\tau^k(\omega))u\|}{L+1}
  \leq \rho(E)+\delta.
\end{equation*}
Now pick $\ell=C\log L$ where $C>0$ is to be chosen later on. We know
that, with probability $\pro$ satisfying
\begin{equation}
  \label{eq:99}
  \pro\geq1-L^2 e^{-\eta\ell},
\end{equation}
for $L\geq L_\delta$ and any $l\in[\ell,L]$ and any $k\in[0,L]$, one
also has
\begin{equation}
  \label{eq:93}
  \forall E\in I,\quad  \forall \|u\|=1,\quad
  \frac{\log\|T^{-1}_l(E;\tau^k(\omega))u\|}{l+1}\leq\rho(E)+\delta.
\end{equation}
Let $\varphi_{j,\omega}$ be a normalized eigenfunction associated to
the eigenvalue $E_{j,\omega}\in I$ with localization center
$x_{j,\omega}$. By the definition of the localization center, one has
\begin{equation}
  \label{eq:89}
  \frac1{L+1}\leq\left\|\begin{pmatrix}
      \varphi_{j,\omega}(x_{j,\omega})\\\varphi_{j,\omega}(x_{j,\omega}-1)
    \end{pmatrix}\right\|^2\leq 1\quad\text{and}\quad
  \frac1{L+1}\leq\left\|\begin{pmatrix}
      \varphi_{j,\omega}(x_{j,\omega}+1)\\\varphi_{j,\omega}(x_{j,\omega})
    \end{pmatrix}\right\|^2\leq 1.
\end{equation}
By the eigenvalue equation, for $x\in\llbracket 0,L\rrbracket$, one
has
\begin{equation}
  \label{eq:91}
  \begin{pmatrix}
    \varphi_{j,\omega}(x)\\\varphi_{j,\omega}(x-1)
  \end{pmatrix}=
  \begin{cases}
    T_{x-x_{j,\omega}}(E;\tau^{x_{j,\omega}}(\omega))
    \begin{pmatrix}
      \varphi_{j,\omega}(x_{j,\omega})\\\varphi_{j,\omega}(x_{j,\omega}-1)
    \end{pmatrix}
    &\text{ if }x\geq x_{j,\omega},\\
    T^{-1}_{x_{j,\omega}-x}(E;\tau^x(\omega))
    \begin{pmatrix}
      \varphi_{j,\omega}(x_{j,\omega})\\\varphi_{j,\omega}(x_{j,\omega}-1)
    \end{pmatrix}
    &\text{ if }x\leq x_{j,\omega}.
  \end{cases}
\end{equation}
Hence, by~\eqref{eq:81} and~\eqref{eq:93}, with probability at least
$1-2L^2 e^{-\eta\ell}-L^{-p}$, if $|x_{j,\omega}-x|\geq\ell$, for
$x_{j,\omega}<x\leq L$, one has
\begin{equation}
  \label{eq:95}
  \begin{split}
    \frac{e^{-(\rho(E_{j,\omega})+\delta)|x-x_{j,\omega}|}}{\sqrt{L+1}}
    &\leq e^{-(\rho(E_{j,\omega})+\delta)|x-x_{j,\omega}|}
    \left\|\begin{pmatrix}\varphi_{j,\omega}(x_{j,\omega})\\
        \varphi_{j,\omega}(x_{j,\omega}-1)
      \end{pmatrix}\right\|\\&\leq
    \left\| T_{x-x_{j,\omega}}(E;\tau^{x_{j,\omega}}(\omega))
      \begin{pmatrix}\varphi_{j,\omega}(x_{j,\omega})\\
        \varphi_{j,\omega}(x_{j,\omega}-1)
      \end{pmatrix}\right\|=
    \left\|\begin{pmatrix}\varphi_{j,\omega}(x)\\\varphi_{j,\omega}(x-1)
      \end{pmatrix}\right\|
  \end{split}
\end{equation}
and, for $0\leq x<x_{j,\omega}$, one has
\begin{equation}
  \label{eq:96}
  \begin{split}
    \left\|\begin{pmatrix}\varphi_{j,\omega}(x)\\\varphi_{j,\omega}(x-1)
      \end{pmatrix}\right\|&=
    \left\|T^{-1}_{x-x_{j,\omega}}(E;\tau^{x_{j,\omega}}(\omega))
      \begin{pmatrix}\varphi_{j,\omega}(x_{j,\omega})\\
        \varphi_{j,\omega}(x_{j,\omega}-1)
      \end{pmatrix}\right\|\\&\geq
    e^{-(\rho(E_{j,\omega})+\delta)|x-x_{j,\omega}|}
    \left\|\begin{pmatrix}\varphi_{j,\omega}(x_{j,\omega})\\
        \varphi_{j,\omega}(x_{j,\omega}-1)
      \end{pmatrix}\right\|\geq
    \frac{e^{-(\rho(E_{j,\omega})+\delta)|x-x_{j,\omega}|}}{\sqrt{L+1}}
  \end{split}
\end{equation}
On the other hand, by the definition of the Dirichlet boundary
conditions, we know that
\begin{equation*}
  \begin{pmatrix}\varphi_{j,\omega}(0)\\\varphi_{j,\omega}(-1)
  \end{pmatrix}=\varphi_{j,\omega}(0)\begin{pmatrix}1\\0
  \end{pmatrix}\quad\text{and}\quad
  \begin{pmatrix}\varphi_{j,\omega}(L+1)\\\varphi_{j,\omega}(L)
  \end{pmatrix}=\varphi_{j,\omega}(L)\begin{pmatrix}0\\1
  \end{pmatrix}.
\end{equation*}
Thus,
\begin{equation*}
  \varphi_{j,\omega}(0)\,  T_{x_{j,\omega}-1}(E;\omega)
  \begin{pmatrix}1\\0
  \end{pmatrix}=\begin{pmatrix}
    \varphi_{j,\omega}(x_{j,\omega})\\\varphi_{j,\omega}(x_{j,\omega}-1)
  \end{pmatrix}
\end{equation*}
and
\begin{equation*}
  \varphi_{j,\omega}(L)\begin{pmatrix}0\\1
  \end{pmatrix}=T_{L-x_{j,\omega}-1}(E;\tau^{x_{j,\omega}+1}(\omega))\begin{pmatrix}
    \varphi_{j,\omega}(x_{j,\omega}+1)\\\varphi_{j,\omega}(x_{j,\omega})
  \end{pmatrix}.
\end{equation*}
Thus, for $\omega\not\in\Omega^+_B(L,\ell,\delta,u_+)\cup\Omega^-_B
(L,\ell,\delta,u_-)$ where we have set $\D u_-:=\begin{pmatrix}0\\1
\end{pmatrix}$ and $\D u_+:=\begin{pmatrix}1\\0
\end{pmatrix}$, we know that
\begin{equation*}
  e^{-(\rho(E_{j,\omega})-\delta)(L-x_{j,\omega})}
  \leq\left\|T^{-1}_{L-x_{j,\omega}-1}(E;\tau^{x_{j,\omega}+1}(\omega))u_+\right\|
\end{equation*}
and
\begin{equation*}
  e^{-(\rho(E_{j,\omega})-\delta)x_{j,\omega}}
  \leq\left\|T_{x_{j,\omega}-1}(E;\omega)u_-\right\|
\end{equation*}
Thus, we obtain that, for
$\omega\not\in\Omega^+_B(L,\ell,\delta,u_+)\cup\Omega^-_B(L,\ell,\delta,u_-)$,
one has
\begin{equation}
  \label{eq:94}
  \begin{split}
    |\varphi_{j,\omega}(L)|&=
    \left\|T^{-1}_{L-x_{j,\omega}}(E;\tau^{x_{j,\omega}+1}(\omega))
      \begin{pmatrix}0\\1 \end{pmatrix}\right\|^{-1}
    \left\| \begin{pmatrix}
        \varphi_{j,\omega}(x_{j,\omega}+1)\\\varphi_{j,\omega}(x_{j,\omega})
      \end{pmatrix}\right\|\\&\leq
    e^{-(\rho(E_{j,\omega})-\delta)(L-x_{j,\omega}-1)}
  \end{split}
\end{equation}
and
\begin{equation}
  \label{eq:82}
  \begin{split}
    |\varphi_{j,\omega}(0)|&=
    \left\|T^{}_{x_{j,\omega}}(E;\tau^{x_{j,\omega}}(\omega))
      \begin{pmatrix}0\\1 \end{pmatrix}\right\|^{-1}
    \left\| \begin{pmatrix}
        \varphi_{j,\omega}(x_{j,\omega})\\\varphi_{j,\omega}(x_{j,\omega}-1)
      \end{pmatrix}\right\|\\&\leq
    e^{-(\rho(E_{j,\omega})-\delta)(x_{j,\omega}-1)}.
  \end{split}
\end{equation}
The estimates given by Lemma~\ref{le:5} on the probability of
$\Omega^+_B(L,\ell,\delta,u_+)$ and $\Omega^-_B(L,\ell,\delta,u_-)$
for $\ell=C\log L$ and the estimate~\eqref{eq:99} then imply that,
with a probability at least $1-4L^2e^{-\eta(\ell-1)}-L^{-p}$, the
bounds~\eqref{eq:95},~\eqref{eq:96},~\eqref{eq:94} and~\eqref{eq:82}
hold. Thus, picking $\ell=C\log L$ for $C>0$ sufficiently large
(depending only on $\eta$, thus, on $\delta$ and $p$), these bounds
hold with a probability at least $1-L^{-p}$. This complete the proof
of Theorem~\ref{thr:10}.\qed
\begin{Rem}
  \label{rem:1}
  One may wonder whether the uniform growth estimate given by
  Lemmas~\ref{le:6} and~\ref{le:5} is actually necessary in the proof
  of Theorem~\ref{thr:10}. That they are necessary is due to the fact
  that both the eigenvalue $E_{j,\omega}$ and the localization center
  $x_{j,\omega}$ (and, thus, the vector
  $\D\left\|\begin{pmatrix}\varphi_{j,\omega} (x_{j,\omega})\\
      \varphi_{j,\omega}(x_{j,\omega}-1)
    \end{pmatrix}\right\|$) depend on $\omega$. Thus,~\eqref{eq:92} is
  not sufficient to estimate the second term in the left hand sides
  of~\eqref{eq:95} and~\eqref{eq:96}.
\end{Rem}
\subsubsection{The proof of Theorem~\ref{thr:20}}
\label{sec:proof-theor-refthr:2}
To prove Theorem~\ref{thr:20}, we follow the strategy that led to the
proof of Theorem~\ref{thr:10}. First, note that~\eqref{eq:95}
and~\eqref{eq:96} provide the expected lower bounds on the
eigenfunction with the right probability. As for the upper bound,
by~\eqref{eq:91}, using the conclusions of Theorem~\ref{thr:10} and
the bounds given by Lemma~\ref{le:6}, we know that, e.g. for $0\leq
x<x_{j,\omega}$
\begin{equation*}
  \begin{split}
    \left\|\begin{pmatrix}\varphi_{j,\omega}(x)\\\varphi_{j,\omega}(x-1)
      \end{pmatrix}\right\|&= \left\|T_{x}(E;\omega)
      \begin{pmatrix}1\\0\end{pmatrix}\right\||\varphi_{j,\omega}(0)|\leq
    e^{(\rho(E_{j,\omega})+\delta)x}
    e^{-(\rho(E_{j,\omega})-\delta)x_{j,\omega}} \\&\leq
    e^{-(\rho(E_{j,\omega})-C\delta)|x-x_{j,\omega}|}
  \end{split}
\end{equation*}
if $(1+C)x\leq (C-1)x_{j,\omega}$, i.e., $2(1+C)^{-1}x_{j,\omega}\leq
x_{j,\omega}-x$.\\
For $x\geq x_{j,\omega}$ one reasons similarly and, thus, completes
the proof of Theorem~\ref{thr:20}.\qed
\begin{Rem}
  \label{rem:2}
  Actually, as the proof shows, the results one obtains are more
  precise than the claims made in Theorem~\ref{thr:20}
  (see~\cite{Kl:12b}).
\end{Rem}
\subsubsection{The proof of Lemma~\ref{le:5}}
\label{sec:proof-lemma-refle:5}
The proofs for the two sets $\Omega^\pm_B(L,\ell,\delta,u)$ are the
same. We will only write out the one for
$\Omega^+_B(L,\ell,\delta,u)$.  Let us first address the measurability
issue for $\Omega^+_B(L,\ell,\delta,u)$. The functions $\omega\mapsto
E_{j,\omega}$ and $\omega\mapsto \varphi_{j,\omega}$ are continuous
(as the eigenvalues and eigenvectors of finite dimensional matrices
depending continuously on the parameter $\omega=(\omega_j)_{0\leq
  j\leq L}$). Thus, for fixed $j$, the sets $\{\omega;\
E_{j,\omega}\in \mathcal{K}^-_j(\omega,L,\delta,u)\}$ and $\{\omega;\
x_{j,\omega}>j\}$ are open (we used the definition of $x_{j,\omega}$
as the left most localization center (see Theorem~\ref{thr:10})). This
yields the measurability of $\Omega^+_B(L,\ell,\delta,u)$.
\vskip.2cm\noindent We claim that
\begin{equation}
  \label{eq:101}
  \frac1{L+1}\car_{\Omega^+_B(L,\ell,\delta,u)}\leq\sum_{j=0}^{L+1-\ell}
  \langle\delta_j,\car_{\mathcal{K}^+_j(\omega,L,\delta,u)}(H_{\omega,L})\delta_j\rangle
\end{equation}
where $\car_{\mathcal{K}^+_j(\omega,L,\delta,u)}(H_{\omega,L})$
denotes the spectral projector associated to $H_{\omega,L}$ on the set
$\mathcal{K}^+_j(\omega,L,\delta,u)$.  Indeed, if one has
$E_{j,\omega}\not\in \mathcal{K}^+_{x_{j,\omega}}(\omega,L,\delta,u)$ for
all $j$ then the left hand side of~\eqref{eq:101} vanishes and the
right hand side is non negative. On the other hand, if, for some $j$,
one has $0\leq x_{j,\omega}\leq L-\ell$ and $E_{j,\omega}\in
\mathcal{K}^+_{x_{j,\omega}}(\omega,L,\delta,u)$ then, we compute
\begin{equation*}
  \begin{split}
    \sum_{l=0}^{L-\ell}
    \langle\delta_l,\car_{\mathcal{K}^+_j(\omega,L,\delta,u)}
    (H_{\omega,L})\delta_l\rangle &=\sum_{l=0}^{L-\ell}
    \sum_{\substack{k\text{ s.t}\\E_{k,\omega}\in
        \mathcal{K}^+_j(\omega,L,\delta,u)}}|\varphi_{k,\omega}(l)|^2
    \geq|\varphi_{j,\omega}(x_{j,\omega})|^2\\&\geq\frac1{L+1}\geq
    \frac1{L+1}\car_{\Omega^+_B(L,\ell,\delta,u)}
  \end{split}
\end{equation*}
by the definition of $x_{j,\omega}$.\\
An important fact is that, by construction (see Lemma~\ref{le:5}), the
set of energies $\mathcal{K}^+_j(\omega,L,\delta,u)$ does not depend
on $\omega_j$. Hence, denoting by $\esp_{\omega_j}(\cdot)$ the
expectation with respect to $\omega_j$ and
$\esp_{\hat\omega_j}(\cdot)$ the expectation with respect to
$\hat\omega_j=(\omega_k)_{k\not=j}$, we compute
\begin{equation*}
  \esp\left(\sum_{j=0}^{L-\ell}
    \langle\delta_j,\car_{\mathcal{K}^+_j(\omega,L,\delta,u)}
    (H_{\omega,L})\delta_j\rangle\right)=
  \sum_{j=0}^{L-\ell}\esp_{\hat\omega_j}\left(\esp_{\omega_j}\left(
      \langle\delta_j,\car_{\mathcal{K}^+_j(\omega,L,\delta,u)}
      (H_{\omega,L})\delta_j\rangle\right)\right)
\end{equation*}
As $\omega_j$ is assumed to have a bounded compactly supported
distribution and as $\mathcal{K}^+_j(\omega,L,\delta,u)$ does not
depend on $\omega_j$, a standard spectral averaging lemma (see,
e.g.,~\cite[Theorem 11.8]{MR2154153}) yields
\begin{equation*}
  \esp_{\omega_j}\left(
    \langle\delta_j,\car_{\mathcal{K}^+_j(\omega,L,\delta,u)}(H_{\omega,L})\delta_j\rangle\right)
  \leq|\mathcal{K}^+_j(\omega,L,\delta,u)|
\end{equation*}
where $|\cdot|$ denotes the Lebesgue measure. Thus, we obtain
\begin{equation}
  \label{eq:102}
  \esp\left(\sum_{j=0}^{L-\ell}
    \langle\delta_j,\car_{\mathcal{K}^+_j(\omega,L,\delta,u)}(H_{\omega,L})
    \delta_j\rangle\right)
  \leq\sum_{j=0}^{L-\ell}\esp_{\hat\omega_j}\left(|\mathcal{K}^+_j
    (\omega,L,\delta,u)|\right)\\
  =\sum_{j=0}^{L-\ell}\esp\left(|\mathcal{K}^+_j(\omega,L,\delta,u)|\right).
\end{equation}
By Lemma~\ref{le:7} and the Fubini-Tonelli theorem, we know that
\begin{equation*}
  \begin{split}
    \esp\left(|\mathcal{K}^+_j(\omega,L,\delta,u)|\right)&=
    \esp\left(\int_I\car_{\mathcal{K}^+_j(\omega,L,\delta,u)}(E)dE
    \right)=
    \int_I\esp\left(\car_{\mathcal{K}^+_j(\omega,L,\delta,u)}(E)
    \right)dE\\&\leq |I|\,\sup_{E\in I}\pro\left( \left|
        \frac{\log\left\|T^{-1}_{L-(j+1)}(E,\omega) u \right\|}{L-j}-
        \rho(E)\right|>\delta\right)\\&\leq |I|\, e^{-\eta(L-j)}.
  \end{split}
\end{equation*}
Taking the expectation of both sides of~\eqref{eq:101} and plugging
this into~\eqref{eq:102}, we obtain
\begin{equation*}
  \pro(\Omega^+_B(L,\ell,\delta,u))\leq (L+1)
  |I| e^{-\eta(\ell-1)}\sum_{j=0}^{L-\ell}e^{-\eta j}
  \leq \frac{(L+1)|I| e^{-\eta(\ell-1)}}{1-e^{-\eta}}.
\end{equation*}
In the same way, one obtains
\begin{equation*}
  \pro(\Omega^-_B(L,\ell,\delta,u))
  \leq \frac{(L+1)|I| e^{-\eta(\ell-1)}}{1-e^{-\eta}}.
\end{equation*}
This completes the proof of Lemma~\ref{le:5}.  \qed
\begin{Rem}
  This proof can be seen as the analogue of the so-called Kotani trick
  for products of finitely many random matrices (see
 , e.g.,~\cite{MR883643}).
\end{Rem}
\subsubsection{The proof of Lemma~\ref{le:6}}
\label{sec:proof-lemma-refle:6}
The basic idea of this proof is to use the estimate~\eqref{eq:92}, in
particular, the exponentially small probability and some perturbation
theory for the cocycles so as to obtain a uniform estimate.\\
Let $\eta$ be given by~\eqref{eq:92}. Fix $\eta'<\eta/2$ and write
\begin{equation}
  \label{eq:80}
  I=\cup_{j\in J} [E_j,E_{j+1}]\text{  where }
  e^{-\eta'(L+1)}/2\leq E_{j+1}-E_j\leq 2e^{-\eta'(L+1)};
\end{equation}
thus, $\#J\lesssim e^{\eta'(L+1)}$.\\
We now want to estimate what happens for
$E\in[E_j,E_{j+1}]$. Therefore, using~\eqref{eq:9} and
\begin{equation*}
  \begin{pmatrix}
    E-V_\omega(n) & -1 \\ 1 &0 \end{pmatrix}-\begin{pmatrix}
    E_j-V_\omega(n) & -1 \\ 1 &0 \end{pmatrix}=(E-E_j)\Delta T
\end{equation*}
where
\begin{equation*}
  \Delta T:= \left|\begin{pmatrix} 1\\0
    \end{pmatrix}\right\rangle\left\langle
    \begin{pmatrix} 1\\0      \end{pmatrix}\right|
\end{equation*}
we compute
\begin{equation}
  \label{eq:83}
  T_L(E,\omega)=T_L(E_j,\omega)+\sum_{l=1}^L(E-E_j)^lS_l      
\end{equation}
where
\begin{equation*}
  \begin{split}
    S_l&:=\sum_{n_1<n_2<\cdots<n_l}
    T_{n_1}(E_j,\tau^{L-n_1}\omega)\times\Delta T\times 
    T_{n_2-n_1-1}(E_j,\tau^{n_2}\omega)\\&\hskip5cm\times\Delta
    T\times \cdots\times \Delta T \times
    T_{L-n_l-1}(E_j,\tau^{n_l}\omega)\\
    &=\sum_{n_1<n_2<\cdots<n_l}\prod_{m=2}^l
    \left\langle\begin{pmatrix} 1\\0
      \end{pmatrix},
      T_{n_m-n_{m-1}-1}(E_j,\tau^{n_m}\omega)\begin{pmatrix} 1\\0
      \end{pmatrix}\right\rangle\\&\hskip3cm
    \left|T_{n_1}(E_j,\tau^{L-n_1}\omega)\begin{pmatrix} 1\\0
      \end{pmatrix}\right\rangle\left\langle
      \begin{pmatrix} 1\\0 \end{pmatrix}\right|
    T_{L-n_l-1}(E_j,\tau^{n_l}\omega)
  \end{split}
\end{equation*}
Clearly, as the random variables have compact support, one has the
uniform bound
\begin{equation}
  \label{eq:84}
  \sup_{\substack{E\in
      I\\\omega\in\Omega}}\|T_L(E;\omega)\|\leq e^{CL}.
\end{equation}
Thus one has
\begin{equation}
  \label{eq:86}
  \sup_{\omega\in\Omega}\|S_l\|\leq L^le^{CL}.
\end{equation}
Hence, for $l_0$ fixed, one computes
\begin{equation}
  \label{eq:87}
  \left\|\sum_{l=l_0}^L(E-E_j)^lS_l\right\|\leq\sum_{l=l_0}^L(E-E_j)^l\|S_l\|
  \leq\sum_{l=l_0}^Le^{-\eta'(L+1)l}L^le^{CL}\leq 1
\end{equation}
if $\eta'l_0>2C$ and $L$ is sufficiently large (depending only on
$\eta'$ and $C$).\\
We now assume that $l_0$ satisfies $\eta'l_0>2C$ and pick $1\leq l\leq
l_0$. Pick $\delta_0\in(0,1)$ small to be fixed later. Assume moreover
that $L$ is so that $\delta_0 L\geq L_\delta$ where $L_\delta$ is
defined in Lemma~\ref{le:7}. Then, by Lemma~\ref{le:7}, for
$m\in\{2,\cdots,l\}$, one has
\begin{enumerate}
\item either $n_m-n_{m-1}\leq L_\delta$; then, one has
  \begin{equation*}
    \left\|T_{n_m-n_{m-1}-1}(E_j,\tau^{n_m-1}\omega)
    \right\|\leq e^{C(n_m-n_{m-1})};
  \end{equation*}
\item or $n_m-n_{m-1}\geq L_\delta$; then, by~\eqref{eq:92}, with
  probability at least equal to $1-e^{-\eta(n_m-n_{m-1})/2}$, one has
  \begin{equation*}
    \left\|T_{n_m-n_{m-1}-1}(E_j,\tau^{n_m-1}\omega)
    \right\|\leq e^{(n_m-n_{m-1})(\rho(E_j)+\delta)}.
  \end{equation*}
\end{enumerate}
Define
\begin{gather*}
  G_{n_1,\cdots,n_l}=\{m\in\{2,\cdots,l\}; n_m-n_{m-1}\geq
  L_\delta\}\\\intertext{ and }
  B_{n_1,\cdots,n_l}=\{2,\cdots,l\}\setminus G_{n_1,\cdots,n_l}.
\end{gather*}
By definition, one has
\begin{equation}
  \label{eq:77}
  \sum_{m\in B_{n_1,\cdots,n_l}}(n_m-n_{m-1})\leq l L_\delta
  \hskip.2cm\text{and}
  \sum_{m\in G_{n_1,\cdots,n_l}}(n_m-n_{m-1})\geq L-l L_\delta.
\end{equation}
For a fixed sequence $n_1<n_2<\cdots<n_m$, the random
variables $\D \left(T_{n_{m'}-n_{m'-1}-1}(E_j,
  \tau^{n_{m'}}\omega)\right)_{1\leq m'\leq m}$ are independent. Hence,
by~\eqref{eq:92}, for a fixed $(m_1,\cdots,m_K)\in
G_{n_1,\cdots,n_l}$, one has
\begin{equation*}
  \pro\left(
    \inf_{1\leq k\leq K}\left\|T_{n_{m_k}-n_{m_k-1}-1}(E_j,
      \tau^{n_{m_k}}\omega)\right\|\geq e^{(\rho(E_j)+\delta)
      (n_{m_k}-n_{m_k-1})}
  \right)\leq e^{-\eta\sum_{k=1}^K n_{m_k}-n_{m_k-1}}.
\end{equation*}
Thus, for $\varepsilon\in(0,1)$, one has
\begin{equation*}
  \pro\left(
    \begin{aligned}
      \exists(m_1,\cdots,m_K)\in G_{n_1,\cdots,n_l}\text{
        s.t. }\sum_{k=1}^K &n_{m_k}-n_{m_k-1}\geq\varepsilon L\\
      \inf_{1\leq k\leq K}\left\|T_{n_{m_k}-n_{m_k-1}-1}(E_j,
        \tau^{n_{m_k}-1}\omega)\right\|&\geq e^{(\rho(E_j)+\delta)
        (n_{m_k}-n_{m_k-1})}
    \end{aligned}
  \right)\leq L^l e^{-\eta\varepsilon L}.    
\end{equation*}
Hence, with probability at least $1-L^l e^{-\eta\varepsilon L}$, we know that
\begin{equation*}
  \begin{split}
    \exists(m_1,\cdots,m_K)\in G_{n_1,\cdots,n_l}\text{
      s.t. }\sum_{k=1}^K n_{m_k}-n_{m_k-1}&\geq L-l L_\delta-\varepsilon L\\
    \forall1\leq k\leq K,\ \left\|T_{n_{m_k}-n_{m_k-1}-1}(E_j,
      \tau^{n_{m_k}-1}\omega)\right\|&\leq e^{(\rho(E_j)+\delta)
      (n_{m_k}-n_{m_k-1})}.
  \end{split}
\end{equation*}
Using estimates~\eqref{eq:77} and~\eqref{eq:84} for the remaining
terms in the product below, for any given $m$-uple $(n_1,\cdots,n_m)$,
one obtains
\begin{equation*}
  \pro\left(
    \begin{aligned}
      &\prod_{m=1}^l\left\|T_{n_m-n_{m-1}-1}(E_j,
        \tau^{n_{m_k}-1}\omega)\right\|\\&\hskip3cm\leq
      e^{(\rho(E_j)+\delta)(1-\varepsilon)(L-l L_\delta) +C(\varepsilon L+l
        L_\delta)}
    \end{aligned}
  \right)\geq 1-L^l e^{-\eta\varepsilon L}.
\end{equation*}
Hence, with probability at least $1- l_0 L^{l_0}\,e^{-\eta\varepsilon
  L}$, for $1\leq l\leq l_0$, we estimate
\begin{equation}
  \label{eq:79}
  \begin{split}
    \|S_l\|&\leq \sum_{n_1<n_2<\cdots<n_l}
    \prod_{m=1}^l\left\|T_{n_m-n_{m-1}-1}(E_j,
      \tau^{n_{m_k}}\omega)\right\| \\&\leq
    L^le^{(\rho(E_j)+\delta)(1-\varepsilon)L +C\varepsilon
      L+(C-(\rho(E_j)+\delta)(1-\varepsilon))l L_\delta} \\&\leq
    L^le^{[\rho(E_j)+\delta+(C-\rho(E_j)-\delta)\varepsilon]L
      +[C-(\rho(E_j)+\delta)(1-\varepsilon)]L_\delta l}\\&\leq
    L^{l_0}e^{[\rho(E_j)+\delta+(C-\rho(E_j)-\delta)\varepsilon]L
      +[C-(\rho(E_j)+\delta)(1-\varepsilon)]L_\delta l_0}.
  \end{split}
\end{equation}
It remains now to choose the quantities $\eta'$, $l_0$ and
$\varepsilon$ so that the following requirements be satisfied
\begin{equation}
  \label{eq:78}
  \begin{split}
    \eta'l_0>2C,\quad (C-\rho(E_j)-\delta)\varepsilon\leq
    \frac\delta2,& \quad l_0
    L^{l_0}\,e^{-\eta\varepsilon L} e^{\eta'(L+1)}\ll 1 \\
    \text{and}\quad\frac{[C-(\rho(E_j)+\delta)(1-\varepsilon)]L_\delta
      l_0}{L+1}&\leq\frac{\delta}{2(\rho(E_j)+\delta)}.
  \end{split}
\end{equation}
Fixing $\varepsilon$ small, picking $0<\eta'<\eta\varepsilon/3$ and
setting $l_0=L^\alpha$ where $\alpha\in(0,1)$, we see that all the
conditions in~\eqref{eq:78} are satisfied for $L$ sufficiently
large. Moreover, one has
\begin{equation*}
  l_0 L^{l_0}\,e^{-\eta\varepsilon L}
  e^{\eta'(L+1)}\leq e^{-\eta\varepsilon L/2}.
\end{equation*}
Plugging this and the last estimate in~\eqref{eq:79}
into~\eqref{eq:83}, we obtain that, with probability at least
$1-e^{-\eta\varepsilon L/2}$, for any $j\in J$ (see~\eqref{eq:80}), for
$E\in[E_j,E_{j+1}]$, one has
\begin{equation}
  \label{eq:90}
  \begin{split}
    \left\|T_L(E,\omega)-T_L(E_j,\omega)\right\|&\leq 1+
    \sum_{l=1}^{l_0}e^{-\eta'l (L+1)} L^le^{(\rho(E_j)+2\delta)L}
    \\&\leq 1+e^{(\rho(E_j)+2\delta)(L+1)}
  \end{split}
\end{equation}
As $\rho$ is continuous (see, e.g.,~\cite{MR88f:60013}), one gets that,
for any $\delta>0$, for $L$ sufficiently large, with probability at
least $1-e^{-\eta\varepsilon L/2}$, one has, for any $E\in I$,
\begin{equation*}
  \left\|T_L(E,\omega)\right\|\lesssim e^{(\rho(E)+2\delta)(L+1)}.
\end{equation*}
Hence, as $T_L(E,\omega)\in SL(2,\R)$, one has
$\D\left\|T^{-1}_L(E,\omega)\right\|\lesssim
e^{(\rho(E)+2\delta)(L+1)}$.\\
Using the fact that the probability measure on $\Omega$ is invariant
under the shift (it is a product measure), we
obtain~\eqref{eq:81}. This completes the proof of
Lemma~\ref{le:6}.\qed
\subsubsection{The proof of Lemma~\ref{le:3}}
\label{sec:proof-lemma6.2}
Assume the realization $\omega$ is such that the conclusions of
Lemma~\ref{le:2} hold in $I$ for the scales $l_L=2\log L$. Fix
$\alpha>0$ and let $\mathcal{E}_{L,\omega}$ be the set of indices of
the eigenvalues $(E_{j,\omega})_{0\leq j\leq L}$ of $H_{\omega,L}$ having a
localization center in $\llbracket L-\ell_L,L\rrbracket$. Fix
$C>\alpha>0$ and consider the projector on the sites in $\llbracket
L-C\ell_L,L\rrbracket$, i.e., $\Pi_C:=\car_{\llbracket
  L-C\ell_L,L\rrbracket}$.\\
Consider the following Gram matrices
\begin{equation*}
  G(\mathcal{E}_{L,\omega})=((\langle\varphi_{j,\omega},
  \varphi_{j,\omega}\rangle))_{(n,m)\in\mathcal{E}_{L,\omega}\times
    \mathcal{E}_{L,\omega}}=Id_N
\end{equation*}
where $N=\#\mathcal{E}_{L,\omega}$ and
\begin{equation*}
  G_\pi(\mathcal{E}_{L,\omega})=((\langle\Pi_C\varphi_{j,\omega},\Pi_C
  \varphi_{j,\omega}\rangle))_{(n,m)\in
    \mathcal{E}_{L,\omega}\times \mathcal{E}_{L,\omega}}.
\end{equation*}
By definition, the rank of $G_\pi(\mathcal{E}_{L,\omega})$ is bounded
by the rank of $\Pi_C$, i.e., by $C\ell_L$. Moreover, as
by~\eqref{eq:75} one has $\|(1-\Pi_C)\varphi_{j,\omega}\|\leq
L^{q}e^{-\rho\eta C\ell_L}$, one has
\begin{equation*}
  \|Id_N-G_\pi(\mathcal{E}_{L,\omega})\|\leq L^{2+q}e^{-\rho\eta C\ell_L}\leq
  L^{2+q-C\rho\eta}.
\end{equation*}
Thus, picking $C\eta\rho>q+2$ yields that, for $L$ sufficiently large,
$G_\pi(\mathcal{E}_{L,\omega})$ is invertible and its rank is
$N$. This yields $\#\mathcal{E}_{L,\omega}=N\leq C\ell_L$ and the proof of
Lemma~\ref{le:3} is complete.\qed
\subsection{The half-line random perturbation: the proof of
  Theorem~\ref{thr:26}}
\label{sec:half-line-rand}
Using the same notations as in section~\ref{sec:half-line-periodic},
we can write
\begin{equation*}
  H^\infty=  \begin{pmatrix}    H_{\omega,-1}^-&
    |\delta_{-1}\rangle\langle\delta_0|
    \\  |\delta_0\rangle\langle\delta_{-1}|&-\Delta_0^+
  \end{pmatrix}
\end{equation*}
where
\begin{itemize}
\item $-\Delta_{0}^+$ is the Dirichlet Laplacian on $\ell^2(\N)$,
\item $H_{\omega,-1}^-=-\Delta+V_\omega$ on $\ell^2(\{n\leq -1\})$ with
  Dirichlet boundary conditions at $0$.
\end{itemize}
Define the operators
\begin{gather*}
  \Gamma_\omega(E):=-\Delta_0^+-E-\langle\delta_{-1}|
  (H_{\omega,-1}^--E)^{-1}|\delta_{-1}\rangle
  \,|\delta_0\rangle\langle\delta_0|,\\
  \tilde\Gamma_\omega(E):=H_{\omega,-1}^--E-\langle\delta_0|
  (-\Delta_0^+-E)^{-1}|\delta_0\rangle
  \,|\delta_{-1}\rangle\langle\delta_{-1}|.
\end{gather*}
For Im$\,E\not=0$, the numbers
$\langle\delta_{-1}|(H_{\omega,-1}^--E)^{-1}|\delta_{-1}\rangle$ and
$\langle\delta_0|(-\Delta_0^+-E)^{-1}|\delta_0\rangle$ have non
vanishing imaginary parts of the same sign; hence, the complex number
$(\langle\delta_{-1}|(H_{\omega,-1}^--E)^{-1}|\delta_{-1}\rangle)^{-1}-
\langle\delta_0|(-\Delta_0^+-E)^{-1}|\delta_0\rangle$ does not
vanish. Thus, by rank one perturbation theory, (see,
e.g.,~\cite{MR2154153}), we thus know that $\Gamma_\omega(E)$ and
$\tilde\Gamma_\omega(E)$ are invertible for Im$\,E\not=0$ and that
\begin{gather}
  \label{eq:55}
  \begin{split}
    \Gamma_\omega^{-1}(E)&=(-\Delta_0^+-E)^{-1}\\&+
    \frac{|(-\Delta_0^+-E)^{-1}|\delta_0\rangle\langle\delta_0|(-\Delta_0^+-E)^{-1}|}{
      (\langle\delta_{-1}|(H_{\omega,-1}^--E)^{-1}|\delta_{-1}\rangle)^{-1}-
      \langle\delta_0|(-\Delta_0^+-E)^{-1}|\delta_0\rangle}
  \end{split}
  \\
  \label{eq:213}
  \begin{split}
    \tilde\Gamma_\omega^{-1}(E)&=(H_{\omega,-1}^--E)^{-1}\\&+
    \frac{|(H_{\omega,-1}^--E)^{-1}|\delta_{-1}\rangle\langle\delta_{-1}|(H_{\omega,-1}^--E)^{-1}|}{
      (\langle\delta_0|(-\Delta_0^+-E)^{-1}|\delta_0\rangle)^{-1}-
      \langle\delta_{-1}|(H_{\omega,-1}^--E)^{-1}|\delta_{-1}\rangle
    }.
  \end{split}
\end{gather}
Thus, for Im$\,E\not=0$, using Schur's complement formula, we compute
\begin{equation}
  \label{eq:58}
  (H_\omega^\infty-E)^{-1}=\begin{pmatrix}
    \tilde\Gamma_\omega^{-1}(E) &\gamma(E) \\\gamma^*\left(\overline{E}\right) &
    \Gamma^{-1}_\omega(E)
  \end{pmatrix}.
\end{equation}
where $\gamma^*\left(\overline{E}\right)$ is the adjoint of
$\gamma\left(\overline{E}\right)$ and
\begin{equation*}
  \gamma(E):=-|(H_{\omega,-1}^--E)^{-1}|\delta_{-1}\rangle
  \langle\delta_0|\Gamma^{-1}_\omega(E)|
\end{equation*}
\subsubsection{The continuation through $(-2,2)\setminus\Sigma$}
\label{sec:continuation-through}
Let us start with the analytic continuation through
$(-2,2)\setminus\Sigma$.\\
One easily checks that the function $E\mapsto
\langle\delta_{-1}|(H_{\omega,-1}^--E)^{-1}|\delta_{-1}\rangle^{-1}$
is analytic outside $\Sigma$, the essential spectrum of
$H_{\omega,-1}^-$ and has simple zeros at the isolated eigenvalues of
$H_{\omega,-1}^-$.  Hence, $E\mapsto \Gamma_\omega^{-1}(E)$ can be
analytically continued near an isolated eigenvalue of
$H_{\omega,-1}^-$ different from $-2$
and $2$.\\
As for $\tilde\Gamma_\omega^{-1}$, using the spectral decomposition of
of $H_{\omega,-1}^--E)^{-1}$, as for any eigenvector of
$H_{\omega,-1}^-$, say, $\varphi$, one has
$\langle\delta_{-1},\varphi\rangle\not=0$, for $E_0$, an isolated
eigenvalue of $H_{\omega,-1}^-$ different from $-2$ and $2$, doing a
polar decomposition of $\tilde\Gamma_\omega^{-1}$ near $E_0$, one
checks that $E\mapsto\tilde\Gamma_\omega^{-1}(E)$ can be analytically
continued to a neighborhood of $E_0$.\\
Finally let us check what happens with $\gamma$. We compute
\begin{equation*}
  \gamma(E)=-\langle\delta_{-1}|(H_{\omega,-1}^--E)^{-1}|\delta_{-1}\rangle^{-1}
  |(H_{\omega,-1}^--E)^{-1}|\delta_{-1}\rangle
  \langle\delta_0|(-\Delta_0^+-E)^{-1}|.
\end{equation*}
As $E\mapsto
\langle\delta_{-1}|(H_{\omega,-1}^--E)^{-1}|\delta_{-1}\rangle^{-1}
(H_{\omega,-1}^--E)^{-1}$ is analytic near any isolated eigenvalue of
$(H_{\omega,-1}^-$, we see that $E\mapsto\gamma(E)$ can be can be
analytically continued to a
neighborhood of an isolated eigenvalue of $H_{\omega,-1}^-$. \\
Hence, the representation~\eqref{eq:58} immediately shows that the
resolvent $(H_\omega^\infty-E)^{-1}$ can be continued through
$(-2,2)\setminus\Sigma$, the poles of the continuation being given by
the zeros of the function
\begin{equation*}
  E\mapsto 1-\langle\delta_0|(-\Delta_0^+-E)^{-1}|\delta_0\rangle
  \langle\delta_{-1}|(H_{\omega,-1}^--E)^{-1}|\delta_{-1}\rangle
  =1-e^{i\theta(E)}\int_\R\frac{dN_\omega(\lambda)}{\lambda-E}.
\end{equation*}
\subsubsection{No continuation through
  $(-2,2)\cap\overset{\circ}{\Sigma}$}
\label{sec:no-continuation}
Let us study the analytic continuation through
$(-2,2)\cap\overset{\circ}{\Sigma}$. Considering the lower right
coefficient of this matrix, we see that, when coming from upper
half-plane through $(-2,2)\cap\overset{\circ}{\Sigma}$,
$E\mapsto(H_\omega^\infty-E)^{-1}$ can be continued meromorphically to
the lower half plane (as an operator from $\ell^2_{\text{comp}}(\Z)$
to $\ell^2_{\text{loc}}(\Z)$) only if $E\mapsto\Gamma_\omega^{-1}(E)$
can be meromorphically (as an operator
from $\ell^2_{\text{comp}}(\N)$ to $\ell^2_{\text{loc}}(\N)$).\\
As $E\mapsto(-\Delta_0^+-E)^{-1}$ can be analytically continued (see
section~\ref{sec:proof-theorem}), by~\eqref{eq:55}, the meromorphic
continuation of $E\mapsto\Gamma_\omega^{-1}(E)$ will exist if and only
if the complex valued map
\begin{equation*}
  E\mapsto g_\omega(E):=\frac1{(\langle\delta_{-1}|(H_{\omega,-1}^--E)^{-1}
    |\delta_{-1}\rangle)^{-1}-
    \langle\delta_0|(-\Delta_0^+-E)^{-1}|\delta_0\rangle}
\end{equation*}
can be meromorphically continued from the upper half-plane through
$(-2,2)\cap\overset{\circ}{\Sigma}$. Fix $\omega$ s.t. the spectrum of
$H_{\omega,-1}^-$ be equal to $\Sigma$ and pure point (this is almost
sure (see, e.g.,~\cite{MR1102675,MR94h:47068}). As $\delta_{-1}$ is a
cyclic vector for $H_{\omega,-1}^-$, for $E$ an eigenvalue of
$H_{\omega,-1}^-$, one then has
\begin{equation}
  \label{eq:57}
  \lim_{\varepsilon\to0^+}
  (\langle\delta_{-1}|(H_{\omega,-1}^--E-i\varepsilon)^{-1}
  |\delta_{-1}\rangle)^{-1}=0.
\end{equation}
Hence, if the analytic continuation of $g_\omega$ would exist, on
$(-2,2)\cap\overset{\circ}{\Sigma}$, it would be equal to
\begin{equation}
  \label{eq:217}
  g_\omega(E+i0)=-\frac1{\langle\delta_0|(-\Delta_0^+-E-i0)^{-1}
    |\delta_0\rangle}.
\end{equation}
By analyticity of both sides, this in turn would imply
that~\eqref{eq:217} holds on the whole upper half-plane, thus, in view
of the definition of $g_\omega$, that~\eqref{eq:57} holds on the whole
upper half plane: this is absurd! Thus, we have proved that, $\omega$
almost surely, $E\mapsto(H_\omega^\infty-E)^{-1}$ does not admit a
meromorphic continuation through $(-2,2)\cap\overset{\circ}{\Sigma}$.
\subsubsection{Absolutely continuity of the spectrum of
  $H_\omega^\infty$ in $(-2,2)\cap\overset{\circ}{\Sigma}$}
\label{sec:absol-cont-spectr}
Let us now prove that the spectral measure of $H_\omega^\infty$ in
$(-2,2)\cap\overset{\circ}{\Sigma}$ is purely absolutely
continuous. Therefore, it suffices (see, e.g.,~\cite[section
2.5]{MR1711536} and~\cite[Theorem 11.6]{MR2154153}) to prove that, for
all $E\in(-2,2)\cap\overset{\circ}{\Sigma}$, one has
\begin{equation*}
  \limsup_{\varepsilon\to0^+}\left|\langle\delta_0,
    (H_\omega^\infty-E-i\varepsilon)^{-1}\delta_0\rangle\right|
  +\left|\langle\delta_{-1},
    (H_\omega^\infty-E-i\varepsilon)^{-1}\delta_{-1}\rangle\right|<+\infty.
\end{equation*}
Using~\eqref{eq:55},~\eqref{eq:213} and~\eqref{eq:58}, for
Im$\,E\not=0$, we compute
\begin{equation}
  \label{eq:206}
  \langle\delta_{-1},(H_\omega^\infty-E)^{-1}\delta_{-1}\rangle
  =\frac{\langle\delta_{-1}|(H_{\omega,-1}^--E)^{-1}|\delta_{-1}\rangle}{
    1-\langle\delta_0|(-\Delta_0^+-E)^{-1}|\delta_0\rangle\cdot
    \langle\delta_{-1}|(H_{\omega,-1}^--E)^{-1}|\delta_{-1}\rangle},
\end{equation}
for $n\geq1$, $m\leq0$,
\begin{equation}
  \label{eq:207}
  \langle\delta_{-n},(H_\omega^\infty-E)^{-1}\delta_m\rangle
  =\frac{-\langle\delta_{-n}|(H_{\omega,-1}^--E)^{-1}|\delta_{-1}\rangle
    \langle\delta_0|(-\Delta_0^+-E)^{-1}|\delta_m\rangle}{
    1-\langle\delta_0|(-\Delta_0^+-E)^{-1}|\delta_0\rangle\cdot
    \langle\delta_{-1}|(H_{\omega,-1}^--E)^{-1}|\delta_{-1}\rangle}
\end{equation}
and
\begin{equation}
  \label{eq:208}
  \langle\delta_0,(H_\omega^\infty-E)^{-1}\delta_0\rangle
  =\frac{\langle\delta_0|(-\Delta_0^+-E)^{-1}|\delta_0\rangle}{
    1-\langle\delta_0|(-\Delta_0^+-E)^{-1}|\delta_0\rangle\cdot
    \langle\delta_{-1}|(H_{\omega,-1}^--E)^{-1}|\delta_{-1}\rangle}.
\end{equation}
Thus, to prove the absolute continuity of the spectral measure of
$H_\omega^\infty$ in $(-2,2)\cap\overset{\circ}{\Sigma}$, it suffices
to prove that, for $E\in(-2,2)\cap\overset{\circ}{\Sigma}$, one has
\begin{multline*}
  \limsup_{\varepsilon\to0^+}\left(
    \left|\frac1{(\langle\delta_{-1}|(H_{\omega,-1}^--E-i\varepsilon)^{-1}|\delta_{-1}\rangle)^{-1}-\langle\delta_0|(-\Delta_0^+-E-i\varepsilon)^{-1}|\delta_0\rangle}\right|\right.\\\left.+\left|
      \frac1{(\langle\delta_0|(-\Delta_0^+-E-i\varepsilon)^{-1}|\delta_0\rangle)^{-1}-\langle\delta_{-1}|(H_{\omega,-1}^--E-i\varepsilon)^{-1}|\delta_{-1}\rangle}
    \right|\right)<\infty.
\end{multline*}
This is the case as
\begin{itemize}
\item the signs of the imaginary parts of
  $-(\langle\delta_{-1}|(H_{\omega,-1}^--E-i\varepsilon)^{-1}|\delta_{-1}\rangle)^{-1}$
  and
  $\langle\delta_0|(-\Delta_0^+-E-i\varepsilon)^{-1}|\delta_0\rangle$
  are the same (negative if Im$\,E<0$ and positive if Im$\,E>0$),
\item for $E\in(-2,2)$,
  $\langle\delta_0|(-\Delta_0^+-E-i\varepsilon)^{-1}|\delta_0\rangle$
  has a finite limit when $\varepsilon\to0^+$,
\item for$E\in(-2,2)$, the imaginary part of
  $\langle\delta_0|(-\Delta_0^+-E-i\varepsilon)^{-1}|\delta_0\rangle$
  does not vanish in the limit $\varepsilon\to0^+$.
\end{itemize}
So, we have proved the part of Theorem~\ref{thr:26} concerning the
absence of analytic conti\-nuation of the resolvent of
$H_\omega^\infty$ through $(-2,2)\cap\overset{\circ}{\Sigma}$ and the
nature of its spectrum in this set.
\subsubsection{The spectrum of $H_\omega^\infty$ is pure point in
  $\overset{\circ}{\Sigma}\setminus[-2,2]$}
\label{sec:spectr-h_om-pure}
Let us now prove the last part of Theorem~\ref{thr:26}. The proof
relies again on~\eqref{eq:58}.  We pick $\beta\in(0,\alpha/2)$ where
$\alpha$ is determined by Theorem~\ref{thr:18} for
$H_{\omega,-1}^-$. Then, for $n\geq1$ and $m\leq0$, using the
Cauchy-Schwartz inequality, for Im$\,E\not=0$, we compute
\begin{multline}
  \label{eq:205}
  \esp\left(\left|\langle\delta_{-n},(H_\omega^\infty-E)^{-1}\delta_m\rangle
    \right|^\beta\right)^2\\\leq
  |\langle\delta_0|(-\Delta_0^+-E)^{-1}|\delta_m\rangle|^2
  \cdot\esp\left(\left|\langle\delta_{-n}|(H_{\omega,-1}^--E)^{-1}|\delta_{-1}\rangle\right|^{2\beta}\right)\\
  \cdot\esp\left(\left|\frac{1}{
        1-\langle\delta_0|(-\Delta_0^+-E)^{-1}|\delta_0\rangle\cdot
        \langle\delta_{-1}|(H_{\omega,-1}^--E)^{-1}|\delta_{-1}\rangle}\right|^{2\beta}\right)
\end{multline}
For $J\subset (-2,2)\setminus\Sigma$ a compact interval, we know that, for
$n\geq1$ and $m\leq0$,
\begin{itemize}
\item $\D\sup_{\text{Im}\,E\not=0}
  |\langle\delta_0|(-\Delta_0^+-E)^{-1} |\delta_m\rangle|\lesssim
  e^{-c m}$ by the Combes-Thomas estimates;
\item $\D\sup_{\text{Im}\,E\not=0}
  \esp\left(\left|\langle\delta_{-n}|(H_{\omega,-1}^--E)^{-1}|
      \delta_{-1}\rangle\right|^{2\beta}\right)\lesssim e^{-2\beta\rho
    n}$ by the characte\-rization~\eqref{eq:204} of localization in
  $\Sigma$ for $H_{\omega,-1}^-$.
\end{itemize}
It suffices now to estimate the last term in~\eqref{eq:205} using a
standard decomposition of rank one perturbations (see,
e.g.,~\cite{MR2154153,MR1244867}), one writes
\begin{equation*}
  \frac{1}{1-\langle\delta_0|(-\Delta_0^+-E)^{-1}|\delta_0\rangle\cdot
    \langle\delta_{-1}|(H_{\omega,-1}^--E)^{-1}|\delta_{-1}\rangle}
  =\frac{\omega_{-1}-b}{\omega_{-1}-a}
\end{equation*}
where $a$ and $b$ only depend on $(\omega_{-n})_{n\geq2}$. Thus, as
$(\omega_{-n})_{n\geq1}$ have a bounded density, for Im$\,E\not=0$,
one has
\begin{equation*}
  \begin{split}
    \esp\left(\left|\frac{1}{
          1-\langle\delta_0|(-\Delta_0^+-E)^{-1}|\delta_0\rangle\cdot
          \langle\delta_{-1}|(H_{\omega,-1}^--E)^{-1}|\delta_{-1}\rangle}
      \right|^{2\beta}\right)&\leq\esp_{(\omega_{-n})_{n\geq2}}
    \esp_{\omega_{-1}} \left(\left|
          \frac{\omega_{-1}-b}{\omega_{-1}-a}\right|^{2\beta}
      \right)\\&\leq C_\beta<+\infty.
  \end{split}
\end{equation*}
Thus, we have proved that, for $J\subset\Sigma\setminus[-2,2]$ a
compact interval, for $\beta\in(0,\alpha/2)$ and some $\tilde\rho>0$,
for $n\geq1$ and $m\leq0$, one has
\begin{equation*}
  \sup_{\substack{\text{Im}\,E\not=0\\ \text{Re}\,E\in I}}
  \esp\left(\left|\langle\delta_{-n},(H_\omega^\infty-E)^{-1}\delta_m\rangle
    \right|^\beta\right)<C_\beta e^{-\tilde\rho(m-n)} .
\end{equation*}
In the same way, using~\eqref{eq:206} and~\eqref{eq:208}, one proves
that
\begin{equation*}
  \sup_{\substack{\text{Im}\,E\not=0\\ \text{Re}\,E\in I}}
  \esp\left(\left|\langle\delta_{0},(H_\omega^\infty-E)^{-1}\delta_0\rangle
    \right|^\beta+
    \left|\langle\delta_{-1},(H_\omega^\infty-E)^{-1}\delta_{-1}\rangle
    \right|^\beta\right)<+\infty
\end{equation*}
Thus, we have proved that, for some $\tilde\rho>0$, one has
\begin{equation*}
  \sup_{\substack{\text{Im}\,E\not=0\\ \text{Re}\,E\in I}} \sup_{m\in\Z}
  \esp\left(\sum_{n\in\Z}e^{\tilde\rho(m-n)}
    \left|\langle\delta_{-n},(H_\omega^\infty-E)^{-1}\delta_m\rangle
    \right|^\beta\right)<+\infty.
\end{equation*}
Hence, we know that the spectrum of $H_\omega^\infty$ in
$\Sigma\setminus[-2,2]$ (as $J$ can be taken arbitrary contained in
this set) is pure point associated to exponentially decaying
eigenfunctions (see,
e.g.,~\cite{MR1244867,MR1301371,MR2002h:82051}). This completes the
proof of Theorem~\ref{thr:26}.
\section{Appendix}
\label{sec:appendix}
In this section we study the eigenvalues and eigenvectors of $H_L$
(see Remark~\ref{rem:6}) near an energy $E'$ that is an eigenvalue of
both $H_0^+$ and $H_k^-$ (see the ends of
sections~\ref{sec:dirichl-eigenv} and~\ref{sec:dirichl-eigenf}). We
keep the notations of
sections~\ref{sec:dirichl-eigenv} and~\ref{sec:dirichl-eigenf}.\\
Let $\varphi^+\in\ell^2(\N)$ (resp. $\varphi^-\in\ell^2(\Z_-)$) be
normalized eigenvectors of $H_0^+$ (resp. $H^-_k$) associated to
$E_-$. Thus, by~\eqref{eq:127} and~\eqref{eq:126}, we can pick, for
$n\geq0$ and $l\in\{0,\cdots,p-1\}$,
\begin{equation}
  \label{eq:129}
  \varphi^+_{np+l}=c a_l(E')\rho^n(E')\text{ and }
  \varphi^-_{-np-l}=c^- b_l(E')\rho^n(E').
\end{equation}
Assume $L=Np+k$ and, for $l\in\{0,\cdots,L\}$, define
$\varphi^{\pm,L}\in\ell^2(\llbracket0,L\rrbracket)$ by
\begin{equation}
  \label{eq:130}
  \begin{aligned}
    \varphi^{+,L}_l:=\varphi^+_l,\quad \varphi^{+,L}_{-1}=
    \varphi^{+,L}_{L+1}:=\varphi^+_{-1}=0\quad\text{and}\\
    \varphi^{-,L}_l:= \varphi^-_{l-L},\quad
    \varphi^{-,L}_{-1}=\varphi^{-,L}_{L+1}:=\varphi^-_0=0.
  \end{aligned}
\end{equation}
Thus, one has
\begin{equation}
  \label{eq:131}
  \begin{aligned}
    H_L\varphi^{+,L}=E'\varphi^{+,L}+\varphi^+_{L+1}\delta_L,&\quad
    H_L\varphi^{-,L}=E'\varphi^{-,L}+\varphi^-_{-L-1}\delta_0
    \\\text{and}\quad
    \langle\varphi^{+,L},\varphi^{-,L}\rangle&=O(N\rho^N(E)).
  \end{aligned}
\end{equation}
Recall that $a_k(E')\not=0\not=b_k(E')$ (see
sections~\ref{sec:dirichl-eigenv} and~\ref{sec:dirichl-eigenf}); thus,
by~\eqref{eq:129}, one has
\begin{equation}
  \label{eq:132}
  |\varphi^-_{-L-1}|\asymp |\rho(E')|^n\asymp |\varphi^+_{L+1}|.
\end{equation}
Moreover, as $H_L$ converges to $H_0^+$ in strong resolvent sense, for
$\varepsilon>0$ sufficiently small, for $L$ sufficiently large, $ H_L
$ has no spectrum in the compact
$E'+[-2\varepsilon,\varepsilon/2]\cup[\varepsilon/2,2\varepsilon]$. Let
$\Pi_L$ be the spectral projector onto the interval
$[\varepsilon/2,\varepsilon/2]$ that is $\D
\Pi_L:=\frac1{2i\pi}\int_{|z-E'|=\varepsilon}( H_L
-z)^{-1}dz$. By~\eqref{eq:131}, one computes
\begin{equation*}
  (1-\Pi_L)\varphi^{+,L}=\frac{\varphi^+_{L+1}}
  {2i\pi}\int_{|z-E'|=\varepsilon}(E'-z)^{-1}(H_L-z)^{-1}\delta_0dz
\end{equation*}
Thus, one gets
\begin{equation}
  \label{eq:133}
  \|(1-\Pi_L)\varphi^{+,L}\|+\|(1-\Pi_L)\varphi^{-,L}\|\lesssim|\rho(E')|^N.
\end{equation}
Define
\begin{equation*}
  \tilde\chi^{+,L}=\frac1{\|\Pi_L\varphi^{+,L}\|}\Pi_L\varphi^{+,L}\quad\text{and}\quad
  \tilde\chi^{-,L}=\frac1{\|\Pi_L\varphi^{-,L}\|}\Pi_L\varphi^{-,L}.
\end{equation*}
The Gram matrix of $(\tilde\chi^{+,L},\tilde\chi^{-,L})$ then reads
Id$+O(N\rho^N(E))$. Orthonormalizing
$(\tilde\chi^{+,L},\tilde\chi^{-,L})$ into $(\chi^{+,L},\chi^{-,L})$
and, computing the matrix elements of $\Pi_L(H_L-E')$ in this basis,
we obtain
\begin{equation*}
  \begin{pmatrix}
    \varphi^+_{L+1}\langle\delta_L,\varphi^{+,L}\rangle&
    \varphi^+_{L+1}\langle\delta_0,\varphi^{+,L}\rangle\\
    \varphi^-_{-L-1}\langle\delta_L,\varphi^{-,L}\rangle&
    \varphi^-_{-L-1}\langle\delta_0,\varphi^{-,L}\rangle
  \end{pmatrix}+O(N^2\rho^{2N}(E))= \alpha\,\rho^N(E)\begin{pmatrix}
    0& 1\\ 1& 0\end{pmatrix}+O(N^2\rho^{2N}(E))
\end{equation*}
Thus, we obtain that the eigenvalues of $H_L$ near $E'$ are given by
$E'\pm\alpha\rho^N(E)+O(N^2\rho^{2N}(E))$ and the eigenvectors by
$\frac1{\sqrt{2}}(\varphi^{+,L}\pm\varphi^{-,L})+O(\rho^{N}(E))$. In
particular, their components at $0$ and $L$ are asymptotic to non
vanishing constants.


\def\cydot{\leavevmode\raise.4ex\hbox{.}}  \def\cprime{$'$}

\end{document}